\documentclass[11pt,twoside,singlespace,vi]{mitthesis}
\usepackage{lgrind}
\usepackage{amsmath,amssymb,amsthm,amsfonts,graphicx}
\begin{document}


\newcommand{\ket}[1]{\left| #1\right\rangle}        
\newcommand{\bra}[1]{\left\langle #1\right|}        
\newcommand{\kets}[1]{| #1 \rangle}        
\newcommand{\bras}[1]{\langle #1 |}        
\newcommand{\braket}[2]{\langle #1 | #2 \rangle} 
\newcommand{\bigket}[1]{\left| #1\right\rangle}        
\newcommand{\bigbra}[1]{\left| #1\right\rangle}        
\newcommand{\scalar}[2]{\left\langle #1 | #2 \right\rangle} 

\newcommand{\ii}{\mathbb{I}}		
\newcommand{\norm}[1]{\left\| #1\right\|}        
\newcommand{\norms}[1]{\| #1\|}        
\newcommand{\normp}[1]{\left\| #1\right\|_p} 
\newcommand{\normtwo}[1]{\left\| #1\right\|_2} 

\newcommand{\combine}[2]{\genfrac{(}{)}{0pt}{}{#1}{#2}}

\newcommand{\ep}{\epsilon}        
\newcommand{\txs}{\:\textrm{\textvisiblespace}\:} 

\newcommand{\tr}{\mathrm{tr}}
\newcommand{\lcm}{\mathrm{lcm}}
\newcommand{\<}{\langle}
\renewcommand{\>}{\rangle}

\newtheorem{definition}{Definition}
\newtheorem{lemma}{Lemma}
\newtheorem{theorem}{Theorem}

\newcommand{\cH}{{\mathcal H}} 
\newcommand{\cC}{{\mathbb{C}}} 
\newcommand{\onemat}{{\bf 1}}  
\newcommand{\goes}{\:\longrightarrow\:}		
\newcommand{\goesU}{\:\stackrel{U_k}{\longrightarrow}\:}

\newcommand{\gA}{\lhd}
\newcommand{\gB}{\blacktriangleright}

\newcommand{\half}{\frac{1}{2}}
\newcommand{\deriv}[2]{\frac{\textrm{d}#1}{\textrm{d}#2}}        
\newcommand{\bigabs}[1]{\Big|\!#1\!\Big|}        
\newcommand{\bignorm}[1]{\Big\|#1\Big\|}        

\newtheorem*{lem1}{Lemma 1}
\newtheorem*{lem1x}{Lemma 1'}
\newtheorem*{lem2}{Lemma 2}
\newtheorem*{lem2x}{Lemma 2'}
\newtheorem*{lem3}{Lemma 3}
\newtheorem*{lem4}{Lemma 4}
\newtheorem*{lem4x}{Lemma 4'}
\newtheorem*{lem5}{Lemma 5}
\newtheorem*{lem5x}{Lemma 5'}

\newcommand{\tav}[1]{       
	\frac{1}{\tau_{10}}\int^{\tau_{10}}_0 
	#1 \,d\tau
}

\newcommand{\tur}{\,\circlearrowleft\,}   
\newcommand{\gat}{\,\blacktriangleright}  
\newcommand{\mov}{\,\vartriangleright}    
\newcommand{\bul}{\:\:\centerdot\:}       
\newcommand{\iga}{\:I\,}                  
\newcommand{\ici}{{\bigcirc \mkern-13mu   
	\mbox{\scriptsize{\textit{I}}}\mkern5mu\,}}
\newcommand{\wga}{W}						
\newcommand{\wci}{{\bigcirc \mkern-17mu 	
	\mbox{\scriptsize{\textit{W}}}\mkern1mu\,}}
\newcommand{\sga}{\:S\,}					
\newcommand{\sci}{{\bigcirc \mkern-14mu 	
	\mbox{\scriptsize{\textit{S}}}\mkern3mu\,}}
\newcommand{\aga}{\,A\,}					
\newcommand{\aci}{{\bigcirc \mkern-14mu 	
	\mbox{\scriptsize{\textit{A}}}\mkern1mu\,}}
\newcommand{\bga}{\,B\,}					
\newcommand{\bci}{{\bigcirc \mkern-15mu 	
	\mbox{\scriptsize{\textit{B}}}\mkern4mu\,}}
\newcommand{\tri}{\,\vartriangleleft}     
\newcommand{\ooo}{\,\Box}     				
\newcommand{\ppp}{\,\blacksquare}			

\newcommand{\goR}{{\bigcirc \mkern-11mu \cdot \mkern7mu}}
\newcommand{\goD}{{\bigcirc \mkern-16mu \times \mkern3mu}}
\newcommand{\goM}{{\bigcirc \mkern-20mu \vartriangleright \mkern2mu}}
\newcommand{\goA}{{\bigcirc \mkern-20mu \blacktriangleright \mkern2mu}}
\newcommand{\goT}{{\mkern1mu \circlearrowleft \mkern4mu}}
\newcommand{\goB}{{\mkern3mu \vartriangleleft \mkern2mu}}
\newcommand{\goF}{{\mkern3mu \vartriangleright \mkern2mu}}
\newcommand{\goO}{{\mkern3mu \boxdot \mkern2mu}}
\newcommand{\goX}{{\mkern3mu \boxtimes \mkern2mu}}
\newcommand{\goo}{{\mkern3mu \boxminus \mkern2mu}}
\newcommand{\gox}{{\mkern3mu \boxplus \mkern2mu}}
\newcommand{\goSX}{\goX\goX\goX\goX\goX\goX}
\newcommand{\goSO}{\goO\goO\goO\goO\goO\goO}
\newcommand{\hA}{\blacktriangleright}
\newcommand{\hO}{\cdot}
\newcommand{\hX}{\times}

\newcommand{\rO}{{\mkern3mu \boxdot \mkern2mu}}
\newcommand{\rX}{{\mkern3mu \boxtimes \mkern2mu}}
\newcommand{\rA}{{\mkern3mu \blacksquare \mkern2mu}}
\newcommand{\rx}{{\mkern3mu \times \mkern2mu}}
\newcommand{\ro}{{\mkern5mu \circ \mkern5mu}}
\newcommand{\ra}{{\mkern3mu \blacktriangleleft \mkern2mu}}
\newcommand{\rb}{{\mkern3mu \blacktriangle \mkern3mu}}
\newcommand{\rc}{{\mkern3mu \blacktriangleright \mkern2mu}}

\newcommand{\mcont}[4]{		
	\left[
	\begin{array}{cccc}
	 1  & 0  & 0  & 0\\
	 0  & 1  & 0  & 0\\
	 0  & 0 & #1 & #2\\
	 0 & 0 & #3 & #4
	\end{array}
	\right]}
\newcommand{\mfour}[4]{		
	\left[
	\begin{array}{cc}
	 #1 & #2 \\
	 #3 & #4 
	\end{array}
	\right]}
\newcommand{\mtwo}[2]{		
	\left[
	\begin{array}{c}
	 #1 \\
	 #2
	\end{array}
	\right]}

\newcommand{\band}[2]{		
	\begin{array}{|r|r|}
	\hline #1 & #2 \\
	\hline
	\end{array}}
\newcommand{\triUR}[3]{ 		
	\begin{array}{r@{}c|c|}
	\cline{2-3} \vline & \:\: #1 & #2 \\ 
	\cline{2-3} & & #3 \\ 
	\cline{3-3} 
	\end{array}\:\:}
\newcommand{\triUL}[3]{ 		
	\:\:\begin{array}{|c|c@{}l}
	\cline{1-2} #1 & #2 \:\: & \vline \\ 
	\cline{1-2} #3 & & \\ 
	\cline{1-1}
	\end{array}}
\newcommand{\four}[4]{ 		
	\begin{array}{|c|c|}
	\hline #1 & #2 \\
	\hline #3 & #4 \\
	\hline
	\end{array}}
\newcommand{\flour}[3]{ 		
	\begin{array}{|c|c|}
	\hline #1 & #2 \\
	\hline \multicolumn{2}{|c|}{#3} \\
	\hline
	\end{array}}
\newcommand{\six}[6]{ 		
	\begin{array}{|c|c|}
	\hline #1 & #2 \\
	\hline #3 & #4 \\
	\hline #5 & #6 \\
	\hline
	\end{array}}
\newcommand{\slix}[5]{ 		
	\begin{array}{|c|c|}
	\hline #1 & #2 \\
	\hline #3 & #4 \\
	\hline \multicolumn{2}{|c|}{#5} \\
	\hline
	\end{array}}

\newlength{\onebox}
\setlength{\onebox}{20pt}
\newcommand\raiseonebox{\raisebox{-.5\onebox} 
  {\rule{0pt}{\onebox}}}

\newcommand{\sve}{\vec{s}\,}
\newcommand{\svp}{\vec{s}\,'}
\newcommand{\Mlim}{M\rightarrow \infty}
\newcommand{\E}{{\cal E}}
\newcommand{\means}[1]{\langle #1 \rangle}        
\newcommand{\corrf}{\means{ \sigma^{(i)}_z \sigma^{(j)}_z } 
- \means{ \sigma^{(i)}_z }\means{ \sigma^{(j)}_z }}
\newcommand{\corrfO}{\means{ O^{(i)} O^{(j)} } - \means{ O^{(i)} }\means{ O^{(j)} }}

\newcommand{\mypaper}[3]{
\vskip6pt
\centerline{\textbf{#1}}
\vskip3pt
\centerline{\small \textsl{#2}}
\vskip6pt
\begin{center}
\begin{minipage}{4.5in}
{\small #3}
\end{minipage}
\end{center}
}
\newcommand{\mypapertwoline}[4]{
\vskip6pt
\centerline{\textbf{#1}}
\centerline{\textbf{#2}}
\vskip3pt
\centerline{\small \textsl{#3}}
\vskip6pt
\begin{center}
\begin{minipage}{4.5in}
{\small #4}
\end{minipage}
\end{center}
}

\title{Local Hamiltonians in Quantum Computation}

\author{Daniel Nagaj}
\department{Department of Physics}
\degree{Doctor of Philosophy in Physics}
\degreemonth{June}
\degreeyear{2008}
\thesisdate{May 16, 2008}

\supervisor{Edward H. Farhi}{Professor of Physics; 
Director, Center for Theoretical Physics}

\chairman{Thomas J. Greytak}{Lester Wolfe Professor of Physics
\\
Associate Department Head for Education}

\maketitle

\cleardoublepage
\pagestyle{empty}
\setcounter{savepage}{\thepage}
\begin{abstractpage}
In this thesis, I investigate aspects of local Hamiltonians in quantum computing. First, I focus on the Adiabatic Quantum Computing model, based on evolution with a time-dependent Hamiltonian. I show that to succeed using AQC, the Hamiltonian involved must have local structure, which leads to a result about eigenvalue gaps from information theory. I also improve results about simulating quantum circuits with AQC. Second, I look at classically simulating time evolution with local Hamiltonians and finding their ground state properties. I give a numerical method for finding the ground state of translationally invariant Hamiltonians on an infinite tree. This method is based on imaginary time evolution within the Matrix Product State ansatz, and uses a new method for bringing the state back to the ansatz after each imaginary time step. I then use it to investigate the phase transition in the transverse field Ising model on the Bethe lattice. Third, I focus on locally constrained quantum problems Local Hamiltonian and Quantum Satisfiability and prove several new results about their complexity. Finally, I define a Hamiltonian Quantum Cellular Automaton, a continuous-time model of computation which doesn't require control during the computation process, only preparation of product initial states. I construct two of these, showing that time evolution with a simple, local, translationally invariant and time-independent Hamiltonian can be used to simulate quantum circuits.

\end{abstractpage}


\cleardoublepage

\section*{Acknowledgments}

I gratefully acknowledge financial support from the MIT Presidential Fellowship, the National Security Agency (NSA) and Advanced Research and Development Activity (ARDA) under Army Research Office (ARO) contract W911NF-04-1-0216, and the W. M. Keck Foundation Center for Extreme Quantum Information Theory during my five years at MIT.

I wasn't alone on my journey to finishing this thesis. First, my thanks go to God for making this world a joy to explore. My parents gave me their love, support and freedom. My undergraduate advisor, Vlado Bu\v{z}ek, introduced me to exciting research. Eddie Farhi, my Ph.D. thesis advisor, taught me to be patient and to ask more questions. He gave me not only good direction, but also plenty of freedom in research. During my time at MIT, I have been honored to coauthor with Jeffrey Goldstone, Sam Gutman, Peter Shor, Iordanis Kerenidis, Shay Mozes, Pawel Wocjan and Igor Sylvester. Besides them, I learned much from discussions about research with Seth Lloyd, Ike Chuang, Guifre Vidal, Peter Love, Oded Regev, Stephen Jordan, Jake Taylor and an anonymous referee. I also benefited from being a teaching assistant for Gunther Roland, Young Lee, John Negele and Bruce Knuteson. Ashdown was a great place to live in because of Noel, Abhi, Jon, Ji-Eun, Marcus, Junsik and the housemasters, Terry and Ann Orlando. Problem sets were do-able with the help of Ambar, Mark, Stephen, Pouyan and Seungeun. Jon taught me to throw a frisbee and juggle, Dave taught me to drive and Ambar taught me to cook Indian food. Helen, Kim and Rachel helped me to grow in photography. Through the MIT Graduate Christian Fellowship, I found friends in Jon, Colin, Jeff, Wouter, Mary, Kevin, John, Ken, Bethany, Cynthia, Eric, Kyle, Stephen, Stephen and others. Looking back at all this today, I am thankful and my joy is great, for I share it with Pavla, the best of wives.


\pagestyle{plain}

\tableofcontents
\newpage
\listoffigures
\newpage
\listoftables

\setcounter{chapter}{-1}
\chapter{Introduction}\label{ch0summary}

Today's computers, working in binary, churn terabytes of classical information for us every day, solving computational problems. Nevertheless, for some of these problems, the required resources for even the best algorithms available today scale unfavorably with the problem size. Such computational difficulty can be turned into a resource of its own, motivating cryptographic schemes. Today's favorite, RSA, is based on the hardness of factoring large numbers. It can withstand attacks by even bigger and faster computers, as making the input size of the problem larger increases the required resources exponentially. Is there then a chance of ever solving hard problems? 
For many optimization problems like 3-SAT, the expected answer is no. 
There are problems thought to be easier than 3-SAT, such as graph isomorphism, the problem to determine whether two graphs are isomorphic or not. Many scientists have devoted many hours to obtain an efficient algorithm for it and did not succeed, but the jury is still out on whether one exists. Another problem thought hard classically is factoring. In 1994, Shor found an algorithm which runs in time polynomial, rather than exponential, in the number of digits of the number to be factored. However, Shor's algorithm requires a {\em quantum computer}, the target of considerable effort by many research groups, with nuclear spins, optical lattices, quantum dots, superconducting circuits or trapped ions as their building blocks. In fact, a small, single-purpose, NMR-based quantum computer was built by Isaac Chuang's group, ran Shor's algorithm, and factored 15. Needing to go much farther than this proof-of-principle quantum computer, today's effort focuses on developing and building scalable and fault-tolerant quantum computers.

The world is quantum-mechanical. It is interesting to investigate what the laws of physics would allow us to compute, if we could control a quantum system and use it as a computer. We would like to know the power and the limits of computers that would use quantum instead of classical information. From understanding the details of chemical reactions to finding ground state properties of quantum spin glass models, there are many problems whose nature is quantum mechanical. It is possible that a quantum computer could help solve these. Not only that, there are many classical problems where algorithmic speedups for quantum computers over their classical counterparts have been found, and there is hope that many are yet to be discovered.

Quantum mechanical systems behave according to the Schr\"odinger equation, governed by the Hamiltonian operator of the system, which describes its interactions. In our experience, the interactions between particles involve only a few particles at a time. A Hamiltonian describing such a system is called {\em local}. More specifically, if the particles are also spatially close, I will call the Hamiltonian {\em spatially local}, and emphasize this distinction throughout the text. In this thesis, I ask---and answer---questions about the power of local Hamiltonians, whether how to simulate them, what could we do with them if we had them, or how hard it is to find their properties in the first place.

\section{Outline and Summary of Results}

After reviewing classical and quantum computing and introducing local Hamiltonians, illustrated by Feynman's Hamiltonian Computer and Kitaev's Local Hamiltonian problem in Chapter \ref{ch1intro}, I present the following new results in the next four chapters:

\vskip6pt\noindent{\bf Chapter \ref{ch2adiabatic}: Adiabatic Quantum Computation} 
investigates the AQC model of quantum computation, which is based on time-evolving a system with a time-dependent Hamiltonian from a
specific initial ground state. I define a general Hamiltonian Computer model, show that it is universal for quantum computing, and show that it also implies the universality of AQC for quantum computation. I improve the estimates on the runtime of AQC simulations of the quantum circuit model using the results I obtain later in Chapter \ref{chQsat}. Then, based on \cite{AQC:fail}, I prove two theorems about what not to do when designing quantum adiabatic algorithms, in essence saying that the Hamiltonians one uses must be local, if one hopes to succeed.

\vskip6pt\noindent{\bf Chapter \ref{ch3mps}: Matrix Product States on Infinite Trees} 
is based on \cite{MPS:infitrees} and contains a new numerical method for investigating the ground state properties of translationally invariant Hamiltonians on an infinite tree (Bethe lattice). This method is based on imaginary time evolution within the Matrix Product State ansatz. I use it to analyze two specific interaction models, first of which is the Ising model in transverse field on the Bethe lattice.

\vskip6pt\noindent{\bf Chapter \ref{chQsat}: Quantum Satisfiability} 
analyzes the complexity of the quantum analogue of classical Satisfiability. I give several novel clock constructions useful for proofs of QMA completeness of Hamiltonian problems and for proofs of universality of Hamiltonian computers. Based on \cite{lh:new3local}, I show that the 3-local Hamiltonian problem stays QMA-complete even without high-norm penalty terms required in previous constructions and that Quantum (3,2,2)-SAT is QMA$_1$ complete. I then present two previously unpublished results. First, Quantum 2-SAT on a line with 11-dimensional particles is QMA$_1$ complete. Second, the Quantum 3-SAT Hamiltonian is universal for quantum computing, when used in my Hamiltonian Computer model.

\vskip6pt\noindent{\bf Chapter \ref{ch5hqca}: Hamiltonian Quantum Cellular Automata} 
is based on \cite{CA:d10}. I define HQCA, a model of quantum computation with a local, time-independent, translationaly invariant and problem-independent Hamiltonian. I present two such automata in 1D geometry, one built from particles with dimension $d=10$, and the other from particles with dimension $d=20$. These results are then related back to the universality of adiabatic quantum computation.

\vskip6pt
\noindent Finally, I summarize the thesis in Chapter \ref{conclusions}, and discuss further research directions. Additional material can be found at the end of the thesis. Appendix \ref{appProp} investigates Kitaev's propagation Hamiltonian,
Appendix \ref{appDyson} works out the details for a Trotter-like formula for classically simulating an adiabatic algorithm, 
Appendix \ref{d20proof} is an analysis of the continuous time quantum walk on a line and cycle,
Appendix \ref{d10proof} concerns the diffusion of free fermions in 1D, and 
Appendix \ref{appTwolevel} looks at transitions in a two-level system.
\chapter{Preliminaries}\label{ch1intro}

This preliminary Chapter is aimed at a reader with little background in classical computing and complexity, and only basic knowledge of quantum mechanics. Kitaev's Classical and Quantum Computation \cite{KitaevBook} and Nielsen \& Chuang's Quantum Information and Computation \cite{NCbook} are great starting points to investigate the topics mentioned herein in much more detail. The advanced reader should start straight with Section \ref{ch1:qlocal} on Local Hamiltonians, where I review the basic building blocks for many results in this thesis: Feynman's Hamiltonian computer and Kitaev's Local Hamiltonian problem.

This Chapter is organized as follows. I start with the classical model of computation in Section \ref{ch1:clas}, review classical complexity in Section \ref{ch1:clascomplex}, and then look at classical locally constrained problems, $k$-Satisfiability and Max-$k$-Satisfiability, in Section \ref{ch1:local}. 

In Section \ref{ch1:qc}, I introduce quantum computation, discussing the quantum circuit model and universality in Section \ref{ch1:qcircuit}, 
quantum complexity in Section \ref{ch1:qcomplex}, and 
the strengths and limits of quantum computers in \ref{ch1:grover}. 

Finally, in Section \ref{ch1:qlocal} I turn my attention to the main topic of this thesis, local Hamiltonians in quantum computing, introducing Feynman's Hamiltonian computer model in Section \ref{ch1:FeynmanHC} and Kitaev's QMA-complete Local Hamiltonian problem in Section \ref{ch1:Kitaev5}.


\section{Classical Computation}\label{ch1:prelim}

\subsection{Introduction to Computation}\label{ch1:clas}

I have a computer that I am typing this on. What can I do (compute) with it? I can surely add and multiply numbers really fast, but could I also use it to factor large numbers? It can effortlessly beat me in chess, but I am up to the challenge in go\footnote{though not against the very recent Monte Carlo programs like MoGo}. Could I have programmed it to write this thesis? How big would the program have to be? How long would it take to execute? The question of computability and effectiveness of computation needs to be looked at in a particular model of computing. In 1937, Alan Turing introduced his remarkable Turing machine (TM) model \cite{Turing1,Turing2}.
\begin{figure}
	\begin{center}
	\includegraphics[width=2.5in]{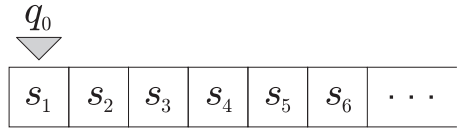}
	\end{center}
	\caption{The Turing Machine.
	}
	\label{ch1:figturing}
\end{figure}
Consider a semi-infinite tape with symbols $s_i \in \{0,1,\txs\}$ on it (see Figure \ref{ch1:figturing}), a read/write head with a finite memory (internal state space) positioned on the left end of the tape, and a finite set of instructions. The head is initialized in the state $q_0$ and positioned on the left of the tape, reading the symbol $s_1$. The computation is then performed 
in sequential steps according to the instruction set. These instructions have the form 
\begin{eqnarray}
	[ current\,\,state ], [symbol\,\,being\,\,read]
	\goes 
	[next\,\,state], [symbol\,to\,write], [where\,to\,move].
	\nonumber
\end{eqnarray}
Depending on the current state of the machine and the symbol that is under the read/write head, the machine can change its state, write a symbol to the tape and move one step to the left (denoted by $\vartriangleleft$), one step to the right ($\vartriangleright$) or stay at its current position ($\:\cdot\:$). It is straightforward to run the machine according to the instruction set. The execution stops only when the internal state of the head changes to a special state \textsc{halt}. The output of the TM is then what is written on the tape (or just the very last symbol written, right under the head).

Here is a simple example of a Turing Machine, with the head starting in the state $q_0$ on the left of the tape: 
\begin{eqnarray}
	q_0,\,\:0 &\goes& \:\:\: q_0\:\:\:,\:0,\:\:\vartriangleright, \label{mov01}\\
	q_1,\,\:0 &\goes& \:\:\: q_1\:\:\:,\:0,\:\:\vartriangleright, \label{mov02}\\
	q_0,\,\:1 &\goes& \:\:\: q_1\:\:\:,\:1,\:\:\vartriangleright, \label{mov11}\\
	q_1,\,\:1 &\goes& \:\:\: q_0\:\:\:,\:1,\:\:\vartriangleright, \label{mov12}\\
	q_0,\txs &\goes& \textsc{halt},\:0,\:\:\:\cdot\:, \label{halt1}\\
	q_1,\txs &\goes& \textsc{halt},\:1,\:\:\:\cdot\:. \label{halt2}
\end{eqnarray}
The first two rules \eqref{mov01}-\eqref{mov02} say that if the head reads a zero, it should keep its state and move to the right. Rules \eqref{mov11}-\eqref{mov12} apply when the head reads a 1. The internal state of the machine then flips ($q_0 \leftrightarrow q_1$), and the head again moves to the right. Finally, when the head reads $\txs$, it assumes it has reached the end of the string, writes down 0 or 1 according to its internal state, and halts. This TM computes the parity of the string such as
\begin{eqnarray}
	\underbrace{\mathtt{10010100110}}_{n} \txs \txs \txs \txs \txs \txs \dots \label{tminput}
\end{eqnarray}
initially written on the tape. The {\em input string size} $n$ is the number of symbols on the tape after which the rest of the tape is blank. Observe that this TM takes at most $n$ steps before it halts on any input string of size $n$.

When a TM halts, it produces an output. In the above example, the output was the parity of the input string. However, a general TM does not have to halt on each possible input, leaving its output undefined. A function on a set $S$ whose domain is not the whole $S$ is called a {\em partial} function. A TM is thus computing some partial function $\phi_M:S\rightarrow S$ on $S$, the set of strings over the alphabet $\{0,1,\txs\}$. 
Taking it from the other end, does a TM computing a partial function $F$ exist for any $F:S \rightarrow S$? There is an uncountable number of such functions, while the set of possible Turing Machines is countable. This means not all functions can be computable, motivating the following definition:
A partial function $F: S \rightarrow S$ is {\em computable}, if there exists a Turing machine $M$ such that $\phi_M = F$.
An example of an uncomputable function is the {\em halting problem}: 
Does a given TM halt on all inputs? \cite{SipserBook} A hypothetical halting problem-solving TM would then have to take the description of the TM whose halting properties it wants to determine as input.

On the other hand, for functions that are computable, it is interesting to ask how hard it is to compute them. For this, let me first ask how long does a specific Turing Machine run before it stops. This question must be related to the input string size $n$. The {\em running time} $T(n)$ of a TM is then defined as the largest number of steps it takes for the TM to stop on any input string of size $n$. There can be several different Turing Machines computing the same function, corresponding to different approaches to computing $F$. Some of them can be more efficient than others, with a better running time $T(n)$. An efficient TM is one whose running time $T(n)$ is small. The most common notion of an {\em efficient} TM is one whose running time scales only like a polynomial in the input size $n$ (i.e. like $T(n)\propto n^3$ and not like $T(n) \propto 2^n$), and I adapt it in this thesis. On the other hand, functions $F$ can be classified into {\em  complexity classes} depending on what would the running time be for the best possible Turing Machine. I return to this topic in Section \ref{ch1:clascomplex}.

\begin{figure}
	\begin{center}
	\includegraphics[width=4.5in]{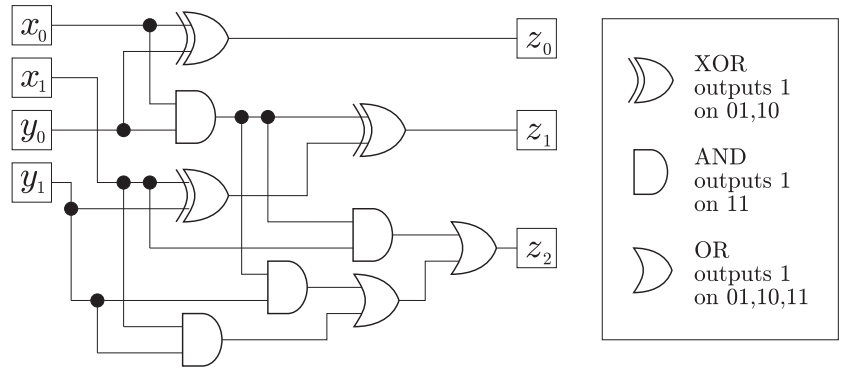}
	\end{center}
	\caption{A binary addition circuit for two two-bit numbers, $z_2 z_1 z_0 = x_1 x_0 \oplus y_1 y_0$.
	}
	\label{ch1:figcircuitadd}
\end{figure}
The notion of an {\em algorithm}, a prescription of what to do to compute something, is dependent on the computational model one uses. Turing himself defined what an algorithm is by his {\em Turing thesis}: 
\begin{quote}
	``Any algorithm can be realized by a Turing machine.''
\end{quote}
Suprisingly, the alternative models of computation proposed in the early days of computer science turned out to be equivalent to the TM model \cite{SipserBook,KitaevBook}. 
Today's computers do not have long tapes with read/write heads inside them. The conventional {\em circuit model} of digital computers (see Figure \ref{ch1:figcircuitadd}), equivalent to the TM model, consists of wires carrying bits of information and logical gates like
\begin{eqnarray}
\begin{tabular}{|c|c|}
\hline
$x$ & NOT $x$ \\
\hline
0 & 1 \\
\hline
1 & 0 \\
\hline
\end{tabular}
\end{eqnarray}
and two-bit gates
\begin{eqnarray}
\begin{tabular}{|c|c|c|c|c|c|}
\hline
$x$ & $y$ & $x$ AND $y$ &  $x$ OR $y$ &  $x$ NAND $y$ &  $x$ XOR $y$ \\
\hline
0 & 0 & 0 & 0 & 1 & 0 \\
\hline
0 & 1 & 0 & 1 & 1 & 1 \\
\hline
1 & 0 & 0 & 1 & 1 & 1 \\
\hline
1 & 1 & 1 & 1 & 0 & 0 \\
\hline
\end{tabular}
\end{eqnarray}
that are applied to them. What is the smallest set of gates one needs to simulate a Turing Machine (compute a computable function $F$)? It turns out that the NAND gate is {\em universal} in the sense that any circuit can be transformed into another built only from wires and NAND gates. 

The single-bit gate NOT is reversible, as its input can be recovered from its output, while the gates AND, OR, XOR and NAND are all non-reversible. It turns out that reversible circuits are as powerful as their irreversible counterparts. The reversible 3-bit Toffoli gate, which flips the third bit if the first two bits are both $1$, is universal \cite{classicalToffoli}. This is important, because quantum computation is a generalization of classical reversible computation (see Section \ref{ch1:qcircuit}).

In the 1930's, analog computers like the {\em differential analyzer} 
for solving differential equations had their programs built in, as their components had to be physically rearranged for solving different problems. 
Similarly, for the TM mentioned so far, a different TM was required for computing each partial function $F$. However, Turing made a big breakthrough, paving the way for digital computation, by proving that it is possible to have a single {\em universal} TM, which can produce the output of any other TM. This is done by having its `program' (the description of which particular computation to perform) given as part of the input. The universal TM then processes the `data' part of the tape according to the `program'. It is interesting to look for a simplest possible universal TM, and people have dedicated quite some effort to it \cite{smallTM,WolframBook}. In Chapter \ref{ch5hqca}, I will return to the idea of a classical universal reversible TM, and bring it to the world of quantum computing. I then investigate simple universal translationally and time-independent Hamiltonians which could be used to simulate any quantum computation in a 1D system.


\subsection{Classical Complexity}
\label{ch1:clascomplex}

Classical complexity examines how hard it is to compute a function $F$ on strings $S$, or how hard it is to answer yes/no questions (decision problems) about strings $S$. A simple yes/no question about a string is: ``Is the input string $S$ a palindrome\footnote{A palindrome is a string that reads the same backwards as it does normally, such as ``livedondecaffacednodevil".}?'' 

For comparing the hardness of problems, it is important to have a notion of reducibility. A {\em reduction} is a way of converting one problem to another in such a way that the solution to the second problem can be used to solve the first problem \cite{SipserBook}. That multiplication can be reduced to addition, using addition as a subroutine polynomially many times, is a simple example of a polynomial reduction. I will frequently use reductions in this thesis, building on previous results about hardness of certain problems. 

The complexity of computing a function or a decision problem can be classified by
\begin{itemize}
\item reducing the problem into another that is already classified,
\item showing that an already classified problem reduces to the one at hand,
\item showing a lower bound $l(n) \leq T(n)$,
\item finding an algorithm for it that runs in some $T(n)$.
\end{itemize}

To prove a lower bound $l(n)$ on the complexity of computing a function $F$, one must show that no algorithm computing $F$ can possibly have running time less than $l(n)$. On the other hand, finding efficient algorithms is hard. We can solve many problems efficiently, but there are plenty others for which today's best known algorithm scales exponentially with the problem size (see Section \ref{ch1:local}). However, we do not know what tomorrow's algorithm designers will come up with. The best situation comes when a problem has a lower bound on its complexity, and an optimal algorithm that matches it. A trivial example is the problem of computing the {\em parity} of a string. The Turing Machine described by \eqref{mov01} is in fact optimal, matching the lower bound of $n$ steps for a string of length $n$. 
A more interesting example is the {\em Ordered search} problem, where one needs to determine whether a specific number appears in an ordered list of $n$ numbers. The list of numbers can be written on the one-dimensional tape of a Turing Machine which in the worst case would have to go through it all. On the other hand, I can view this problem differently. Consider the list of numbers to be hidden in a black box which the TM can query by writing a number $m$ on a separate `query' tape. Upon a query, the black box subroutine writes the $m$-th number from the list on the tape. The TM can then compare this number to the one it searches for, do some computation, and run another query. The complexity of the algorithm is then determined by the number of queries to the black box the TM needs to ask before it determines whether a number is in the black box list or not. The complexity of such black-box problems is called {\em query complexity}. The best classical algorithm for ordered search matches the corresponding lower bound, requiring to query $\log_2 n$ numbers in the list. In Section \ref{ch1:grover}, I look at {\em Unstructured search}, another problem for which we have both a lower bound on the query complexity and an optimal algorithm, in the classical and in the quantum case as well.

I now introduce a few important classical complexity classes, all related to {\em decision problems}, which are simply yes/no questions.
\begin{definition}[Complexity class P]
The class of decision problems solvable in polynomial time by a Turing machine.
\end{definition}
This is the basic class of {\em efficiently solvable} problems. It contains easy problems like determining whether two numbers are coprime, but also the problem of determining whether a number is prime \cite{primesinp}. On the other hand, it is not known whether it contains {\em factoring}, the problem of determining whether a number $p$ has a divisor $d$ in some range $a\leq d \leq b$. Today's best known classical algorithm for factoring has runtime $2^\frac{n}{3}$, where $n$ is the number of bits of $p$. On the other hand, in 1994, Shor \cite{ShorFact} discovered a breakthrough {\em quantum} algorithm
with running time $T(n) \propto n^3$, whose use of the Fourier transform resulted in many of today's quantum algorithms. I refer the reader to Section \ref{ch1:grover} for more details about efficient quantum algorithms. 

All the problems in P are contained in another class, NP, described by Scott Aaronson in the Complexity ZOO \cite{complexityZoo} as `the class of of dashed hopes and idle dreams'. Determining whether NP is in fact bigger than P or equal to it is a long standing open question in Computer Science with a million dollar prize from the Clay foundation awaiting the solver. The definition of NP is
\begin{definition}[Complexity class NP (Non-deterministic, Polynomial time)]
NP is the class of decision problems such that, if the answer is `yes', then there is a proof of this fact, of length polynomial in the size of the input, that can be verified by a deterministic polynomial-time algorithm. On the other hand, if the answer is `no', then the algorithm must declare invalid any purported proof that the answer is `yes'.
\end{definition}

It is possible to expand the original Turing Machine model by allowing the TM to use an additional random generator, making it a {\em probabilistic Turing Machine}. This can be done by giving the TM access to a second tape with a random string written on it. The question whether this model is equivalent to the original, deterministic TM model is open, but believed to be true. 
The probabilistic Turing Machine model gives rise to the class BPP:
\begin{definition}[Complexity class BPP (Bounded-error, Probabilistic, Polynomial time)]
The class of decision problems solvable in polynomial time by a randomized Turing machine.
\end{definition}
Great progress has been made in randomized algorithms recently, and it is a challenge for researchers to derandomize them. I give an example of a randomized algorithm in Section \ref{ch1:local} (the WALKSAT algorithm for solving 3-SAT). Moreover, quantum computing, the topic of this thesis, is inherently randomized.

Analogously to generalizing P to BPP, the class NP motivates the class MA. Consider an interactive protocol between a powerful wizard Merlin and a logic, rational verifier Arthur, whose role is to verify purported proofs prepared by Merlin. The difference from NP is that Arthur can use a randomized algorithm, rather than verify the proof deterministically\footnote{If all random algorithms can be made deterministic (BPP=P), MA equals NP \cite{complexityZoo}.}:
\begin{definition}[Complexity class MA (Merlin-Arthur)]
The class of decision problems solvable by a Merlin-Arthur protocol, which goes as follows. Merlin, who has unbounded computational resources, sends Arthur a polynomial-size purported proof that the answer to the problem is yes. Arthur must verify the proof in BPP (i.e. probabilistic polynomial-time), so that if the answer to the decision problem is
\begin{enumerate}
\item `yes', then there exists a proof which Arthur accepts with probability at least 2/3.
\item `no', then Arthur accepts any ``proof'' with probability at most 1/3.
\end{enumerate}
\end{definition}
An alternative definition requires that if the answer is `yes', then there exists a proof such that Arthur accepts with certainty. However, the definitions with one-sided and two-sided error are equivalent \cite{SipserBook}. Moreover, the probabilities $2/3$ and $1/3$ in the definition of MA can be equivalently substituted by $1-\ep$ and $\ep$ with arbitrarily small $\ep$.

I introduce the quantum analogues of these classes in Section \ref{ch1:qcomplex}.


\subsection{Classical locally constrained problems}\label{ch1:local}

Whether deciding on the best class schedule, finding a good route to spend a day touring a city or pairing people up at a dance class, our lives are full of local optimization problems. Some of them are simple, such as picking a single box of cereal with the lowest price/taste ratio from the shelf in the supermarket, some of them harder, like deciding which lenses to pack into my camera bag for the next trip. The latter is a question concerning a cost function involving only a few elements at a time (how useful will a particular lens in my bag be, considering that I already bagged some others) and a global constraint (the size of the bag and my willingness to carry it all). In fact, this is an example of the {\em Knapsack} problem. Another locally constrained problem, {\em Satisfiability}, and its quantum analogue, are the focus of a large part of this thesis. 

A boolean {\em clause} is an expression of the form $s_1 \vee s_2 \vee \dots \vee s_k$, where $s_i$ are boolean variables or their negations. This clause is true for all bit strings $s_1 s_2 \dots s_k$ except for $\{00\dots0\}$. Consequently, the clause $s_1 \vee s_2'$, where $s_2'$ is the negation of bit $s_2$, is true for all assignments $s_1 s_2$ except for $\{01\}$. Satisfiability is then defined as
\begin{definition}[$k$-SATISFIABILITY, $k$-SAT]
Consider a collection $C=\{c_1,c_2,\dots,c_m\}$ of boolean clauses for a finite binary string $S = s_1 s_2 \cdots s_n$, such that each $c_i$ involves $k$ of the string variables. Does a string $S$ satisfying all the clauses in $C$ exist?
\label{defSAT}
\end{definition}
\begin{figure}
	\begin{center}
	\includegraphics[width=2in]{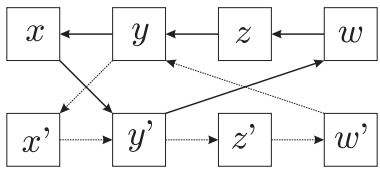}
	\end{center}
	\caption{The implication graph for a 2-SAT instance 
	$(x \vee y') \wedge 
	(x' \vee y') \wedge 
	(y \vee w) \wedge 
	(y \vee z') \wedge 
	(z \vee w')$.
	There is a directed loop containing both $y$ and $y'$, so
	the instance is not satisfiable.
	}
	\label{ch1:fig2sat}
\end{figure}
This is an instance of 2-SAT on four bits:
\begin{eqnarray}
	c_1 &=& x\: \vee y', \nonumber\\
	c_2 &=& x' \vee y', \nonumber\\
	c_3 &=& y\: \vee w, \\
	c_4 &=& y\: \vee z', \nonumber\\
	c_5 &=& z\: \vee w'.\nonumber
\end{eqnarray}
Here $y'$ labels the negation of variable $y$, and the first clause, $c_1$, is false if $x=0$ and $y=1$. 
The complexity of $k-$SAT varies with $k$. For 2-SAT, we have an algorithm with complexity scaling as $O(n^2)$. Draw a directed graph of implications between the bits and their negations as in Figure \ref{ch1:fig2sat}, where for example clause $c_4$ implies the connections $z\implies y$ and $y' \implies z'$. If a loop containing both a variable and its negation exists (it is easy to check this), the instance of 2-SAT is unsatisfiable. 2-SAT is thus in P. However, no known polynomial algorithm for 3-SAT exists. The Cook-Levin theorem
says that 3-SAT is NP-complete. A problem $\Pi$ is {\em complete} for a complexity class, when 
\begin{itemize}
\item
the problem $\Pi$ is contained in the class, and 
\item
every other problem in the class can be reduced to $\Pi$ (solved by using the solution to $\Pi$ as a subroutine polynomially many times). 
\end{itemize}
It is straightforward to see that 3-SAT is in NP. The basic idea for the second direction of the NP-completeness proof is to encode the verification circuit in the definition of NP using a 3-SAT instance. It is instructive to see how this type of proof works in the classical case, before I introduce its quantum analogue in Section \ref{ch1:Kitaev5}. 
\begin{figure}
	\begin{center}
	\includegraphics[width=2.5in]{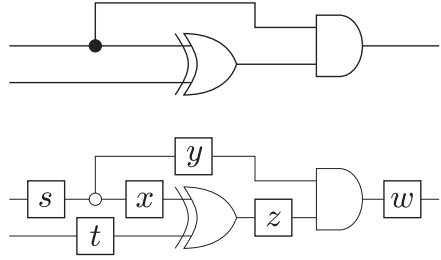}
	\end{center}
	\caption{Illustration for the Cook-Levin theorem. Transforming a verifier circuit into a 3-SAT instance.
	}
	\label{ch1:figcooklevin}
\end{figure}
An instance $x$ of a problem $\Pi$ in NP is a yes/no question, denoted as $x\in (L_{yes} \cup L_{no})$, where $L_{yes}$ is the set of instances with the answer `yes'. Each problem instance has a verifier circuit $V$ associated with it. When $V$ checks a purported proof that $x \in L_{yes}$, it accepts only a valid proof, never outputting 1 if the ``proof'' is wrong. This verifying procedure can be reduced to finding a satisfying assignment for a certain 3-SAT instance. Without loss of generality, circuit $V$ consists of wires and gates involving at most 3 bits (two inputs, one output). Assign a boolean variable to each wire of the circuit as in Figure \ref{ch1:figcooklevin}. The 3-SAT clauses then check the proper evaluation of the circuit. For example, the boolean clauses for $y,z,w$ in Figure \ref{ch1:figcooklevin} 
(where $w$ is the output of the AND gate on $y$ and $z$) read 
\begin{eqnarray}
c_1 &=& y\phantom{'} \vee z\phantom{'} \vee w', \\
c_2 &=& y\phantom{'} \vee z' \vee w', \\
c_3 &=& y' \vee z\phantom{'} \vee w', \\
c_4 &=& y' \vee z' \vee w,
\end{eqnarray}
ruling out the assignments $yzw \in \{001,011,101,110\}$.
The last ingredient in the proof is a clause verifying that the circuit outputs 1:
\begin{eqnarray}
c_{out} = w.
\end{eqnarray}
If a proof of $x \in L_{yes}$ which the circuit $V$ accepts exists, an assignments of the bits corresponding to the computation satisfying the 3-SAT instance exists as well. On the other hand, if there is no ``proof'' that $V$ accepts, there is no satisfying assignment to all the clauses in the 3-SAT instance. Therefore, if one could solve 3-SAT, one could solve any problem in NP. This is also true for $k$-SAT with $k\geq 3$.

\begin{figure}
	\begin{center}
	\includegraphics[width=6in]{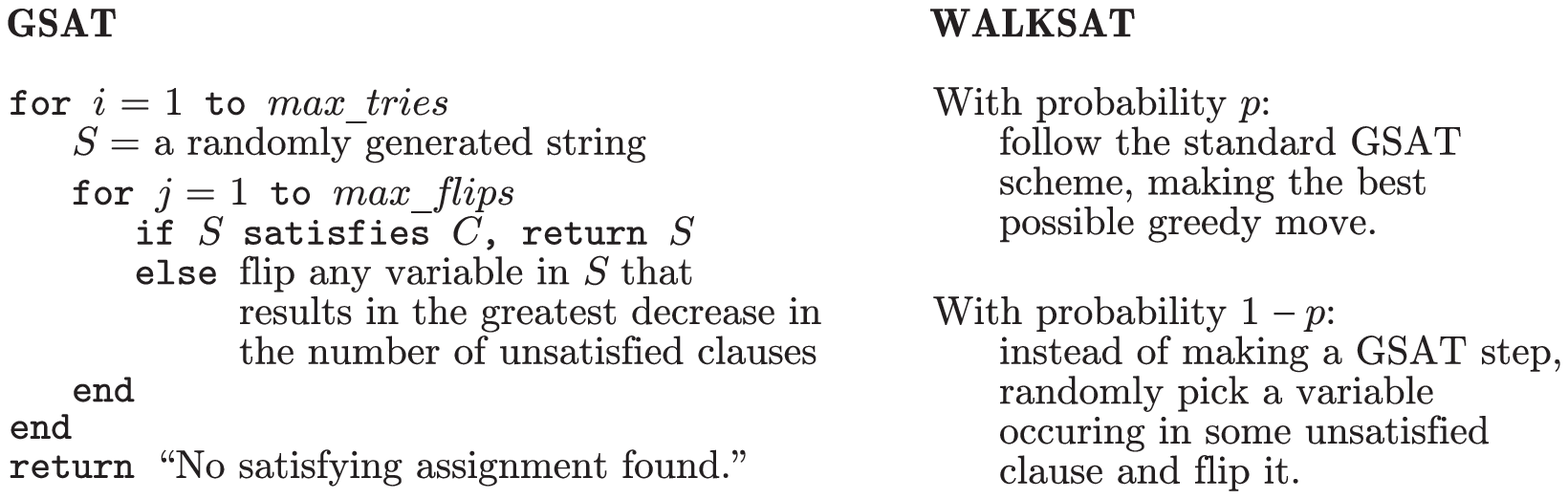} 
	\end{center}
	\caption{The GSAT and WALKSAT algorithms.
	\label{ch1:figurewalksat}}
\end{figure}
To illustrate the idea of randomized algorithms mentioned in Section \ref{ch1:clascomplex}, let me present two randomized algorithms for the SAT problem. The procedures are given in Figure \ref{ch1:figurewalksat}. The greedy algorithm GSAT \cite{gsat} starts with a random bit string and proceeds to flip bits which give the greatest decrease in the number of broken clauses. It works great for 2-SAT, but not for 3-SAT, as it gets stuck in local minima. The second algorithm, WALKSAT \cite{walksat}, goes around this problem by randomly choosing between performing a step of GSAT or flipping a random bit in some unsatisfied clause. This randomized algorithm performs surprisingly well on many real-life instances of 3-SAT.

\subsubsection{SAT for non-bit strings}
\label{ch1:sat23}
\begin{figure}
	\begin{center}
	\includegraphics[width=5in]{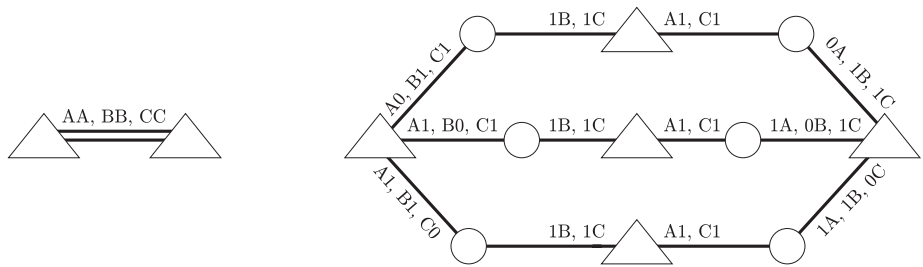} 
	\end{center}
	\caption{A reduction of Graph 3-colorability to 2-Satisfiability for trit-bit pairs, also called (3,2)-SAT. The trits with possible values \{A,B,C\} are denoted by triangles, while the bits are denoted by circles. Each clause is represented by the pair of symbols which it rules out. For Graph 3-colorability, the excluded assignments are AA, BB and CC, ruling out the same color for two neighboring vertices.
	\label{ch1:figure32sat}}
\end{figure}
SAT can be generalized to non-bit strings. The simplest variant of this is (3,2)-SAT, with each clause involving a trit (a three-valued letter) and a bit. Although it doesn't seem much more complicated than regular 2-SAT, already this variant is NP-complete. It belongs to NP, because it can be reduced to 3-SAT by encoding the trit into two bits. On the other hand, it is also NP-hard, as the NP-complete {\em Graph 3-colorability} problem can be reduced to (3,2)-SAT as depicted in Figure \ref{ch1:figure32sat}.

\subsubsection{MAX-$k$-SAT}
The last variant of Satisfiability I introduce is MAX-$k$-SAT. In regular $k$-SAT, the question is whether an assignment satisfying all clauses exists or not. Sometimes, showing that no such assignment exists can be easy. However, the MAX-$k$-SAT problem poses a harder question: does an assignment breaking fewer than $a$ clauses exist? The clauses in MAX-$k$-SAT define a cost function, and solving MAX-$k$-SAT for different $a$ would determine its minimum value. Although regular 2-SAT is in P, MAX-2-SAT is NP-complete, so the existence of a polynomial algorithm for it is highly unlikely.

In Table \ref{SATtable}, I summarize the known results about classical satisfiability problems. The goal of  Chapter \ref{chQsat} is to classify the complexity of the quantum analogues of these. I refer the reader interested in NP-complete problems and other variants of SAT to the book of Garey and Johnson \cite{Intractability}.
\begin{table}
\begin{center}
\begin{tabular}{|r|l|l|}
\hline 
Classical & \ bit strings & \ general strings \\
\hline 
$k$-SAT & \ $k=2$ : in P	& \ $(3,2)$-SAT : NP-complete  \\
& \ $k\geq 3$ : NP-complete & \  \ \\
\hline	
MAX-$k$-SAT & \ $k\geq 2$ : NP-complete &  \\
\hline
\end{tabular}
\caption{Known complexity for classical satisfiability problems. $(3,2)$-SAT in the general strings column is the problem where each clause involves one trit (letter with three values) and one bit.}
\label{SATtable}
\end{center}
\end{table}


\section{Quantum Computing}\label{ch1:qc}

Up until now, I have considered the model of computing governed by the laws of classical physics. One needs to ask what is the computing power of systems governed by quantum mechanics. Alternatively, how hard is it to simulate quantum-mechanical systems with classical computers? Already Feynman stated these questions in \cite{lh:Feynman}. 
I assume the reader is familiar with quantum mechanics, so the following exposition will be brief, introducing only those topics which I later use.

Classical data can be easily measured, copied and erased. If someone gave me a binary string like  \textsc{1001001}, I could easily read it (measure it) and make many copies for myself. However, the world is quantum-mechanical. The state space of the simplest quantum system -- a spin-$\half$ particle, a {\em qubit}, is much bigger than the two possible values $\{0,1\}$ of a classical bit. The first postulate of quantum mechanics is
\begin{itemize}
\item Associated to any isolated physical system is a complex vector space with inner product (i.e. a Hilbert space) known as the {\em state space} of the system. The system is completely described by its {\em state vector}, which is a unit vector in the system's state space.
\end{itemize}
The state space of a qubit is a $2$-dimensional complex vector space $\cH=\cC^2$. The state of a qubit is given by
\begin{eqnarray}
	\ket{\psi} &=& c_0 \mtwo{1}{0} + c_1 \mtwo{0}{1} = 
	c_0 \ket{0} + c_1 \ket{1},
	\label{ch1:qubit}
\end{eqnarray}
where $\ket{0}$ and $\ket{1}$ are basis vectors in the usual Dirac notation and $c_0,c_1$ are complex numbers for which $|c_0|^2 + |c_1|^2 = 1$. 

How can a quantum computer be built from a system of qubits? In a classical computer, one can apply irreversible gates like \textsc{NAND} or reversible operations like the Toffoli gate. Besides measurements which I discuss below, the available operations for a quantum system are unitary transformations. Quantum mechanical systems obey the time-dependent Schr\"odinger equation. The next postulate of quantum mechanics is
\begin{itemize}
\item The time evolution of the state of a quantum system is described by the time-dependent {\em Schr\"odinger equation}
\begin{eqnarray}
	i \hbar \deriv{}{t} \ket{\psi(t)} &=& H(t) \ket{\psi(t)},
	\label{ch1:schrodinger}
\end{eqnarray}
where $H(t)$ is a Hermitian operator called the Hamiltonian.
Note that the value of the constant $\hbar$ can be absorbed into the definition of $H$.
\end{itemize}
Because the Hamiltonian is Hermitian, the time evolved state $\ket{\psi(t)}$ and the initial state $\ket{\psi(0)}$ are related by a unitary transformation $U_{t,0}$:
\begin{eqnarray}
	\ket{\psi(t)} = U_{t,0} \ket{\psi(0)}.
\end{eqnarray}
When the Hamiltonian $H$ is time-independent, this unitary transformation is
\begin{eqnarray}
	U_{t,0} = e^{-iHt}.
\end{eqnarray}
Unitarity of $U$ means that $U_{t,0}^\dagger = U_{t,0}^{-1}$, and implies that time evolution according to the Schr\"odinger equation is reversible. The transformations one can perform on a quantum computer are thus unitary. 

The final ingredient for quantum computation are measurements, used to read out the result. They could be used throughout the computation, but in the next Section I argue why it is sufficient to use measurements only at its end. Quantum measurements work in the following fashion: 
\begin{itemize}
\item Consider a quantum observable $\hat{M}$, a Hermitian operator acting on the state space of the system being measured, expressed in terms of its eigenvectors as
\begin{eqnarray}
	\hat{M} = \sum_i \lambda_i \kets{\phi_i^M}\bras{\phi_i^M}.
\end{eqnarray}
If the state of the quantum system is $\ket{\psi}$ immediately before the measurement then the probability to obtain the result $\lambda_i$ is given by
\begin{eqnarray}
	p(i) = \left|\braket{\phi_i^M}{\psi}\right|^2,
\end{eqnarray}
and the state of the system after the measurement is the corresponding eigenstate $\kets{\phi_i^M}$.
\end{itemize}
Imagine now that I want to determine what state the qubit \eqref{ch1:qubit} is in. As an example, let me choose to measure $\hat{\sigma}_z = \ket{0}\bra{0} - \ket{1}\bra{1}$, the $z$-component of the spin. I will get the result $+1$ with probability $|c_0|^2$, and the result $-1$ with probability $|c_1|^2$. It takes many copies of $\ket{\psi}$ and many measurements to determine $c_0$ and $c_1$ in general. Even if I made different measurements, distinguishing nonorthogonal states perfectly is impossible \cite{NCbook}.


\subsection{The Quantum Circuit model}\label{ch1:qcircuit}

Quantum circuits are a direct generalization of classical reversible circuits, using unitary transformations instead of classical reversible gates.  
A few basic single-qubit unitary gates, expressed as $2\times2$ matrices acting on column vectors in $\cC^2$ are the Pauli matrices
\begin{eqnarray}
	\sigma_x = \mfour{0}{1}{1}{0}, \qquad 
	\sigma_y = \mfour{0}{-i}{i}{0}, \qquad
	\sigma_z = \mfour{1}{0}{0}{-1},
\end{eqnarray}
and the Hadamard gate $H$, the phase gate $S$ and the $\frac{\pi}{8}$ gate $T$:
\begin{eqnarray}
	H = \frac{1}{\sqrt{2}} \mfour{1}{1}{1}{-1}, \qquad
	S = \mfour{1}{0}{0}{i}, \qquad
	T = \mfour{e^{-i\frac{\pi}{8}}}{0}{0}{e^{i\frac{\pi}{8}}}.
\end{eqnarray}
A general single qubit unitary gate can be expressed using the Pauli matrices as
\begin{eqnarray}
	U = e^{-i\frac{\phi}{2} \hat{n}\cdot \vec{\sigma}},
\end{eqnarray}
where $\phi$ is a real number $\hat{n} = (n_x,n_y,n_z)$ is a unit vector and $\hat{n}\cdot \vec{\sigma} = n_x \sigma_x + n_y \sigma_y + n_x \sigma_z$

At least one two-qubit gate is required for quantum computation. The following two gates apply $\sigma_x$ or $\sigma_z$ to the target qubit, if the control qubit is in the state $\ket{1}$:
\begin{eqnarray}
	\mathrm{CNOT} = \mcont{0}{1}{1}{0}, \qquad
	\mathrm{C-Z} = \mcont{1}{0}{0}{-1}.
\end{eqnarray}
The elements of a quantum circuit are usually depicted as in 
Figure \ref{ch1:figurecontrol}.
\begin{figure}
	\begin{center}
	\includegraphics[width=4.5in]{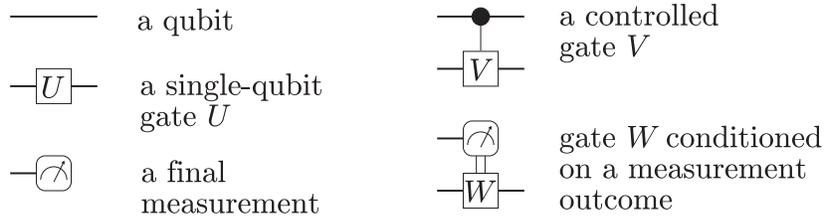}
	\end{center}
	\caption{The components of a quantum circuit.
	}
	\label{ch1:figurecontrol}
\end{figure}
At the end of the computation, the output of a quantum circuit is read out by performing a few final measurements. Although I could perform measurements throughout the computation and apply gates conditioned on the outcomes, it is possible to defer measurements until the very end of the computation by using controlled gates, as shown in Figure \ref{ch1:figureqcirc}. 
\begin{figure}
	\begin{center}
	\includegraphics[width=4in]{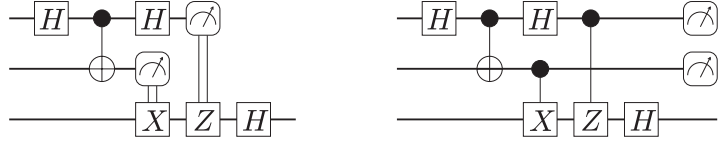}
	\end{center}
	\caption{A quantum circuit with measurements performed during the computation and its version with measurements deferred until the end of the computation.
	}
	\label{ch1:figureqcirc}
\end{figure}

\subsubsection{Universality}
\label{ch1:universal}
The NAND gate is universal for classical circuits and the Toffoli gate is universal for classical reversible circuits as mentioned in Section \ref{ch1:clas}. On the other hand, can every unitary transformation on $n$ qubits be decomposed into many applications of a small finite set of gates? Indeed, universal gate sets for quantum computing exist in the sense that they can approximate any unitary operation arbitrarily well. A universal gate set must be dense in the group $SU(n)$, and the approximations must converge fast. I point the interested reader to Appendix 3 of \cite{NCbook} for a thorough discussion of the Solovay-Kitaev theorem concerning the convergence rate.
Some well known gate sets universal for quantum computing are:
\begin{itemize}
\item{single qubit gates, CNOT} \cite{NCbook},
\item{single qubit gates, C\,--\,Z} \cite{NCbook},
\item{$H$, $S$, $T$, CNOT} \cite{NCbook}.
\end{itemize}
Another universal gate set uses the quantum Toffoli gate:
\begin{itemize}
\item{$H$, Toffoli \cite{CA:Shi:02}}
\end{itemize}
It is intriguing, because it uses only real matrices. It is possible to do quantum mechanics with real numbers, as superpositions and minus signs, not imaginary numbers are essential. 
Consider the state \eqref{ch1:qubit} of a qubit,
and add to the system another qubit labeled $r$. This qubit can be used to work around the requirement of imaginary amplitudes. 
\begin{eqnarray}
	\ket{\Psi} &=& \textrm{Re}(c_0) \ket{0}\ket{0}_r + \textrm{Im}(c_0) \ket{0}\ket{1}_r \\
			&+& \textrm{Re}(c_1) \ket{1}\ket{0}_r +
			\textrm{Im}(c_1) \ket{1}\ket{1}_r
\end{eqnarray}
is a state of the larger system with only real amplitudes, and has all the information contained in the original state $\ket{\psi}$ \eqref{ch1:qubit}. Arbitrary unitary transformations on $\ket{\psi}$ can be then mapped onto real unitary transformations on the larger system \cite{CA:Shi:02}.
Finally, a single two-qubit gate 
\begin{itemize}
\item{$W$,}
\end{itemize}
the controlled $\frac{\pi}{2}$-rotation about the $y$-axis is universal for quantum computation \cite{CA:Shi:02}. In Chapter \ref{ch5hqca}, I use this single universal gate in my Hamiltonian Quantum Cellular Automaton constructions.


\subsection{Quantum Complexity}\label{ch1:qcomplex}
In recent years, quantum complexity classes have been defined and studied in an attempt to understand the capacity and limitations of quantum computers and quantum algorithms and their relation to classical complexity classes. The basic quantum complexity class is BQP, the analogue of BPP.
\begin{definition}[Complexity class BQP (Bounded error, Quantum, Polynomial time]
\label{defbqp}
Consider a decision problem $\Pi = (L_{yes}\cup L_{no})$, where $L_{yes}$ is the set of all instances $x$ of the problem $\Pi$ with the answer `yes'. The problem $\Pi$ is in BQP, if there exists a uniform family of circuits\footnote{Here {\em uniform family of circuits} means that there is one circuit for each problem size (and not for each problem instance), with a number of gates polynomial in the problem size.} $U$, with the following properties. For a problem instance $x$ with the answer `yes' ($x\in L_{yes}$), the probability of the circuit $U$ outputting `yes' is greater than $\frac{2}{3}$. On the other hand, when $x\in L_{no}$, the probability of the circuit outputting `yes' is less than $\frac{1}{3}$.
\end{definition}
In short, BQP is the class of problems one can solve efficiently using a quantum computer. The complexity class QMA, also known as BQNP, studied and defined in \cite{lh:Kni96} and \cite{KitaevBook}, is the quantum analogue of the classical complexity class MA, the probabilistic setting of NP. 
\begin{definition}[Complexity class QMA]
\label{ch1:qma}
A decision problem $\Pi = (L_{yes}\cup L_{no})$ of size $n$ is in the class QMA if there exists a polynomial time quantum verifier circuit $V$ such that for every instance $x$ of the problem
\begin{enumerate}
	\item $\forall x\in L_{yes}$ : there exists a witness state $\ket{\varphi(x)}$ such that 
		the computation $V\ket{x} \otimes \ket{\varphi(x)} \otimes \ket{0\dots0}_{ancilla}$
		yields the answer 1 with probability at least $p$ ;
	\item $\forall x\in L_{no}$ : for any witness state $\ket{\varphi(x)}$, the computation 
		$V\ket{x} \otimes \ket{\varphi(x)} \otimes \ket{0\dots0}_{ancilla}$
		yields the answer 1 with probability at most $p-\epsilon$,
\end{enumerate}
where $p>0$ and $\epsilon = \Omega(1/\textrm{poly}(n))$.
\end{definition}
The first known QMA-complete problem, local Hamiltonian, is a quantum analogue of classical MAX-k-SAT. Building on ideas that go back to Feynman \cite{lh:Feynman}, Kitaev \cite{KitaevBook} has shown that 5-local Hamiltonian is QMA complete and I present this result in Section \ref{ch1:Kitaev5}. A special case of Local Hamiltonian, Quantum $k$-SAT, is the topic of Chapter \ref{chQsat}. There I show that several of its variants are complete for QMA$_1$, the class QMA with one-sided error:
\begin{definition}[Complexity class QMA$_1$]
\label{ch1:qma1}
The class QMA$_1$ is the class QMA with single sided error, i.e. with $p=1$ in the above definition.
\end{definition}
In the classical case, MA with one-sided error is equivalent to regular MA. It is not known whether QMA=QMA$_1$.


\subsection{The Strengths and Limits of Quantum Computing}\label{ch1:grover}
There are quite a few problems on which quantum computers would algorithmically outperform classical ones. In quantum computing, we can use resources unknown in classical computing: superpositions,  interference and entanglement. Much has been discovered since Shor's factoring algorithm in 1994 \cite{ShorFact}. Stephen Jordan has a thorough and up-to-date summary of the algorithmic successes of quantum computing in his thesis \cite{AQC:JorThesis}. However, the speedups are rarely exponential. Not only that, solving problems such as {\em Parity}\footnote{{\em Parity}: determine the sum modulo 2 of $n$ bits.} cannot be sped up on a quantum computer at all \cite{parity}. I refer the reader interested in the limitations of quantum computers to \cite{lowbound1,lowbound2,lowbound3,lowbound4,BBBV}.

Algorithm designers bound the complexity of a problem from above by finding better algorithms, while lower bound enthusiasts limit it from below. The quantum variant of the black-box problem {\em Unstructured search} is an example of a problem for which these two approaches meet. The lower bound on its query complexity is known, and an optimal quantum algorithm, Grover's search, exists. 
\begin{definition}[Unstructured Search]
\label{defGrover}
Consider a quantum black-box (oracle) 
\begin{eqnarray}
	\hat{O} \ket{s} =
	\ii - 2\ket{w}\bra{w}. 
\end{eqnarray}
which flips the phase of a single, unknown to you, canonical basis quantum state $\ket{w}$, and acts trivially on any other canonical basis state. Find the state $\ket{w}$.
\end{definition}

Looking at a single state of $n$ bits unsuccessfully does not give me any information about where to look next, making this problem unstructured. Classically, the required number of queries scales like the number of possible $n$-bit strings, $2^n$, and is matched by the simplest randomized algorithm one can think of---check the strings randomly until you find the marked one. 
The query complexity of a quantum algorithm for unstructured search is the number of necessary calls to $\hat{O}$ in the worst case. In 1996, Grover discovered his famous algorithm requiring only $\sqrt{2^n}$ queries, getting a square root speedup over the best possible classical algorithm. Also around that time, Bennett, Bernstein, Brassard and Vazirani found a matching lower bound $l(n) \propto \sqrt{2^n}$ \cite{BBBV}. Therefore, Grover's algorithm is the best possible quantum algorithm for unstructured search. Unstructured search can be also recast in a continuous setting, where quantum computation is done by continuous-time Hamiltonian evolution instead of application of discrete unitary gates. Instead of a quantum black box, in \cite{AQC:FG98analog} Farhi and Gutmann use an oracle Hamiltonian  
\begin{eqnarray}
	H_{w} = - E \ket{w}\bra{w},
\end{eqnarray}
where $E$ determines an energy scale. This oracle Hamiltonian is always on in the system. What one can do to find $\ket{w}$ in this Hamiltonian setting is to prepare the system in some initial state,  and let the system evolve according to a Hamiltonian 
\begin{eqnarray}
	H(t) = H_{D}(t) + H_{w},
\end{eqnarray}
where $H_D(t)$ is an additional arbitrary time-dependent Hamiltonian. If measuring the system after some time produces $\ket{w}$, the algorithm is successful. Farhi and Gutmann show that the required running time for any such Hamiltonian algorithm (any $H_D(t)$ one uses) necessarily scales like $\sqrt{2^n}/E$. In Chapter \ref{ch2adiabatic}, I use this lower bound to show what not to do when designing quantum adiabatic algorithms.



\section{Local Hamiltonians\\ in Quantum Computation}\label{ch1:qlocal}

Finding the properties of relatively simple quantum systems is the topic of much condensed matter physics. Even though all the interactions are local, involving only a few neighboring particles, the resulting possibilities are amazingly rich. In this thesis I focus on two properties of systems governed by Local Hamiltonians. First,
a (quantum) computer can be constructed from a system governed by a simple Hamiltonian. Second, the ground states of local Hamiltonians can be hard to find, as they can encode solutions to interesting optimization problems. 

Hamiltonian time evolution according to the Schr\"odinger equation is the underlying principle of all quantum computing, producing the unitary operations required the quantum circuit model. Controlling the Hamiltonian of a quantum system means changing the interaction of the particles with each other and their surroundings. For example, in addition to the natural interaction of the particles in a quantum system, we can apply magnetic field to spins on a lattice, laser pulses to ions in an ion trap or electric fields to electrons in quantum dots. External control can allow us to perform the desired unitary transformations in the system, making it into a computer. Feynman \cite{lh:Feynman} thought of how to make a Hamiltonian Computer for simulating other quantum systems long before people invented quantum circuits. I introduce his idea in Section \ref{ch1:FeynmanHC}, and build on it throughout this thesis, especially in Chapters \ref{chQsat} and \ref{ch5hqca}. 

Besides investigating the time evolution of a quantum system, analyzing the spectrum of its eigenvalues is another source of interesting problems. It is long known that finding ground states of certain Hamiltonians is hard. If we could find ground states of a system governed by the Ising model 3D, we could solve  NP-complete problems \cite{Barahona}. Finding minima of locally constrained optimization problems can be recast into finding the ground state energy of a local Hamiltonian. 
I am interested in showing how hard it is to find it for some Hamiltonians. In one case, in Chapter \ref{ch3mps}, I give a numerical method for finding the ground state energy of the translationally invariant Ising model in transverse field on a Bethe lattice, based on the Matrix Product State ansatz. On the other hand, building on Kitaev's Local Hamiltonian problem I introduce in Section \ref{ch1:Kitaev5}, I classify the complexity of finding the ground state of several Hamiltonians in Chapter \ref{chQsat}.


\subsection{Feynman's Hamiltonian Quantum Computer}
\label{ch1:FeynmanHC}
Already in 1985, Feynman \cite{lh:Feynman} proposed the following interpretation/implementation of a computer using a quantum-mechanical systems. What he originally had in mind was to simulate a sequence of reversible classical gates in a quantum mechanical system. However, his construction applies equally well to implementing a quantum circuit. Take a sequence of $L$ unitary transformations $U_t$ on $n$ work qubits 
\begin{eqnarray}
	U=U_L U_{L-1} \dots U_2 U_1.
	\label{ch1:Ucircuit}
\end{eqnarray}
Consider now a system with two registers, 
\begin{eqnarray}
	\cH = \cH_{work} \otimes \cH_{clock}.
	\label{ch1:workclock}
\end{eqnarray}
The first one holds the $n$ work qubits, 
which I label $q_1, \dots, q_n$ throughout the thesis, while the second register
holds a pointer particle hopping on a line $0,1,\dots,L$, serving as a {\em clock} for the computation. I can realize the clock register using $L+1$ qubits labeled $c_0,\dots,c_L$, and initialize it in the state
\begin{eqnarray}
	\ket{0}_c = \ket{1_{c_0}0_{c_1}0_{c_2}\dots 0_{c_L}}.
\end{eqnarray}
The state of the clock register corresponding to ``time'' $t$ is
\begin{eqnarray}
	\ket{t}_c = \ket{0_{c_0}\dots 0_{c_{t-1}}1_{c_t}0_{c_{t+1}}\dots 0_{c_L}}.
	\label{ch1:pulse}
\end{eqnarray}
A simple time-independent Hamiltonian can facilitate the evaluation of $U$ on the work register of this system.
\begin{figure}
	\begin{center}
	\includegraphics[width=4in]{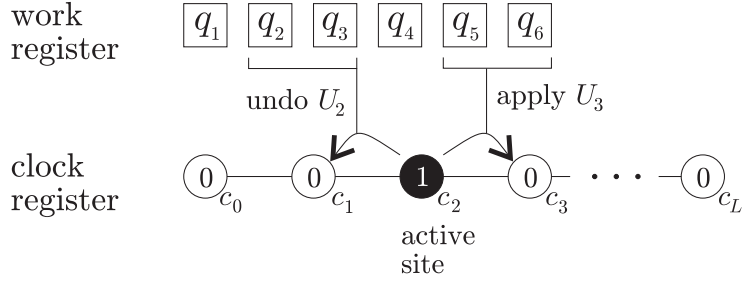}
	\end{center}
	\caption{Feynman's computer with a hopping pointer particle (a single spin up). When the active site moves from $c_2$ to $c_3$, gate $U_3$ gets applied to work qubits $q_5$ and $q_6$. When the active spin jumps back from $c_2$ to $c_1$, gate $U_2$ is uncomputed on qubits $q_2$ and $q_3$.
	}
	\label{ch1:figfeynman}
\end{figure}
As the pointer particle hops on the line
$0,1,\dots,L$, performing a quantum walk, the corresponding unitary operations are applied (or uncomputed) on the work qubits as depicted in Figure \ref{ch1:figfeynman}. The Hamiltonian with the desired dynamics is
\begin{eqnarray}
	H_{Fnmn} = \sum_{t=1}^{L}
		   U_t \otimes \sigma_{(t)}^{+} \sigma_{(t-1)}^{-}
		+  U_t^\dagger \otimes \sigma_{(t)}^{-} \sigma_{(t-1)}^{-},
	\label{ch1:HFeynman1}
\end{eqnarray}
where $\sigma^{+}_{(t)} = \ket{1}\bra{0}_{c_t}$ and $\sigma^{-}_{(t)} = \ket{0}\bra{1}_{c_t}$ are the raising and lowering operators for the clock register spin $c_t$ and $U_t$ acts on the corresponding work qubits.
When the system is initialized in the state 
\begin{eqnarray}
	\ket{\Psi_0} = \ket{\psi_0} \otimes \ket{0}_c 
	= \ket{\psi_0} \otimes \ket{10\dots 0}_c,
	\label{ch1:initstate}
\end{eqnarray}
where $\ket{\psi_0}$ is some initial state of the work qubits, time evolution with $H_{Fnmn}$ according to the Schr\"odinger equation brings the state into a superposition of states 
\begin{eqnarray}
	\ket{\Psi_t} = \ket{\psi_t} \otimes \ket{t}_c,
\end{eqnarray}
where $\ket{\psi_t}$ is the state of the work qubits after the first $t$ gates of the circuit:
\begin{eqnarray}
	\ket{\psi_t} = U_t U_{t-1} \dots U_1 \ket{\psi_0}.
	\label{ch1:psit}
\end{eqnarray}
After some time $\tau$, when I measure the clock register of $\ket{\Psi}$ and obtain $L$, the state of the system after the measurement becomes $\ket{\psi_L} \otimes \ket{L}_c$. The work register of this state contains the desired output of the circuit $U$. 

I can obtain $\ket{\psi_L}$ with high probability using a slightly modified system. First, I pad the sequence $U$ \eqref{ch1:Ucircuit} with $2L$ extra identity operations
\begin{eqnarray}
	U' = \ii_{3L} \ii_{3L-1} \dots \ii_{L+1} U,
\end{eqnarray}
and then expand the clock register to $3L+1$ qubits accordingly.
All the states $\ket{\psi_t}$ \eqref{ch1:psit} for $t\geq L$ now are the same, equal to the desired output state $\ket{\psi_L}$.
Starting from \eqref{ch1:initstate}, I let the system evolve for a time $\tau$ chosen uniformly at random between $0$ and $O(L\log L)$. Because the quantum walk on a line is rapidly mixing, the probability to measure a state with the clock register $\ket{t\geq L}_c$ at time $\tau$ is close to $\frac{2}{3}$, as shown in Appendix \ref{d20proof}. This yields the desired state $\ket{\psi_L}$ of the work qubits. If a measurement of the clock register results in $\ket{\psi_t}\otimes\ket{t<L}_c$, repeat the experiment.

The clock register representation of the states $\ket{t}_c$ given in \eqref{ch1:pulse} is only one of the possibilities. I describe a different one in the next Section, and several others in Chapter \ref{chQsat}. Considering this, Feynman's Hamiltonian can be rewritten in a general form, encompassing different implementations of $\ket{t}_c$, as
\begin{eqnarray}
	H_{Fnmn} &=& \sum_{t=1}^{T} H_{F}^{t}, \\
	H_{F}^{t}   &=& U_t \otimes X_{t,t-1}
		   +
		   \left(U_t \otimes X_{t,t-1} \right)^\dagger,
	\label{ch1:HFeynman}
\end{eqnarray}
where $U_t$ acts on the corresponding work register qubits, and the operator $X_{t,t-1}$ acts on the clock register as 
\begin{eqnarray}
	X_{t,t-1} &=& \ket{t}\bra{t-1}_c, 
	\label{ch1:XopsLOG}
\end{eqnarray}
increasing the time from $t-1$ to $t$. In Feynman's original formulation, this corresponds to moving the position of the single up spin to the right as $\ket{\dots 0100 \dots}_c \rightarrow \ket{\dots 0010 \dots}_c$, and the operator $X_{t,t-1}$ can be written as $X_{t,t-1} = \ket{01}\bra{10}_{c_{t-1},c_t}$, where $c_{t-1}$ and $c_t$ are two of the clock register qubits. Analogously, the Hermitian conjugate of $X_{t,t-1}$ moves the time register backwards as
\begin{eqnarray}
	\left(X_{t,t-1}\right)^\dagger &=& \ket{t-1}\bra{t}_c.
	\label{ch1:XopsLOG2}
\end{eqnarray}
This is my preferred formulation of Feynman's Hamiltonian throughout this thesis. It shows up in Chapter \ref{ch2adiabatic} where I use it to prove results about the universality of adiabatic quantum computing, while in Chapter \ref{chQsat} it serves a basic building block for my results about the complexity of Quantum Satisfiability. Finally, in Chapter \ref{ch5hqca} it inspires my Hamiltonian Quantum Cellular Automaton model. 


\subsection{Kitaev's Local Hamiltonian (LH) Problem}
\label{ch1:Kitaev5}
In this section I introduce Kitaev's {\em Local Hamiltonian} and sketch the proof that it is QMA complete. 
In \cite{KitaevBook} Kitaev defined the following promise problem:
\begin{definition}[Local Hamiltonian]
\label{defLH}
Take a $k$-local Hamiltonian $H = \sum_{j=1}^r H_j$ acting on $n$ qubits, composed of $r=\textrm{poly}(n)$ terms, each of those acting nontrivially on a constant number $k$ of qubits. Given two numbers $a$, $b$, where $b>a$ and the separation $b-a$ is greater than $n^{-\alpha}$ for some constant $\alpha$, determine whether 
	\begin{enumerate}
	\item $H$ has an eigenvalue not exceeding $a$, or 
	\item all eigenvalues of $H$ are greater than $b$.
	\end{enumerate}
\end{definition}

The two numbers $b$ and $a$ are separated by at least an inverse polynomial in the number of qubits in the system, $n$. LH is a promise problem, therefore we know that either $1$ or $2$ are true. Kitaev proved that
\begin{theorem}[LH is QMA-complete]
5-Local Hamiltonian is QMA-complete.
\end{theorem}
One of the directions of the proof is to show that LH is in QMA. 
The verifier circuit $U$, saying whether Arthur accepts a given state or not, is constructed as follows. Add an `answer' qubit to the system and for each term $H_j = \sum \lambda_s \ket{\psi_s}\bra{\psi_s}$, construct a measurement operator 
\begin{eqnarray}
	W_j : \ket{\psi_s}\otimes \ket{0}_{ans} 
		\rightarrow 
		\ket{\psi_s} \otimes \left(
			\sqrt{\lambda_s} \ket{0}
		+			
			\sqrt{1-\lambda_s} \ket{1}
		\right)_{ans}. 
\end{eqnarray}
Taking a general state $\ket{\eta} = \sum_s y_s \ket{\psi_s}$, 
the probability the circuit accepts it (i.e. that I measure 1 on the output of $U$) is
\begin{eqnarray}
	p_j(1) &=& \bra{\eta}\otimes \bra{0}_{ans}
		W_j^\dagger \left( \ii \otimes \ket{1}\bra{1}_{ans} \right) W_j
	\ket{\eta}\otimes\ket{0}_{ans} \\
	&=& \sum_{s} (1-\lambda_s) y_s^* y_s = 1-\bra{\eta}H_j \ket{\eta}.
\end{eqnarray}
Add now another `operator' register to the system. The general circuit for checking all of the terms in the Hamiltonian combines the checking operators $W_j$ into
\begin{eqnarray}
	W = \sum_{j=1}^{r} \ket{j}\bra{j}_{op} \otimes W_j,
\end{eqnarray}
where $r$ is the number of terms in the Hamiltonian.
After applying $W$ to the system in the state
\begin{eqnarray}
	\left(\sum_{j=1}^r \ket{j}_{op}\right) \otimes \ket{\eta} \otimes \ket{0}_{ans},
\end{eqnarray}
one measures the answer qubit. The probability of getting the outcome 1 is then $p_1 = 1- \frac{\bra{\eta}H\ket{\eta}}{r}$.
This is enough to show that LH is in QMA.

For the other direction, one needs to construct a corresponding instance of the Hamiltonian problem for every instance of a problem in QMA. Each of those has a verifier circuit $U$ with the required accept/decline properties associated with it. Given a quantum circuit $U$, Kitaev showed how to construct $H$ with the properties required in the definition of LH (Definition \ref{defLH}). If the circuit $U$ accepts the state $\ket{\psi}$, this Hamiltonian's ground state can be constructed from it, and the corresponding expectation value in this state will be smaller than $a$. On the other hand, if the circuit does not accept any state with high probability, the ground state of $H$ will have energy higher than $b$.

\subsubsection{Kitaev's Hamiltonian}
Consider a quantum system consisting of a work register with $n$ qubits $q_1,\dots,q_n$,
and a clock register. 
\begin{eqnarray}
	\cH = 
	\cH_{work}
		 \otimes \cH_{clock}.
\end{eqnarray}
Kitaev's Hamiltonian is a sum of three terms,
\begin{eqnarray}
	H_{Kitaev} = H_{prop} + H_{out} + H_{input}.
	\label{ch1:HKitaevLOG}
\end{eqnarray}

The first term checks the proper progression of the circuit $U$. Taking the idea from Feynman's Hamiltonian computer (see Section \ref{ch1:FeynmanHC}), Kitaev used \eqref{ch1:HFeynman} to construct a Hamiltonian whose ground state is the uniform superposition over the computation called the {\em history state}:
\begin{eqnarray}
	\ket{\phi}_{history} &=& \frac{1}{\sqrt{L+1}}
			\sum_{t=0}^{L} \ket{\psi_t} \otimes \ket{t}_c, 
	\label{ch1:historystate}
\end{eqnarray}
where 
\begin{eqnarray}
	\ket{\psi_t} &=& U_t U_{t-1} \dots U_2 U_1 \ket{\psi_0},
\end{eqnarray}
and $\ket{\psi_0}$ is any initial state of the work qubits.
This Hamiltonian is
\begin{eqnarray}
	H_{prop} &=& \frac{1}{2} \sum_{t=1}^{L}
		H_{prop}^{t},
		\label{ch1:Hprop}
\end{eqnarray}
where each term $H_{prop}^{t}$ is a projector constructed from Feynman's Hamiltonian \eqref{ch1:HFeynman} as
\begin{eqnarray}
	H_{prop}^t &=&
			\ii \otimes \left( P_{t-1} + P_{t} \right) 
			-  \underbrace{\left(U_t \otimes X_{t,t-1} 
			+  \left(U_t \otimes X_{t,t-1}\right)^\dagger
			\right)}_{H_{F}^{t}}
	,
\end{eqnarray}
with $P_{t} = \ket{t}\bra{t}_c$ the projector onto the state $\ket{t}_c$ of the clock register, and the time-increasing operator $X_{t,t-1}$ changing the clock from $t-1$ to $t$ given by \eqref{ch1:XopsLOG}. The only states that are not obviously annihilated by $H_{prop}^t$ live in the subspace of states with the clock register in the states $\ket{t-1}_c$ or $\ket{t}_c$. Within that subspace, the states with zero energy have the form
\begin{eqnarray}
	\ket{\phi} &=& \ket{\alpha}\otimes \ket{t-1}_c 
			+ U_t \ket{\alpha} \otimes \ket{t}_c,
\end{eqnarray}
making it easy to check that the history state \eqref{ch1:historystate} is annihilated by the non-negative Hamiltonian $H_{prop}$. It can be thought of as checking the correct propagation of the computation, as the expectation value of $H_{prop}$ in any state that is not a history state is greater than zero. In fact, the second lowest eigenvalue can be bounded by $\lambda_1 \geq \frac{c_1}{L^2}$ for some constant $c_1$ (see Appendix \ref{appProp} for a detailed analysis of the Hamiltonian \eqref{ch1:Hprop}).

The last of the work register qubits, $q_n$, is the designated output qubit for the circuit $U$. The second term $H_{out}$ in \eqref{ch1:HKitaevLOG} adds an energy penalty to states 
whose $q_n$ at time $L$ is $\ket{0}$, meaning that the circuit did not accept them. 
\begin{eqnarray}
	H_{out} = \ket{0}\bra{0}_{q_n} \otimes \ket{L}\bra{L}_c.
\end{eqnarray}
This term ensures that $H_{Kitaev}$ \eqref{ch1:HKitaevLOG} can have a low eigenvalue only if the circuit $U$ outputs $1$ on some input state $\ket{\psi_0}$ with high probability.

The work register has $n$ qubits $q_1,\dots,q_n$. Out of these, at time $t=0$, some are initialized in a purported proof state, while others are ancilla qubits required for the verification circuit. The last part of $H_{Kitaev}$,  
\begin{eqnarray}
	H_{input} = \sum_{k\in ancilla} \ket{1}\bra{1}_{q_k}
		\otimes \ket{0}\bra{0}_c,
	\label{ch1:Hinput}
\end{eqnarray}
checks the correct initialization of these ancilla qubits by adding an energy penalty to states whose ancillae at time $t=0$ (when the clock register is in the state $\ket{0}_c$) are not in the state $\ket{0}$.

These three terms together make up Kitaev's Hamiltonian \eqref{ch1:HKitaevLOG}. Let me now sketch the highlights of the proof that \eqref{ch1:HKitaevLOG} has the properties described in Definition \ref{defLH}. First, the expectation value of $H_{Kitaev}$ in the history state \eqref{ch1:historystate} corresponding to an initial state $\ket{\psi_{yes}}$ that the circuit $U$ accepts can be shown to be upper bounded by $\lambda_0^{yes} \leq \frac{\ep}{L}$ for some $\ep$.
Second, one needs to lower bound the lowest eigenvalue of $H_{Kitaev}$ in the case the circuit $U$ doesn't accept any state with high probability, and show that it is separated from $\lambda_0^{yes}$ at least as an inverse polynomial in $L$. All three terms in \eqref{ch1:HKitaevLOG} are non-negative. 
The proof uses a geometrical argument, showing that the angle between the null spaces of $H_{prop}$ and $(H_{input} + H_{out})$ is lower bounded by
$\theta$, an inverse polynomial in $L$. The smallest eigenvalue of $H_{Kitaev}$ can then be bounded from below by the smallest nonzero eigenvalue of $H_{prop}$ multiplied by a geometric factor coming from $\theta$. The resulting lower bound on the lowest eigenvalue of \eqref{ch1:HKitaevLOG} is then $\lambda_0^{no} \geq c(1-\sqrt{\ep})L^{-3}$. $H_{Kitaev}$ thus has the desired properties.

The final step in the proof that 5-local Hamiltonian is QMA complete is to implement the clock register. I describe several clock register constructions in detail in Section \ref{ch4:clocks}. One of these is Kitaev's {\em domain wall} clock. Similarly to Feynman \eqref{ch1:pulse}, he encoded his clock register in unary using $L+1$ clock qubits $c_t$. However, the states $\ket{t}_c$ are now encoded as
\begin{eqnarray}
	\ket{t}_c	&=&	\ket{1_{c_0} 1_{c_1} \dots 1_{c_t} 0_{c_{t+1}} \dots 0_{c_{L-1}}0_{c_L}}
\label{ch1:unary}
\end{eqnarray} 
The Hilbert space of the $L+1$ clock qubits $c_0,\dots,c_L$ is much bigger than the space spanned by the proper clock states $\ket{t}_c$ given by \eqref{ch1:unary}, which I call the {\em legal clock subspace}, $\cH_{legal}$.
By adding a clock-checking Hamiltonian 
\begin{eqnarray}
	H_{clock} &=& \sum_{t=0}^{L-1} \ket{01}\bra{01}_{c_{t},\,c_{t+1}}
	\label{ch1:Hclock5}
\end{eqnarray}
to $H_{Kitaev}$ \eqref{ch1:HKitaevLOG}, one can ensure that the low-energy eigenstates of $H$ are close to the subspace of states with a proper unary-encoded clock, i.e.
whose clock register states are close to the legal clock subspace $\cH_{legal}$. The final form of the Hamiltonian is then
\begin{eqnarray}
	H_{Kitaev5} &=& H_{prop} + H_{out} + H_{input} + H_{clock}.
	\label{ch1:HKitaev5}
\end{eqnarray}
In this clock implementation, $P_{t}$, the projector onto the state $\ket{t}$, is only a 2-local operator
\begin{eqnarray}
	P_t^c = \ket{t}\bra{t}_c = \ket{10}\bra{10}_{c_{t},c_{t+1}},
\end{eqnarray}
with a special case $P_{L}^c = \ket{1}\bra{1}_{c_L}$, acting trivially on the rest of the qubits. On the other hand, the operator $X_{t,t-1}$ \eqref{ch1:XopsLOG} in this encoding is 3-local (acting nontrivially on three clock qubits):
\begin{eqnarray}
	X_{t,t-1} &=& \ket{t}\bra{t-1}_c = \ket{110}\bra{100}_{c_{t-1},c_{t},c_{t+1}}, 
\end{eqnarray}
with a special case at the right end of the clock register,
\begin{eqnarray}
 X_{L,L-1} = \ket{L}\bra{L-1}_c = \ket{11}\bra{10}_{c_{L-1},c_{L}}. \end{eqnarray} 
Recalling that 2-qubit gates are universal for quantum computation, the terms in $H_{prop}$ \eqref{ch1:Hprop}, and thus also $H_{Kitaev}$, become at most 5-local. It remains to show that in the case the circuit $U$ rarely outputs $1$, the lowest eigenvalue of \eqref{ch1:HKitaev5} is also bounded from below by an inverse polynomial in $n$. This proof follows from the geometric arguments outlined above without much complication. Therefore, 5-local Hamiltonian is QMA-complete.

In Chapter \ref{chQsat}, I return to this problem, focusing on a special case called Quantum Satisfiability, defined by Bravyi in \cite{lh:BravyiQ2SAT06}. When all the terms in $H$ are projectors (non-negative), one can ask whether there exists a state with energy exactly 0, or whether the ground state energy of $H$ is greater than $b$. I investigate the complexity of several variants of Q-SAT and utilize the techniques used in Kitaev's proof I just outlined.

\chapter{Adiabatic Quantum Computation}
\label{ch2adiabatic}

Farhi et al. \cite{AQC:FGGS00} introduced Adiabatic Quantum Computing (AQC) as a method to find a minimum of a classical cost function by using a quantum system with a time-dependent, local Hamiltonian. I describe this model in Section \ref{ch2:aqcintro}, and then review the proof of Aharonov et al. \cite{AQC:AvDKLLR05}, who showed that AQC is equivalent to the quantum circuit model, in the light a previously unpublished result of Seth Lloyd \cite{AQC:Lloydunpublished}
in Section \ref{ch2:aqcbqp}. This motivates my definition of a Hamiltonian Computer model in Section \ref{ch2:hc}, which I later use to prove universality of the Quantum 3-SAT Hamiltonian in Section \ref{ch4:train} and of two Hamiltonian Quantum Cellular Automata in Chapter \ref{ch5hqca}.

There is hope that there may be combinatorial search problems, defined on $n$ bits so that $N=2^n$, where for certain ``interesting'' subsets of the instances the run time of the Quantum Adiabatic Algorithm grows sub-exponentially in $n$. A positive result of this kind would greatly expand the known power of quantum computers. At the same time it is worthwhile to understand the circumstances under which the algorithm is doomed to fail. Section \ref{ch2:local}, based on the paper \cite{AQC:fail},

\mypaper{How to make the Quantum Adiabatic Algorithm Fail}{Edward Farhi, Jeffrey Goldstone, Sam Gutmann and Daniel Nagaj}{ 
The quantum adiabatic algorithm is a Hamiltonian based quantum algorithm designed to find the minimum of a classical cost function whose domain has size $N$. We show that poor choices for the Hamiltonian can guarantee that the algorithm will not find the minimum if the run time grows more slowly than square root of $N$. These poor choices are nonlocal and wash out any structure in the cost function to be minimized and the best that can be hoped for is Grover speedup. These failures tell us what not to do when designing quantum adiabatic algorithms.
}
contains two results about the necessity of local structure for successful AQC algorithms. First, in Section \ref{ch2:projector} I show a connection between a bad choice of the starting Hamiltonian in an adiabatic algorithm and an information theoretical lower bound on the running time of any quantum algorithm for unstructured search (see Section \ref{ch1:grover}). Second, I show in Section \ref{ch2:scrambled} that knowing the spectrum of the final Hamiltonian is not enough to determine whether finding its ground state via an AQC algorithm is feasible. The local structure determining how the final spectrum arises is essential.

\section{Introduction to AQC}
\label{ch2:aqcintro}

Consider a classical optimization problem with cost function $h(z)$ on $n$ bits, with $z=0,\dots,2^n-1$. I can construct a Hamiltonian for $n$ qubits, diagonal in the computational ($z$) basis, according to this cost function as
\begin{eqnarray}
	H_P = \sum_{z=0}^{N-1} h(z) \ket{z}. \label{problemHam}
\end{eqnarray}
The goal is now to find the ground state of $H_P$, which encodes the solution to the classical optimization problem with cost function $h(z)$. If I could prepare the ground state of $H_P$, a measurement of the spins in the $z$-basis would produce the desired optimal assignment $z$ (or one of them, if there are several). 

To prepare the ground state of $H_P$, a `beginning' Hamiltonian $H_B$ is 
introduced with a known and easy to construct ground state $\ket{g_B}$. The quantum computer is a system governed by the time dependent Hamiltonian 
\begin{eqnarray}
	H(t) &=& \left(1-\frac{t}{T}\right) H_B + \left(\frac{t}{T}\right) H_P
		= (1-s)H_B + s H_P, \label{adiabaticHam}
\end{eqnarray}
where the time $t$ runs from $0$ to $T$. It is also sometimes useful to denote $s=t/T$, with $0\leq s \leq 1$. This defines 
a path in Hamiltonian space between $H_B$ and $H_P$,
with the rate of change of $H(t)$ controlled by $T$.
The Schr\"odinger equation for a state $\ket{\psi(t)}$ reads
\begin{eqnarray}
	i \deriv{}{t} \ket{\psi(t)} = H(t) \ket{\psi(t)}. \label{Schrodinger1}
\end{eqnarray}
I now choose the state at $t=0$ to be the ground state of $H_B$
\begin{eqnarray}
	\ket{\psi(0)} = \ket{g_B}
\end{eqnarray}
and run the algorithm for time $T$ (let the system evolve according to $H(t)$). 
The power of the algorithm comes from the Adiabatic theorem.
Suited for our investigation, in this work I choose to refer to its formulation due to Jeffrey Goldstone \cite{AQC:Goldstone}:
\begin{theorem}[Adiabatic Theorem]\label{thm:adiabatic}
	Consider a time-dependent Hamiltonian $H(t)$, and the time evolution of the state $\ket{\phi_0}$, the ground state of $H(0)$, from time $t=0$ to $t=T$. 
Then
\begin{eqnarray}
	\ep(T) = \norm{ \ket{U(T,0)\phi_0(0)} - \ket{\phi_0(T)}} 
	= O \left( \frac{1}{T} \max_t \frac{M(t)^2}{g(t)^{3}}\right), 
\end{eqnarray}
where $U(T,0)$ is the time evolution operator from time $0$ to $T$ corresponding to $H(t)$, while $M(t)= T \norm{\frac{dH}{dt}}$ and $g(t)=E_1(t)-E_0(t)$ is the energy gap of $H(t)$.
\end{theorem}
Thus, by the adiabatic theorem, if the gap is large enough, $T$ is large enough, and the norm of the derivative of $H(t)$ is bounded, $\ket{\psi(T)}$ will have a large component in the ground state subspace of $H_P$. A measurement of $z$ can then be used to find the minimum of the classical cost function $h(z)$ I wanted to find. The algorithm is useful if the required run time $T$ does not grow exponentially with the number of qubits $n$.

Rescaling the Hamiltonian by a factor $E$ would result in a shorter required running time. However, this is a known tradeoff between energy and time. It is thus usual to think the required resources of an AQC algorithm as $T\cdot \norm{H}$, or as the required running time for a rescaled Hamiltonian whose norm is $\norm{H}=1$.

Note that if the eigenvalue $\lambda_0(t)$ is degenerate, i.e. when we have a subspace $\mathcal{L}_0(t)$ instead of a single state $\ket{\phi_0(t)}$ corresponding to the eigenvalue $\lambda_0(s)$, the formulation of the theorem can be generalized to involve the projection onto the subspace $\mathcal{L}_0(T)$ instead of the overlap with the single eigenstate $\ket{\phi_0(T)}$.

It used to be common to state the adiabatic theorem with $\ep(T)=const.$ for $T \propto \Delta^{-2}$, where $\Delta = \min_t g(t)$ is the minimum energy gap. However, several authors have recently rigorously investigated the sufficient conditions for the adiabatic theorem and only proved results in which the scaling of $T$ with the gap was worse than $\Delta^{-2}$ \cite{AQC:RuskaiGap,AQC:Ruskai2,AQC:AmbainisRegev}.
In this work, I choose to use Goldstone's formulation which for the particular Hamiltonian I investigate \eqref{adiabaticHam} implies constant overlap of the time evolved state with the ground state of $H_P$ if the control time depends on the gap like $T \propto \Delta^{-3}$.

Also note that other time-dependent paths in Hamiltonian space from $H_B$ to $H_P$ are possible, each of them defining a specific Quantum Adiabatic Algorithm. In fact, it has been shown \cite{AQC:pathchange}, that a path change in the Hamiltonian space can be used to beat certain counterexamples in which the standard, linear-path QAC requires exponential time to run. 

\subsection{The power of AQC}

Although AQC was introduced with the hope of solving local optimization problems, the verdict on the practical use of AQC for this purpose is so far 
`not discouraging'. Several numerical studies provided results that did not show exponential behavior in the required running time \cite{AQC:FarhiScience,AQC:Hogg}. Nevertheless, because of computer resource limitations, these could not be done for bit numbers greater than $n=24$, while behavior change might be expected for much larger values of $n$ \cite{surveyPropagation}. Only very recently, Young et al. \cite{AQC:SmelyanskiyEC} utilized a different numerical method, investigating the gap in a QAC algorithm for the NP-complete {\em Exact Cover} problem using a Quantum Monte Carlo method, and found polynomial behavior of the gap for up to $n=128$ qubits.

On the other hand, soon after AQC was introduced, examples of cases in which a simple linear adiabatic algorithm failed were constructed by van Dam et al. \cite{AQC:failureDMV}. However, all of these problems were circumvented in Farhi et al. \cite{AQC:pathchange} by changing the path in Hamiltonian space by adding a random term $H_C$ to the Hamiltonian as
\begin{eqnarray}
	H(t) = (1-s) H_B + s H_P + s(1-s) H_C.
	\label{ch2:pathchange}
\end{eqnarray}
In the cleverly constructed example of \cite{AQC:failureDMV}, the gap above the ground state is exponentially small only at a single point on the interpolation path in $s$. Adding the term $H_C$ avoids the bad point on the Hamiltonian path where energy levels cross, making the scaling of the required runtime favorable again.

Recently, more papers investigating cases where certain adiabatic algorithms fail appeared \cite{AQC:Znidaric06,AQC:WeiYing06}. Nevertheless, they share a common design flaw, a bad choice of the starting Hamiltonian. In Section \ref{ch2:local}, I show that poor choices for the Hamiltonian can guarantee that the algorithm will not find the minimum if the run time grows more slowly than $\sqrt{N}$. These poor choices are nonlocal and wash out any structure in the cost function to be minimized and the best that can be hoped for is Grover speedup. These failures tell us what not to do when designing quantum adiabatic algorithms.

AQC is interesting on its own as an implementation of a quantum computer, as Aharonov et al. in \cite{AQC:AvDKLLR05} (see Section \ref{ch2:aqcbqp}) proved that its computational power is equivalent to that of the conventional quantum circuit model. Because of the energy gap, the computation in the AQC model is inherently protected against noise. The robustness and fault tolerance of the AQC has been recently investigated by Childs et al. \cite{AQC:Childs}, Jordan et al. \cite{AQC:codes}, Lidar \cite{AQC:Lidar} and Lloyd \cite{AQC:Lloydunpublished}. The challenge today is to find more new, entirely adiabatic algorithms, such as the adiabatic state preparation of Aharonov and Ta-Shma \cite{AQC:AharonovPrepare}.


\section{AQC is equivalent to BQP}
\label{ch2:aqcbqp}
Aharonov et al. proved the universality of the AQC model in \cite{AQC:AvDKLLR05}:
\begin{theorem}[Equivalence of AQC and the quantum circuit model]
The model of adiabatic computation is polynomially equivalent to the standard model of quantum computation.
\end{theorem}
One of the proof directions is straightforward. One can simulate time evolution with an AQC Hamiltonian on a standard quantum computer, dividing the time into small intervals and approximating the time evolution in each slice. The usual procedure utilizes the finite $n$ approximation of the Trotter-Suzuki formula \cite{Trotter1,Trotter2,Trotter3,Trotter4,Trotter5} 
\begin{eqnarray}
	e^{i(A+B)\Delta t} = \lim_{n\rightarrow \infty} \left(e^{iA\Delta t/n} e^{iB\Delta t/n}\right)^n,
\end{eqnarray}
and the resulting unitary transformations are applied using a circuit-based quantum computer. The precision of this approximation can be improved, as the AQC Hamiltonian is changing only linearly in time by integrating the Dyson series for the time-dependent Hamiltonian $H(t)$. Assume that $H(t)$ consists of two types of terms, the first containing only terms built from $\sigma^{x}_i$ and the second containing only terms built from $\sigma^z_i$. The Hamiltonian can then be written as 
\begin{eqnarray}
	H\left(t\right) = \left(1-\frac{t}{T}\right) H_x + \left(\frac{t}{T}\right) H_z,
\end{eqnarray}
where 
$H_x$ and $H_z$ contain only terms built from the respective type of operator.
As shown in Appendix \ref{appDyson}, the approximation of the time evolution operator $U(t+\Delta t,t)$ in the form of a product of exponentials, correct to order $(\Delta t)^2$, is 
\begin{eqnarray}
	W_2 &=& e^{-i a_2 H_x} e^{-i b_2 H_z} e^{-i c_2 H_x},
\end{eqnarray}
with
\begin{eqnarray}
	a_2 = c_2 &=& \left(1-\frac{t}{T} - \frac{\Delta t}{2T}\right) \frac{\Delta t}{2}, \\
	b_2 &=& \left(\frac{t}{T} + \frac{\Delta t}{2T}\right) \Delta t.
\end{eqnarray}

On the other hand, one can encode the evaluation of any quantum circuit $U$ into time evolution with a sufficiently slowly changing time-dependent AQC Hamiltonian. The proof in \cite{AQC:AvDKLLR05} uses ideas from Kitaev's construction I described in Section \ref{ch1:Kitaev5}. Consider a system with two registers, work and clock \eqref{ch1:workclock}, with the work register holding $n$ qubits $q_1,\dots, q_n$ and with the clock register consisting of $L+1$ qubits $c_0,\dots,c_L$, where $L$ is the number of gates in the circuit $U$. Encode the clock states states $\ket{t}_c$ in unary as in \eqref{ch1:unary}.
Aharonov et. al. choose the initial Hamiltonian
\begin{eqnarray}
	H_{init} = H_{input} + H_{clock} + H_{clockinit}.	
\end{eqnarray}
The term $H_{input}$ \eqref{ch1:Hinput} checks that when the clock is in the state $\ket{0}_c = \ket{1_{c_0} 0_{c_1} \dots 0_{c_L}}$, the state of the work qubits is initialized in $\ket{00\dots 0}$, by adding an energy penalty to all states for which this is not true. Second, $H_{clock}$ \eqref{ch1:Hclock5} checks whether the clock register states are proper unary-clock states. Finally, in the unary clock encoding, the term
\begin{eqnarray}
	H_{clockinit} = \ket{1}\bra{1}_{c_1}	
\end{eqnarray}
is a projector onto all states whose clock register is not in the state $\ket{0}_c$, the only legal clock state whose second clock qubit ($c_1$) is in the state $\ket{0}$. Because the three terms $H_{init}$ are non-negative, the ground state of $H_{init}$ is the state annihilated by all of them:
\begin{eqnarray}
	\ket{\phi_{init}} = \ket{\phi_0(0)} = \ket{00\dots 0} \otimes \ket{1_{c_1} 0_{c_2} \dots 0_{c_T}}.	
	\label{ch2:init}
\end{eqnarray}
The final Hamiltonian is chosen to be a modification of \eqref{ch1:HKitaev5}
\begin{eqnarray}
	H_{final} = H_{prop} + H_{input} + H_{clock},	
	\label{ch2:Hcomputer}
\end{eqnarray}
omitting the term $H_{out}$.
The term $H_{prop}$ \eqref{ch1:Hprop} has a degenerate ground state subspace containing all the possible history states \eqref{ch1:historystate} for a given quantum circuit $U$. 
The ground state of the complete $H_{final}$ is the specific history state \eqref{ch1:historystate} for the quantum circuit $U$ starting in the initial state \eqref{ch2:init}. Using classical Markov chain methods, the authors then show that the gap above the ground state of the AQC Hamiltonian 
\begin{eqnarray}
	H(\tau) = \left(1-\frac{\tau}{T}\right) H_{init} + \frac{\tau}{T} H_{final}
	\label{ch2:Haqcbqp}
\end{eqnarray}
is bounded from below by an inverse polynomial in $L$ (the length of the computation $U$). This is possible, because it is enough to look for this gap in the subspace of legal unary clock states. Moreover, we know that the eigenvalue gap between the ground state subspace and the first excited state of $H_{prop}$ is polynomially large (see Appendix \ref{appProp} for details).

The authors then encode the clock register in unary, using the 3-local construction of  \cite{lh:KR03}, and prove that $3$-local AQC is universal for quantum computation. However, the cost of this translation from a quantum circuits to AQC is steep. First, the required runtime for the adiabatic algorithm is proportional to at least the inverse gap squared. Second, the gap of the final Hamiltonian is small, because of the penalty terms necessary for the 3-local clock construction of \cite{lh:KR03}. 
The required runtime, for  a rescaled Hamiltonian with $\norm{H}=O(1)$, is then claimed to be
$T  \propto L^{14}$. However, the authors in \cite{AQC:AvDKLLR05} obtained this using a formulation of the Adiabatic theorem which was later found to be wrong, with a $\Delta^{-2}$ dependence of the runtime. If, in fact, I use the adiabatic theorem with $\Delta^{-{2+\delta}}$, their result becomes $T \propto L^{14+3\delta}$. The proper adiabatic theorem I stated above has $\delta=1$, estimating the running time for the simulation of quantum circuits by AQC as $T \propto L^{17}$. Deift, Ruskai and Spitzer \cite{AQC:RuskaiGap} later showed that one of the gap estimates in the proof (not in the adiabatic theorem) of \cite{AQC:AvDKLLR05} can be made better, resulting in the total required runtime scaling $T \propto L^{12+3\delta}$ (or, with the adiabatic theorem I use, $L^{15}$). It is possible that the adiabatic theorem with $\delta<1$ in the condition $T\propto \Delta^{-(2+\delta)}$ is provable \cite{AQC:Ruskai2}. However, for now, I stay on the safe side and use the adiabatic theorem with $\delta = 1$, as proved by Goldstone \cite{AQC:Goldstone}.

Later, Mizel et al. \cite{AQC:Mizel} showed that the required running time for simulating quantum circuits using AQC scales like $n^2 L^2$, where $L$ is the number of gates in the circuit, and $n$ is the number of work qubits. However, this running time is for a Hamiltonian whose norm is $\norm{H}=L$, so that the real running time for a rescaled Hamiltonian scales like $n^2 L^3$. 
Their construction is built on 2-particle, 4-site interactions of electrons hopping through an array of quantum dots. If we think of their model as implemented using qubits instead of hopping particles, it becomes 4-local. Also, they use the knowledge about where in their algorithm the gap is small, running the algorithm slower around that spot, allowing them to have their runtime scale as good as $T\propto \Delta^{-1}$. 

When I want to keep the interactions 3-local (for qubits), much improvement is necessary. For this, I suggest using my new 3-local Hamiltonian construction of Section \ref{ch4:q3} instead of \cite{lh:KR03}. The required running time for the AQC algorithm then becomes $T \propto L^{5+2\delta} = L^7$. This scales worse than Mizel's required runtime, but my Hamiltonian is only 3-local for qubits. Also, when implementing my static-qubit construction using real hopping particles like Mizel \cite{AQC:Mizel}, the required interactions couple 2 particles, one of which hops between two sites, giving a 2-particle, 3-site interaction. The improvement over the previous 3-local constructions is possible because my clock register encoding uses only constant norm penalty terms for illegal clock states, as opposed to penalty terms scaling like high polynomials in $L$ as in \cite{lh:KR03,lh:KKR04}. Moreover, the Hamiltonians used can be further restricted using the results of Sections \ref{ch4:triangleXZ} and \ref{ch4:train}, as explained therein. 

This, however, is not the strongest result I want to present here. Using an unpublished result of Seth Lloyd \cite{AQC:Lloydunpublished}, in the next Section I show that the resources necessary to perform universal quantum computation using 3-local AQC scale only as $
\tau \cdot \norm{H} = O(L^2\log^2 L)$, where $\tau$ is the algorithm runtime. However, in this case it is not the adiabatic theorem that gives the computational model its power, it is rather the underlying structure of the final Hamiltonian whose dynamics is a quantum walk on a line. 


\subsection{The Hamiltonian Computer model and universality}
\label{ch2:hc}
In this section, I look at the dynamics of the Hamiltonian \eqref{ch2:Haqcbqp} in more detail. The result presented here was pointed out to me by Seth Lloyd during my Part III doctoral exam\footnote{I was made to work it out on the board, and was quite puzzled by it at the time, thinking I must have gotten it wrong.} and I thank him for the access to the write-up of his unpublished work \cite{AQC:Lloydunpublished}.
He showed that even when the control time $T$ for the change in the Hamiltonian is short, the result of the quantum circuit $U$ can still be obtained from the state of the system with high probability after waiting for time $\tau$ upper bounded by a polynomial in $L$, the number of gates in the circuit $U$. 

What can happen to the state $\ket{\phi_{init}}$ \eqref{ch2:init} when we change the Hamiltonian from $H_{init}$ to $H_{final}$ faster than the required conditions in the statement of the adiabatic theorem? Let me consider an extreme scenario, jumping from $H_{init}$ to $H_{final}$ instantly.
The system will not be in the ground state of $H_{final}$, but in a superposition of some of its eigenstates, those which it could transition to under the action of the Hamiltonian \eqref{ch2:Haqcbqp}. To see what these are, let me first analyze the time evolution of $\ket{\phi_{init}}$ directly with $H_{final}$ \eqref{ch2:Hcomputer}. The term determining the dynamics of the system is $H_{prop}$, as the terms $H_{input}$ and $H_{clock}$ annihilate the initial state and all the states $H_{final}$ brings it to. 
Consider now the subspace of states which one can transition to from $\ket{s}\otimes\ket{0}_c$ using arbitrary powers of $H_{prop}$. There is one such subspace for each work register basis state $\ket{s}$, and these subspaces are not connected by $H_{prop}$. As described in much more detail in Appendix \ref{appProp}, after a unitary basis change
\begin{eqnarray}
	\ket{s} \otimes \ket{t}_c 
	\qquad \rightarrow \qquad 
	\ket{s_t} \otimes \ket{t}_c = (U_t U_{t-1} \dots U_2 U_1) \ket{s} \otimes \ket{t}_c,
\end{eqnarray}
the Hamiltonian $H_{prop}$ becomes block diagonal. The matrix form of $H_{prop}$ in each of these blocks is 
\begin{eqnarray}
	H_{prop}' = \ii - \half B - \half P_{0} - \half P_{L}, 
\end{eqnarray}
where the matrix $B$ (see Appendix \ref{d20proof}) is the adjacency matrix for a line of length $L+1$ and $P_{0}$ and $P_{L}$ are projectors onto the endpoints of this line.
This gives the system the dynamics of a quantum walk on a line.

Because the blocks for different work register basis states $\ket{s}$ are decoupled,
the specific initial state $\ket{\phi_{init}}= \ket{0}\otimes\ket{0}_c$ time evolved for some time $\tau$ with $H_{prop}$ can thus end up only in some superposition of the states 
\begin{eqnarray}
	\ket{\Psi_t} = \ket{\psi_t}\otimes \ket{t}_c = (U_t U_{t-1} \dots U_2 U_1) \ket{00\dots 0} \otimes \ket{t}_c,
	\label{ch2:psit}
\end{eqnarray}
where $\ket{\psi_t}$ is the state of the work qubits after after $t$ gates of the quantum circuit $U$ have been applied to an initial all-zero state $\ket{00\dots0}$. 
The span of all these states $\ket{\Psi_t}$ defines the subspace $\cH_{0}$.
The state of the system $\ket{\phi(\tau)}$ after time $\tau$ has no overlap with states outside of $\cH_{0}$, because the Hamiltonian $H_{prop}$ is constructed so that it doesn't couple $\cH_{0}$ to the subspaces generated by $H_{prop}$ from some different state $\ket{s\neq 0}\otimes \ket{0}_c$.

Note that the state $\ket{\psi_L}$ contains the result of the quantum circuit $U$ in its work register. I would like to find the system with the computation already done, so I boost the probability by padding the initial circuit as in Figure \ref{ch2:figurepadcircuit}.
\begin{figure}
	\begin{center}
	\includegraphics[width=3.5in]{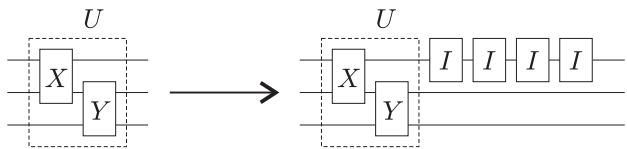}
	\end{center}
	\caption{Padding the circuit $U$ with $L$ gates by $2L$ extra identity gates. The state of the system after $L$ steps stays unchanged.}
	\label{ch2:figurepadcircuit}
\end{figure}
After the original $L$ gates of the circuit, let me add $2L$ extra identity gates, taking $L\rightarrow 3L$. Now, all the states $\ket{\Psi_t}$ with $t>L/3$ contain the output of the circuit $U$ in their work register, as the work register of $\ket{\Psi_t}$ is $U \ket{00\dots 0}$ and stays unchanged for all $t>L/3$ because I padded the circuit with identity gates.

The time evolution of $\ket{\phi_{init}}$ with $H_{prop}$ corresponds to a continuous-time quantum walk on a line with specific boundary conditions. I can get rid of the boundary conditions by wrapping the walk around on a circle, which simplifies the analysis, given in detail in Appendix \ref{d20proof}. There I show that this walk mixes so fast, that when I choose
a time $\tau$ uniformly at random between $0$ and $\tau_1 \propto L \log^2 L$  the probability to find the system in a state $\ket{t}_c$ with $t>L/3$ is close to $2/3$. Therefore, when I let the state $\ket{\phi_{init}}$ evolve for a time chosen uniformly at random between $0$ and $O(L\log^2 L)$ and measure the clock qubit $c_{L/3}$, I will obtain the result $1$ with probability close to 2/3. This signifies `success', as it implies that the clock register of the state I measured was one of the states with $t>L/3$. The work qubits then contain the result of the quantum computation $U$.

To conclude, when I change $H_{init}$ to $H_{final}$ quickly, the resulting state of the system will be some superposition of states $\ket{\Psi_t}$, close to $\ket{\Psi_0}$. When I then let the system evolve for extra time $\tau < O(L\log^2 L)$ with the Hamiltonian $H_{final}$, I can measure the clock register, and with high probability I will obtain a result that tells me that the work register now contains the result of the quantum circuit $U$. The overall resources necessary for this procedure then scale like $\tau \cdot \norm{H} = O(L^2\log^2 L)$, as the Hamiltonian I use has norm $O(L)$.

Let me now formalize the model of computation I just described. 
\begin{definition}[Hamiltonian Computer (HC) model]
A Hamiltonian Computer aiming to simulate a quantum circuit $U$ with $L$ gates is defined by a system consisting of two registers, work and clock, a $k$-local Hamiltonian $H$, a simple to prepare initial state $\ket{\psi_0}$, and a bound on the required evolution time. The clock register states must be constructed in such a way that a measurement distinguishing $\ket{t\leq L/3}_c$ (failure) and $\ket{t> L/3}_c$ (success) is simple. The computation procedure goes as follows:
\begin{enumerate}
\item Prepare the system in the initial state $\ket{\psi_0}$, the ground state of a $k$-local Hamiltonian $H_0$.
\item Turn off the initial Hamiltonian $H_{0}$, and let the system evolve with a time independent Hamiltonian $H$ for a time $\tau$ chosen uniformly at random between $0$ and $poly(L)$. 
\item Make a measurement determining whether the clock register 
is in a state $\ket{t> L/3}$, which signifies success. The probability to obtain such state is greater than $p=\half$. The result of the quantum circuit $U$ is then in the work register of the system. Restart otherwise. 
\end{enumerate}
\end{definition}

Feynman's Hamiltonian Computer with a unary clock $\ket{t}_c = \ket{0\dots 010\dots 0})_c$ (see Section \ref{ch1:FeynmanHC}) is one of the possible implementations of a HC. The initial Hamiltonian $H_0$ is one whose ground state is $\ket{\psi_0}=\ket{0_{q_1}\dots 0_{q_n}}\otimes \ket{0_{c_0}\dots 0_{c_L}}$, and the time independent Hamiltonian $H$ is given by \eqref{ch1:HFeynman1}. To determine success/failure, I need to measure the first third of the clock qubits. When I don't find the spin up there, it means I have succeeded. In the case of Feynman's Hamiltonian, the Hamiltonian itself is 4-local, and Lloyd \cite{AQC:Lloydunpublished} bounds the required runtime from above by $L^2$. As seen in the previous Section, Aharonov et al. \cite{AQC:AvDKLLR05} use the 3-local Hamiltonian $H_{final}$ of Kempe and Regev \cite{lh:KR03} to simulate quantum circuits using AQC, and require resources of the order $L^{12+3\delta}$. However, when I encode the clock register and the clock state transitions as in Section \ref{ch4:train}, the gap of $H_{final}$ whose norm is $O(L)$ scales like $L^{-2}$. The mixing of the quantum walk is then such that the required running time becomes only $O(L \log^2 L)$. When I rescale my Hamiltonian to $\norm{H}=1$, the resources necessary for this HC model are $O(L^2 \log^2 L)$. For more detail on this construction, see Section \ref{ch4:train}.

What remains is the question how is one protected against noise in such Hamiltonian Computer models. In AQC, it was the energy gap. However, as I have shown, the energy gap computed by Aharonov et al. \cite{AQC:AvDKLLR05} is not relevant, because the transitions to higher momenta eigenstates of the quantum walk still contain the result of the computation. The relevant energy barrier is the energy cost to transition from the subspace $\cH_{legal}$ to $\cH_{legal}^\perp$, in which case we would lose the computation. 
For this purpose, I use a clock-checking Hamiltonian like \eqref{ch1:Hclock5}, whose norm is $O(L)$. Lloyd \cite{AQC:Lloydunpublished} showed, that the energy gap for such Hamiltonian scales like $L^{-1}$. In summary, when simulating quantum circuits in a HC model, one requires resources $\tau \cdot \norm{H} = O(L^2\log^2 L)$, and is protected by an energy gap scaling like $L^{-1}$, where $L$ is the number of gates in the circuit. This energy barrier to decoherence can be increased, at the cost of increasing the locality of interactions, by using error correcting codes like the ones developed for AQC by Jordan et al. \cite{AQC:codes}.



\section{How to make the Quantum Adiabatic Algorithm Fail}
\label{ch2:local}

After looking at the power of AQC in general and showing how it can be used to simulate any quantum circuit effectively, let me now return to its original intended purpose: solving locally constrained optimization problems. 

Recently, a paper announcing the failure of AQC in certain cases appeared \cite{AQC:Znidaric06}, putting considerable effort into showing that the gap in a specific AQC algorithm is exponentially small. However, the authors make a bad choice of the starting Hamiltonian $H_B$ in \eqref{adiabaticHam}. In this section I prove some general results which show that with certain choices of $H_B$ or $H_P$ the algorithm will not succeed if $T$ is $o(\sqrt{N})$, that is $T/\sqrt{N}\rightarrow 0$ as $N\rightarrow\infty$, so that improvement beyond Grover speedup is impossible. I view these failures as due to poor choices for $H_B$ and $H_P$, which teach us what not to do when looking for good algorithms. I guarantee failure by removing any structure which might exist in $h(z)$ from either $H_B$ or $H_P$. By structure I mean that $z$ is written as a bit string and both $H_B$ and $H_P$ are sums of terms involving only a few of the corresponding qubits. After the preprint \cite{AQC:fail} was posted on the quant-ph archive, two similar papers \cite{AQC:WeiYing06,AQC:MoscaIoannou07} concerning this topic also appeared.

In Section \ref{ch2:projector} I show that regardless of the form of $h(z)$ if $H_B$ is a one dimensional projector onto the uniform superposition of all the basis states $\ket{z}$, then the quantum adiabatic algorithm fails. Here all the $\ket{z}$ states are treated identically by $H_B$ so any structure contained in $h(z)$ is lost in $H_B$. In Section \ref{ch2:scrambled} I consider a scrambled $H_P$ that I get by replacing the cost function $h(z)$ by $h(\pi(z))$ where $\pi$ is a permutation of $0$ to $N-1$. Here the values of $h(z)$ and 
$h(\pi(z))$ are the same but the relationship between input and output is scrambled by the permutation. This effectively destroys any structure in $h(z)$ and typically results in algorithmic failure. 

The quantum adiabatic algorithm is a special case of Hamiltonian based continuous time quantum algorithms, where the quantum state obeys (\ref{Schrodinger1}) and the algorithm consists of specifying $H(t)$, the initial state $\ket{\psi(0)}$, a run time $T$ and the operators to be measured at the end of the run. In the Hamiltonian language, the Grover problem can be recast as the problem of finding the ground state of
\begin{eqnarray}
	H_w = E(\ii-\ket{w}\bra{w}), \label{groverHam}
\end{eqnarray}
where $w$ lies between $0$ and $N-1$. The algorithm designer can apply $H_w$, but in this oracular setting, $w$ is not known. In 
\cite{AQC:FG98analog}
 the following result was proved. Let 
\begin{eqnarray}
	H(t) = H_D(t) + H_w, \label{analogHam}
\end{eqnarray}
where $H_D$ is any time dependent ``driver'' Hamiltonian independent of $w$. Assume also that the initial state $\ket{\psi(0)}$ is independent of $w$. For each $w$ we want the algorithm to be successful, that is $\ket{\psi(T)}=\ket{w}$. It then follows that
\begin{eqnarray}
	T \geq \frac{\sqrt{N}}{2E}. \label{Groverscaling}
\end{eqnarray}
The proof of this result is a continuous-time version of the BBBV oracular proof \cite{BBBV}. My proof techniques in this paper are similar to the methods used to prove the result just stated.


\subsection{General search starting with a one-dimensional projector}
\label{ch2:projector}

In this section I consider a completely general cost function $h(z)$ with $z=0,\dots,N-1$. The goal is to use the quantum adiabatic algorithm to find the ground state of $H_P$ given by (\ref{problemHam}) with $H(t)$ given by (\ref{adiabaticHam}). Let
\begin{eqnarray}
	\ket{s} = \frac{1}{\sqrt{N}} \sum_{z=0}^{N-1} \ket{z}
\end{eqnarray}
be the uniform superposition over all possible values $z$. If I pick
\begin{eqnarray}
	H_B = E (\ii- \ket{s}\bra{s}) \label{groverstart}
\end{eqnarray}
and $\ket{\psi(0)}=\ket{s}$, then the adiabatic algorithm fails in the following sense: 
\begin{theorem}
\label{ch2:t1}
Let $H_P$ be diagonal in the $z$ basis with a ground state subspace of dimension $k$. Let
\[
	H(t) = (1-t/T) E \left( \ii - \ket{s}\bra{s}\right) + (t/T) H_P.
\]
Let $P$ be the projector onto the ground state subspace of $H_P$ and let $b>0$ be the success
probability, that is, $b=\bra{\psi(T)}P\ket{\psi(T)}$. Then
\[
	T \geq \frac{b}{E}\sqrt{\frac{N}{k}} - \frac{2\sqrt{b}}{E}.
\]
\end{theorem}

\begin{proof}
Keeping $H_P$ fixed, introduce $N-1$ additional beginning Hamiltonians as follows.
For $x=0,\dots,N-1$ let $V_x$ be a unitary operator diagonal in the $z$ basis with
\[
	\bra{z}V_x\ket{z} = e^{2\pi i z x /N}
\]
and let
\[
	\ket{x} = V_x \ket{s} = \frac{1}{\sqrt{N}} \sum_{z=0}^{N-1} e^{2\pi i z x /N}\ket{z}
\]
so that the $\{\ket{x}\}$ form an orthonormal basis. Note also that
\[
	\ket{x=0}=\ket{s}.
\]
We now define 
\[
  H_x(t) = ( 1-t/T ) E (\ii-\ket{x}\bra{x}) + (t/T) H_P,
\]
with corresponding evolution operator $U_x (t_2,t_1)$. 
Note that $H(t)$ above is $H_0(t)$ with the corresponding evolution operator $U_0$. For each $x$ let the system 
evolve with $H_x(t)$ from the ground state of $H_x(0)$, which is $\ket{x}$. Note that 
$H_x=V_x H_0 V_x^{\dagger}$ and $U_x=V_x U_0 V_x^{\dagger}$. Let $\ket{f_x}=U_x(T,0)\ket{x}$. For
each $x$ the success probability is $\bra{f_x}P\ket{f_x}$, which is equal to $b$ since $P$ commutes
with $V_x$.
The key point is that if I run the Hamiltonian evolution with $H_x$ backwards in time, I would then be finding $x$, that is, solving the Grover problem. However, this should not be possible unless the run time $T$ is of order $\sqrt{N}$. 

Let $U_R$ be the evolution operator corresponding to an $x$-independent reference Hamiltonian
\[
  H_R(t) = (1-t/T) E + (t/T) H_P.
\]
Let $\ket{g_x} = \frac{1}{\sqrt{b}}P\ket{f_x}$ be the normalized component of $\ket{f_x}$ in the
ground state subspace of $H_P$. Consider the difference in backward evolution from $\ket{g_x}$
with Hamiltonians $H_x$ and $H_R$, and sum on $x$,
\[
  S(t) = \sum_x \norm{U_x^{\dagger}(T,t)\ket{g_x}-U_R^{\dagger}(T,t)\ket{g_x}}^2.
\]
Clearly $S(T)=0$, and
\begin{eqnarray*}
  S(0) &=& \sum_x \norm{ U_x^{\dagger}(T,0)\ket{g_x} - U_R^{\dagger}(T,0)\ket{g_x} }^2.
\end{eqnarray*}
Now $\ket{g_x}=\sqrt{b}\ket{f_x}+\sqrt{1-b}\ket{f_x^{\perp}}$ where $\ket{f_x^{\perp}}$ is orthogonal
to $\ket{f_x}$. Since $U_x^{\dagger}(T,0)\ket{f_x}=\ket{x}$,
\begin{eqnarray*}
       S(0) = \sum_x \norm{ \sqrt{b}\ket{x} + \sqrt{1-b}\ket{x^\perp} - \ket{i_x} }^2,
\end{eqnarray*}
where for each $x$, $\ket{x^\perp}$ and $\ket{i_x}$ are normalized states with $\ket{x^\perp}$ orthogonal 
to $\ket{x}$. Since $H_R$ commutes with $H_P$, $\ket{i_x}=U_R^{\dagger}(T,0)\ket{g_x}$ 
is an element of the $k$-dimensional ground state subspace of $H_P$. Then
\begin{eqnarray*}
  S(0) &=& 2N - \sum_x \left[ \sqrt{b} \scalar{x}{i_x} + \sqrt{1-b}\langle x^\perp | i_x\rangle+c.c.\right] \\
	&\geq& 2N - 2\sqrt{b} \sum_x \bigabs{\scalar{x}{i_x}} - 2N\sqrt{1-b}.
\end{eqnarray*}
Choosing a basis $\{\ket{G_j}\}$ for the $k$ dimensional ground state subspace of $H_P$ and writing 
$\ket{i_x} = a_{x1}\ket{G_1} + \cdots + a_{xk}\ket{G_k}$ gives
\begin{eqnarray}
  \sum_x \bigabs{\scalar{x}{i_x}} 
  &\leq& \sum_{x,j} \left|a_{xj}\right| \cdot \bigabs{\scalar{x}{G_j}} \label{boundix} \\
  &\leq& \sqrt{ \sum_{x,j} \left|a_{xj}\right|^2 \sum_{x',j'} \bigabs{\scalar{x'}{G_{j'}}}^2 } 
	= \sqrt{Nk}. \nonumber
\end{eqnarray}
Thus
\begin{eqnarray}
  S(0)\geq 2N(1-\sqrt{1-b})- 2\sqrt{b}\sqrt{Nk}. \label{s0eqn}
\end{eqnarray}
Using the Schr\"{o}dinger equation, one can find the time derivative of $S(t)$:
\begin{eqnarray*}
  \deriv{}{t}S(t) &=& -\sum_x \deriv{}{t} \left[ \bra{g_x}U_x(T,t)U_R^{\dagger}(T,t) \ket{g_x} + c.c. \right]  \\
	&=& -i \sum_x \bra{g_x}U_x(T,t)[H_x(t)-H_R(t)]U_R^{\dagger}(T,t) \ket{g_x} + c.c. \\
	&=& -2\,\textrm{Im} \sum_x (1-t/T)E \bra{g_x}U_x(T,t) \ket{x}\bra{x} U_R^{\dagger}(T,t) \ket{g_x}.
\end{eqnarray*}
Now 
\begin{eqnarray*}
	\left| \deriv{}{t} S(t) \right| &\leq& 2E (1-t/T) \sum_x
		\left| \bra{g_x}U_x(T,t) \ket{x}\bra{x} U_R^{\dagger}(T,t) \ket{g_x} \right| \\
	&\leq& 2E (1-t/T) \sum_x
		\left| \bra{x} U_R^{\dagger}(T,t) \ket{g_x} \right|.
\end{eqnarray*}
The same technique as in (\ref{boundix}) gives
\begin{eqnarray*}
	\left| \deriv{}{t} S(t) \right| &\leq& 2E (1-t/T) \sqrt{Nk}.
\end{eqnarray*}
Therefore
\begin{eqnarray*}
	\int^T_0 \left| \deriv{}{t} S(t) \right| \textrm{d}t \leq ET \sqrt{Nk}.
\end{eqnarray*}
Now $S(0)\leq S(T) + \int^T_0 \left| \deriv{}{t} S(t) \right| \textrm{d}t$ and $S(T)=0$ so 
\[
	S(0)\leq ET\sqrt{Nk}.
\]
Combining this with (\ref{s0eqn}) gives
\[
   ET\sqrt{Nk} \geq 2N (1-\sqrt{1-b}) - 2\sqrt{b} \sqrt{Nk},
\]
which implies what I wanted to prove:
\[
   T \geq \frac{b}{E}\sqrt{\frac{N}{k}} - \frac{2\sqrt{b}}{E}.
\]
\end{proof}

How to interpret Theorem \ref{ch2:t1}? The goal is to find the minimum of the cost function $h(z)$ using the quantum adiabatic algorithm. It is natural to pick for $H_B$ a Hamiltonian whose ground state is $\ket{s}$, the uniform superposition of all $\ket{z}$ states. However
if I pick $H_B$ to be the one dimensional projector $E(\ii-\ket{s}\bra{s})$ the algorithm will not find the ground state if $T/\sqrt{N}$ goes to $0$ as $N$ goes to infinity. The problem is that $H_B$ has no structure and makes no reference to $h(z)$. My hope is that the algorithm might be useful for interesting computational problems if $H_B$ has structure that reflects the form of $h(z)$. 

Note that Theorem \ref{ch2:t1} explains the algorithmic failure discovered by 
\v{Z}nidari\v{c} and Horvat \cite{AQC:Znidaric06} for a particular set of $h(z)$.

For a simple but convincing example of the importance of the choice of $H_B$, suppose I take a decoupled $n$ bit problem which consists of $n$ clauses each acting on one bit, say for each bit $j$ 
\begin{eqnarray*}
  h_j(z)=\left\{
    \begin{array}{rl}
	0 & \quad \textrm{if} \, z_{j}=0, \\
	1 & \quad \textrm{if} \, z_{j}=1,
    \end{array}
  \right. 
\end{eqnarray*}
so 
\begin{eqnarray}
	h(z)=z_1+z_2+\dots+z_n  \label{decoupledcost}.
\end{eqnarray}
Let me pick a beginning Hamiltonian reflecting the bit structure of the problem,
\begin{eqnarray}
	H_B = \sum_{j=1}^{n} \frac{1}{2}\left(1-\sigma_x^{(j)}\right). \label{sumSX}
\end{eqnarray}
The ground state of $H_B$ is $\ket{s}$,
The quantum adiabatic algorithm acts on each bit independently, producing a success probability of
\[
	p = \left(1-q(T)\right)^n,
\]
where $q(T)\rightarrow 0$ as $T\rightarrow\infty$ is the transition probability between the ground state and the excited state of a single qubit. As long as $n q(T)\rightarrow const.$ I have a constant probability of success. This can be achieved for $T$ of order $\sqrt{n}$, because for a two level system with a nonzero gap, the probability of a transition is $q(T)=O(T^{-2})$. (For details, see Appendix \ref{appTwolevel}.) However, I know from Theorem \ref{ch2:t1} that a poor choice of $H_B$ would make the quantum adiabatic algorithm fail on this simple decoupled $n$ bit problem by destroying the bit structure. 

Next, suppose the satisfiability problem I am trying to solve has clauses involving say 3 bits. If clause $c$ involves bits $i_c$, $j_c$ and $k_c$ I may define the clause cost function
\begin{eqnarray*}
  h_c(z)=\left\{
    \begin{array}{rl}
	0 & \quad \textrm{if} \,\, z_{i_c},z_{j_c},z_{k_c} \, \textrm{satisfy clause} \, c, \\
	1 & \quad \textrm{otherwise}.
    \end{array}
  \right.
\end{eqnarray*}
The total cost function is then
\begin{eqnarray*}
  h(z) = \sum_c h_c(z).
\end{eqnarray*}
To get $H_B$ to reflect the bit and clause structure I may pick
\begin{eqnarray*}
  H_{B,c} = \frac{1}{2}\left[(1-\sigma_x^{(i_c)}) + (1-\sigma_x^{(j_c)}) + (1-\sigma_x^{(k_c)}) \right]
\end{eqnarray*}
with
\begin{eqnarray}
  H_B = \sum_c H_{B,c}. \label{clauseham}
\end{eqnarray}
In this case the ground state of $H_B$ is again $\ket{s}$. With this setup, Theorem \ref{ch2:t1} does not apply.

\begin{figure}
	\begin{center}
	\includegraphics[width=4.5in]{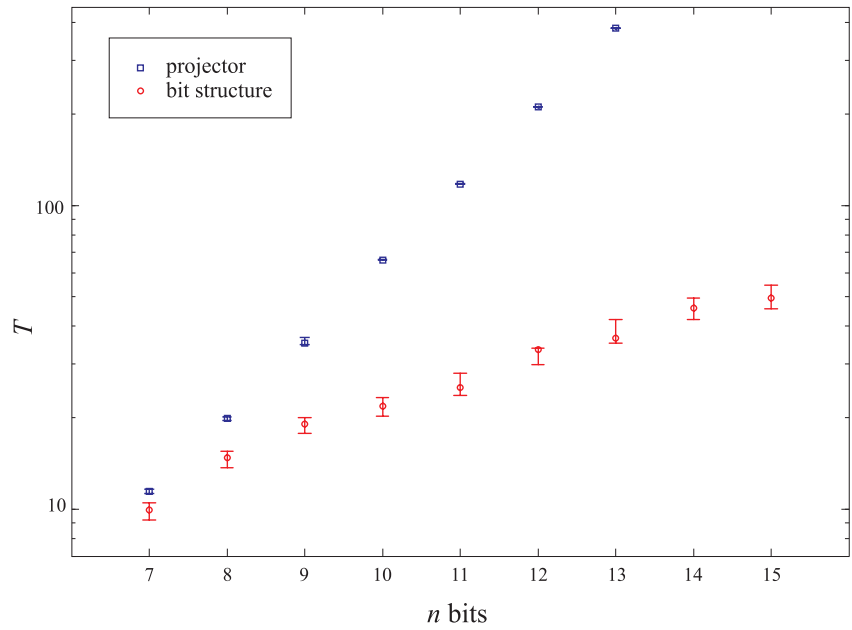}
	\end{center}
	\caption{Median required run time $T$ versus bit number. At each bit number there are 50 random instances of Exact Cover with a single satisfying assignment. I choose the required run time to be the value of $T$ for which quantum adiabatic algorithm has success probability between 0.2 and 0.21.	 For the projector beginning Hamiltonian I use \eqref{groverstart} with $E=n/2$.	 The plot is log-linear. The error bars show the 95\% confidence interval for the true medians.}
	\label{loglin}
\end{figure}

I did a numerical study of a particular satisfiability problem, Exact Cover. For this problem if clause $c$ involves bits $i_c$, $j_c$ and $k_c$, the cost function is 
\begin{eqnarray*}
  h_c(z)=\left\{
    \begin{array}{rl}
	0 & \quad \textrm{if} \,\, z_{i_c}+z_{j_c}+z_{k_c}=1, \\
	1 & \quad \textrm{otherwise}.
    \end{array}
  \right.
\end{eqnarray*}
Some data is presented in FIG. 1. Here we see that with a structured beginning Hamiltonian the required run times are substantially lower than with the projector $H_B$.


\subsection{Search with a scrambled problem hamiltonian}
\label{ch2:scrambled}

In the previous section I showed that removing all structure from $H_B$ dooms the quantum adiabatic algorithm to failure. In this section I remove structure from the problem to be solved ($H_P$) and show that this leads to algorithmic failure. Let $h(z)$ be a cost function whose minimum I seek. Let $\pi$ be a permutation of $0,1,\dots,N-1$ and let
\[
	h^{[\pi]}(z)=h\left(\pi^{-1}(z)\right).
\]
I will show that no continuous time quantum algorithm (of a very general form) can find the minimum
of $h^{[\pi]}$ for even a small fraction of all $\pi$ if $T$ is $o(\sqrt{N})$. Classically, this problem
takes order $N$ calls to an oracle.

Without loss of generality let $h(0)=0$, and $h(1),h(2),\dots,h(N-1)$ all be positive. For any permutation $\pi$ of $0,1,\dots,N-1$, define a problem Hamiltonian $H_{P,\pi}$, diagonal in the $z$ basis, as 
\[
	H_{P,\pi} = \sum_{z=0}^{N-1} h^{[\pi]}(z) \ket{z}\bra{z} 
		=\sum_{z=0}^{N-1} h(z) \ket{\pi(z)}\bra{\pi(z)}.
\]
Now consider the Hamiltonian 
\begin{eqnarray}
	H_{\pi}(t)=H_D(t)+c(t) H_{P,\pi} \label{continuoustimealgorithm}
\end{eqnarray}
for an arbitrary $\pi$-independent driving Hamiltonian
$H_D(t)$ with $|c(t)|\leq1$ for all $t$. Using this composite Hamiltonian, 
evolve the $\pi$-independent starting state $\ket{\psi(0)}$ for time $T$, reaching the state $\ket{\psi_\pi(T)}$. This setup is more general than the quantum adiabatic algorithm since I do not require $H_D(t)$ or $c(t)$ to be slowly varying. Success is achieved if the overlap of $\ket{\psi_{\pi}(T)}$ 
with $\ket{\pi(0)}$ is large.

I first show
\begin{lemma}
\begin{eqnarray}
	\sum_{\pi,\pi'} \Big\| \ket{\psi_\pi(T)}-\ket{\psi_{\pi'}(T)} \Big\|^2 \leq 4h^* T N!\sqrt{N-1}, \label{permutationlemma}
\end{eqnarray}
where the sum is over all pairs of permutations $\pi,\pi'$ that differ by a single transposition 
involving $\pi(0)$, and $h^*=\sqrt{\sum h(z)^2 / (N-1)}$. 
\end{lemma}

\begin{proof}
For two different permutations $\pi$ and $\pi'$ let $\ket{\psi_{\pi}(t)}$ be the state obtained
by evolving from $\ket{\psi(0)}$ with $H_{\pi}$ and let $\ket{\psi_{\pi'}(t)}$ be the state obtained
by evolving from $\ket{\psi(0)}$ with $H_{\pi'}$.

Now
\begin{eqnarray*}
  \deriv{}{t} \Big\| \ket{\psi_\pi(t)}-\ket{\psi_{\pi'}(t)} \Big\|^2 
	&=& -\deriv{}{t} \scalar{\psi_{\pi}(t)}{\psi_{\pi'}(t)} + c.c. \\
	&=& i \bra{\psi_{\pi}(t)}(H_{\pi}(t)-H_{\pi'}(t))\ket{\psi_{\pi'}(t)} + c.c. \\
	&\leq& 2 \Big| \bra{\psi_{\pi}(t)}(H_{\pi}(t)-H_{\pi'}(t))\ket{\psi_{\pi'}(t)} \Big|.
\end{eqnarray*}
Consider the case when $\pi$ and $\pi'$ differ by a single transposition involving $\pi(0)$. 
Specifically, $\pi'=\pi \circ (a\leftrightarrow 0)$ for some $a$. 
Now if $\pi(0) = i$ and $\pi(a) =j$, we have $\pi'(0) =j$ and $\pi'(a) =i$. 
Therefore, since $h(0)=0$, 
\begin{eqnarray*}
  H_{P,\pi}-H_{P,\pi'} = c(t) h(a) \left( \ket{j}\bra{j} - \ket{i}\bra{i} \right)
			= c(t) h(a) \left( \ket{\pi(a)}\bra{\pi(a)} - \ket{\pi'(a)}\bra{\pi'(a)} \right),
\end{eqnarray*}
so that
\begin{eqnarray*}
  \deriv{}{t} \sum_{\pi,\pi'} \Big\| \ket{\psi_\pi(t)}-\ket{\psi_{\pi'}(t)} \Big\|^2 
  &\leq&
	2 |c(t)| \sum_{\pi,\pi'} h(a) \Big| \bra{\psi_{\pi}(t)} 
		\left( \ket{\pi(a)}\bra{\pi(a)} - \ket{\pi'(a)}\bra{\pi'(a)} \right)  
		\ket{\psi_{\pi'}(t)} \Big|.
\end{eqnarray*}
This further simplifies to
\begin{eqnarray*}
  \deriv{}{t} \sum_{\pi,\pi'} \Big\| \ket{\psi_\pi(t)}-\ket{\psi_{\pi'}(t)} \Big\|^2
  &\leq&
	2 \sum_{\pi,\pi'} h(a) \left(
		\bigabs{\scalar{\psi_{\pi}(t)}{\pi(a)}} +
		\bigabs{\scalar{\pi'(a)}{\psi_{\pi'}(t)}} \right) \\
  &=&
	2 \sum_{\pi} \sum_{a\neq 0} h(a) 
		\bigabs{\scalar{\psi_{\pi}(t)}{\pi(a)}} +
	2 \sum_{\pi'} \sum_{a\neq 0} h(a) 
		\bigabs{\scalar{\pi'(a)}{\psi_{\pi'}(t)}}  \\
  &=&
	4 \sum_{\pi} \sum_{a\neq 0} h(a) \bigabs{\scalar{\psi_{\pi}(t)}{\pi(a)}} \\
  &=&
	4 \sum_{\pi} \sum_{a} h(a) \bigabs{\scalar{\psi_{\pi}(t)}{\pi(a)}} \\
  &\leq&
	4 \sum_{\pi} \sqrt{\sum_a h(a)^2} = 4h^* N! \sqrt{N-1} .
\end{eqnarray*}
where I used the Cauchy-Schwartz inequality to obtain the last line. Integrating this inequality
for time $T$, I obtain the result I wanted to prove,
\begin{eqnarray*}
  \sum_{\pi,\pi'} \Big\| \ket{\psi_\pi(T)}-\ket{\psi_{\pi'}(T)} \Big\|^2 &\leq& 4 h^* T N! \sqrt{N-1},
\end{eqnarray*}
where the sum is over $\pi$ and $\pi'$ differing by a single transposition involving $\pi(0)$.
\end{proof}

Next I establish
\begin{lemma}
  Suppose $\ket{1}$, $\ket{2}$, $\ket{L}$ are orthonormal vectors and $\bigabs{\scalar{\psi_i}{i}}^2 \geq b$ for 
normalized vectors $\ket{\psi_i}$, where $i=1,\dots,L$. Then for any normalized $\ket{\varphi}$,
\begin{eqnarray}
   \sum_{i=1}^{L} \bignorm{\ket{\psi_i}-\ket{\varphi}}^2 \geq bL-2\sqrt{L}. \label{vectorlemma}
\end{eqnarray}
\end{lemma}
\begin{proof}
Write
\begin{eqnarray*}
  \sum_i \bignorm{\ket{\psi_i}-\ket{\varphi}}^2 &\geq& \sum_i \bigabs{\scalar{i}{\psi_i}-\scalar{i}{\varphi}}^2 \\
	&\geq& \sum_i \bigabs{\scalar{i}{\psi_i}}^2 - 2 \sum_i \bigabs{\scalar{i}{\psi_i}} \bigabs{\scalar{i}{\varphi}}
\end{eqnarray*}
and use the Cauchy-Schwartz inequality to obtain
\begin{eqnarray*}
  \sum_i \bignorm{\ket{\psi_i}-\ket{\varphi}}^2 &\geq& b L 
	- 2 \sqrt{\sum_i \bigabs{\scalar{i}{\psi_i}}^2 } \sqrt{\sum_i \bigabs{\scalar{i}{\varphi}}^2} \\
    &\geq& bL - 2 \sqrt{L}.
\end{eqnarray*}
\end{proof}

Let me now state the main result of this Section. 
\begin{theorem}
\label{ch2:t2}
  Suppose that a continuous time algorithm of the form \eqref{continuoustimealgorithm} 
  succeeds with probability at least $b$, i.e. 
$\bigabs{\scalar{\psi_{\pi}(T)}{\pi(0)}}^2\geq b$, for a set of $\epsilon N!$ permutations.
Then
\begin{eqnarray}
   T\geq \frac{\epsilon^2 b}{16 h^*} \sqrt{N-1} - \frac{\epsilon\sqrt{\epsilon/2}}{4h^*}.
\end{eqnarray}
\end{theorem}
\begin{proof}

For any permutation $\pi$, there are $N-1$ permutations $\pi'_a$ obtained from $\pi$ by first transposing $0$ and $a$. For each $\pi$ let $\mathcal{S}_{\pi}$ be the subset of those $N-1$ permutations on which the algorithm succeeds with probability at least $b$.  Any such permutation appears in exactly $N-1$ of the sets $\mathcal{S}_{\pi}$, therefore
\[
	\sum_\pi \left|\mathcal{S}_\pi\right| = (N-1)\epsilon N!.
\]
Let $M$ be the number of sets $\mathcal{S}_{\pi}$ with $\left|\mathcal{S}_{\pi}\right|\geq\frac{\epsilon}{2}(N-1)$. Now
\begin{eqnarray*}
  \sum_{\pi} \left| \mathcal{S}_{\pi} \right| &=& \sum_{\left| \mathcal{S}_{\pi} \right|\geq \frac{\epsilon}{2}(N-1)} \left| \mathcal{S}_{\pi} \right|
	+ \sum_{\left| \mathcal{S}_{\pi} \right|< \frac{\epsilon}{2}(N-1)} \left| \mathcal{S}_{\pi} \right| \\
    \sum_{\pi} \left| \mathcal{S}_{\pi} \right| &\leq& M(N-1) + (N!-M)\frac{\epsilon}{2}(N-1), \\
  (N-1)\epsilon N! &\leq& M(N-1) + N!\frac{\epsilon}{2}(N-1),
\end{eqnarray*}
so $M\geq\frac{\epsilon}{2}N!$, i.e. at least 
$\frac{\epsilon}{2} N!$ of the sets $\mathcal{S}_{\pi}$ must contain 
at least $\frac{\epsilon}{2}(N-1)$ permutations on which the algorithm succeeds with probability at least $b$. For the corresponding $\pi$, I have
\[
   \sum_{\pi'_a} \bignorm{\ket{\psi_\pi (T)}-\ket{\psi_{\pi'_a} (T)}}^2  
	\geq b \frac{\epsilon}{2}(N-1) - 2\sqrt{\frac{\epsilon}{2}(N-1)}.
\]
by Lemma 2. (Note that the algorithm is not assumed to succeed with probability $b$ on $\pi$.) Since there are at least $\frac{\epsilon}{2}N!$ such $\pi$, 
\begin{eqnarray*}
  \sum_{\pi,\pi'} \bignorm{\ket{\psi_\pi (T)}-\ket{\psi_{\pi'}(T)}}^2 
	\geq \frac{\epsilon}{2}N! \left(b \frac{\epsilon}{2}(N-1) - 2\sqrt{\frac{\epsilon}{2}(N-1)}\right),
\end{eqnarray*}
where the sum is over all permutations $\pi$ and $\pi'$ which differ by a single transposition involving $\pi(0)$. Combining this with Lemma 1 I obtain
\begin{eqnarray*}
   T\geq \frac{\epsilon^2 b}{16 h^*} \sqrt{N-1} - \frac{\epsilon\sqrt{\epsilon/2}}{4 h^*},
\end{eqnarray*}
which is what I wanted to prove.
\end{proof}

What I have just shown is that no continuous time algorithm of the form \eqref{continuoustimealgorithm} can find the minimum of $H_{P,\pi}$ with a constant success probability for even a fraction $\epsilon N!$ of all permutations $\pi$ if $T$ is $o(\sqrt{N})$. 
A typical permutation $\pi$ yields an $H_{P,\pi}$ with no structure relevant to any fixed $H_D$ and the algorithm cannot find the ground state of $H_{P,\pi}$ efficiently. 

\begin{figure}
	\begin{center}
	\includegraphics[width=4.5in]{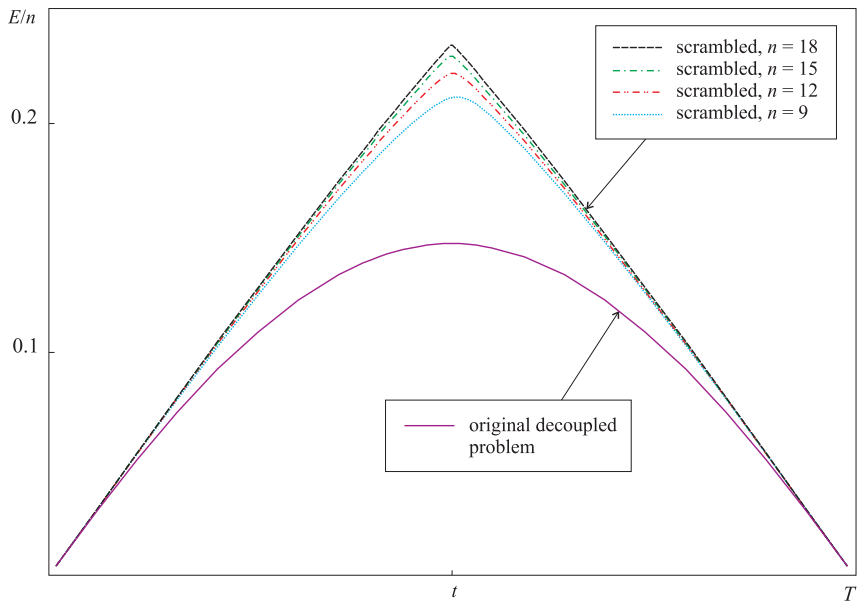}
	\end{center}
	\caption{The scaled ground state energy $E/n$ for a quantum adiabatic algorithm Hamiltonian of a decoupled problem. The lowest curve corresponds to the original decoupled problem. The upper ``triangular'' curves correspond to single instances of the $n$-bit decoupled problem, where the problem Hamiltonian was scrambled.}
	\label{triangle}
\end{figure}

To illustrate the nature of this failure for the quantum adiabatic algorithm for a typical permutation, consider again the decoupled
$n$ bit problem with $h(z)$ given by \eqref{decoupledcost} and $H_B$ given by \eqref{sumSX}. The lowest curve in FIG. 2 shows the ground state energy divided by $n$ as a function of $t$. (Since the system is decoupled this is actually the ground state energy of a single qubit.) 
I then consider the $n$ bit scrambled problem for different values of $n$. At each $n$ I pick a single random permutation $\pi$ of $0,\dots,(2^n-1)$ and apply it to obtain a cost function $h(\pi^{-1}(z))$ while keeping $H_B$ fixed. The ground state energy divided by $n$ is now plotted for $n=9,12,15$ and $18$. From these scrambled problems it is clear that if I let $n$ get large the typical curves will approach a triangle with a discontinuous first derivative at $t=T/2$. For large $n$, the ground state changes dramatically as $t$ passes through $T/2$. In order to keep the quantum system in the ground state we need to go very slowly near $t=T/2$ and this results in a long required run time.


\subsection{Summary}

I have shown two main results about the performance of the quantum adiabatic algorithm when used to find the minimum of a classical cost function $h(z)$ with $z=0,\dots,N-1$. Theorem \ref{ch2:t1} says that for any cost function $h(z)$, if the beginning Hamiltonian is a one dimensional projector onto the uniform superposition of all the $\ket{z}$ basis states, the algorithm will not find the minimum of $h$ if $T$ is less then of order $\sqrt{N}$. This is true regardless of how simple it is to classically find the minimum of $h(z)$.

In Theorem \ref{ch2:t2} I start with any beginning Hamiltonian and classical cost function $h$. Replacing $h(z)$ by a scrambled version, i.e. $h^{[\pi]}(z)=h(\pi(z))$ with $\pi$ a permutation of $0$ to $N-1$, will make it impossible for the algorithm to find the minimum of $h^{[\pi]}$ in time less than order $\sqrt{N}$ for a typical permutation $\pi$. For example suppose we have a cost function $h(z)$ and have chosen $H_B$ so that the quantum algorithm finds the minimum in time of order $\textrm{log}\,N$. Still scrambling the cost function results in algorithmic failure.

These results do not imply anything about the more interesting case where $H_B$ and $H_P$ are structured, i.e., sums of terms each operating only on several qubits.


\chapter{Matrix Product States \\on Infinite Trees}\label{ch3mps}

When can one actually simulate local Hamiltonians or find their ground states efficiently? It is possible for some Hamiltonians in 1D and in tree geometry, by using an approximate method based on Matrix Product States. This is what I show in this Chapter, based on the paper \cite{MPS:infitrees}

\mypapertwoline{The Quantum Transverse Field Ising Model on an Infinite Tree}{from Matrix Product States}{Daniel Nagaj, Edward Farhi, Jeffrey Goldstone, Peter Shor and Igor Sylvester}
{We give a generalization to an infinite tree geometry of Vidal's infinite time-evolving block decimation (iTEBD) algorithm \cite{MPS:Vidal1Dinfinite} for simulating an infinite line of quantum spins. We numerically investigate the quantum Ising model in a transverse field on the Bethe lattice using the Matrix Product State ansatz. We observe a second order phase transition, with certain key differences from the transverse field Ising model on an infinite spin chain. We also investigate a transverse field Ising model with a specific longitudinal field. When the transverse field is turned off, this model has a highly degenerate ground state as opposed to the pure Ising model whose ground state is only doubly degenerate.}

This Chapter is organized as follows. Section \ref{MPSsection} is a review of the MPS ansatz and contains its generalization to the tree geometry. In Section \ref{updatesection}, I review the numerical procedure for unitary updates and give a recipe for applying imaginary time evolution within the MPS ansatz. In Section \ref{TIsection}, I 
adapt Vidal's iTEBD method for simulating translationally invariant one-dimensional systems to systems with tree geometry. Section \ref{NRsection} contains my numerical results for the quantum Ising model in a transverse field for translationally invariant systems. In Section \ref{NRsection1}, I first test my method for the infinite line, and then in Section \ref{NRsection2} I present new results for the infinite tree. I turn to the \textsc{not} 00 model in Section \ref{not00section} and show that my numerics work well for this system even when there is a high ground state degeneracy. In Section \ref{stabilitysection} I investigate the stability of my tree results and conjecture that they may be good approximations to a local description far from the boundary of a large finite tree system.


\section{Introduction}

The matrix product state (MPS) description \cite{MPS:OstlundRommerMPS1}\cite{MPS:OstlundRommerMPS2} has brought a new way of approaching many-body quantum systems. Several methods of investigating spin systems have been developed recently combining state of the art many-body techniques 
such as White's Density Matrix Renormalization Group \cite{MPS:DMRG}\cite{MPS:DMRGreview} (DMRG) with quantum information motivated insights. Vidal's Time Evolving Block Decimation (TEBD) algorithm 
\cite{MPS:VidalMPS} \cite{MPS:VidalMPS2} uses MPS and emphasizes entanglement (as measured by the Schmidt number), directing the computational resources into that bottleneck of the simulation. It provides the ability to simulate time evolution and it was shown that MPS-inspired methods handle periodic boundary conditions well in one dimension \cite{MPS:DMRGandPBC}, areas where the previous use of DMRG was limited. TEBD has been recast into the language of DMRG in \cite{MPS:TEBDandDMRG1}\cite{MPS:TEBDandDMRG2} and adapted to finite systems with tree geometry in \cite{MPS:ShiTrees}. DMRG is especially successful in describing the properties of quantum spin chains, the application of basic DMRG-like methods is limited for quantum systems with higher dimensional geometry. New methods like PEPS \cite{MPS:PEPS1}\cite{MPS:PEPS2} generalize MPS to higher dimensions, opening ways to numerically investigate systems that were previously inaccessible.

Here I am interested in investigating infinite translationally invariant systems. Several numerical methods to investigate these were developed recently. The iTEBD algorithm \cite{MPS:Vidal1Dinfinite} (see also Sec.\ref{TIsection}) is a generalization of TEBD to infinite one-dimensional systems. A combination of PEPS with iTEBD called iPEPS \cite{MPS:iPEPS} provides a possibility of investigating infinite translationally invariant systems in higher dimensions.
\begin{figure}
	\begin{center}
	\includegraphics[width=1.1in]{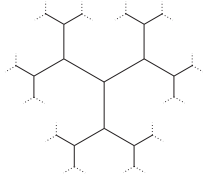}
	\caption{The Bethe lattice (infinite Cayley tree).}
	\label{figinfitree}
	\end{center}
\end{figure}
My contribution is a method to investigate the ground state properties of infinite translationally invariant quantum systems on the Bethe lattice using imaginary time evolution with Matrix Product States. The Bethe lattice is an infinite tree with each node having three neighbors, as depicted in Fig.\ref{figinfitree}. It is translationally invariant in that it looks the same at every vertex.
This geometry is interesting, because of the following connection to large random graphs with fixed valence. Moving out from any vertex in such a random graph, you need to go a distance of order $\textrm{log}\,n$, where $n$ is the number of vertices in the graph, before you detect that you are not on the Bethe lattice, that is, before you see a loop.

I choose to investigate the quantum transverse field Ising model on the Bethe lattice. Note that I work directly on the infinite system, never taking a limit. First I test the iTEBD method on a system with a known exact solution, the infinite line. Then I turn to the Bethe lattice with the new method I provide. In both cases, the Hamiltonian is given by
\begin{eqnarray} 
	H = \frac{J}{2} \sum_{\langle i,j\rangle} (1-\sigma_z^{i} \sigma_z^{j}) 
		+ \frac{h}{2} \sum_i \left(1-\sigma_x^{i}\right), \label{ourH1}
\end{eqnarray}
where the sum over $i$ is over all sites, and the sum over $\langle i,j \rangle$ is over all bonds (nearest neighbors).
I show that imaginary time evolution within the MPS ansatz provides a very good approximation for the exact ground state on an infinite line, resulting in nearly correct critical exponents for the magnetization and correlation length as one approaches the phase transition. I obtain new results for the quantum Ising model in transverse field on the infinite tree. Similarly to the infinite line, I observe a second order phase transition and obtain the critical exponent for the magnetization, $\beta_{T}\approx 0.41$ (different than the mean-field result). However, the correlation length does not diverge at the phase transition for this system and I conjecture that it has the value $1/\ln 2$. 

I also investigate a model where besides an antiferromagnetic interaction of spins I add a specific longitudinal field $\frac{1}{4}\sigma_z^i$ for each spin: 
\begin{eqnarray}
	H_{\textsc{not}\,00} 
	&=& J 
	 \sum_{\langle i,j\rangle} 
		\frac{1}{4}\left( 1 + \sigma_z^i + \sigma_z^j + \sigma_z^i \sigma_z^j  \right)
		+ \frac{h}{2} \sum_i \left(1-\sigma_x^{i}\right).
		\label{ourH00}
\end{eqnarray}
I choose the longitudinal field in such a way that the interaction term in the computational basis takes a simple form, $\ket{00}\bra{00}_{ij}$, giving an energy penalty to the $\ket{00}$ state of neighboring spins. (I follow the usual convention that spin up in the $z$-direction is called 0.) I call it the \textsc{not} 00 model accordingly. This model is interesting from a computational viewpoint. 
The degeneracy of the ground state of $H_{\textsc{not}\,00}$ at $h=0$ is high for both infinite line and infinite tree geometry of interactions. I am interested in how my numerical method deals with this case, as opposed to the double degeneracy of the ground state of \eqref{ourH1} at $h=0$. I do not see a phase transition in this system as I vary $J$ and $h$. 


\section{Matrix Product States} 
\label{MPSsection}

If one's goal is to numerically investigate a system governed by a local Hamiltonian, it is convenient to find a local description and update rules for the system. A Matrix Product State description is particularly suited to spin systems for which the connections do not form any loops. Given a state of this system, I will first show how to obtain its MPS description, and then how to utilize this description in a numerical method for obtaining the time evolution and approximating the ground state (using imaginary time evolution). I begin with matrix product states on a line (a spin chain), and then generalize the description to a tree geometry. In \ref{updatesection}, I give a numerical method of updating the MPS description for both real and imaginary time simulations.

\subsection{MPS for a spin chain}

Given a state $\ket{\psi}$ of a chain of $n$ spins
\begin{eqnarray}
  \ket{\psi} &=& \sum_{\dots s_i s_{i+1} \dots}
  c_{\dots,s_i,s_{i+1},\dots}
  \ket{s_1}_1 \dots \ket{s_i}_i \ket{s_{i+1}}_{i+1} \dots \ket{s_n}_n, 
  \label{psiMPS}
\end{eqnarray}
I wish to rewrite the coefficients $c_{s_1,\dots,s_n}$ as a matrix product (see \cite{MPS:MPSreview} for a review of MPS)
\begin{eqnarray}
  c_{\dots,s_i,s_{i+1},\dots} &=& 
  \sum_{\dots a b c \dots} \dots \lambda^{(i-1)}_a 
  \Gamma^{(i),s_i}_{a,b} \lambda^{(i)}_b 
  \Gamma^{(i+1),s_{i+1}}_{b,c} \lambda^{(i+1)}_c
  \dots 
  \label{cMPS}
\end{eqnarray}
using $n$ tensors $\Gamma^{(i)}$ and $n-1$ vectors $\lambda^{(i)}$. 
The range of the indices $a,b,\dots$ will be addressed later. 
After decomposing the chain into two subsystems, one can rewrite the state of the whole system
in terms of orthonormal bases of the subsystems. 
$\lambda^{(i)}$ is the vector of Schmidt coefficients 
for the decomposition of the state of the chain onto the subsystems $1\dots i$ and $i+1 \dots n$. 

In order to obtain the $\lambda$'s and the $\Gamma$'s for a given state $\ket{\psi}$, one has to perform the following steps. First, perform the Schmidt decomposition of the chain between sites $i-1$ and $i$ as
\begin{eqnarray}
  \ket{\psi} = \sum_{a=1}^{\chi_{i-1}} \ket{\phi_a}_{1,\dots,i-1} 
  \lambda^{(i-1)}_a \ket{\phi_a}_{i,\dots,n},
\end{eqnarray}
where the states on the left and on the right of the division form
orthonormal bases required to describe the respective subsystems 
of the state $\ket{\psi}$.
The number $\chi_{i-1}$ (the Schmidt number) is the minimum number of terms 
required in this decomposition. 
\begin{figure}
	\begin{center}
	\includegraphics[width=3.3in]{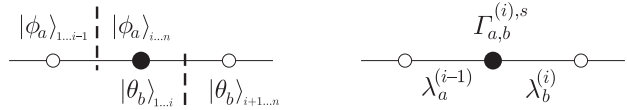}
	\caption{Two successive Schmidt decompositions on a line allow me to find the $\Gamma$
	tensor for the marked site and the two $\lambda$ vectors for the bonds coming out of it.}
	\label{figlineschmidt}
	\end{center}
\end{figure}
The Schmidt decomposition for a split between sites $i$ and $i+1$ gives
\begin{eqnarray}
  \ket{\psi} = \sum_{b=1}^{\chi_{i}} \ket{\theta_b}_{1,\dots,i} 
  \lambda^{(i)}_b \ket{\theta_b}_{i+1,\dots,n}.
  \label{division2}
\end{eqnarray}
These two decompositions (see FIG.\ref{figlineschmidt}) describe the same state, allowing me to combine them to express 
the basis of the subsystem $i,\dots,n$ using the spin at site $i$ and the basis of the subsystem $i+1,\dots,n$ as
\begin{eqnarray}
   \ket{\phi_a}_{i,\dots,n} = \sum_{s=0,1} \sum_{b=1}^{\chi_{i}} \Gamma^{(i),s}_{a,b} 
   \lambda^{(i)}_b \ket{s}_i \ket{\theta_b}_{i+1,\dots,n}, \label{combineSchmidt}
\end{eqnarray}
where I inserted the $\lambda_b^{(i)}$ for convenience. 
This gives me the tensor $\Gamma^{(i)}$. 
It carries an index 
$s$ corresponding to the state $\ket{s}$ of the $i$-th spin, and indices $a$ and $b$, corresponding
to the two consecutive divisions of the system (see FIG.\ref{figlineschmidt}). 
Because $\ket{\phi_a}$ (and $\ket{\theta_b}$) are orthonormal states,
the vectors $\lambda$ and tensors $\Gamma$ obey the following normalization conditions. From \eqref{division2} I have 
\begin{eqnarray}
	\sum_{b=1}^{\chi_i} \lambda_b^{(i)2} = 1,
	\label{normal0}
\end{eqnarray}
while
\eqref{combineSchmidt} implies 
\begin{eqnarray}
	\braket{\phi_{a'}}{\phi_a}_{i,\dots,n} = \sum_{s=0,1} \sum_{b=1}^{\chi_i}  \Gamma^{(i),s *}_{a',b} \lambda_b^{(i)}\Gamma^{(i),s}_{a,b}  \lambda_b^{(i)} = \delta_{a,a'} \, ,
	\label{normal1}
\end{eqnarray}
and 
\begin{eqnarray}
	\braket{\theta_{b'}}{\theta_b}_{1,\dots,i} =
	\sum_{s=0,1} \sum_{a=1}^{\chi_{i-1}} \lambda_a^{(i-1)} \Gamma^{(i),s *}_{a,b'}
	\lambda_a^{(i-1)} \Gamma^{(i),s}_{a,b} = \delta_{b,b'} \,  .
	\label{normal2}
\end{eqnarray}


\subsection{MPS on Trees} 
\label{MPStreeSubsection}

Matrix Product States are natural not just on chains, but also on trees, because these can also be split into two subsystems by cutting a single bond, allowing for the Schmidt-decomposition interpretation
as described in the previous section. The Matrix Product State description of a state of a spin system on a tree, i.e. such that the bonds do not form loops, is a generalization of the above procedure. 
Tree-tensor-network descriptions such as the one given here have been previously described in \cite{MPS:ShiTrees}.

Specifically, for the Bethe lattice with 3 neighbors per spin,
I introduce a vector $\lambda^{(k)}_{a_k}$ for each bond $k$ and a four-index (one for spin, three for bonds) tensor $\Gamma^{(i),s_i}_{a_k,a_l,a_m}$ for each site $i$. I can then rewrite the state $\ket{\psi}$ analogously to \eqref{psiMPS},\eqref{cMPS} as
\begin{eqnarray}
	\ket{\psi} = 
		\Bigg(
			\prod_{k\in\textrm{bonds}} 
			\sum_{a_k = 1}^{\chi_k}
			\lambda^{(k)}_{a_k} 
		\Bigg)
		\Bigg(
			\prod_{i\in \textrm{sites}} 
			\sum_{s_i} 
		\Gamma^{(i),s_i}_{a_l, a_m, a_n} 
		\Bigg)
	\ket{\dots}\ket{s_i}\ket{\dots},
\end{eqnarray}
where $a_l,a_m,a_n$ are indices corresponding to the three bonds $l, m$ and $n$ coming out of site $i$. Each index $a_l$ appears in two $\Gamma$ tensors and one $\lambda$ vector. To obtain this description, one needs to perform a Schmidt decomposition across each bond. This produces 
the vectors $\lambda^{(l)}$.
\begin{figure}
	\begin{center}
	\includegraphics[width=3.5in]{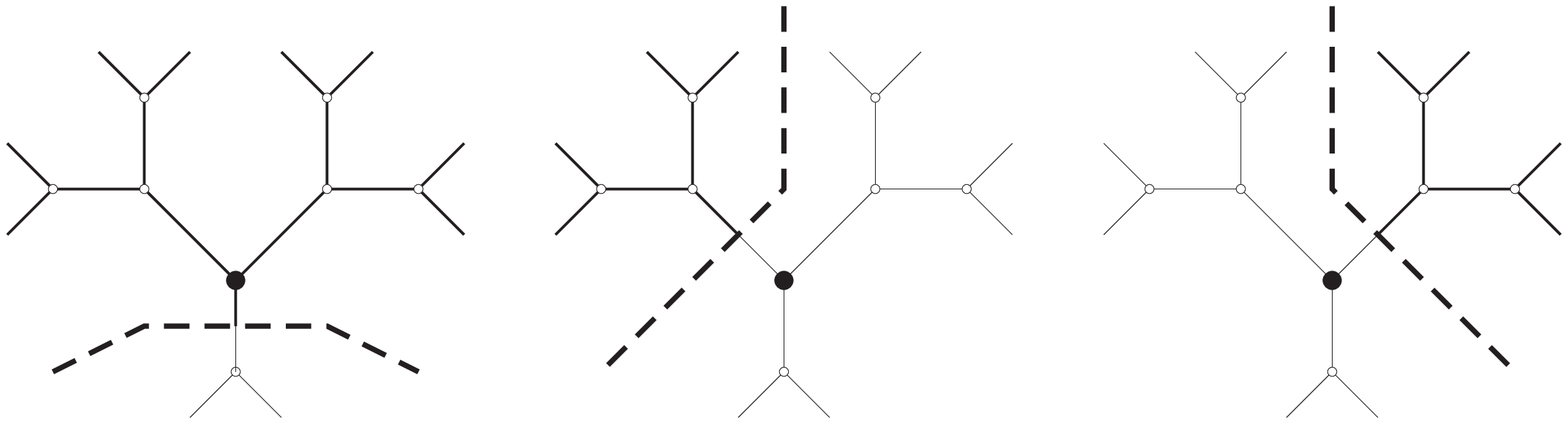}
	\caption{The three Schmidt decompositions on a tree required to obtain the $\Gamma$ tensor
	for the marked site and the three $\lambda$ vectors for the bonds emanating from it.}
	\label{figtreeschmidt}
	\end{center}
\end{figure}
To obtain the tensor $\Gamma^{(i)}$ for site $i$, 
one needs to combine the three decompositions corresponding to the bonds of site $i$
as depicted in Fig.\ref{figtreeschmidt}.
Analogously to \eqref{combineSchmidt}, expressing the orthonormal basis
for the first subsystem marked in Fig.\ref{figtreeschmidt} in terms of the state of the spin $\ket{s_i}$ and the orthonormal bases for the latter two subsystems in Fig.\ref{figtreeschmidt}, one obtains the tensor $\Gamma^{(i),s}_{a_l,a_m,a_n}$ for site $i$. 

The normalization conditions for a MPS description of a state on a tree are analogous to \eqref{normal0}-\eqref{normal2}.
I have
\begin{eqnarray}
	\sum_{a_k} \lambda^{(k)2}_{a_k} = 1, 
		\label{normalT0}
\end{eqnarray}
\begin{eqnarray}
	\sum_{s=0,1} \sum_{a_k = 1}^{\chi_k} \sum_{a_l = 1}^{\chi_l} \Gamma^{(i),s *}_{a_k,a_l,a_{m'}} \lambda^{(k)2}_{a_k}\lambda^{(l)2}_{a_l} \Gamma^{(i),s}_{a_k,a_l,a_m} = \delta_{a_m,a_{m'}} \, , 
	\label{normalT1}
\end{eqnarray}
and two other variations of \eqref{normalT1} with $k,l$ and $m$ interchanged.


\section{Simulating Quantum Systems with MPS}
\label{updatesection}

I choose to first describe the numerical procedures for a chain of spins. Then, at the end of the respective subsections, I note how to generalize these to tree geometry.


\subsection{Unitary Update Rules}

The strength of the MPS description of the state lies in the efficient 
application of local unitary update rules such as $U = e^{-i A \Delta t}$ (where $A$ is an operator acting only on a few qubits). 
First, I describe the numerical procedure in some detail, and then, in the next Section,  discuss how to modify the procedure to also implement imaginary time evolution.

Given a state $\ket{\psi}$ as a Matrix Product State, I want to know what happens after 
an application of a local unitary. In particular, for a 1-local $U$ acting on the $i$-th spin, 
it suffices to update the local tensor
\begin{eqnarray}
	\Gamma^{(i),s}_{a,b} 
	\buildrel{U}\over\longrightarrow U^{s}_{s'} \Gamma^{(i),s'}_{a,b}.
\end{eqnarray} 
The update rule for an application of a 2-local unitary $V$ acting on neighboring spins $i$ and $i+1$,
requires several steps. First, using
a larger tensor 
\begin{eqnarray}
	\Theta^{s,t}_{a,c} = \lambda^{(i-1)}_a 
	  \sum_b \left( \Gamma^{(i),s}_{a,b}
	  \lambda^{(i)}_b \Gamma^{(i+1),t}_{b,c} \right) \lambda^{(i+1)}_c,
	  \label{Thetatensor}
\end{eqnarray}
I rewrite the state $\ket{\psi}$ as
\begin{eqnarray}
  	\ket{\psi} &=& \sum_{a,c} \sum_{s,t} \Theta^{s,t}_{a,c} \ket{\phi_a}_{1\dots i-1}
	  \ket{s}_{i} \ket{t}_{i+1} \ket{\phi_c}_{i+1 \dots n}.
	  \label{Thetatensor2}
\end{eqnarray}
After the application of $V$, the tensor $\Theta$ in the description of $\ket{\psi}$ 
changes as
\begin{eqnarray}
	\Theta^{s,t}_{a,c} \buildrel{V}\over\longrightarrow 
	\sum_{s' t'} V^{s,t}_{s',t'} \Theta^{s',t'}_{a,c}.
	\label{Thetaupdate}
\end{eqnarray}
One now needs to decompose the updated tensor $\Theta$ to obtain the updated 
tensors $\Gamma^{(i)}$, $\Gamma^{(i+1)}$ and the vector $\lambda^{(i)}$.
I use the indices $a,s$ and $c,t$ of $\Theta$
to introduce combined indices $(as)$ and $(ct)$ and form a matrix $T_{(as),(ct)}$ 
with dimensions $2\chi_{i-1} \times 2\chi_{i+1}$ as
\begin{eqnarray}
	T_{(as),(ct)} =  \Theta_{a,c}^{s,t}.
	\label{Tmatrix1}
\end{eqnarray}
Using the singular value decomposition (SVD), this matrix can be decomposed into $T = Q \Lambda W$, where
$Q$ and $W$ are unitary and $\Lambda$ is a diagonal matrix. 
In terms of matrix elements, this reads
\begin{eqnarray}
	T_{(as),(ct)} = 
  \sum_b Q_{(as),b} 
  D_{b,b} W_{b,(ct)}.
  \label{Tmatrix}
\end{eqnarray}
The diagonal matrix $D=\textrm{diag}(\lambda^{(i)})$
gives me the updated Schmidt vector 
$\lambda^{(i)}$. 
The updated tensors $\Gamma^{(i)}$ and $\Gamma^{(i+1)}$ can be obtained from the 
matrices $Q,W$ and the definition of $\Theta$ \eqref{Thetatensor} 
using the old vectors $\lambda^{(i-1)}$ and $\lambda^{(i+1)}$ which do 
not change with the application of the local unitary $V$. 
After these update procedures, the conditions \eqref{normal0}-\eqref{normal2}
are maintained.

The usefulness/succintness of this description depends crucially on the amount of entanglement 
across the bipartite divisions of the system as measured by the Schmidt numbers $\chi_i$.
To exactly describe a general quantum state $\ket{\psi}$ of a chain of $n$ spins, the Schmidt 
number for the split through the middle of the chain is necessarily $\chi_{n/2} = 2^{n/2}$.
Suppose I start my numerical simulation in a state that is exactly described by a MPS with only low $\chi_i$'s.
The update step described above involves an interaction
of two sites, and thus could generate more entanglement across the $i,i+1$ division.
After the update, the index $b$ in $\lambda^{(i)}_b$ 
would need to run from $1$ to $2\chi_i$ to keep 
the description exact (unless $\chi_i$ already is at its maximum required value 
$\chi_i=2^{\textrm{min}\{i,n-i\}}$).
This makes the number of parameters in the MPS description grow exponentially 
with the number of update steps.

So far, this description and update rules have been exact. 
Let me now make the description an approximate one (use a block-decimation step) instead. 
First, introduce the parameter $\chi$, which is the maximum number of Schmidt terms I keep after each update step.
If the amount of entanglement in the system is low, the Schmidt coefficients 
$\lambda^{(i)}_b$ decrease rapidly with $b$ (I always take the elements of $\lambda$ sorted in decreasing order).
A MPS ansatz with restricted $\chi_i = \chi$ will hopefully be a good approximation to 
the exact state $\ket{\psi}$. 
However, I also need to keep the restricted $\chi$ throughout the simulation.
After a two-local unitary update step, the vector $\lambda^{(i)}$ can have $2\chi$ entries. However, if the $b>\chi$ entries in $\lambda^{(i)}_b$ after the update are small,
I am justified to truncate $\lambda^{(i)}$ to have only $\chi$ entries and multiply it 
by a number so that it satisfies \eqref{normal0}.
I also truncate the $\Gamma$ tensors so that they keep dimensions $2\times\chi\times\chi$. The normalization condition \eqref{normal2} for $\Gamma^{(i)}$ will be still satisfied exactly, while the error in the normalization condition \eqref{normal1} will be small. This normalization error can be corrected as discussed in the next section.
This procedure keeps the state within the MPS ansatz with restricted $\chi$.

The procedure described above allows me to efficiently approximately implement local unitary evolution. To simulate time evolution 
\begin{eqnarray}
	\ket{\psi(t)} = e^{-iHt}\ket{\psi(0)}, \label{realtime}
\end{eqnarray}
with a local Hamiltonian like \eqref{ourH1}, I first divide the time $t$ into small slices $\Delta t$ and split the Hamiltonian into two groups of commuting terms $H_k^{(x)}$ and $H_m^{(z)}$. Each time evolution step $e^{-iH\Delta t}$ can then be implemented as a product of local unitaries using the second order Trotter-Suzuki formula
\begin{eqnarray}
	U_2 = \left(\prod_{k} e^{-i H^{(x)}_k \frac{\Delta t}{2}}\right)
		\left(\prod_{m} e^{-i H^{(z)}_m \Delta t}\right)
		\left(\prod_{k} e^{-i H^{(x)}_k \frac{\Delta t}{2}}\right).
		\label{trotter}
\end{eqnarray}
The application of the product of the local unitaries within each group can be done almost in parallel (in two steps, as described in Section \ref{TIsection}), as they commute with each other.

These update rules allow me to efficiently approximately simulate the real 
time evolution \eqref{realtime}
with a local Hamiltonian $H$ 
for a state $\ket{\psi}$ within the MPS ansatz with parameter $\chi$. 
The number of parameters in this MPS description with restricted $\chi$  
is then $n (2\chi^2)$ for the tensors $\Gamma^{(i)}$ and $(n-1) \chi$ for the vectors 
$\lambda^{(i)}$.
The simulation cost of each local update step scales like $O(\chi^3)$, 
coming from the SVD decomposition of the matrix $\Theta$.
For a system of $n$ spins, I thus need to store $O(2n\chi^2+n\chi)$ numbers and each update will take $O(n\chi^3)$ steps.

The update procedure generalizes to tree geometry by taking the tensors $\Gamma$ with dimensions $2\times\chi\times\chi\times \chi$ as in Section \ref{MPStreeSubsection}. 
For a local update (on two neighboring spins $i$ and $i+1$ with bonds labeled by $l,m,n$ and $n , o, p$) I rewrite the state $\ket{\psi}$ analogously to \eqref{Thetatensor2} as
\begin{eqnarray}
  	\ket{\psi} &=& \sum_{a_k,a_l,a_o,a_p} \sum_{s,t} \Theta^{s,t}_{(a_k a_l),(a_o a_p)} 
	\ket{\phi_{a_k}} \ket{\phi_{a_l}}
	  \ket{s}_{i} \ket{t}_{i+1} \ket{\phi_{a_o}}\ket{\phi_{a_p}}.
\end{eqnarray}
using the tensor 
\begin{eqnarray}
	\Theta^{s,t}_{(a_k a_l),(a_o a_p)} = \lambda^{(k)}_{a_k} \lambda^{(l)}_{a_l} 
	  \sum_{a_m} \left( \Gamma^{(A),s}_{a_k,a_l,a_m}
	  \lambda^{(m)}_{a_m} \Gamma^{(i+1),t}_{a_m,a_o,a_p} \right) \lambda^{(o)}_{a_o} \lambda^{(p)}_{a_p},
\end{eqnarray}
with combined indices $(a_k a_l)$ and $(a_o a_p)$.
One then needs to update the tensor $\Theta$ as described above \eqref{Thetaupdate}-\eqref{Tmatrix}. The decomposition procedure to get the updated vector $\lambda^{(m)}$ and the new tensors $\Gamma^{(i)}$ and $\Gamma^{(i+1)}$ now requires 
$O(\chi^6)$ computational steps. The cost of a simulation on $n$ spins thus scales like $O(n\chi^6)$.


\subsection{Imaginary Time Evolution}
\label{imaginarysection}

Using the MPS ansatz, I can also use imaginary time evolution with $e^{-Ht}$ instead of \eqref{realtime} to look for the ground state of systems governed by local Hamiltonians. One needs to replace each unitary term
$e^{-i A \Delta t}$ 
in the Trotter expansion \eqref{trotter} of the time evolution with
$e^{- A \Delta t}$ followed by a normalization procedure. However, the usual normalization procedure for imaginary time evolution (multiplying the state by a number to keep $\braket{\psi}{\psi}=1$) is now not enough to satisfy the MPS normalization conditions \eqref{normal0}-\eqref{normal2} for the tensors $\Gamma$ and vectors $\lambda$ I use to describe the state $\ket{\psi}$.

The unitarity of the real time evolution automatically implied
that the normalization conditions \eqref{normal0},\eqref{normal2} were satisfied after an
exact unitary update. While there already was an error in \eqref{normal1} introduced by the truncation of the $\chi+1 \dots 2\chi$ entries in $\Gamma^{(i)}$, the non-unitarity of imaginary time evolution update steps introduces further normalization errors. It is thus important to properly normalize the state after every application of terms like $e^{-A\Delta t}$ to keep it within the MPS ansatz. 

In \cite{MPS:Vidal1Dinfinite}, Vidal dealt with this problem by taking progressively shorter and shorter steps $\Delta t$ during the imaginary time evolution. This procedure results in a properly normalized state only at the end of the evolution, after the time step decreases to zero (and not necessarily during the evolution). 
I propose a different scheme in which I follow each local update $e^{-A \Delta t}$ by a normalization procedure (based on Vidal's observation) to bring the state back to the MPS ansatz at all times. 
The simulation I run (evolution for time $t$) thus consists of many short time step updates $e^{-H \Delta t}$, each of which is implemented using a Trotter expansion as a product of local updates $e^{-A \Delta t}$. Each of these local updates is followed by my normalization procedure.
 
I now describe the iterative normalization procedure in detail for the case of an infinite chain, where it can be applied efficiently, as the description of the state $\ket{\psi}$ requires only two different tensors $\Gamma$ (see Section \ref{TIline}). One needs to apply the following steps over and over, until the normalization conditions are met with chosen accuracy. 

First, for each nearest neighbor pair $i,i+1$ with even $i$, combine the MPS description of these two spins \eqref{Thetatensor}-\eqref{Thetatensor2}, forming the matrix $T$ \eqref{Tmatrix1}. Do a SVD decomposition of $T$ \eqref{Tmatrix} to obtain a new vector $\lambda^{(i)}$. The decomposition does not increase the number of nonzero elements of $\lambda^{(i)}$, as the rank of the $2\chi\times2\chi$ matrix $T$ \eqref{Tmatrix} was only $\chi$ (coming from \eqref{Thetatensor}).
I thus take only the first $\chi$ values of $\lambda^{(i)}$
and rescale the vector to obey $\sum_{a=1}^{\chi} \lambda^{(i)2}_a = 1$.
using this new $\lambda^{(i)}$, I obtain tensors $\Gamma^{(i)}$, $\Gamma^{(i+1)}$ from \eqref{Tmatrix}, and truncate them to have dimensions $\chi\times\chi\times2$.  
Second, I repeat the previous steps for all nearest neighbor pairs of spins $i,i+1$ with $i$ odd. 

I observe that repeating the above steps over and over results in exponential decrease in the error in the normalization of the $\Gamma$ tensors. Note though, that the rate of decrease in normalization errors becomes much slower near the phase transition for the transverse field Ising model on an infinite line (see Section \ref{NRsection1}).

In practice, I apply this normalization procedure by using the same subroutine for the local updates $e^{-A\Delta t}$, except that I skip the step 
\eqref{Thetaupdate}, which is equivalent to applying the local update with $\Delta t=0$.
The normalization procedure is thus equivalent to evolving the state repeatedly with zero time step (composing two tensors $\Gamma$ and decomposing them again) and imposing the normalization condition on the vectors $\lambda$. Note though, following from the definition of the SVD, that each decomposition assures that one of the conditions \eqref{normal1},\eqref{normal2} is retained exactly for the updated tensors $\Gamma$. The errors in the other normalization condition for the $\Gamma$ tensors 
are decreased in each iteration step.

The numerical update rules for a system with tree geometry are a simple analogue of the update rules for MPS on spin chains. Every interaction couples two sites, 
with tensors $\Gamma^{(A),s}_{a,b,c}$ and $\Gamma^{(B),t}_{c,d,e}$, with the three lower indices corresponding to the bonds emanating from the sites. 
One only needs to reshape the tensors into $\Gamma^{(A),s}_{(ab),c}$ 
and $^{(B),t}_{c,(de)}$ and proceed as described in \eqref{Thetatensor}
and below.


\section{MPS and Translationally Invariant Systems} 
\label{TIsection}

\subsection{An Infinite Line} 
\label{TIline}

For systems with translational symmetry such as an infinite line 
all the sites are equivalent. 
I assume that the ground state is translationally invariant, and
furthermore pick the tensors $\Gamma^{(i)}$ and vectors $\lambda^{(i)}$ 
to be site independent.
For fixed $\chi$ the number of complex parameters in the translationally 
invariant MPS ansatz on the infinite line scales as $2\chi^2$. 

When using imaginary time evolution to look for the ground state of this system, within this ansatz, it is technically hard to keep the translational symmetry and the normalization conditions after each update. Numerical instabilities plagued my efforts to impose the symmetry in the procedures described above.
In \cite{MPS:Vidal1Dinfinite}, Vidal devised a method to deal with this problem. Let us
break the translational symmetry of the ansatz by labeling the sites $A$ and $B$ as in FIG.\ref{figlineAB}.
This doubles the number of parameters in the ansatz.
\begin{figure}
	\begin{center}
	\includegraphics[width=2.6in]{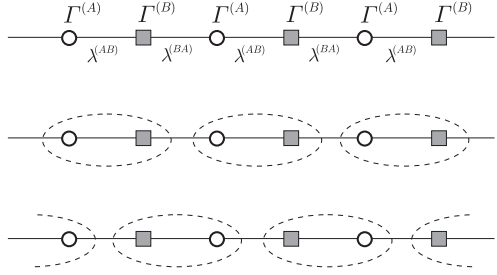}
	\caption{The parametrization and update rules for the infinite line.}
	\label{figlineAB}
	\end{center}
\end{figure}
The state update now proceeds in two steps. Let the site pairs $AB$ 
interact and update the tensors $\Gamma^{(A)}$, $\Gamma^{(B)}$ and the vector $\lambda^{(AB)}$. 
Then let the neighbor pairs $BA$ interact, after which I update the tensors 
$\Gamma^{(B)}$, $\Gamma^{(A)}$ and the vector $\lambda^{(BA)}$.
What I observe is that after many state updates the elements of the resulting $\Gamma^{(A)}$ and the $\Gamma^{(B)}$ tensors differ at a level which is way below my numerical accuracy (governed by the normalization errors) and I am indeed obtaining a translationally invariant description of the system.

One of the systems easily investigated with this method (iTEBD) is the Ising model in a transverse field \eqref{ourH} on an infinite line.
Vidal's numerical results for the real time evolution and imaginary time evolution \cite{MPS:Vidal1Dinfinite} of this system show remarkable agreement with the exact solution. I take a step further and also numerically obtain the critical exponents for this system.
Further details can be found in Section \ref{NRsection}, where I compare these results for the infinite line to the results I obtain for the Ising model in transverse field on the Bethe lattice.


\subsection{An Infinite Tree} 
\label{infitreesection}

For the infinite Bethe lattice, My approach is a modification of the above procedure introduced by Vidal. In order to avoid the numerical instabilities associated with imposing site-independent $\Gamma$ and $\lambda$ after the update steps, 
I break the translational symmetry by labeling the ``layers'' of the tree $A$ and $B$ 
(denoted by half-circles and triangles), as in FIG.\ref{figtreeAB}. 
\begin{figure}
	\begin{center}
	\includegraphics[width=2.5in]{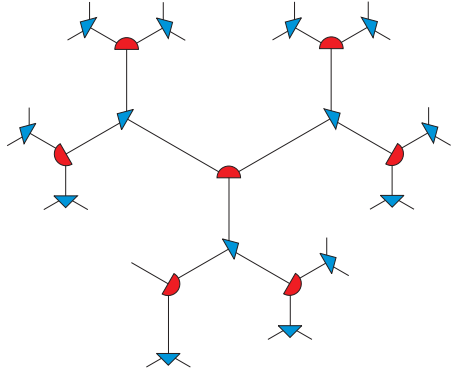}
	\caption{The two-layer, directed labeling of the tree.}
	\label{figtreeAB}
	\end{center}
\end{figure}

The Bethe lattice is also symmetric under the 
permutation of directions. Tensors $\Gamma$ with full directional symmetry obey
$\Gamma_{a,b,c} = \Gamma_{b,c,a}=\Gamma_{c,a,b}=\Gamma_{c,b,a}=\Gamma_{b,a,c}=\Gamma_{a,c,b}$.
However, for the purpose of simple organization of interactions, I will also partially break this symmetry by consistently labeling an `inward' bond for each node,
as denoted by the flat sides of the semi-circles and the longer edges of the triangles in Fig.\ref{figtreeAB}.
This makes the first of the three indices of $\Gamma_{a,b,c}$ special. 
However, I keep the residual symmetry $\Gamma_{a,b,c}=\Gamma_{a,c,b}$.
I can enforce this by interacting a spin with both of the spins from the next layer at the same time. The update procedure for the interaction between the spins now splits into two steps,
interacting the layers in the $AB$ order first, 
and then in the $BA$ order as in Fig.\ref{figtreeABinteract}.
\begin{figure}
	\begin{center}
	\includegraphics[width=5in]{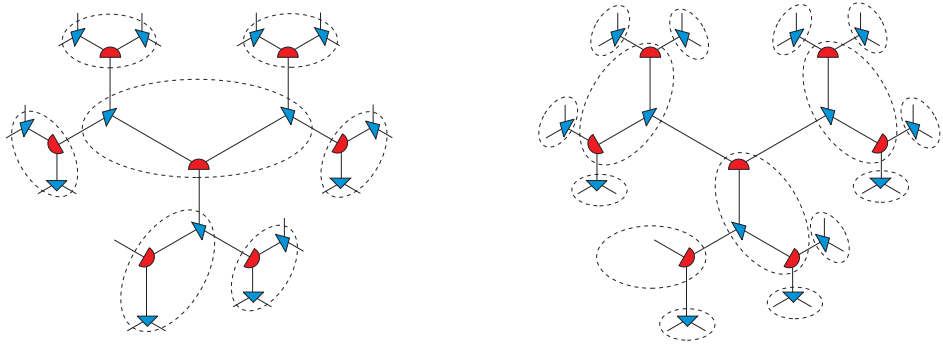}
	\caption{The two-step interactions for the infinite tree.}
	\label{figtreeABinteract}
	\end{center}
\end{figure}
Similarly to what I discovered for the line, the differences 
in the elements of the final $\Gamma^{(A)}$ and $\Gamma^{(B)}$ are
well below the numerical accuracy of my procedure.

The scaling of this procedure is more demanding than the $O(\chi^3)$ simulation for a line.
The number of entries in the matrix $\Theta$ used in each update step 
is $2\chi^2 \times 4\chi^4$, therefore the SVD decomposition requires  $O(\chi^8)$
steps. The scaling of my numerical method is thus $O(\chi^8)$ for each update step.


\subsection{Expectation Values}
A nice property of the MPS state description is that it allows efficient computation of expectation values of local operators. 
First, for a translationally invariant system on a line (with only one tensor $\Gamma$ and one vector $\lambda$), I have
for an operator $O^{(i)}$ acting only on the $i$-th spin 
\begin{eqnarray}
	\bra{\psi}O^{(i)} \ket{\psi} = 
	\sum_{s_i,s_i'=0,1} O_{s_i,s_i'}^{(i)}
	\sum_{a=1}^{\chi} 
	\sum_{b=1}^{\chi} 
	(\lambda_{a}
	\Gamma^{s_i' *}_{a,b}
	\lambda_{b} )
	(\lambda_{a}
	\Gamma^{s_i}_{a,b} 
	\lambda_{b} )
	,
	\label{expect1}
\end{eqnarray}
where $O^{(i)}_{s_i,s_i'} = \bras{s_i'}O^{(i)}\kets{s_i}$.
Similarly, for the expectation values of $O^{(i)}O^{(j)}$ (assuming $j>i$), 
\begin{eqnarray}
	\bra{\psi}O^{(i)}O^{(j)} \ket{\psi} 
	&=& 
	\sum_{s_i,s'_i,\dots,s_j,s'_j}
	  \sum_{a,e} \sum_{b,b', \dots} 
	 O_{s_i,s_i'}^{(i)}O_{s_j, s_j'}^{(j)} \label{twoexpect} \\
	&&\times \, \big(\lambda_{a}
			\Gamma^{s_i'*}_{a,b'} 
			\lambda_{b'}
			\Gamma^{s_{i+1}*}_{b',c'} 
			\lambda_{c'}
			\cdots
			\lambda_{d'}
			\Gamma^{s_j'*}_{d',e} 
			\lambda_{e}\big) \nonumber\\
	 &&\times \, 
		\big(\lambda_{a}
			\Gamma^{s_i}_{a,b} 
			\lambda_{b}
			\Gamma^{s_{i+1}}_{b,c} 
			\lambda_{c}
			\cdots
			\lambda_{d}
			\Gamma^{s_j}_{d,e} 
			\lambda_{e}\big). 
	\nonumber
\end{eqnarray}
Defining a $\chi^2 \times \chi^2$ matrix $B$ (where one should think of $(bb')$ as one combined index ranging from $1$ to $\chi^2$) as 
\begin{eqnarray}
	B_{(bb'), (cc')} = 
	\sum_{s} 
			\Gamma^{s}_{b,c}
			\Gamma^{s *}_{b',c'}
			\lambda_{c}
			\lambda_{c'},
\end{eqnarray}
and
vectors $v$ and $w$ with elements again denoted by a combined index $(bb')=1\dots\chi^2$ as
\begin{eqnarray}
	v_{(bb')} &=& 
		\sum_{s_i,s_i'} O_{s_i,s_i'}^{(i)}
	\sum_{a} 
			(\lambda_{a})^2
			\Gamma^{s_i}_{a,b} 
			\Gamma^{s_i' *}_{a,b'} 
			\lambda_{b}
			\lambda_{b'}
			, \\
	w_{(dd')} &=& 
		\sum_{s_j,s_j'} O_{s_j,s_j'}^{(j)}
	\sum_{e} 
			\Gamma^{s_j}_{d,e} 
			\Gamma^{s_j' *}_{d',e} 
			(\lambda_{e})^2
			,
\end{eqnarray}
I can rewrite \eqref{twoexpect} as
\begin{eqnarray}
	\bra{\psi}O^{(i)}O^{(j)} \ket{\psi} = v^T \underbrace{B B \cdots B}_{j-i-1} w,
	\label{Bcorrfunction}
\end{eqnarray}

There is a relationship between the eigenvalues of the matrix $B$ and the correlation function $\corrfO$. One of the eigenvalues of $B$ is $\mu_1=1$, with the corresponding right eigenvector  
\begin{eqnarray}
	\beta^{(1R)}_{(cc')} =   
	\sum_{s} 
	\sum_{c}
			\Gamma^{s}_{b,c}
			\Gamma^{s *}_{b',c}
			(\lambda_{c})^2,
\end{eqnarray}
and left eigenvector
\begin{eqnarray}
	\beta^{(1L)}_{(bb')} =   
		\delta_{b,b'} \lambda_{b}^2,
\end{eqnarray}
which can be verified using the normalization conditions \eqref{normal1} and \eqref{normal2}.
I numerically observe that $\mu_1=1$ is also the largest eigenvalue.  
(Note that $|\mu_k|>1$ would result in correlations unphysically growing with distance.)
Denote the second largest eigenvalue of $B$ as $\mu_2$. 
Using the eigenvectors of $B$, I can express $B^{j-i-1}$ in \eqref{Bcorrfunction}, as
\begin{eqnarray}
	B^{j-i-1} = \beta^{(1L)} \beta^{(1R)T} + \mu_2^{j-i-1} \beta^{(2L)} \beta^{(2R)T} + \dots.
	\label{eigexpand}
\end{eqnarray}
When computing the correlation function, 
the term that gets subtracted exactly cancels the leading term involving $\mu_1=1$. 
Therefore, if $|\mu_2|$ is less than 1, \eqref{eigexpand} implies
\begin{eqnarray}
	\corrfO \propto \mu_2^{|j-i|}. \label{secondeig}
\end{eqnarray} 
The correlation function necessarily falls of exponentially in this case,
and the correlation length $\xi$ is related to $\mu_2$ as $\xi = -1/\ln \mu_2$.

The computation of expectation values for a MPS state on a system with a tree geometry can be again done efficiently. For single-site operators $O^{(i)}$, the formula is an analogue of \eqref{expect1} with three $\lambda$ vectors for each $\Gamma$ tensors which now have three lower indices. For two-site operators, the terms in \eqref{Bcorrfunction} now become
\begin{eqnarray}
	B_{(cc'), (dd')} &=& 
	\sum_{s} \sum_{e}
			\Gamma^{s}_{c,e,d}
			\Gamma^{s *}_{c',e,d'}
			(\lambda_{e})^2
			\lambda_{d}
			\lambda_{d'},\\
	v_{(cc')} &=& 
		\sum_{s_i,s_i'} O_{s_i,s_i'}^{(i)}
	\sum_{a,b} 
			(\lambda_{a})^2 (\lambda_{b})^2
			\Gamma^{s_i}_{a,b,c} 
			\Gamma^{s_i' *}_{a,b,c'} 
			\lambda_{c}
			\lambda_{c'}
			, \\
	w_{(dd')} &=& 
		\sum_{s_j,s_j'} O_{s_j,s_j'}^{(j)}
	\sum_{e,f} 
			\Gamma^{s_j}_{d,e,f} 
			\Gamma^{s_j' *}_{d',e,f} 
			(\lambda_{e})^2
			(\lambda_{f})^2
			.
\end{eqnarray}
The correlation length is again related to the second eigenvalue of the $B$ matrix as
in \eqref{secondeig}.


\section{Quantum Transverse Field Ising Model} 
\label{NRsection}

My goal is to investigate the phase transition for the Ising model in
transverse magnetic field \eqref{ourH1} on the infinite line and on the Bethe lattice. 
I choose to parametrize the Hamiltonian as  
\begin{eqnarray}
	H = \frac{s}{2} \sum_{\langle i,j\rangle} \left(1-\sigma_z^{i} \sigma_z^{j}\right) 
		+ \frac{b(1-s)}{2} \sum_i \left(1-\sigma_x^{i}\right), \label{MPSourH}
\end{eqnarray}
where $0\leq s\leq 1$
and $b$ is the number of bonds for each site ($b=2$ for the line, $b=3$ for the tree). I will investigate the ground state properties of \eqref{MPSourH} as we vary $s$. The point $s=0$ corresponds to a spin system in transverse magnetic field, while $s=1$ corresponds
to a purely ferromagnetic interaction between the spins.


\subsection{The Infinite Line.}
\label{NRsection1}
I present the results for the case of an infinite line and compare
them to exact results obtained via fermionization 
(see e.g. \cite{MPS:Sachdevbook}, Ch.4).
Vidal has shown \cite{MPS:Vidal1Dinfinite} that imaginary time evolution within the MPS ansatz is capable of providing a very accurate approximation for the ground state energy and correlation function. I show that even using $\chi$ smaller than used in \cite{MPS:Vidal1Dinfinite}, I obtain the essential information about the nature of the phase transition in the infinite one-dimensional system. I also obtain the critical exponents for the magnetization and the correlation length.

In FIG.\ref{lineEd}, I show the how the ground state energy obtained using imaginary time evolution with MPS converges to the exact energy as $\chi$ increases.
\begin{figure}
	\begin{center}
	\includegraphics[width=4.5in]{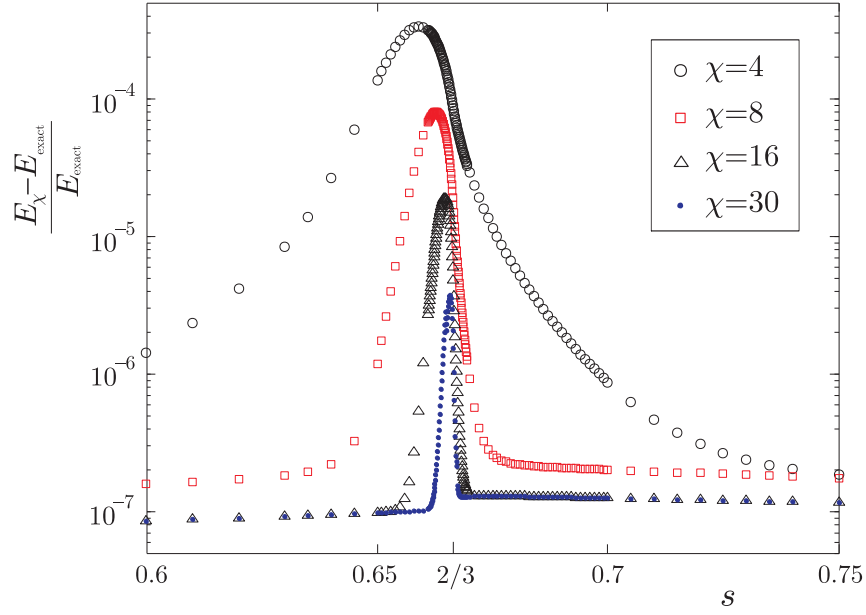}
	\caption{Transverse Ising model on an infinite line. Fractional difference of the ground state energy obtained using MPS and the exact ground state, near the phase transition at $s_L=\frac{2}{3}$. The energy scale is logarithmic.
	}
	\label{lineEd}
	\end{center}
\end{figure}
The exact solution for a line has a second order phase transition 
at the critical value of $s$, $s_L=\frac{2}{3}$, and the ground state energy 
and its first derivative are continuous,
while the second derivative diverges at $s=s_L$.
I plot the first and second derivative of $E$ with respect to $s$ obtained numerically and compare them to the exact values in FIG. \ref{lineE12},
observing the expected behavior already for low $\chi$.
\begin{figure}
	\begin{center}
	\includegraphics[width=6in]{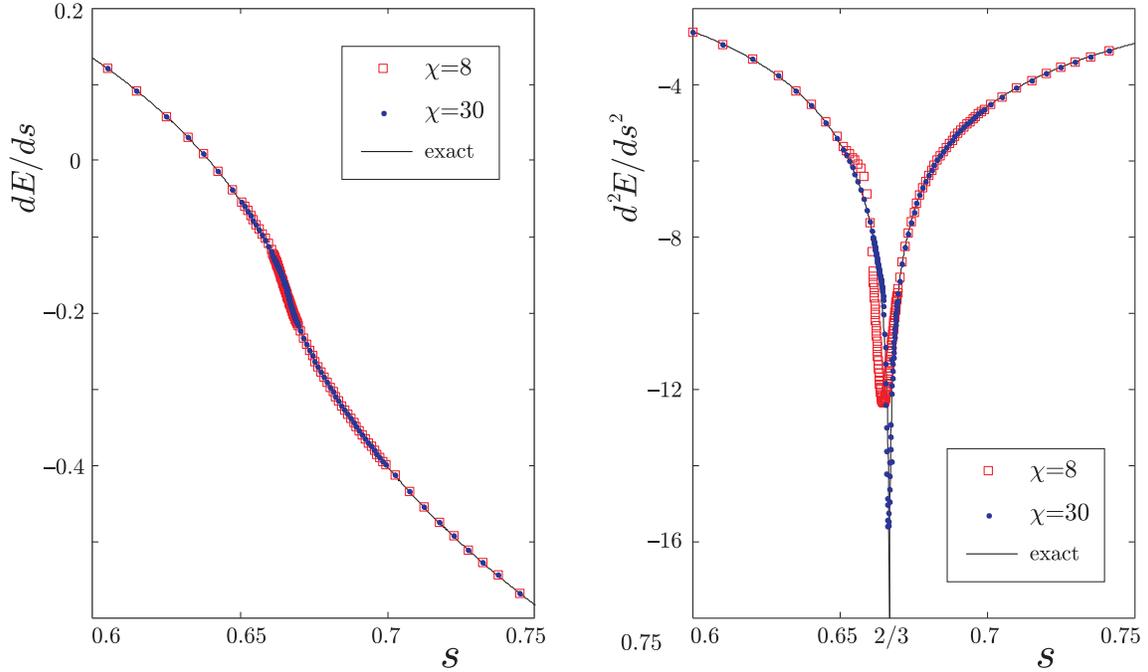}
	\caption{Transverse Ising model on an infinite line. The first and second derivative with respect to $s$ of the
	ground state energy obtained via MPS compared with the exact result.
	}
	\label{lineE12}
	\end{center}
\end{figure}

\begin{figure}
	\begin{center}
	\includegraphics[width=3.5in]{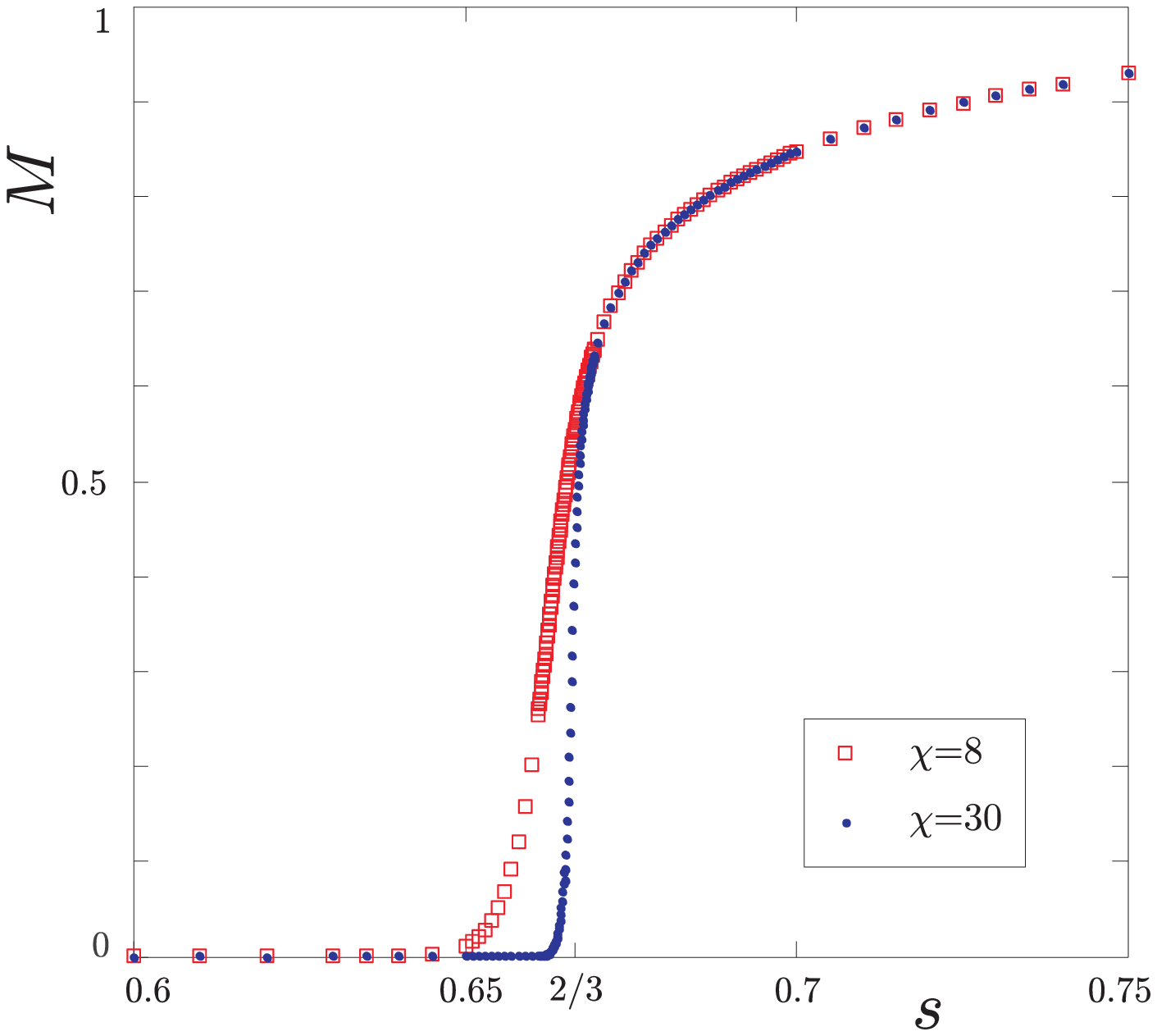}
	\caption{Transverse Ising model on an infinite line. Magnetization obtained using MPS vs $s$, with $B_z=10^{-8}$.
	}
	\label{lineM}
	\end{center}
\end{figure}
The derivative of the exact magnetization $M = \langle \sigma_z \rangle$ is 
discontinuous at $s_L$, with the magnetization starting to rise steeply
from zero as 
\begin{eqnarray}
	M \propto (x-x_L)^{\beta},
	\label{Mexponentline}
\end{eqnarray}
with the critical exponent $\beta_L = \frac{1}{8}$. Here $x$ is the ratio of the ferromagnetic interaction strength
to the transverse field strength in \eqref{MPSourH} and is
\begin{eqnarray}
	x = \frac{s}{2(1-s)},
\end{eqnarray}
with the value $x=x_L=1$ at the phase transition ($s_L=\frac{2}{3}$). 
I plot the magnetization obtained with my method 
in FIG.\ref{lineM}. To obtain the magnetization depicted in the plot, I used a small symmetry breaking longitudinal field with magnitude $B_z = 10^{-8}$.
\begin{figure}
	\begin{center}
	\includegraphics[width=4in]{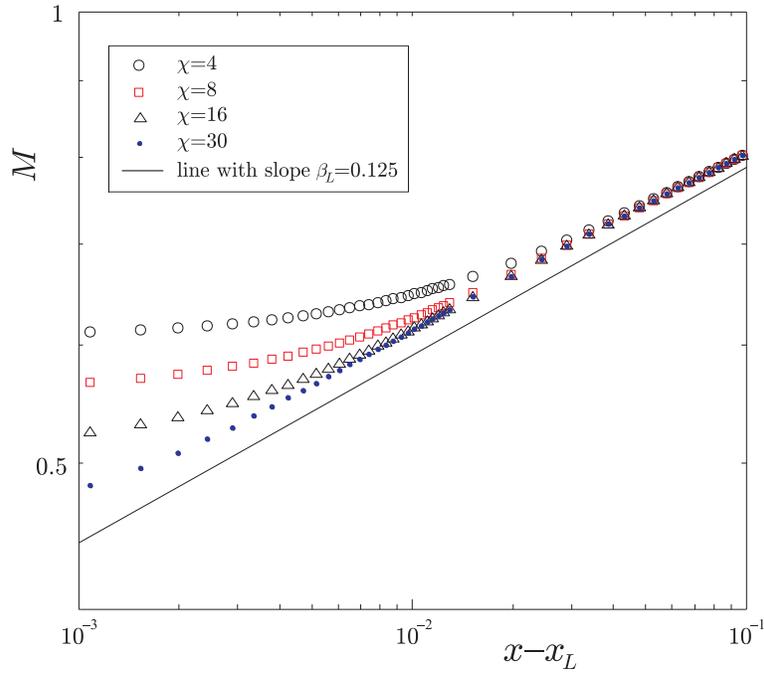}
	\caption{Transverse Ising model on an infinite line. Log-log plot of magnetization vs. $x-x_L$. 
	I also plot a line with slope $\beta_L = 0.125$.
	}
	\label{lineMexp}
	\end{center}
\end{figure}
In FIG.\ref{lineMexp}, I plot $M$ vs. $x-x_L$ on a log-log scale. I also plot a line with slope $0.125$. Observe that as $\chi$ increases, the data is better represented by a line down to smaller values of $x-x_L$. 
For the largest $\chi$ displayed, a straight line fit of the data between $4\times 10^{-3} \leq x-x_L \leq 10^{-1}$ gives me a slope of $0.120$. Note that the mean field value of the critical exponent for magnetization is $0.5$.

\begin{figure}
	\begin{center}
	\includegraphics[width=5in]{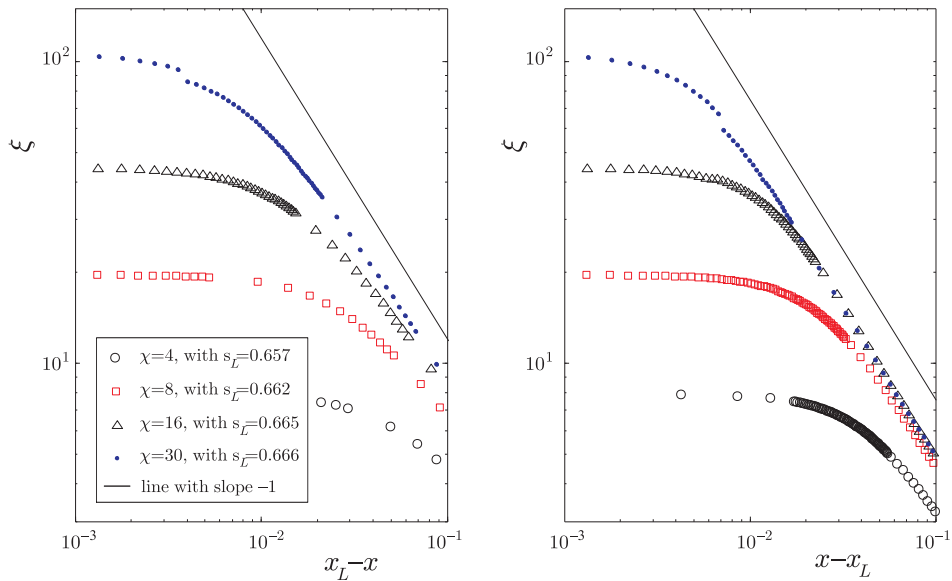}
	\caption{Transverse Ising model on an infinite line. Log-log plot of the correlation length vs. $|x-x_L|$. 
	I also plot a line with slope $-\nu_L = -1$.
	For each $\chi$ in the plot, I choose a numerical value of the critical point $s_L$ 
	as the point at which the correlation length is maximal.
	}
	\label{lineCexp}
	\end{center}
\end{figure}
The correlation function $\corrf$ 
can be computed efficiently using \eqref{twoexpect}.
Away from criticality, it falls off
exponentially as $e^{- |i-j|/\xi}$.
The falloff of the correlation function is necessarily exponential as long as $\mu_2$,
the second eigenvalue of $B$, is less than $1$.
The exact solution for the correlation length $\xi$ near the critical point has the form
\begin{eqnarray}
  \xi \propto |x-x_L|^{-1}
  \label{xidiverge}
\end{eqnarray}
as $x$ approaches $x_L$.
Already at low $\chi$ the iTEBD method captures the divergence of the correlation length. In FIG. \ref{lineCexp}, I plot $\xi$ vs. $|x-x_L|$ on a log-log plot, together with a line with slope $-1$.
Again, as $\chi$ increases, the data is better represented by a line closer to the phase transition.
For the highest $\chi$ displayed, a straight line fit of the data between $2\times 10^{-2} \leq x_L-x \leq 4\times 10^{-1}$ gives me a slope of $-0.92$.
Note that the mean-field value of the critical exponent for the correlation length is $0.5$ (corresponding to slope $-0.5$ in the graph).


\subsection{The Infinite Tree (Bethe Lattice).}
\label{NRsection2}

The computational cost of the tree simulation is more expensive with growing $\chi$ than the line simulation, so I give my results only up to $\chi=8$. I run the imaginary time evolution with 10000 iterations (each iteration followed by several normalization steps) for each point $s$, taking a lower $\chi$ result as the starting point for the procedure. I also add a small symmetry-breaking longitudinal field with magnitude $B_z = 10^{-8}$. 

I see that the energy and its first derivative with respect to $s$ are continuous. However, I now observe a finite discontinuity in the second derivative of the ground state energy (see FIG. \ref{treeEE}), as opposed to the divergence on the infinite line. This happens near $s=s_T \approx 0.5733$.
\begin{figure}
	\begin{center}
	\includegraphics[width=5.5in]{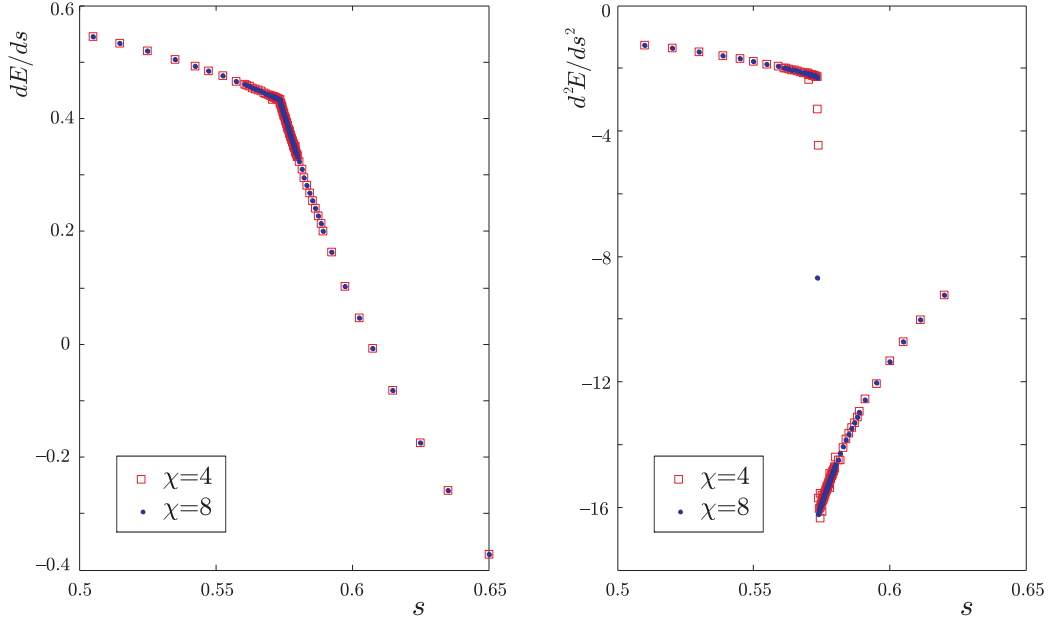}
	\caption{Transverse Ising model on an infinite tree. The first and second derivative with respect to $s$ of the ground state energy obtained via MPS.
	}
	\label{treeEE}
	\end{center}
\end{figure}

\begin{figure}
	\begin{center}
	\includegraphics[width=4.5in]{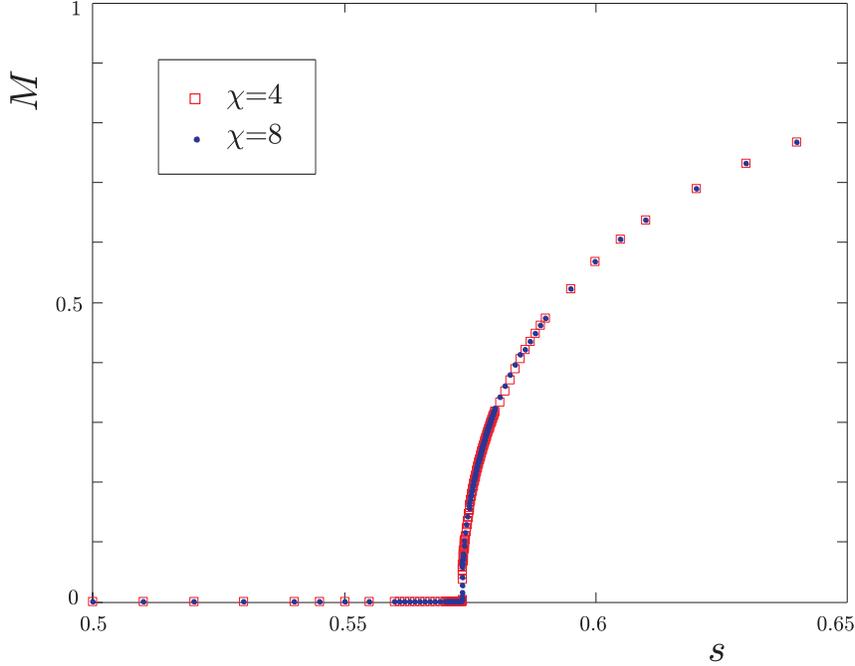}
	\caption{Transverse Ising model on an infinite tree. Magnetization vs. $s$, with $B_z=10^{-8}$.
	}
	\label{treeM}
	\end{center}
\end{figure}
Similarly to the one-dimensional case, the magnetization quickly grows for $s>s_T$,
while it has (nearly) zero value for $s<s_T$ (see FIG.
\ref{treeM}). 
In FIG.\ref{treeMexp}, I plot the magnetization vs. $x-x_T$ on a log-log scale, where
$x$ is 
\begin{eqnarray}
	x = \frac{s}{3(1-s)},
\end{eqnarray}
with the value $x=x_T\approx 0.451$ at the phase transition (where $s=s_T\approx 0.5733$). 
I want to test whether the magnetization behaves like
\begin{eqnarray}
	M\propto (x-x_T)^{\beta}
\end{eqnarray}
for $x$ close to $x_T$, which would appear as a line on the log-log plot. 
As $\chi$ grows, the data is better represented by a straight line closer to the phase transition. 
If I fit the $\chi=8$ data for
$4\times 10^{-4} \leq x-x_L 4\times \leq 10^{-3}$, I get $\beta=0.41$. I add a line
with this slope to the plot.
Note that the mean-field value for the exponent $\beta$ is $0.5$, just as it is for the infinite line.
\begin{figure}
	\begin{center}
	\includegraphics[width=4.5in]{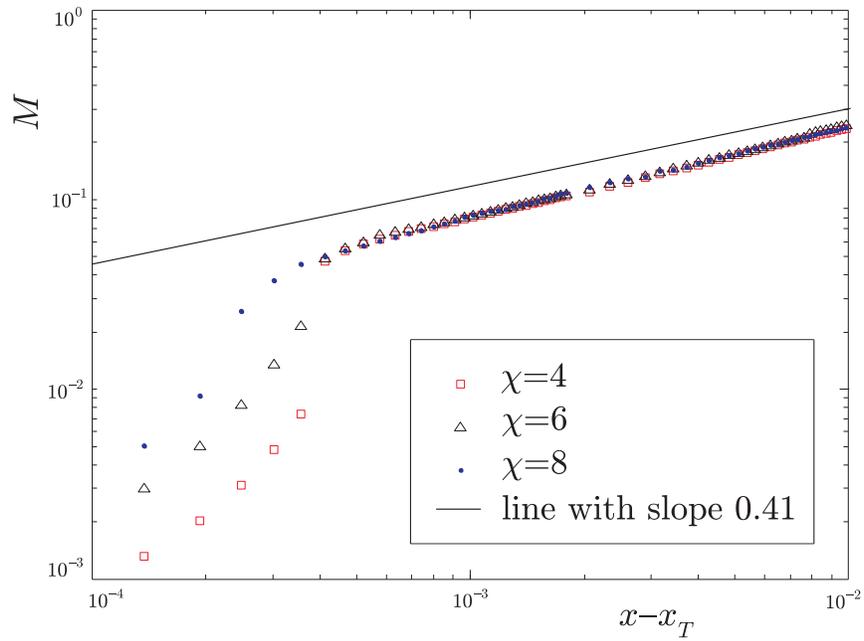}
	\caption{Transverse Ising model on an infinite tree. Log-log plot of magnetization vs. $x-x_T$, with $B_z = 10^{-8}$. 
	I also plot a line with slope $0.41$.
	}
	\label{treeMexp}
	\end{center}
\end{figure}

I observe that the correlation length now rises up only to a finite value (see FIG. \ref{treeCexp}). As I increase $\chi$, the second eigenvalue of the $B$ matrix, $\mu_2$, approaches a maximum value close to $\frac{1}{2}$.
I conjecture that the limiting value of $\mu_2$ is indeed $\frac{1}{2}$, which corresponds to a finite correlation length with value $(\ln 2)^{-1}$. Note that for the infinite line, the second eigenvalue of $B$ approaches $1$, and so the correlation length is seen to diverge at the phase transition.
\begin{figure}
	\begin{center}
	\includegraphics[width=5in]{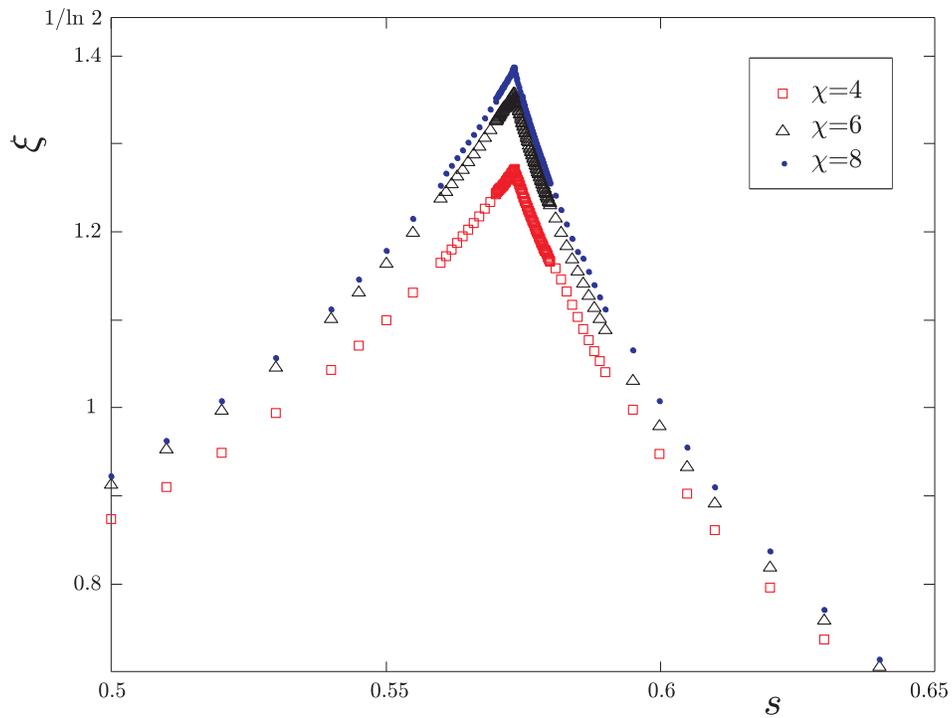}
	\caption{Transverse Ising model on an infinite tree. A linear plot of the correlation length vs. $s$.}
	\label{treeCexp}
	\end{center}
\end{figure}


\section{The Not $00$ Model}
\label{not00section}

I now look at a model with a different interaction term.
Starting with an antiferromagnetic interaction,
I add a specific longitudinal field at each site. As in the previous section, I parametrize the Hamiltonian \eqref{ourH00} with a single parameter $s$:
\begin{eqnarray}
	H_{\textsc{not}\,00} &=& s \sum_{\langle i,j\rangle} 
		\underbrace{\frac{1}{4}\left( 1 + \sigma_z^i + \sigma_z^j + \sigma_z^i \sigma_z^j  \right)}_{H_{ij}}
		+ \frac{b(1-s)}{2} \sum_i \left(1-\sigma_x^{i}\right), \label{Hnot00}
\end{eqnarray}
with $b=2$ on the line and $b=3$ on the tree.
I choose the longitudinal field in such a way that the nearest-neighbor interaction term $H_{ij}$ becomes a projector, expressed in the computational basis as  
\begin{eqnarray}
	H_{ij} = \ket{00}\bra{00}_{ij},
\end{eqnarray}
thus penalizing only the $\ket{00}$ configuration of neighboring spins. Accordingly, I call this model \textsc{not} 00. The ground state of the transverse Ising model \eqref{MPSourH} at $s=1$ has degeneracy 2. For \eqref{Hnot00} on the infinite line or the Bethe lattice, the degeneracy of the ground state at $s=1$ is infinite, as any state that does not have two neighboring spins in state $\ket{0}$ has zero energy.

\subsection{Infinite Line}
I use my numerics to investigate the properties of \eqref{Hnot00} 
on the infinite line as a function of $s$.
\begin{figure}
	\begin{center}
	\includegraphics[width=6in]{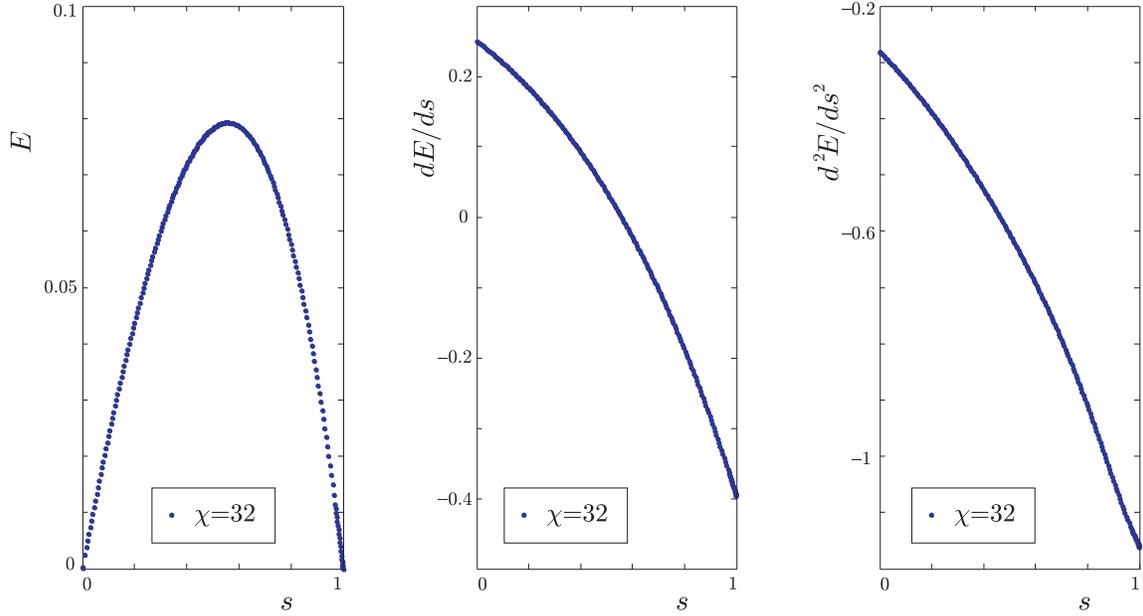}
	\caption{The \textsc{not} $00$ model on an infinite line. The ground state energy and its first two derivatives with respect to $s$.}
	\label{fig00lineE3}
	\end{center}
\end{figure}
\begin{figure}
	\begin{center}
	\includegraphics[width=5in]{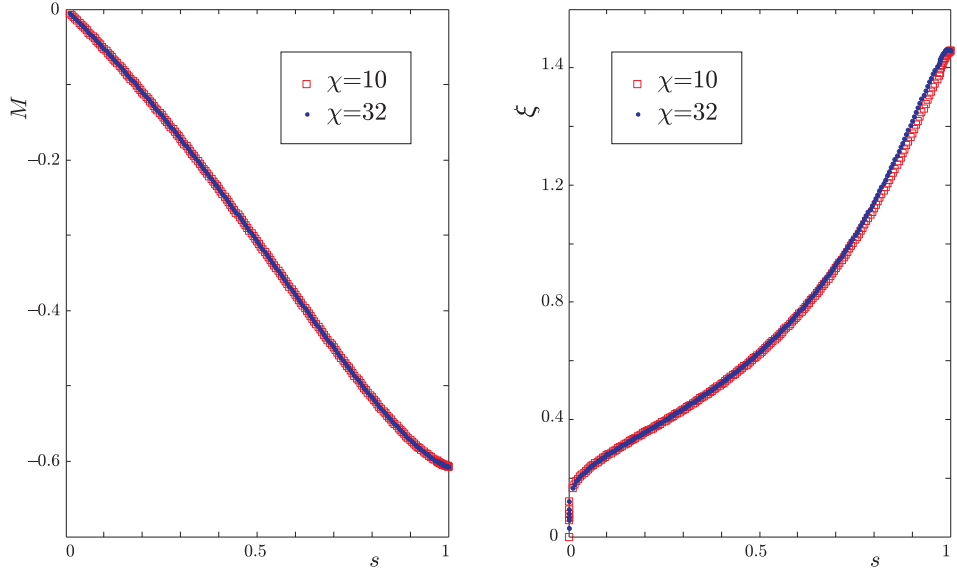}
	\caption{The \textsc{not} $00$ model on an infinite line. Magnetization as a function of $s$ and correlation length as a function of $s$.}
	\label{fig00lineM}
	\end{center}
\end{figure}
My numerical results show continuous first and second derivatives of the energy 
with respect to $s$ (see FIG.\ref{fig00lineE3}). The magnetization $M=\means{\sigma_z}$ decreases 
continuously and monotonically from $0$ at $s=0$ to a final value of $-0.606$ at $s=1$ (see FIG.\ref{fig00lineM}). 
The second eigenvalue of the $B$ matrix \eqref{secondeig} rises continuously from $0$ at $s=0$,
approaching $0.603$ at $s=1$ (see FIG.\ref{fig00lineM}). Because $\mu_2 < 1$, the 
correlation length $\xi$ is finite for all values of $s$ in this case.
These results imply that there is no phase transition for this model as I vary $s$.

As a test of my results, I compute the magnetization at $s=1$ exactly for this model on a finite chain (and ring) of up to $n=17$ spins. I maximize the expectation value of $H_B = \sum_i \sigma_x^i$ within the subspace of all allowed states at $s=1$ (with no two zeros on neighboring spins), thus minimizing the expectation value of the second term in \eqref{Hnot00} for $s$ approaching 1. 
I compute the magnetization $M = \langle \sigma_z^{i}\rangle$
for the middle $i = \lfloor \frac{n}{2} \rfloor$ spin for the ground state of the \textsc{not} 00 model 
exactly
for a finite chain and ring of up to $17$ spins at $s=1$. As I increase $n$, the value of $M$ converges to $-0.603$ (much faster for the ring, as the values of $M$ for $n=14,17$ differ by less than $10^{-4}$). Recall that I obtained $M=-0.606$ from my MPS numerics with $\chi = 16$ for the \textsc{not} 00 model on an infinite line.
I also compare the values of the Schmidt coefficients across the central division of the finite chain ($n=16$) to the elements of the $\lambda$ vector obtained using my MPS numerics with $\chi=32$. I observe very good agreement for the 11 largest values of $\lambda_k$, with the difference that my MPS values keep decreasing (exponentially), while the finite-chain values flatten out at around $\lambda_{k>14} \approx 10^{-9}$ (see FIG.\ref{not00lamfig}). The behavior of the components of $\lambda$ from MPS doesn't change with increasing $\chi$.
\begin{figure}
	\begin{center}
	\includegraphics[width=4in]{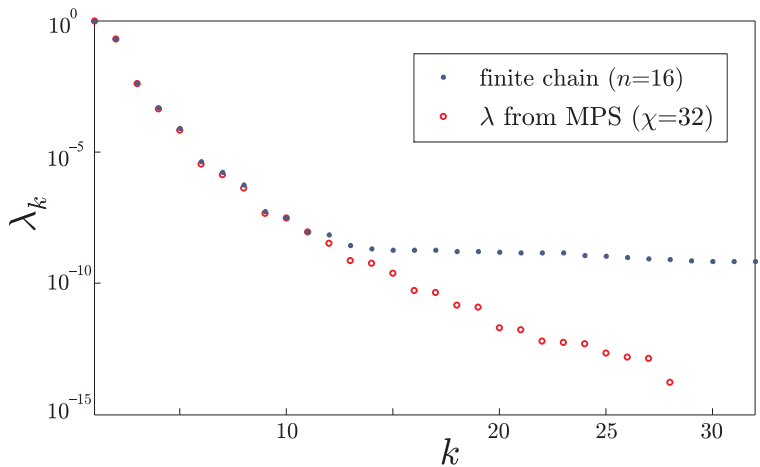}
	\caption{Ground state of the \textsc{not} $00$ model on a line at $s=1$. Comparison of the 	exact Schmidt coefficients for a division across the middle of a finite chain
	 and of the MPS values ($\chi=32$) for an infinite line.}
	\label{not00lamfig}
	\end{center}
\end{figure}

In the ground state of \eqref{Hnot00} at $s=1$, the overlap with the $\ket{00}$ state of any two neighboring spins is exactly $0$. If the $\Gamma$ tensors are the same at every site, the component of the state $\ket{\psi}$ that has overlap with the state $\ket{00}$ on nearest neighbors can be expressed as 
\begin{eqnarray}
	\sum_{a,b,c}   
	\left( \lambda_a \Gamma_{ab}^{0} \lambda_b \Gamma_{bc}^{0} \lambda_c \right) 
	\ket{\phi_a}\ket{00}\ket{\phi_c}.
\end{eqnarray}
Furthermore, when the $\Gamma$ tensors are symmetric, the elements of the $\lambda$ vectors must be allowed to take negative values
to make this expression equal to zero. Note that until now, I used only positive $\lambda$ vectors, knowing that they come from Schmidt decompositions, which give me the freedom to choose the components of $\lambda$ to be positive and decreasing. 

The negative signs in the $\lambda$ vector can be absorbed into every other $\Gamma$ tensor, resulting in a state with two different $\Gamma$ (for the even and odd-numbered sites) and only positive $\lambda$'s. In fact, this is what I observe in my numerics, which assume positive $\lambda$, but allow two different $\Gamma$ tensors (see \ref{TIline}). If I allow the elements of $\lambda$ to take negative values, my numerically obtained $\Gamma$ tensors are identical.

\subsection{Infinite Tree}
Here, I numerically investigate the \textsc{not} $00$ model \eqref{Hnot00} on the Bethe lattice. 
\begin{figure}
	\begin{center}
	\includegraphics[width=6in]{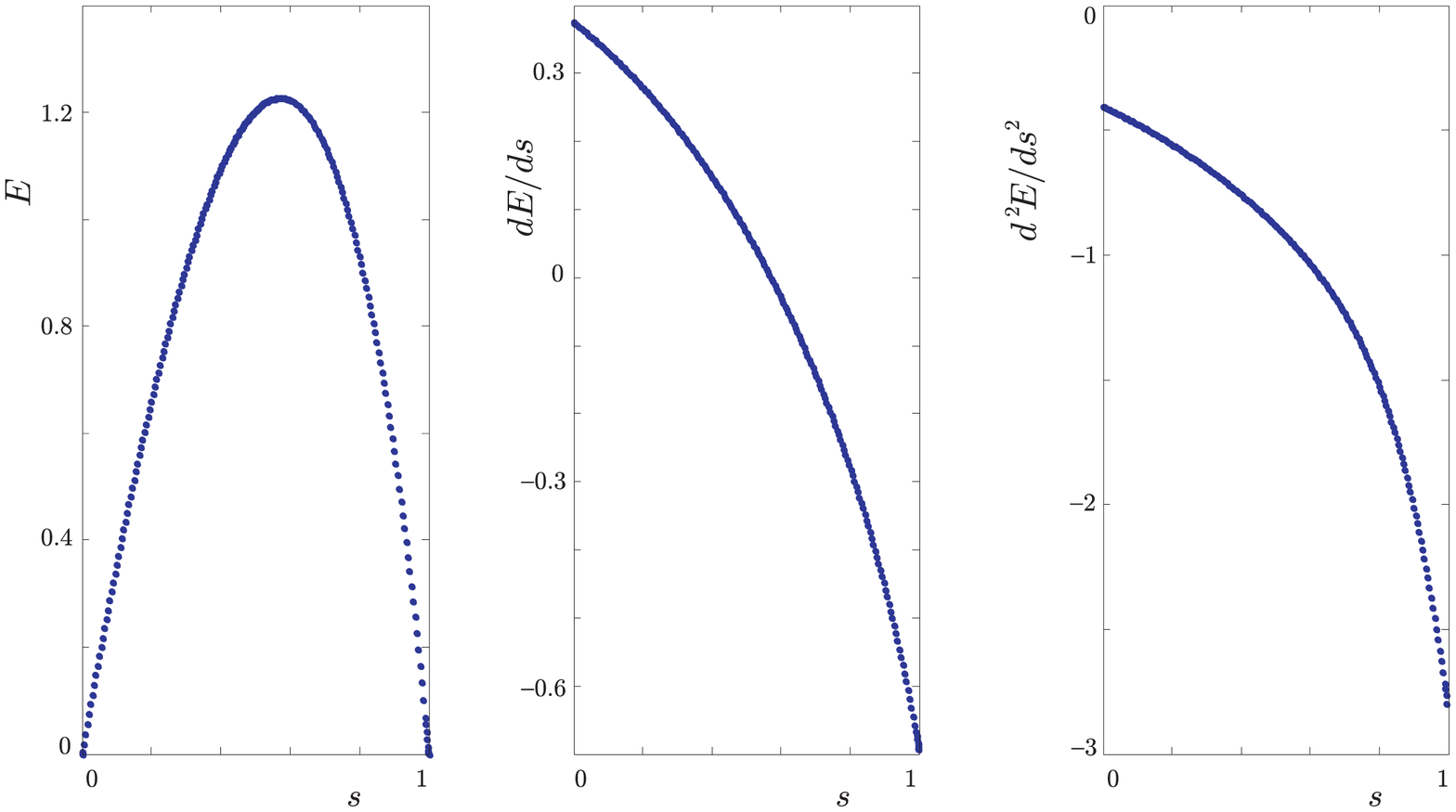}
	\caption{The \textsc{not} $00$ model on an infinite tree. The ground state energy and its first two derivatives with respect to $s$.}
	\label{fig00treeE3}
	\end{center}
\end{figure}
\begin{figure}
	\begin{center}
	\includegraphics[width=5in]{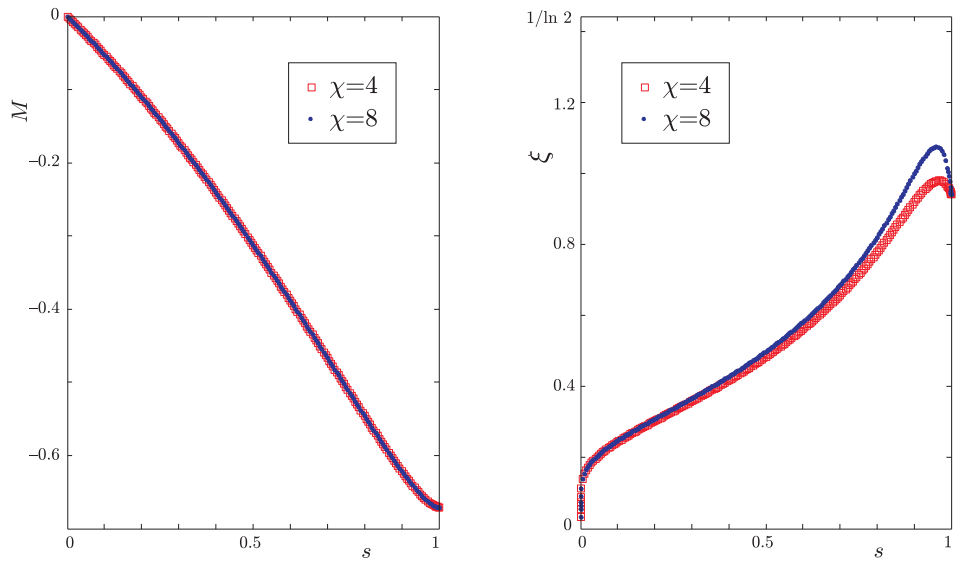}
	\caption{The \textsc{not} $00$ model on an infinite tree. Magnetization as a function of $s$ and correlation length as a function of $s$.}
	\label{fig00treeM}
	\end{center}
\end{figure}
As on the line, the numerics show continuous first and second derivatives of the energy with respect to $s$ (see FIG.\ref{fig00treeE3}) and a continuous decrease in the magnetization from $0$ at $s=0$ to $-0.671$ at $s=1$ (see FIG.\ref{fig00treeM}). The correlation length behaves similarly as on the line, increasing with $s$, but it reaches a maximum at $s=0.96$ for $\chi=8$. The maximum value of $\xi$ is apparently lower than $1/\textrm{ln}\,2$, (see FIG.\ref{fig00treeM}), meaning that 
on the tree, the correlation function $\corrf$ falls off with distance faster than $2^{-|i-j|}$ for all $s$. 



\section{Stability and Correlation Lengths on the Bethe Lattice}
\label{stabilitysection}

I have found that on the Bethe lattice, for both my models, the second eigenvalue $\mu_2$ of the matrix $B$ \eqref{secondeig}, which determines the correlation length, apparently is never greater than $\frac{1}{2}$. In this section I argue that this is a model-independent, and calculation method independent, consequence of assuming that a translation-invariant ground state is the stable limit of a sequence of ground states of finite Cayley trees as the size of the tree grows. 
For a related problem, the stability of recursions for Valence Bond States on Cayley trees has been investigated by Fannes et.al. in \cite{MPS:VBScayley}.

The Hamiltonians \eqref{ourH1} and \eqref{ourH00} each consist of sums of terms 
$H^{(x)}_k$, $H^{(z)}_m$, as in \eqref{trotter}, where each term
$H^{(x)}_k$ depends on a single $\sigma_x$ and each $H^{(z)}_m$ on a neighboring pair of $\sigma_z$. I calculate the quantum partition function
\begin{eqnarray}
	Z(\beta) = \tr e^{-\beta H}
	\label{partf1}
\end{eqnarray}
as the limit of 
\begin{eqnarray}
	Z(N,\Delta t) = \tr \left[ 
		\prod_k e^{-\Delta t H_k^{(x)} }
		\prod_m e^{-\Delta t H_m^{(z)}} \right]^N
	\label{partf2}
\end{eqnarray}
as $\Delta t \rightarrow 0$, $N \rightarrow \infty$ with $N\Delta t = \beta$. To find the properties of the ground state, I take $\beta \rightarrow \infty$ so that I need $Z(N,\Delta t)$ as $\Delta t \rightarrow 0$, $N \rightarrow \infty$ with $N\Delta t \rightarrow \infty$ and $N(\Delta t)^3 \rightarrow 0$ (to make the error in using the Trotter-Suzuki formula go to zero). I interpret \eqref{partf2} as giving the {\em classical} partition function of a system of Ising spins ($s=\pm 1$) on a lattice consisting of $N$ horizontal layers, each of which is a Cayley tree of radius $M$ (i.e with a central node and concentric rings of $3, 3\times2, 3\times 2^2, \dots, 3\times 2^{M-1}$ nodes). 
I can write \eqref{partf2} as
\begin{eqnarray}
	Z(N,\Delta t) = \sum_{\{s\}} \prod A \prod B,
	\label{Zprod}
\end{eqnarray}
where the sum is over all configurations of $N\times(3\times 2^M-2)$ spins $s=\pm 1$
and the products are of a Boltzmann factor $A$ for each horizontal link in the Cayley trees, and a Boltzmann factor $B$ for each vertical link between corresponding nodes in neighboring layers
(see FIG.\ref{figboltzmann})
(layer $N$ is linked to layer $1$ to give the trace).
\begin{figure}
	\begin{center}
	\includegraphics[width=2in]{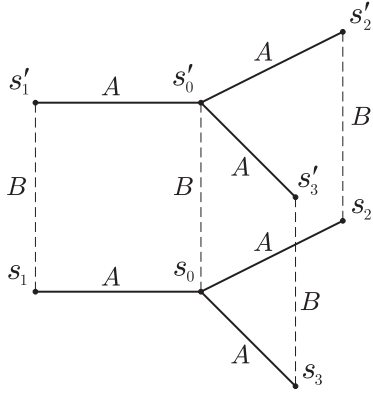}
	\caption{A system of classical spins on a lattice whose layers are Cayley trees. $A$ and $B$ denote the Boltzmann factors.}
	\label{figboltzmann}
	\end{center}
\end{figure}
The factor $A$ for the link between nodes $i,j$ in the same horizontal layer is given by  
\begin{eqnarray}
	A(s_i,s_j) &=& e^{-\Delta t H^{(z)}_{(ij)}(s_i,s_j)}.
	\label{defA} 
\end{eqnarray}
The factor $B$ for the link between nodes $i,i'$ in the same vertical column is given by
\begin{eqnarray}
	B(s_i,s_i') &=& \bra{\sigma_z=s_i'} e^{-\Delta t H^{(x)}_{(i)}} \ket{\sigma_z=s_i}.
	\label{defB}
\end{eqnarray}
When all the terms $H^{(z)}_{(ij)}$ are of the same form, as are all the terms $H^{(x)}_{(i)}$, the form of the factors $A$ and $B$ does not depend on which particular links they belong to.

Each term in the sum, divided by $Z$, can be thought of as the probability of a configuration $\{s\}$. In what follows I will keep $N$ and $\Delta t$ fixed and consider the limit $\Mlim$, i.e. finite Cayley tree $\rightarrow$ Bethe lattice. I will then suppose that my results, which are independent of the form of $A$ and $B$ (provided $A,B > 0$) will also hold after the $N\rightarrow \infty$ limit is taken, i.e. for the quantum ground state.

I will think of the lattice as a single tree, with each node being a vertical column of $N$ spins. (A recent use of this technique to investigate the quantum spin glass on the Bethe lattice is in \cite{MPS:LSSspinglass}.)
Denote by $\vec{s}$ the vector of $N$ values of $s$ along a column. Let $K(\vec{s})$ be the
product of the $N$ factors $B(s,s')$ along a column and let $L(\sve,\svp)$ be the product of 
$N$ factors $A(s,s')$ on the horizontal links between nearest neighbor columns. Let $Z_M$
be the partition function for a tree of radius $M$. 
I can calculate $Z_M$ by a recursion on $M$ as follows:
\begin{eqnarray}
	Z_M &=& \sum_{\sve} K(\sve) [ F_M(\sve) ]^3, \label{Zdefinition} \\
	F_M(\sve) &=& \sum_{\svp} L(\sve,\svp) K(\svp) [ F_{M-1}(\svp) ]^2, 	
			\label{Frecursion} \\
	F_0(\sve) &=& 1.	
			\label{Finit} 
\end{eqnarray}
It is easy to see that this recursion gives the correct $Z_M$ (the case $M=2$ is shown in
 FIG.\ref{figm2tree}).
\begin{figure}
	\begin{center}
	\includegraphics[width=3in]{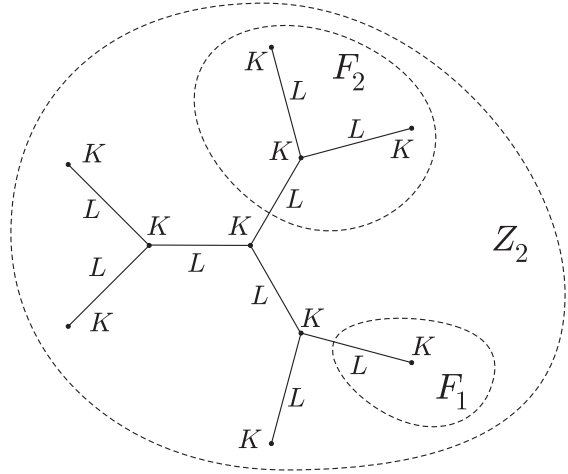}
	\caption{The $M=2$ Cayley tree.}
	\label{figm2tree}
	\end{center}
\end{figure}
I also see that
\begin{eqnarray}
	P_M(\sve) &=& \frac{1}{Z_M} K(\sve) [ F_{M}(\sve) ]^3 
			\label{probability} 
\end{eqnarray}
is the probability of the configuration $\sve$ along the central column.
For the case $N=1$, i.e. classical statistical mechanics on a tree, this is the well-known method to find an exact solution \cite{MPS:classicalexact}.

In order to have a well-defined translationally invariant limit as $\Mlim$ I would like the recursion \eqref{Frecursion} for $F_M$ to have an attractive fixed point $F$ which $F_M$ approaches as $\Mlim$.
`Attractive' means that if I start the recursion with a different $F_0(\sve)$, sufficiently close to $F_0(\sve)=1$, the limiting value of $F_M(\sve)$ will be the same fixed point. This in turn implies that on a Cayley tree with large $M$, small changes in the Hamiltonian on the outer edge will have small effects on the properties of the central region.

First however I need to fix the overall normalization of $F_M(\sve)$, since if $F_M(\sve)$ satisfies \eqref{Frecursion}, so does $a^{2^M} F_M(\sve)$ which rules out an attractive fixed point.

Let
\begin{eqnarray}
	F_M(\sve) &=& Z_M^{\frac{1}{3}} \hat{F}_{M}(\sve),
			\label{Fhat} 
\end{eqnarray}
so that
\begin{eqnarray}
	\sum_{\sve} K(\sve) [ \hat{F}_{M}(\sve) ]^3 &=& 1,
			\label{Fhatnorm} 
\end{eqnarray}
and
\begin{eqnarray}
	P_M(\sve) &=& K(\sve) [ \hat{F}_{M}(\sve) ]^3.
			\label{Phat} 
\end{eqnarray}
The recursion relation becomes
\begin{eqnarray}
	\hat{F}_M(\sve) &=& \lambda_M \sum_{\svp} L(\sve,\svp) K(\svp) 
		[ \hat{F}_{M-1}(\svp) ]^2, 	
			\label{Fhatrecursion}
\end{eqnarray}
with $\lambda_M$ determined by the normalization condition \eqref{Fhatnorm}. 
I can now suppose that
\begin{eqnarray}
	\hat{F}_M(\sve) \rightarrow \hat{F}(\sve) \quad \textrm{as} \quad \Mlim, 	
			\label{Fgoes} 
\end{eqnarray}
with
\begin{eqnarray}
	\hat{F}(\sve) &=& \lambda \sum_{\svp} L(\sve,\svp) K(\svp) [ \hat{F}(\svp) ]^2, 	
			\label{Flimit} 
\end{eqnarray}
and
\begin{eqnarray}
	\sum_{\sve} K(\sve) [ \hat{F}(\sve) ]^3 &=& 1.
			\label{Flimnorm} 
\end{eqnarray}
To determine whether $\hat{F}$ is an {\em attractive} fixed point, let
\begin{eqnarray}
	\hat{F}_M(\sve) &=& \hat{F}(\sve) + f_M(\sve), \label{Fperturb} \\
	\lambda_M &=& \lambda (1 + \ep_M), \label{lperturb}
\end{eqnarray}
with $f_M \rightarrow 0$ and $\ep_M \rightarrow 0$ as $\Mlim$.
To first order in $f_M$, $\ep_M$, \eqref{Fhatrecursion} and \eqref{Fhatnorm} become
\begin{eqnarray}
	f_M(\sve) &=& \ep_M \hat{F}(\sve) + 2 \sum_{\svp} T(\sve,\svp) f_{M-1}(\svp),
		\label{frecursion} \\
	\sum_{\sve} K(\sve) [ \hat{F}(\sve)]^2 f_M(\sve) &=& 0,
		\label{frecursion2}	
\end{eqnarray}
where
\begin{eqnarray}
	T(\sve,\svp) = \lambda L(\sve,\svp) K(\svp) \hat{F}(\svp).
			\label{Tdef} 
\end{eqnarray}
From \eqref{Flimit},
\begin{eqnarray}
	\sum_{\svp}  T(\sve,\svp) \hat{F}(\svp) = \hat{F}(\sve),
	\label{TFhat}
\end{eqnarray}
and since $L(\sve,\svp) = L(\svp,\sve)$,
\begin{eqnarray}
	\sum_{\svp}  K(\svp) [ \hat{F}(\svp) ]^2 T(\svp,\sve) 
		= K(\sve) [\hat{F}(\sve) ]^2,
\end{eqnarray}
i.e. the linear operator $T$ has an eigenvalue one, with right eigenvector $\hat{F}$ and left eigenvector $K\hat{F}^2$ (which from \eqref{Flimnorm} have scalar product one).
\eqref{frecursion} now gives
\begin{eqnarray}
	\sum_{\sve} K(\sve) [ \hat{F}(\sve) ]^2 f_M(\sve) 
		= \ep_M  + 2\sum_{\sve} K(\sve) [ \hat{F}(\sve) ]^2 f_{M-1}(\sve),
\end{eqnarray}
so from \eqref{frecursion2}, $\ep_M = 0$. Let 
\begin{eqnarray}
	T^{\perp}(\sve,\svp) = T(\sve,\svp) - \hat{F}(\sve) K(\svp) [\hat{F}(\svp)]^2,
	\label{Tperp1}
\end{eqnarray}
so that 
\begin{eqnarray}
	\sum_{\svp} T^{\perp}(\sve,\svp) \hat{F}(\svp) &=& 0,
	\label{Tperp2}
\end{eqnarray}
and
\begin{eqnarray}
	\sum_{\svp} K(\svp) [ \hat{F}(\svp) ]^2 T^{\perp}(\svp,\sve) 
		= 0.
	\label{Tperp2b}
\end{eqnarray}
\eqref{frecursion} now becomes 
\begin{eqnarray}
	f_M(\sve) = 2 \sum_{\svp} T^{\perp}(\sve,\svp) f_{M-1}(\svp).
		\label{ffromT}
\end{eqnarray}
\eqref{ffromT} shows that $\hat{F}(\sve)$ is an attractive fixed point if and only if 
\begin{eqnarray}
	\norms{T^{\perp}} < \frac{1}{2}
	\label{Tbound}
\end{eqnarray}
(for a tree with valence $p+1$ at each vertex, $\frac{1}{2}$ is replaced by $\frac{1}{p}$).

I can in fact prove that there does exist an $\hat{F}(\sve)$ satisfying \eqref{Flimit} and \eqref{Flimnorm}, for which the corresponding $T^\perp$ has a maximum eigenvalue less than $\frac{1}{2}$. Define a function $\Phi$ of $\hat{F}(\sve)$ by
\begin{eqnarray}
	\Phi[\hat{F}] = \sum_{\sve,\svp} [\hat{F}(\sve)]^2 K(\sve)
		L(\sve,\svp)K(\svp)[\hat{F}(\svp)]^2.
\end{eqnarray}
I look for a maximum of $\Phi$ with $\hat{F}(\sve)$ restricted to the region
\begin{eqnarray}
	\sum_{\sve} K(\sve)[\hat{F}(\sve)]^3 &=& 1,\\
	\hat{F}(\sve) &\geq& 0.
\end{eqnarray}
A maximum must exist, but it might be on the boundary of the region, i.e. it might have $\hat{F}(\sve)=0$ for some values of $\sve$. Elementary calculations (omitted here) establish that stationary values of $\Phi$ on the boundary cannot be maxima. At stationary points 
in the interior of the region, i.e. with $\hat{F}(\sve)>0$ for all $\sve$, \eqref{Flimit} and \eqref{Flimnorm} must be satisfied. If such a stationary point is a maximum, all the eigenvalues of $1-2T^\perp$ are $\geq 0$, i.e. all the eigenvalues of $T^\perp$ are $\leq \frac{1}{2}$. This is weaker than the attractive fixed point condition, which also requires that no eigenvalue is less than $-\frac{1}{2}$, but does correspond to the observed property of $\mu_2$.

I now examine the joint probability distribution of $\vec{s}_0$ and $\vec{s}_d$, where $0$ denotes the central column and $d$ a column distance $d$ from the center.
\begin{figure}
	\begin{center}
	\includegraphics[width=4in]{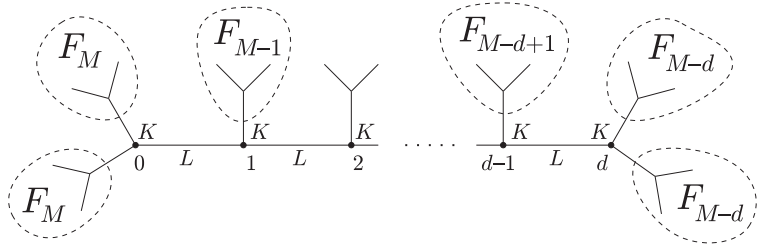}
	\caption{Computing the probability distribution of $\vec{s}_0$ and $\vec{s}_d$.}
	\label{figdistanced}
	\end{center}
\end{figure}
One can see from FIG.\ref{figdistanced} that
\begin{eqnarray}
	P_M(\vec{s}_0,\vec{s}_d) = \frac{1}{Z_M} \sum_{\vec{s}_1,\cdots,\vec{s}_{d-1}}
		[ F_M(\vec{s}_0) ]^2 
			K(\vec{s}_0) L(\vec{s}_0,\vec{s}_1) 
			F_{M-1}(\vec{s}_1)
			\dots && \\
			\dots L(\vec{s}_{d-1},\vec{s}_d) K(\vec{s}_d)
			[ F_{M-d}(\vec{s}_d) ]^2.&& \nonumber
\end{eqnarray}
As $\Mlim$ (with $d$ fixed) this becomes (using \eqref{Tdef})
\begin{eqnarray}
	P(\vec{s}_0,\vec{s}_d) &=& 
		C [ \hat{F}(\vec{s}_0) ]^2 K(\vec{s}_0) 
		T^{d} (\vec{s}_0,\vec{s}_d)
		 \hat{F}(\vec{s}_d),
	 \label{PfromT}
\end{eqnarray}
where the normalization $C$ is determined by 
\begin{eqnarray}
	\sum_{\vec{s}_0,\vec{s}_d} P(\vec{s}_0,\vec{s}_d) = 1.
\end{eqnarray}
Using \eqref{TFhat} I find
\begin{eqnarray}
	P(\vec{s}_0) = \sum_{\vec{s}_d} P(\vec{s}_0,\vec{s}_d) 
		= C K(\vec{s}_0) [ \hat{F}(\vec{s}_0) ]^3,
\end{eqnarray}
so from \eqref{Flimnorm}, $C=1$ (and $P(\vec{s}_0)$ agrees with the limit of \eqref{Phat}). Expressing \eqref{PfromT} in terms of $T^{\perp}$ using \eqref{Tperp1}, \eqref{Tperp2} and \eqref{Tperp2b}, 
\begin{eqnarray}
	P(\vec{s}_0,\vec{s}_d) - P(\vec{s}_0) P(\vec{s}_d)
	 	&=& K(\vec{s}_0) [ \hat{F}(\vec{s}_0) ]^2
			(T^{\perp})^d (\vec{s}_0,\vec{s}_d)  \hat{F}(\vec{s}_d).
\end{eqnarray}
Thus the correlation between $\vec{s}_0$ and $\vec{s}_d$ falls off as $\mu^{d}$, where $\mu$ is the eigenvalue of $T^\perp$ with maximum modulus, and so from \eqref{Tbound}, faster than $1/2^d$.

If this conclusion is correct (and clearly the argument is less than rigorous), it establishes more than my experimental observation that $\mu_2<\frac{1}{2}$. The quantum limit of $P(\vec{s}_0,\vec{s}_d)$ encodes not only the static correlation 
$\bra{\psi_0} s_0 s_d \ket{\psi_0}$ in the ground state, but also the 
imaginary time dependent correlation
$\bra{\psi_0}e^{Ht}s_0 e^{-Ht} s_d \ket{\psi_0}$ which in turn determines the linear response as measured by $s_d$ to a time-dependent perturbation proportional to $s_0$. If this indeed falls off faster than $1/2^d$, then there is some hope that the Bethe lattice can be used as a starting point for investigation of fixed valence random lattices.

\chapter{Quantum Satisfiability}\label{chQsat}

The topic of this chapter is {\em Quantum Satisfiability}, the quantum analog of classical {\em Satisfiability}. I discussed the locally constrained classical problem in Section \ref{ch1:local}, where I also summarized the well known complexity of its many variants (see Table \ref{SATtable} in Section \ref{ch1:clascomplex}). On the other hand, there are still several open questions about Quantum Satifiability, and I choose to investigate some of them here.

First, in Section \ref{ch4:qks} I introduce Quantum Satisfiability and summarize what is presently known about its complexity. Then I review the proof techniques I use throughout this chapter, focusing on several clock constructions in Section \ref{ch4:clocks}. 
Section \ref{ch4:q3} is based on the paper \cite{lh:new3local}

\mypaper{A New Construction for a QMA Complete 3-local Hamiltonian}{Daniel Nagaj, Shay Mozes}{
We present a new way of encoding a quantum computation into a 3-local Hamiltonian. Our construction is novel in that it does not include any terms that induce legal-illegal clock transitions. Therefore, the weights of the terms in the Hamiltonian do not scale with the size of the problem as in previous constructions. This improves the construction by Kempe and Regev \cite{lh:KR03}, who were the first to prove that 3-local Hamiltonian is complete for the complexity class QMA, the quantum analogue of NP. 

Quantum k-SAT, a restricted version of the local Hamiltonian problem using only projector terms, was introduced by Bravyi \cite{lh:BravyiQ2SAT06} as an analogue of the classical k-SAT problem. Bravyi proved that quantum 4-SAT is complete for the class QMA with one-sided error (QMA$_1$) and that quantum 2-SAT is in P. We give an encoding of a quantum circuit into  a quantum 4-SAT Hamiltonian using only 3-local terms.  As an intermediate step to this 3-local construction,  we show that quantum 3-SAT for particles with dimensions $3\times 2\times 2$  (a qutrit and two qubits) is QMA$_1$ complete. The complexity of quantum 3-SAT with qubits remains an open question. 
}
One of the implications of this construction is that Adiabatic Quantum Computation
can be used for simulating quantum circuits much more effectively than 
Aharonov et al. showed in \cite{AQC:AvDKLLR05} (see Section \ref{ch2:aqcbqp}). 

Next, in Section \ref{ch4:q1}, I review the result of Aharonov et al. \cite{CA:GottesmanLine} who proved that Quantum 2-SAT on a line with 12-dimensional particles is QMA$_1$ complete. I decrease the required dimensionality of particles in their construction to $d=11$ in Section \ref{ch4:d11}.

In Section \ref{ch4:q2general}, I present a very recent result by Eldar and Regev, who proved that Quantum 2-SAT for a cinquit-qutrit pair (particles with $d=5$ and $d=3$) is QMA$_1$ complete, using a novel {\em triangle clock} construction. In Section \ref{ch4:triangleXZ}, I use their result to obtain another QMA-complete 3-local Hamiltonian construction built from constant norm, 3-local terms of a restricted type. 

Finally, I take yet another step towards figuring out the complexity of Quantum 3-SAT in Section \ref{ch4:train}. Using my {\em train switch} clock construction, I show that a simple Quantum 3-SAT Hamiltonian is universal for quantum computation in the Adiabatic Quantum Computing model of Chapter \ref{ch2adiabatic} and in the more general Hamiltonian Computer model of Section \ref{ch2:hc}. Moreover, the required running time for this model scales only slightly worse than linearly with the length of the computation.


\section{Introduction}

\subsection{Quantum k-SAT}
\label{ch4:qks}

The Quantum $k$-SAT (Q-$k$-SAT) promise problem was introduced by Bravyi \cite{lh:BravyiQ2SAT06} as an analogue of classical $k$-SAT. In classical {\em Satisfiability}, one needs to determine whether there exists a bit string satisfying all the boolean clauses of a given problem instance. In Q-$k$-SAT, the problem is to determine whether a 
Hamiltonian acting on $n$ qubits has a zero eigenvalue, or whether all its eigenvalues are higher than $\epsilon \geq n^{-\alpha}$ for some constant $\alpha$. Moreover, the Hamiltonian 
\begin{eqnarray}
	H_{\textrm{Q-}k\textrm{-SAT}} = \sum P_i, \label{QKS}
\end{eqnarray}
is composed of $k$-local projector terms
\begin{eqnarray}
	P_i = P_i^2 
		=  \ii^{\otimes(n-k)} \otimes  \ket{\psi_i}\bra{\psi_i}_{\{q^{i}_1\dots q^{i}_k\}},
\end{eqnarray} 
where each $P_i$ acts nontrivially on $k$ qubits $\{q^{i}_1\dots q^{i}_k\}$. 
For a `yes' answer to a problem instance, the ground state of $H_{\textrm{Q-}k\textrm{-SAT}}$ must be {\em exactly zero}, i.e. a state which is annihilated by {\em all} the projectors $P_i$ must exist. This shows the analogy 
to classical SAT, where all the boolean clauses have to be satisfied. 
If the states $\ket{\psi_i}$ are computational basis states and $\epsilon=1$, the problem reduces to classical $k$-SAT. The special case of commuting projectors $P_i$ has been analyzed in \cite{lh:BravyiCommute05}.

Local Hamiltonian, the problem I introduced in Section \ref{ch1:Kitaev5}, is more general than Q-$k$-SAT. It is the quantum analogue of classical MAX-$k$-SAT, where one is interested in the properties of the ground state of the {\em total} Hamiltonian (the {\em sum} of the terms). 
Viewing Quantum $k$-SAT and $k$-local Hamiltonian as quantum analogues of classical $k$-SAT and MAX-$k$-SAT, it is interesting to investigate  and compare their complexities. Known complexity results about the classical problems are summarized in Section \ref{ch1:clascomplex}, Table \ref{SATtable}. In Table \ref{QSATtable}, I present the currently known results about Quantum Satisfiability and Local Hamiltonian problems. In this Chapter, I prove those marked with $^*$.

\begin{table}
\begin{center}
\begin{tabular}{|r|l|l|}
\hline 
Quantum & \ qubits & \ qudits \\
\hline 
Q-$k$-SAT & \ $k=2$ : in P	& \ Q-$(5,3)$-SAT : QMA$_1$-complete \\
	  & 		& \ Q-$(11,11)$-SAT in 1D : \\
	  & 		& \ \qquad \qquad \qquad QMA$_1$-	complete$^*$ \\
& \ $k= 3$ : contains NP, & \ \\
	& \ \qquad \qquad \qquad universal for BQP$^*$ & 
	\ Q-$(3,2,2)$-SAT : QMA$_1$-complete$^*$ \ 
 \\
& \ $k\geq 4$ : QMA$_1$-complete & \ \\
\hline	
$k$-local  & \ $k\geq 2$ : QMA-complete &  \\
Hamiltonian  & \ \qquad \qquad \qquad even on a 2D grid &  \\
 & \ $k=3$ with constant norm terms : &  \\
 & \ \qquad \qquad \qquad QMA-complete$^*$ &  \\
\hline
\end{tabular}
\caption{Known complexity of Quantum Satisfiability and Local Hamiltonian problems. My own results, presented later in this Chapter, are marked with $^ *$.}
\label{QSATtable}
\end{center}
\end{table}

In \cite{lh:BravyiQ2SAT06}, Bravyi gave a classical polynomial algorithm for Quantum $2$-SAT showing that Q-$2$-SAT belongs to P. He then proved that Quantum $k$-SAT is QMA$_1$ complete for $k\geq 4$. The class QMA$_1$ is a special case of QMA, where only one sided error is allowed (see Definition \ref{ch1:qma1} in Section \ref{ch1:qcomplex}). When the answer to the problem instance is `yes', the verifier circuit $U$ can output `yes' on some state with certainty. Thus, the ground state energy of a Quantum $k$-SAT Hamiltonian must be exactly zero in the `yes' case. When I release the certainty requirement on the verifier circuit $U$, it becomes a verifier circuit for a QMA problem instead. Therefore, when I prove that a certain Q-SAT problem is QMA$_1$ complete, the QMA-completeness of the mother Local-Hamiltonian problem immediately follows as a side result. I present a result of this type in Section \ref{ch4:q3}, where I give a new 3-local Hamiltonian construction with constant norm terms.

Bravyi's original definition also required all of the terms in the Hamiltonian to be projectors. However, using $k$-local positive semidefinite operator terms $H^{+}_i$ with zero ground state and constant norm instead of projectors $P_i$ in \eqref{QKS} is an equivalent problem. Quantum $k$-SAT with positive semidefinite operators contains Q-$k$-SAT with projectors. On the other hand, if one is able to solve Q-$k$-SAT with projectors, one can solve Q-$k$-SAT with positive semidefinite operators as well. 
For each positive semidefinite operator $H^{+}_i$, define a projector $P^{+}_i$ with the same ground state subspace. If $H_P = \sum P^{+}_i$ has a zero ground state, so does $H_{+}=\sum H^{+}_i$. If the ground state energy of $H_P = \sum P^{+}_i$ is greater than $\ep$, the ground state energy of $H_{+}$ is greater than $c\ep$, where $c$ is a constant. Therefore, instead of only projector terms, I can equivalently use positive semidefinite operators in Q-$k$-SAT Hamiltonians.


\subsection{Constructing the Hamiltonians}

The two types of results about Quantum k-SAT Hamiltonians that I present here are proofs of QMA (or QMA$_1$) completeness and proofs of universality of using a Hamiltonian for simulating quantum circuits. In both cases, the quest is to prove these results for ever simpler Hamiltonians. Better locality (lower $k$, the number of particles involved in each term), simpler geometry of interactions (2D grid, 1D), lower dimensionality of particles involved, and using only terms of a restricted type 
are what I am aiming at.

To prove that a Hamiltonian problem is QMA-complete, one first needs to show that it belongs to QMA, and then to show that it also contains QMA. I have presented Kitaev's proof of QMA-completeness for the 5-Local Hamiltonian problem in Section \ref{ch1:Kitaev5}, and because a lot of the results I present in the next two Chapters build on it, I now quickly summarize it. 
Before considering the special case of $5$-local Hamiltonian, Kitaev proved that $k$-local Hamiltonian is in QMA for any constant $k$. Therefore, the question whether a Hamiltonian problem belongs to QMA
is usually answered quickly by a reduction to $k$-Local Hamiltonian.
For the other direction, one needs to reduce every instance of a problem in QMA into an instance of the Hamiltonian problem in question. Every instance of a problem in QMA is a `yes/no' question. Moreover, every problem instance has a verifier circuit $U$ which verifies claims (supposed proofs) about the answer. This circuit must have the completeness and soundness properties. Completeness means that if the answer to the problem instance is `yes', there exists a state that the circuit $U$ will accept with high probability. The soundness property means that if the answer to the problem is `no', there is no state that fools the verifier circuit $U$ with high probability. I then need to construct a Hamiltonian whose ground state would have low energy (or exactly zero energy for Quantum $k$-SAT) if the answer is `yes', and whose ground state would have considerably higher energy if the answer is `no'. 

All the results in this Chapter use the same general form of this Hamiltonian, but differ in the implementation of some of the terms. 
Let me first look at what they have in common and focus on the differences in Section \ref{ch4:clocks}.
The verifier quantum circuit $U$ with $L$ gates acts on $n$ qubits. Consider now
a quantum system whose Hilbert space $\cH$ is larger than $\cH_{work} = \cC^{\otimes n}$. A way to do this is to append a clock register to the system as
\begin{eqnarray}
	\cH &=& \cH_{work} \otimes \cH_{clock}.
	\label{ch4:workclock}
\end{eqnarray}
Then, encode the progression of the circuit $U$
\begin{eqnarray}
	\ket{\psi_t} = U_t U_{t-1} \dots U_2 U_1 \ket{\psi_0}, 
	\quad t=0,\dots,L
	\label{ch4:Uprogress}
\end{eqnarray}
into states 
\begin{eqnarray}
	\ket{\Psi_t} = \ket{\psi_t} \otimes \ket{t}_c
\end{eqnarray}
of the larger system. How exactly the states $\ket{t}_c$ and transitions between them are implemented constitutes a `clock construction' and I explain the possibilities in detail in Section \ref{ch4:clocks}.
Following Kitaev, the Hamiltonian I then construct for the larger system is 
\begin{eqnarray}
	H = H_{prop} +  H_{out} + H_{input} + H_{clock} + H_{clockinit}.
	\label{ch4:completeHam}
\end{eqnarray}
It is built in such a way that the history state for the quantum circuit $U$ corresponding to some initial state $\ket{\psi_0}$
\begin{eqnarray}
	\ket{\Psi_{history}} = \frac{1}{\sqrt{L+1}} \sum_{t=0}^{L} \ket{\Psi_t}
		= \frac{1}{\sqrt{L+1}} \sum_{t=0}^{L} \ket{\psi_t} \otimes \ket{t}_c
		\label{ch4:psihistory}
\end{eqnarray}
is close to the ground state of $H$. In fact, for a `yes' instance of Quantum $k$-SAT, it {\em is} the ground state of $H$. 

The first term, $H_{prop}$ is a Hamiltonian whose ground state has the form $\ket{\Psi_{history}}$:
\begin{eqnarray}
	H_{prop} 
	&=& 
		\sum_{t=1}^{L} H_{prop}^t \\
	H_{prop}^t	&=& 
	\half \left(\ii \otimes \left( P_t + P_{t-1} \right)
						- U_t \otimes X_{t,t-1} - \left(U_t \otimes X_{t,t-1}\right)^\dagger \right),
\end{eqnarray}
where $P_t = \ket{t}\bra{t}_c$ is a projector onto the state $\ket{t}_c$ of the clock register, the operator  $X_{t,t-1} = \ket{t}\bra{t-1}_c$ increases the time register from $t-1$ to $t$ and the unitary gate $U_t$ acts on the corresponding work qubits.
Observe that it is built from terms which check whether the amplitude of each $\ket{\Psi_{t-1}}$ and $\ket{\Psi_{t}}$ is the same. Therefore, the history state \eqref{ch4:psihistory} is its eigenvector with eigenvalue 0. 

The last of the $n$ work qubits, $q_n$, is the output qubit for the circuit $U$. The term 
\begin{eqnarray}
	H_{out} = \ket{0}\bra{0}_{q_n} \otimes P_L
\end{eqnarray}
takes care of whether the circuit $U$ accepts some state or not. It gives an energy penalty to any history state built from a state $\ket{\psi_{no}}$ that is not accepted by $U$, i.e. whose output qubit after $L$ gates is not in the state $\ket{1}$. The term 
\begin{eqnarray}
	H_{input} = \sum_{k\in ancilla}\ket{1}\bra{1}_{q_k} \otimes P_0
\end{eqnarray}
takes care of proper initialization of ancilla qubits, in the case they are required for $U$. The final two terms $H_{clock}$ and $H_{clockinit}$ add an energy penalty to all states of the clock register $\cH_{clock}$ that are not in the subspace spanned by the states $\ket{t}_c$. The implementation of these two terms and how the states $\ket{t}_c$ are actually encoded is the topic of the next Section. 

For the other type of result, universality of a Hamiltonian $H$ for simulating a quantum circuit, one needs to show that time evolution (for not too long a time) from a simple initial state with the Hamiltonian $H$ can produce the output of the circuit $U$ with high probability. This type of proof again requires the encoding of the progression of a quantum circuit $U$ into states $\ket{\Psi_t}$ of a larger system. However, proving this type of result is different in that I need to analyze the dynamics of the system, not just the ground state of $H$. However, what simplifies the problem is that I can choose the initial state myself. I present one such result in Section \ref{ch4:train}. Later, the whole Chapter \ref{ch5hqca} is dedicated to showing universality for Hamiltonians which are both time and translation-invariant.


\subsection{Constructing Clocks}\label{ch4:clocks}
In this Section, I review several methods of implementing the clock register in \eqref{ch4:workclock}. The goal of a {\em clock construction} is to encode the progression \eqref{ch4:Uprogress} of a quantum circuit $U$ with $L$ gates on $n$ work qubits into a sequence of orthogonal states $\ket{\Psi_t}$ of a larger system $\cH$. It is desirable for a clock construction to have:
\begin{itemize}
\item Easily implementable unique transitions 
\end{itemize}
One needs to construct a Hamiltonian inducing transitions between the states $\ket{\Psi_t}$. It is important to ensure that each state $\ket{\Psi_t}$ can transition directly only to the states $\ket{\Psi_{t+1}}$ or $\ket{\Psi_{t-1}}$, i.e. 
\begin{eqnarray}
	H \ket{\Psi_t} &=& \alpha_{t-1} \ket{\Psi_{t-1}} 
	+ \beta_t \ket{\Psi_t}
	+ \gamma_{t+1} \ket{\Psi_{t+1}},
\end{eqnarray}
and not into any other state. Undesirable transitions can be suppressed by adding large penalty terms as in the work of Kempe and Regev \cite{lh:KR03}, but such constructions can not be used for proving results about Quantum $k$-SAT. In Sections \ref{ch4:q3}-\ref{ch4:train} I give ways to construct the clock register with unique transition rules, allowing me to prove new results about Q-$k$-SAT and strengthen old results about the Local Hamiltonian problem. 

\begin{itemize}
\item Locally checkable encoding, initialization and final detection
\end{itemize}
The full Hilbert space $\cH$ can be much larger than the subspace $\cH_{legal}$ spanned by the states $\ket{\Psi_t}$. I must be able to detect whether a state 
belongs to the subspace $\cH_{legal}$ using only local terms. In practice, this is done by adding an energy penalty to any state outside $\cH_{legal}$. Also, I need to be able to make a local projection measurement onto $\ket{\Psi_L}$ and $\ket{\Psi_L}^\perp$.

\begin{itemize}
\item Local transition rules involving only a few (clock) particles,
\item Low dimensionality of the (clock) particles,
\end{itemize}
Simple local interactions in the system $\cH$ and low dimensionality of the particles involved are essential for my goal of obtaining stronger results about ever simpler local Hamiltonians.

\begin{itemize}
\item Simple geometry of interactions,
\end{itemize}
Finally, for practical implementation, the $k$-particle interactions in $H$ need to involve only spatially close particles. The 1D construction of \cite{CA:GottesmanLine} presented in Section \ref{ch4:q1} is one example of this.

Let me now present a few of clock register encodings and discuss their advantages and disadvantages.

\subsubsection{Adding a clock register}
The first option is to directly add a clock register to the system as in Feynman's construction in Section \ref{ch1:FeynmanHC} 
\begin{eqnarray}
	\cH = \cH_{work} \otimes \cH_{clock}.
\end{eqnarray}
A set of orthogonal states $\ket{t}_c$ of the clock register is then used to label the progress of the computation, while the contents of the work register holds the state of the work qubits after $t$ gates have been applied to an initial state $\ket{\psi_0}$
\begin{eqnarray}
	\ket{\Psi_t} &=& \underbrace{\left(
		U_t U_{t-1} \dots U_1 \ket{\psi_0}
		\right)}_{\ket{\psi_t}}
	\otimes \ket{t}_c.
	\label{ch4:psit}
\end{eqnarray}
The space spanned by the clock states $\ket{t}_c$ then defines the {\em legal clock subspace} $\cH_{legal}$. An essential ingredient of the clock constructions I present here is a local way of checking whether a state of the clock register belongs to the legal clock subspace. It is usually implemented by projector terms adding an energy penalty to states outside of $\cH_{legal}$.

One could think of implementing the clock register by simply labeling the basis states of $\cH_{clock}$ by $\ket{t}_c$. The minimum required number of qubits in $\cH_{clock}$ would then be $\log L$. However, such a construction is not local, as the transition rules between different clock states necessarily involve $(\log L)$ qubits.

The simplest local implementation of the clock register is the {\em pulse clock} used by Feynman \cite{lh:Feynman}. A state corresponding to time $t$ is the state of $L+1$ qubits $c_0,\dots,c_L$, with a single clock qubit $c_{t}$ in the state $\ket{1}$
\begin{eqnarray}
	\ket{t}_c = \ket{00 \dots 00100 \dots 00}.
	\label{ch4:clockpulse}
\end{eqnarray}
This single up-spin denotes the {\em active site} in the clock register. To transition between clock states
$\ket{t}_c$ and $\ket{t+1}_c$ (and vice versa), I need to look only at 2 clock qubits
and apply the operator
\begin{eqnarray}
	X_{t+1,t} + X_{t+1,t}^\dagger = \ket{01}\bra{10}_{c_{t},c_{t+1}} + \ket{10}\bra{01}_{c_{t},c_{t+1}}.
	\label{ch4:pulsetran}
\end{eqnarray}
This low locality of transitions is a big advantage of the pulse clock. Proper encoding of the clock register is assured by operators $\ket{11}\bra{11}$ which give an energy penalty to incorrect clock states with two active spots (two spin up particles) in the system. These are also 2-local. On the other hand, the big disadvantage of the pulse clock is that it is impossible to check whether the clock is not in the state $\ket{D} = \ket{0\dots0}_c$, without a single spin up, using only local, projector terms. This `dead' state is annihilated by the transition operators \eqref{ch4:pulsetran}, so it does not `move' anywhere. For this reason, the pulse clock is suited only for constructing Hamiltonian computers (see Section \ref{ch2:hc}), where one can pick the initial state of the clock and thus rule out the state $\ket{D}$. It is unusable for proving QMA-completeness results, unless one rules out the state $\ket{D}$ by adding terms that energetically favor a single spin up in the system over the state $\ket{D}$.

I already presented Kitaev's unary {\em domain wall} clock in Section \ref{ch1:Kitaev5}
\begin{eqnarray}
	\ket{t}_c = \ket{11\dots 1100 \dots 00}.
	\label{ch4:clockwall}
\end{eqnarray}
The active spot in the clock register is now determined by the position of the domain wall $10$. The unique transition rules for this clock are now 3-local, $\ket{\dots 100\dots}\leftrightarrow \ket{\dots 110 \dots}$, but the advantage one gains are simple check operators for the clock. Illegal clock states are detected by the projector $\ket{01}\bra{01}$ on consecutive clock qubits, and the dead state $\ket{D}=\ket{0\dots 0}_c$ can be easily ruled out using the projector $\ket{0}\bra{0}_{c_0}$ on the first clock qubit. Also, one can use a 1-local transition rule $\ket{0}\leftrightarrow\ket{1}$ between clock states as in \cite{lh:KR03} or \cite{AQC:AvDKLLR05}, with the caveat that the illegal transitions such as $\ket{11100000} \rightarrow \ket{11100100}$ can occur and need to be taken care of by using large penalty terms. I work around this problem while keeping the locality low in the following constructions.

It is desirable to decrease the locality of transition rules as much as possible. In Bravyi \cite{lh:BravyiQ2SAT06} introduced a clock made out of particles with $d=4$. The upside of his construction is that some of the transition rules involve only one clock particle (albeit 4-dimensional). The interaction required to increment the clock and to apply a two-qubit gate at the same time thus involves particles with dimension $4\times 2 \times 2$, and can be thought of as a 4-qubit interaction. I now present a combined domain wall and pulse clock made out of particles with $d=2$ and $d=3$, building on the idea of \cite{lh:BravyiQ2SAT06}. Consider first a pulse clock made out of particles with dimension $d=3$, with the progression of states
\begin{eqnarray}
	\ket{0}_c &=& \ket{\goF\goO\goO\goO\goO}, \\
	\ket{1}_c &=& \ket{\goB\goO\goO\goO\goO}, \\
	\ket{2}_c &=& \ket{\goO\goF\goO\goO\goO}, \\
	\ket{3}_c &=& \ket{\goO\goB\goO\goO\goO}, \\
	\ket{4}_c &=& \ket{\goO\goO\goF\goO\goO}, \\
	\ket{5}_c &=& \ket{\goO\goO\goB\goO\goO}. 
\end{eqnarray}
Some of the transitions are 1-local ($\goF\leftrightarrow\goB$) while others are 2-local ($\goB\goO \leftrightarrow \goO\goF$). Therefore, one can engineer a 3-local interaction for a qutrit and two qubits which increments the clock and concurrently applies a unitary gate to two work qubits. The initialization problem of this pulse clock is then solved by combining it with a domain wall clock $\ket{111000}$ as follows:
\begin{eqnarray}
	\ket{0}_c &=& \ket{10 \goO 00 \goO 00 \goO 00}, \\
	\ket{1}_c &=& \ket{11 \goF 00 \goO 00 \goO 00}, \\
	\ket{2}_c &=& \ket{11 \goB 00 \goO 00 \goO 00}, \\
	\ket{3}_c &=& \ket{11 \goO 10 \goO 00 \goO 00}, \\
	\ket{4}_c &=& \ket{11 \goO 11 \goB 00 \goO 00}, \\
	\ket{5}_c &=& \ket{11 \goO 11 \goF 00 \goO 00}, \\
	\ket{6}_c &=& \ket{11 \goO 11 \goO 10 \goO 00}.
\end{eqnarray}
The transition rules are 3-local for the transition of the active spot from the domain wall to the pulse clock $10\goO \leftrightarrow 11\goB$, while they stay 1-local for the $\goB\leftrightarrow\goF$ transition. Also, note that 3-local operators suffice to detect whether the clock register state has a single active spot (either $10$ in neighboring domain wall qubits or $1\goB 0$ or $1\goF 0$ when there is a pulse clock particle at the domain wall). I implement this clock using qubits rather than qutrits in Section \ref{ch4:q3}, which results in a new QMA-complete 3-local Hamiltonian construction without the requirement for large norm terms penalizing illegal clock states.

\begin{figure}
	\begin{center}
	\includegraphics[width=4in]{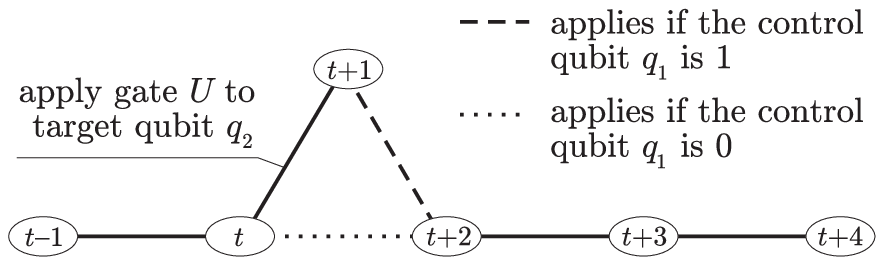} 
	\end{center}
	\caption{The `triangle' clock transition rules allowing to apply a controlled gate to two work qubits using interactions that involve one work qubit at a time.
	\label{ch4:figuretriangle1}}
\end{figure}
Very recently, Eldar and Regev \cite{lh:RegevTriangle} had a revolutionary idea and abandoned the linear progression of states $\ket{\Psi_{t}}$ to create a `triangle' clock as in Figure \ref{ch4:figuretriangle1}. In this construction, the progression of states $\ket{\Psi_{t}}$ can take two different paths, depending on the state of one of the work qubits. Eldar and Regev use it to prove interesting results about Quantum 2-SAT with qudits, shattering the misconception that to apply a two-qubit gate, one has to interact with both qubits in the same step. I review their work in Section \ref{ch4:triangle} and present my own results using this construction in Section \ref{ch4:triangleXZ}. I then take this idea further, creating a more complex, but more symmetric, {\em train switch} clock in Section \ref{ch4:train}. There I use it to prove the final result of this Chapter: one can simulate BQP in the Hamiltonian Computer model of Section \ref{ch2:hc} using a Quantum 3-SAT Hamiltonian.

\subsubsection{Geometric clocks}

Besides directly adding a separate clock register to the system and leaving the work qubits at a particular location in the system, one has the option of having the $n$ work qubits move around by using higher dimensional particles. Consider a simple 2D construction using qutrits, where the state space of each particle is 
\begin{eqnarray}
	\cH_{p} = \cH_{\goR} \oplus \cH_{\goX},
\end{eqnarray}
with $\cH_{\goR}$ a 2-dimensional subspace (can hold the state of a qubit) and $\cH_{\goX}$ a one-dimensional subspace. When a particle is in the state $\ket{\goX}$, it indicates that no qubit is present at the site. The location of the qubits along the strip then determines the states $\ket{\Psi_t}$:
\begin{eqnarray}
\begin{array}{c}
\goX \goX \goX \goR \goSX  \\
\goX \goX \goX \goR \goSX  \\
\goX \goX \goX \goR \goSX  \\
\goX \goX \goX \goR \goSX 
\end{array}
\end{eqnarray}
Observe that the state $\ket{\Psi_t}$ is necessarily orthogonal to any state $\ket{\Psi_{t'}}$ for which the location of the qubits is different. However, synchronization in this na\"{\i}ve geometric clock construction is a problem, resulting in transition rules which are not unique and need to be at least 4-local, if one wants to apply two qubit gates $U_t$ while moving the qubits forward. 
The challenge is to make a geometric clock encoding with unique transition rules which would involve only a few particles with low dimension. In other words, the goal is to get the work qubits to move along in a coordinated fashion and to have the corresponding gates $U_t$ applied to them along the way. 

In the papers of Mizel \cite{AQC:Mizel}, a geometric clock like this is planned to be implemented in a system of electrons hopping between quantum dots. The interactions in \cite{AQC:Mizel} involve 2 particles hopping through 4 sites. 

Another geometric clock construction appeared in \cite{AQC:AvDKLLR05}, where the authors used 6-dimensional particles on a 2D grid. However, it doesn't have unique transition rules and thus requires large penalty terms for illegal clock states. Their progression of states $\ket{\Psi_t}$ is 
\begin{eqnarray}
\left.
\begin{array}{c}
\goX \goD \goO \goO  \\
\goX \goR \goO \goO  \\
\goX \goR \goO \goO  
\end{array}
\right|
\left.
\begin{array}{c}
\goX \goD \goO \goO  \\
\goX \goD \goO \goO  \\
\goX \goR \goO \goO  
\end{array}
\right|
\left.
\begin{array}{c}
\goX \goD \goO \goO  \\
\goX \goD \goO \goO  \\
\goX \goD \goO \goO  
\end{array}
\right|
\left.
\begin{array}{c}
\goX \goD \goO \goO  \\
\goX \goD \goO \goO  \\
\goX \goX \goR \goO  
\end{array}
\right|
\left.
\begin{array}{c}
\goX \goD \goO \goO  \\
\goX \goX \goR \goO  \\
\goX \goX \goR \goO  
\end{array}
\right|
\left.
\begin{array}{c}
\goX \goX \goR \goO  \\
\goX \goX \goR \goO  \\
\goX \goX \goR \goO  
\end{array}
\right|
\begin{array}{c}
\goX \goX \goD \goO  \\
\goX \goX \goR \goO  \\
\goX \goX \goR \goO  
\end{array}
\nonumber
\end{eqnarray}
where the state space of each 6-dimensional particle is a direct sum $\cH_\goD \oplus \cH_\goR \oplus \cH_\goX \oplus \cH_\goO$, with the 2-dimensional subspaces $\cH_\goD$ and $\cH_\goR$ able to hold the state of a qubit. The transitions in this type of clock can be made unique and checkable without large penalty terms if one makes the particles 12-dimensional.

The recent construction of Aharonov, Gottesman, Irani and Kempe \cite{CA:GottesmanLine} shows that a geometric clock can be constructed even in 1D, on a chain of 12-dimensional particles. What is even better, the legal clock states are locally checkable and the transition rules are unique and 2-local. The underlying idea in getting the qubits to move is a `pass the hat forward and then jump over it' technique illustrated in Figure \ref{ch4:figurehats}. I explain this  construction in detail in Section \ref{ch4:q1}, where I also decrease the required dimensionality of particles to $d=11$. 
\begin{figure}
	\begin{center}
	\includegraphics[width=4.5in]{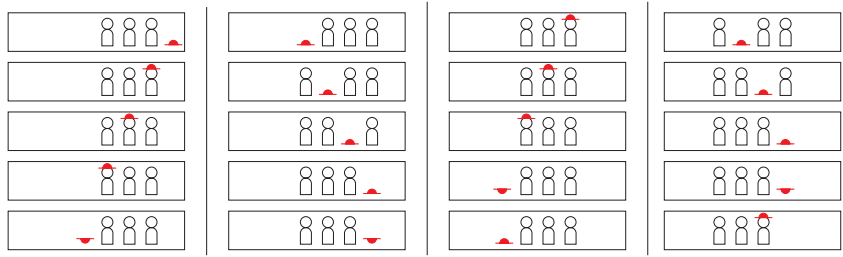} 
	\end{center}
	\caption{The `pass the hat forward and then jump over it' technique for a geometric clock (qubit transport) on a line from \cite{CA:GottesmanLine}. The progression of 20 states is shown.}
	\label{ch4:figurehats}
\end{figure}

Janzing and Wocjan \cite{CA:JW:05} also use a geometric clock in their complicated Hamiltonian Computer construction. Recently, myself and Wocjan \cite{CA:d10} simplified their model greatly, and use a novel geometric clock construction with transition rules reminiscent of diffusion instead of the conventional quantum walk on a line. I present this result in Chapter \ref{ch5hqca} on Hamiltonian Quantum Cellular Automata.


\section{A new 3-local QMA complete Local Hamiltonian \\(Quantum 5-SAT from 3-local terms)}\label{ch4:q3}

Quantum $k$-SAT, the special case where the Hamiltonian is a sum of local projectors was defined and studied by Bravyi \cite{lh:BravyiQ2SAT06} as a natural analogue of classical $k$-SAT. It is in P for $k=2$, and it is QMA$_1$ complete for $k\geq 4$. However, the classification of quantum 3-SAT is still an open question. 
Its mother problem, 3-local Hamiltonian was shown to be QMA complete by Kempe and Regev in \cite{lh:KR03}. This result was further improved, showing that 2-local Hamiltonian is QMA complete in \cite{lh:KKR04,lh:Terhal2D}. These constructions use 
the Hamiltonian of the form 
\eqref{ch4:completeHam}, and encode the clock register using the unary domain wall clock \eqref{ch4:clockwall}. Since the terms in their Hamiltonian are no longer 5-local as in Kitaev's proof in Section \ref{ch1:Kitaev5}, the corresponding terms in $H_{prop}^t$ do not only verify proper application of $U_t$, but also induce transitions from legal clock states into illegal ones. The subspace $\mathcal{H}_{work} \otimes \mathcal{H}_{legal}$ is thus no longer invariant under the action of $H$. To fix this, the penalty associated with illegal clock states is made high (scaling as a high polynomial in $n$), effectively forcing the ground state of the Hamiltonian to reside in the subspace of legal clock states. The Hamiltonian then contains $O(L)$ terms with weights that scale as $O(L^{12})$ and has thus norm $\norm{H}=O(L^{13})$, where $L$ is the number of gates in the computation \cite{lh:KR03}.
This was later improved to weights of order $O(L^{6})$, resulting in 
$\norm{H}=O(L^{7})$ in \cite{lh:KKR04}.

Note that there are two energy scales in this problem. The norm of the Hamiltonian $\norm{H}$, and the energy difference $b-a \geq 1/poly(L)$ in the definition of local-Hamiltonian. It is important to keep track of both of these. One can always rescale the Hamiltonian such that $\norm{H}$ is a constant, but this also shrinks the energy difference $b-a$. The quantity $\frac{b-a}{\norm{H}}$ is an indicator of the strength of the result. Physically speaking, it relates the precision required in obtaining the ground state energy (finding whether it is below $a$ or above $b$) to the strength of the interactions (as given by $\norm{H}$). Proving QMA completeness for the local-Hamiltonian where $\frac{b-a}{\norm{H}}=const.$ would give a quantum analogue of the famous PCP theorem \cite{PCP1,PCP2}.

My goals are much simpler. I show in section \ref{3loc} that the 3-local Hamiltonian problem is QMA complete using a Hamiltonian with $O(L)$ terms with only constant operator norms. This makes the norm of my Hamiltonian scale as $\norm{H}=O(L)$, while I keep the same energy difference $b-a = 1/poly(L)$ as in the previous constructions. The $O(L^{6})$ increase in the ratio of the two energy scales of the problem thus makes my construction physically much more interesting.

In this work, I show a new reduction from a verifier quantum circuit to a 3-local Hamiltonian. The novelty of my construction is that it leaves the space of legal clock-register states invariant. Therefore, the weights of the terms in my Hamiltonian do not scale with the size of of the input problem. Such terms do appear in the constructions of \cite{lh:KR03} and \cite{lh:KKR04}. As an intermediate step in the construction, I prove that quantum 3-SAT for qutrits is QMA$_1$-complete.

The rest of Section \ref{ch4:q3} is organized as follows. After reviewing Bravyi's proof that Quantum $4$-SAT is QMA$_1$-complete in Section \ref{ch4:prel}, I present a qutrit-clock construction in Section \ref{3sat} and show that quantum 3-SAT for particles with dimensions $3 \times 2 \times 2$ (the interaction terms in the Hamiltonian couple one qutrit and two qubits) is QMA$_1$-complete. The existence of such a construction was previously mentioned but not specified by Bravyi and DiVincenzo in \cite{lh:BravyiQ2SAT06} as \cite{lh:DiVincenzo}. In Section \ref{3loc} I show how to encode the qutrit clock particles from Section \ref{3sat} into a pair of qubits in such a way that the Hamiltonian remains 3-local, obtaining a new construction of a QMA complete 3-local Hamiltonian. This Hamiltonian is composed of 4-local positive semidefinite operator terms. However, each of these 4-local operators is composed of only 3-local interaction terms. It is not a Quantum 3-SAT Hamiltonian, since the 3-local terms by themselves are not positive semidefinite operators. I discuss the complexity of Quantum 3-SAT and further directions in Section \ref{ch4:new3conclusions}.

\subsection{Bravyi's Quantum 4-SAT} \label{ch4:prel}

In \cite{lh:BravyiQ2SAT06}, Bravyi proved that quantum $k$-SAT belongs to QMA$_1$ for any constant $k$.
Furthermore, he showed that quantum 4-SAT is QMA$_1$ complete using a new realization of the clock. In his construction Bravyi uses $L+1$ clock particles with 4 states: unborn, active 1 ($a_1$, input for a gate), active 2 ($a_2$, output of a gate), and dead. These 4 states of a clock particle are easily realized by two qubits per clock particle. 
There are $2L$ legal clock states: 
\begin{eqnarray}
	\ket{C_{2k-1}} = \kets{\underbrace{d\dots d}_{k-1} a_1 \underbrace{u \dots u}_{L-k}},
	\qquad \textrm{and} \qquad 	
	\ket{C_{2k}} = \kets{\underbrace{d \dots d}_{k-1} a_2 \underbrace{u \dots u}_{L-k} },
\end{eqnarray}
for $1\leq k \leq L$. A clock Hamiltonian $H_{clock} = H_{clockinit} + \sum_{k=1}^{L-1} H^{(k)}_{clock}$ 
is required to check whether the states of the clock are legal.
\begin{eqnarray}
	H^{(k)}_{clock} &=& \ket{d}\bra{d}_k \otimes \ket{u}\bra{u}_{k+1} \\ 
				&+& \ket{u}\bra{u}_k \otimes \Big(\ket{d}\bra{d}  + 
				  \ket{a_1}\bra{a_1} + \ket{a_2}\bra{a_2} \Big)_{k+1} \nonumber \\
		&+& \Big( \ket{a_1}\bra{a_1} + \ket{a_2}\bra{a_2} \Big)_k \otimes 
					\Big( \ket{a_1}\bra{a_1} + \ket{a_2}\bra{a_2} + \ket{d}\bra{d} \Big)_{k+1}, \nonumber \\
	H_{clockinit} &=& \ket{u}\bra{u}_1 + \ket{d}\bra{d}_L. \label{clockinitbravyi}
\end{eqnarray}
The Hamiltonian checking the correct application of gates is $H_{prop} = \sum_{k=1}^{L} H_{prop}^{(k)}$, with
\begin{eqnarray}
	H_{prop}^{(k)} = \frac{1}{2} \Big( \ii \otimes \ket{a_1}\bra{a_1}_k + \ii \otimes \ket{a_2}\bra{a_2}_k 
	- U_k \otimes \ket{a_2}\bra{a_1}_k - U^{\dagger}_k \otimes \ket{a_1}\bra{a_2}_k \Big). \label{prop4}
\end{eqnarray}
Each such term verifies the correct application of the gate $U_k$ 
between the states $\ket{a_1}$ and $\ket{a_2}$ of the $k$-th clock particle.
This only requires interactions of the $k$-th clock particle (qubit pair) and the two work 
qubits the gate $U_k$ is applied to. Each of the terms is thus a 4-local projector. 

I need another Hamiltonian term to propagate the clock state $\ket{C_{2k}}$ into $\ket{C_{2k+1}}$ 
while leaving the work qubits untouched (that is, for the ground state 
$\ket{\psi_{2k}}_{work} = \ket {\psi_{2k+1}}_{work}$). This is done by the 4-local clock-propagation Hamiltonian 
$H_{clockprop} = \sum_{k=1}^{L-1} H_{clockprop}^{(k)}$, with
\begin{eqnarray}
	H_{clockprop}^{(k)} &=& 
		\frac{1}{2} \Big( \ket{a_2}\bra{a_2}_k \otimes \ket{u}\bra{u}_{k+1} 
					+ \ket{d}\bra{d}_k \otimes \ket{a_1}\bra{a_1}_{k+1} \Big) \\
		&-& \frac{1}{2} \Big( \ket{d}\bra{a_2}_k \otimes \ket{a_1}\bra{u}_{k+1} 
			+ \ket{a_2}\bra{d}_k \otimes \ket{u}\bra{a_1}_{k+1} \Big). \nonumber
\end{eqnarray}
The final ingredients in this construction are
\begin{eqnarray}
	H_{init} &=&  \sum_{n=1}^{N_a}\ket{1}\bra{1}_n \otimes \ket{a_1}\bra{a_1}_1, \\
	H_{out} &=&  \ket{0}\bra{0}_{out} \otimes \ket{a_2}\bra{a_2}_L.
\end{eqnarray}
Applying Kitaev's methods \cite{KitaevBook} to this construction, Bravyi shows that the quantum 4-SAT Hamiltonian (a sum of 4-local projectors)
\begin{eqnarray}
	H &=& H_{clock} + H_{clockprop} + H_{init} + H_{out} + H_{prop} \label{ham4}
\end{eqnarray} 
is QMA$_1$ complete.

Bravyi's original definition required all of the terms in the Hamiltonian to be projectors. However, as I have shown at the end of Section \ref{ch4:qks}, using positive semidefinite operator terms in $H$ instead of just restricting ourselves to projectors is an equivalent problem.


\subsection{A qutrit clock implementation} \label{3sat}

In this section I present a new realization of the clock which builds on Bravyi's quantum 4-SAT realization described above. Using this clock construction, I prove that quantum 3-SAT for qutrits is QMA$_1$-complete. 

First, I need to show that quantum 3-SAT with qutrits is in QMA. I can use Bravyi's proof that quantum $k$-SAT for qubits is in QMA$_1$ for any constant $k$. Given an instance of quantum 3-SAT for qutrits, I convert it into an instance of quantum 6-SAT for qubits by encoding each qutrit in two qubits and projecting out one of the four states. According to Bravyi, this problem is in QMA$_1$ and therefore so is the original
quantum 3-SAT problem with qutrits.

For the other direction in the proof, I need to construct a quantum 3-SAT Hamiltonian for qutrits, corresponding to a given quantum verifier circuit $U$ for a problem in QMA. 
The terms in the Hamiltonian I will construct act on the space of one qutrit and two qubits (particles with dimensions $3\times 2\times 2$). 

\subsection{Clock register construction}

The clock-register construction in the previous section required 4 states for each clock particle: $\ket{u},\ket{a_1},\ket{a_2}$ and $\ket{d}$. Let me first explain why Bravyi's construction requires two ``inactive'' states: $\ket{d}$ and $\ket{u}$. If I only use $\ket{d}$ (i.e., have legal clock states of the form $\ket{d\dots d a_1 d \dots d}$ and $\ket{d\dots d a_2 d \dots d}$), I immediately get a 3-local Hamiltonian for qutrits. However, in Bravyi's construction, the first clock particle is never in the state $\ket{u}$, and the last one is never in the state $\ket{d}$ (see \eqref{clockinitbravyi}). This ensures that at least one clock particle is in an active state. When not using the state $\ket{u}$, I can no longer exclude the state with no active particles $\ket{dd\dots d}$ in a simple local fashion.

\begin{figure}
	\begin{center}
	\includegraphics[width=3.5in]{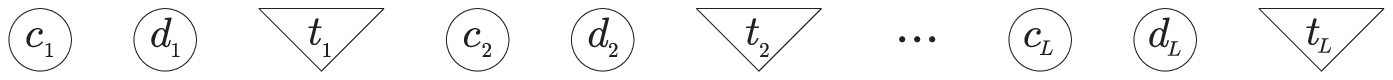} 
	\end{center}
	\caption{Clock register consisting of $2L$ qubits and $L$ qutrits. \label{construct3d}}
\end{figure}

I fix this by modifying the clock register as shown in in Fig.\ref{construct3d}.
The clock register now consists of $2L$ qubits and $L$ qutrits.
The $2L$ qubits $c_1,d_1,\dots,c_k,d_k$ play the role of the 
usual unary $\ket{1\dots 11 00 \dots0}$ clock representation, 
while the $L$ qutrits $t_1,t_2, \dots t_k$ play the role of Bravyi's clock with just three states
($d,a_1,a_2$).

I define the legal clock space $\mathcal{H}_{legal}$ as the space spanned by the $3L$ states $\ket{C_m}$.
These states are defined for $1\leq k \leq L$ as follows:
\begin{eqnarray}
	\ket{C_{3k-2}} &=& \kets{\underbrace{(11d)(11d)\dots(11d)}_{k-1\,\,\textrm{times}} (10d) 
		\underbrace{(00d)(00d)\dots(00d)}_{L-k \,\,\textrm{times}}}, \label{qutritlegal} \\
	\ket{C_{3k-1}} &=& \kets{\underbrace{(11d)(11d)\dots(11d)}_{k-1\,\,\textrm{times}} (11a_1) 
		\underbrace{(00d)(00d)\dots(00d)}_{L-k \,\,\textrm{times}}}, \nonumber \\
	\ket{C_{3k}} &=& \kets{\underbrace{(11d)(11d)\dots(11d)}_{k-1\,\,\textrm{times}} (11a_2) 
		\underbrace{(00d)(00d)\dots(00d)}_{L-k \,\,\textrm{times}}}. \nonumber
\end{eqnarray}
The first state ($\ket{C_{3k-2}}$) corresponds to the time when the
qubits are ``in transport'' from the previous gate to the current ($k$th) gate.
The second one corresponds to the time right before application of gate $U_k$ 
and the third corresponds to the time right after the gate $U_k$ was applied.
The structure of such clock register can be understood as two coupled
``unary'' clocks, the qubit one ($c_k,d_k$) of the $11\dots11100\dots00$ type and the qutrit one 
(the $t_k$'s) of the $00\dots00100\dots00$ type. 
Formally, the legal clock states satisfy the following constraints:
\begin{enumerate}
\item if $d_k$ is 1, then $c_k$ is 1.
\item if $c_{k+1}$ is 1, then $d_k$ is 1.
\item if $t_k$ is active ($a_1/a_2$), then $d_k$ is 1.
\item if $t_k$ is active ($a_1/a_2$), then $c_{k+1}$ is 0.
\item if $d_k$ is 1 and $c_{k+1}$ is 0, then $t_{k}$ is not dead $(d)$.
\item $c_1$ is 1.
\item if $d_L$ is 1, then $t_{L}$ is not dead.
\end{enumerate}
The last two conditions are required to exclude the clock states $\ket{(00d)(00d)\dots (00d)}$ 
and $\ket{(11d)(11d) \dots(11d)}$ that have no active clock terms.
The clock Hamiltonian $H_{clock} = H_{clockinit} + \sum_{k=1}^{L} H_{clock1}^{(k)} + \sum_{k=1}^{L-1} H_{clock2}^{(k)}$ 
verifies the above constraints.
\begin{eqnarray}
		H_{clock1}^{(k)} &=& 
				\ket{01}\bra{01}_{c_k,d_k} 
				+ \ket{0}\bra{0}_{d_k} \otimes \Big( \ket{a_1}\bra{a_1} + \ket{a_2}\bra{a_2} \Big)_{t_k}, \label{clock3} \\
		H_{clock2}^{(k)} &=& 
				\ket{01}\bra{01}_{d_k,c_{k+1}} 
				+ \Big( \ket{a_1}\bra{a_1} + \ket{a_2}\bra{a_2} \Big)_{t_k} \otimes \ket{1}\bra{1}_{c_{k+1}} \nonumber \\
			&+& \ket{1d0}\bra{1d0}_{d_k,t_k,c_{k+1}}, \nonumber \\
		H_{clockinit} &=&
				\ket{0}\bra{0}_{c_1} 
			+ \ket{1}\bra{1}_{d_{L}} \otimes \ket{d}\bra{d}_{t_{L}}. \nonumber
\end{eqnarray}
Only the last term in $H_{clock2}^{(k)}$ is a 3-local projector, 
acting on the space of two qubits and one qutrit. 
The rest of the terms are 2-local projectors on two qubits, or a qubit and a qutrit.
The space of legal clock states $\mathcal{H}_{legal}$
is the kernel of the clock Hamiltonian $H_{clock}$.

\subsection{Checking correct application of gates and clock propagation}

The gate-checking Hamiltonian $H_{prop} = \sum_{k=1}^{L} H_{prop}^{(k)}$ is an analogue of \eqref{prop4}, with
\begin{eqnarray}
	H_{prop}^{(k)} = \frac{1}{2} \Big( \ii_{work} \otimes \Big( \ket{a_1}\bra{a_1}_{t_k} 
		+ \ket{a_2}\bra{a_2}_{t_k} \Big)
	- U_k \otimes \ket{a_2}\bra{a_1}_{t_k} - U^{\dagger}_k \otimes \ket{a_1}\bra{a_2}_{t_k} \Big). \label{prop3}
\end{eqnarray}

\begin{figure}
	\begin{center}
	\includegraphics[width=3.6in]{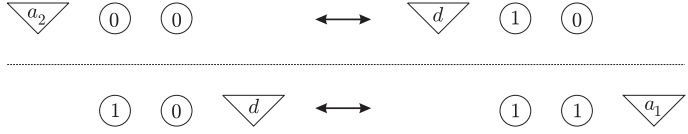} 
	\end{center}
	\caption{Illustration of the two-step clock pointer propagation. \label{clockprop3d}}
\end{figure}
The clock propagation proceeds in two steps. First, the ``active'' spot in the clock register moves from
the state $\ket{a_2}$ of the qutrit $t_k$ to the $\ket{10}$ state of the next two qubits $c_{k+1},d_{k+1}$.
After this, it moves into the state $\ket{a_1}$ of the next qutrit $t_{k+1}$, as in Fig.\ref{clockprop3d}.
The Hamiltonian checking whether this happened, while the work qubits were left untouched, is 
$H_{clockprop} = \sum_{k=1}^{L} H_{clockprop1}^{(k)} + \sum_{k=1}^{L-1} H_{clockprop2}^{(k)}$, with
\begin{eqnarray}
	H_{clockprop1}^{(k)} &=& \frac{1}{2} \Big( 
				\ket{10}\bra{10}_{c_k,d_k} \otimes \ket{d}\bra{d}_{t_{k}} 
			+ \ket{11}\bra{11}_{c_k,d_k} \otimes \ket{a_1}\bra{a_1}_{t_{k}}
				 	\Big) \label{clockprop3} \\
	&-& \frac{1}{2} \Big( 
				\ket{11}\bra{10}_{c_k,d_{k}} \otimes \ket{a_1}\bra{d}_{t_{k}} 
			+ \ket{10}\bra{11}_{c_k,d_{k}} \otimes \ket{d}\bra{a_1}_{t_{k}} 
					\Big), \nonumber \\
	H_{clockprop2}^{(k)} &=& \frac{1}{2} \Big( 
				 \ket{a_2}\bra{a_2}_{t_{k}} \otimes \ket{00}\bra{00}_{c_{k+1},d_{k+1}} 
			+	 \ket{d}\bra{d}_{t_{k}} \otimes \ket{10}\bra{10}_{c_{k+1},d_{k+1}} 
				 	\Big) \nonumber \\
	&-& \frac{1}{2} \Big( 
				 \ket{d}\bra{a_2}_{t_{k}} \otimes \ket{10}\bra{00}_{c_{k+1},d_{k+1}} 
			+	 \ket{a_2}\bra{d}_{t_{k}} \otimes \ket{00}\bra{10}_{c_{k+1},d_{k+1}} 
					\Big). \nonumber
		\end{eqnarray}
The input Hamiltonian checks whether the computation has properly initialized ancilla qubits.
\begin{eqnarray}
	H_{init} &=& \sum_{n=1}^{N_a} \ket{1}\bra{1}_n \otimes \ket{a_1}\bra{a_1}_{t_1}. \label{init3} 
\end{eqnarray}
Finally, the output Hamiltonian checks whether the result of the computation was 1.
\begin{eqnarray}
	H_{out} &=& \ket{0}\bra{0}_{out} \otimes \ket{a_2}\bra{a_2}_{t_{L}}. \label{out3} 
\end{eqnarray}

All of the terms coming from \eqref{clock3} -- \eqref{out3} in the Hamiltonian 
\begin{eqnarray}
	H &=& H_{clock} + H_{clockprop} + H_{init} + H_{out} + H_{prop}. \label{ham3}
\end{eqnarray}
are projectors. Therefore, 
the ground state has energy zero if and only if there exists 
a zero energy eigenstate of all of the terms.
If there exists a witness $\ket{\varphi}$ 
on which the computation $U$ gives the result 1 with probability 1,
I can construct a computational history state \eqref{ch4:psihistory} 
for a modified circuit $\tilde{U}= U_L \cdot \ii \cdot \ii \cdot U_{L-1} \cdot \ii\cdot \ii \cdots U_1 \cdot \ii$,
where the ``identity'' gates correspond to the clock propagation in my construction, with nothing
happening to the work qubits.
This state is a zero eigenvector of all of the terms in the Hamiltonian \eqref{ham3}.

I now need to prove that if no witness exists (the
answer to the problem is ``no''), then the ground state energy of \eqref{ham3} is
lower bounded by $1/poly(L)$. 
Let me decompose the Hilbert space into 
\begin{eqnarray}
	\mathcal{H} = \left( \mathcal{H}_{work} \otimes \mathcal{H}_{legal} \right)
			\oplus ( \mathcal{H}_{work} \otimes  \mathcal{H}^{\perp}_{legal} ).
\end{eqnarray}
where $\mathcal{H}_{legal}$ is the space of legal clock states (on which $H_{clock}\ket{\alpha} =0$).
The Hamiltonian \eqref{ham3} leaves this decomposition invariant, because it does not
induce transitions between legal and illegal clock states.
Since any state in $\mathcal{H}_{legal}^{\perp}$ violates at least one
term in $H_{clock}$, the lowest eigenvalue of the 
restriction of \eqref{ham3} to $\mathcal{H}_{work}\otimes \mathcal{H}_{legal}^{\perp}$
is at least 1.
On the other hand, the restriction of $H$ to the legal clock space 
is identical to the legal clock space restriction of Bravyi's Hamiltonian \eqref{ham4} 
from the previous section. Therefore, his proof using the methods
of Kitaev \cite{KitaevBook} applies to this case as well. He shows that if a no witness state for the quantum circuit $U$ exists, then
the ground state energy of the restriction of \eqref{ham4} to $\mathcal{H}_{work} \otimes \mathcal{H}_{legal}$ lower bounded by $1/poly(L)$. This means that if there is no witness state for the verifier circuit $U$, the ground state of \eqref{ham3} is lower bounded by $1/poly(L)$. This concludes the proof that quantum 3-SAT with qutrits (in fact, quantum 3-SAT on particles with dimensions $3\times 2\times 2$, a qutrit and two qubits) is QMA$_1$ complete.

The existence of another $3 \times 2 \times 2$ construction for quantum 3-SAT
(i.e., a Hamiltonian with terms acting on one qutrit and two qubits)
was already mentioned in \cite{lh:BravyiQ2SAT06} as \cite{lh:DiVincenzo}, though that construction was not specified. 
I choose to write out my result explicitly as it serves as a natural intermediate step towards the new 3-local Hamiltonian construction described in the following section.

\subsection{The new 3-local QMA complete construction (for qubits)} \label{3loc}

In Bravyi's Quantum 4-SAT construction \cite{lh:BravyiQ2SAT06}, the clock particles (qubit pairs) 
can be in 4 states. In the previous section, I required only 3 states of the clock particles and used qutrits as particles with these three states. I start with the clock-register construction (see Fig.\ref{construct3d}) from the previous section, with legal states as in \eqref{qutritlegal}. However, I now encode the three states 
of every clock qutrit $t_k$ using a pair of qubits $r_k, s_k$. The new clock register is depicted 
in Fig.\ref{construct3loc}.
\begin{eqnarray}
	\ket{a_1}_{t_k} &\rightarrow& \frac{1}{\sqrt{2}}\left(\ket{01}-\ket{10}\right)_{r_k,s_k}, \qquad 
			\ket{d}_{t_k} \rightarrow \ket{00}_{r_k,s_k}, \\
	\ket{a_2}_{t_k} &\rightarrow& \frac{1}{\sqrt{2}}\left(\ket{01}+\ket{10}\right)_{r_k,s_k}. \nonumber
\end{eqnarray}
This encoding allows me to obtain a new 3-local Hamiltonian construction (for the QMA-complete 3-local Hamiltonian problem). This Hamiltonian is a quantum 4-SAT Hamiltonian whose 4-local positive semidefinite operator terms consist of just 3-local interactions.

I am looking for 3-local terms that flip between the clock states 
$\ket{a_1} \leftrightarrow \ket{a_2}$, while simultaneously (un)applying a 2-qubit gate $U_k$
on two work qubits. 
I encode the active states of a clock particle into
entangled states, and thus I am able to flip between these clock states with a term like $Z_1$ involving only one of the clock particles. Thus my 2-qubit gate checking Hamiltonian involves only 3-local terms 
(acting on one clock qubit $r_k$ or $s_k$ and the two work qubits on which the gate $U_k$ acts).

\begin{figure}
		\begin{center}
		\includegraphics[width=3.6in]{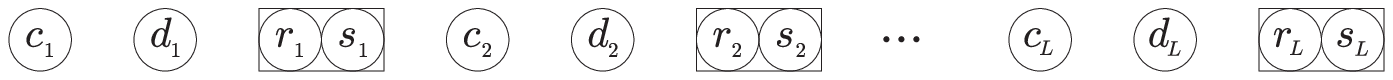}
		\end{center}
	\caption{Clock register construction with $2L+2L$ qubits. \label{construct3loc}}
\end{figure}
First, I define the legal clock space $\mathcal{H}_{legal}$ as the space spanned by the 
$3L$ states $\ket{C_m}$. These states are defined for $1\leq k \leq L$ as follows (compare to \eqref{qutritlegal}):
\begin{eqnarray}
	\ket{C_{3k-2}} &=& \kets{\underbrace{(11)(00)\dots (11)(00)}_{k-1\,\,\textrm{times}}}
		\otimes \ket{10}_{c_k,d_k} \otimes \ket{00}_{r_k, s_k} \otimes 
		\kets{\underbrace{(00)(00)\dots (00)}_{L-k \,\,\textrm{times}}}, \label{threeloclegal} \\
	\ket{C_{3k-1}} &=& 
		\kets{\underbrace{(11)(00)\dots (11)(00)}_{k-1\,\,\textrm{times}}}
		\otimes \kets{11}_{c_k,d_k} \otimes \frac{1}{\sqrt{2}}\left(\ket{01}-\ket{10}\right)_{r_k,s_k}
		\otimes \kets{\underbrace{(00)(00)\dots (00)}_{L-k \,\,\textrm{times}}}, \nonumber \\
	\ket{C_{3k}} &=& 
		\kets{\underbrace{(11)(00)\dots (11)(00)}_{k-1\,\,\textrm{times}}}
		\otimes \kets{11}_{c_k,d_k} \otimes \frac{1}{\sqrt{2}}\left(\ket{01}+\ket{10}\right)_{r_k,s_k}
		\otimes \kets{\underbrace{(00)(00)\dots (00)}_{L-k \,\,\textrm{times}}}. \nonumber
\end{eqnarray}
Similarly to the construction of the previous section, the first state ($\ket{C_{3k-2}}$) corresponds to the time when the
qubits are ``in transport'' from the previous gate to the current ($k$th) gate.
The second one corresponds to the time right before application of gate $U_k$ 
and the third corresponds to the time right after the gate $U_k$ was applied.

Formally, the legal clock states for this construction satisfy the following constraints:
\begin{enumerate}
\item if $d_k$ is 1, then $c_k$ is 1.
\item if $c_{k+1}$ is 1, then $d_k$ is 1.
\item the pair $r_k, s_k$ is not in the state $\ket{11}$.
\item if the pair $r_k, s_k$ is active (in the state $(\ket{01}\pm \ket{10})/\sqrt{2}$), then $d_k$ is 1.
\item if the pair $r_k, s_k$ is active (in the state $(\ket{01}\pm \ket{10})/\sqrt{2}$), then $c_{k+1}$ is 0.
\item if $d_k$ is 1 and $c_{k+1}$ is 0, then the pair $r_k, s_k$ is not dead (in the state $\ket{00}$).
\item $c_1$ is 1.
\item if $d_L$ is 1, then the pair $r_L,s_L$ is not dead (in the state $\ket{00}$).
\end{enumerate}
The last two conditions are required to make the clock states $\ket{(00)(00)\dots (00)(00)}$ 
and $\ket{(11)(00)\dots(11)(00)}$ with no active spots illegal.
The clock Hamiltonian $H_{clock} = H_{clockinit} + \sum_{k=1}^{L} H_{clock1}^{(k)} + \sum_{k=1}^{L-1} H_{clock2}^{(k)}$ 
verifies the above constraints.
\begin{eqnarray}
		H_{clock1}^{(k)} &=& 
				\ket{01}\bra{01}_{c_k,d_k} 
				+ \ket{0}\bra{0}_{d_k} \otimes \big( \ket{1}\bra{1}_{r_k} +\ket{1}\bra{1}_{s_k} \big)
				+ \ket{11}\bra{11}_{r_k,s_k}, \label{clock4} \\
		H_{clock2}^{(k)} &=& 
				\ket{01}\bra{01}_{d_k,c_{k+1}} 
				+ \big( \ket{1}\bra{1}_{r_k} +\ket{1}\bra{1}_{s_k} \big) \otimes \ket{1}\bra{1}_{c_{k+1}} + h_4^{(k)}, \nonumber \\
		h_4^{(k)} &=& 
				\ket{1}\bra{1}_{d_k} \otimes \frac{1}{2}\,(Z_{r_k}+Z_{s_k}) \otimes \ket{0}\bra{0}_{c_{k+1}} 
				+ \ket{11}\bra{11}_{r_k,s_k}, \nonumber \\
		H_{clockinit} &=&
				\ket{0}\bra{0}_{c_1} 
			+ \ket{1}\bra{1}_{d_{L}} \otimes \ket{00}\bra{00}_{r_L,s_L}. \nonumber
\end{eqnarray}
All of the terms involve only 3-local interactions.
All terms in $H_{clock1}^{(k)}$, $H_{clock2}^{(k)}$ and $H_{clockinit}$, are are projectors. 
The term $h_4^{(k)}$ corresponds to the sixth legal state condition.
It is a 4-local projector onto the space spanned by (illegal clock) states $\ket{1_{d_k}(00)_{r_k,s_k}0_{c_{k+1}}}$, 
$\ket{0(11)0}$, $\ket{0(11)1}$ and $\ket{1(11)1}$.
Note that even though $h_4^{(k)}$ is a 4-local projector, it is only constructed of 3-local terms.

\subsection{Checking gate application with 3-local terms}
Let me start by writing out a Hamiltonian that checks the correct application of a single-qubit gate $U_k$.
\begin{eqnarray}
	H_{prop}^{(k),\,one-qubit} = \frac{1}{2}\left( 
		\begin{array}{rr}
			\ii\otimes \ket{01-10}\bra{01-10}_{r_k, s_k} 
				- & U_k \otimes \ket{01+10}\bra{01-10}_{r_k, s_k}\\
			\ii\otimes \ket{01+10}\bra{01+10}_{r_k, s_k}
				- & U_k^{\dagger} \otimes \ket{01-10}\bra{01+10}_{r_k, s_k}
			\end{array}
		\right),
\end{eqnarray}
where $\ket{01\pm 10}$ is a shortcut notation for the normalized entangled states $(\ket{01}\pm \ket{10})/\sqrt{2}$.
This Hamiltonian is a 3-local projector. 
Note that in the case $U_k=\ii$, this Hamiltonian becomes the projector $(\ii-X)/2$
on the space of active clock states $\{\ket{01-10},\ket{01+10}\}$.

For a two-qubit gate $U_k$, the above construction would be 4-local. However, I am be able to 
construct this 4-local projector using
only 3-local terms. To do this, I require the 2-qubit gate to be symmetric"' $U_k=U_k^{\dagger}$. 
This is a universal construction, 
since the symmetric gate \textsc{CNOT} (or \textsc{C}$_{\phi}$) is universal.
Now I can write 
\begin{eqnarray}
	H^{(k),\,two-qubit}_{prop} 
		= \frac{1}{2}\Big( \ii \otimes \frac{1}{2}\,( \ii - Z_{r_k} Z_{s_k} ) -  U_k \otimes  \frac{1}{2}\,(Z_{r_k}-Z_{s_k}) \Big).
\end{eqnarray}
The first term in this Hamiltonian, $( \ii - Z_{r_k} Z_{s_k} )/2$, is a projector onto the space of active clock states,
$\ket{01\pm10}_{r_k,s_k}$, as I needed. The second term contains $(Z_{r_k}-Z_{s_k})/2$, which has zero eigenvalues
for the states $\ket{00}_{r_k,s_k}$ and $\ket{11}_{r_k,s_k}$, and flips between the states 
$\ket{01-10}_{r_k,s_k}\leftrightarrow\ket{01+10}_{r_k,s_k}$.
Altogether, this is a 4-local projector made out of only 3-local terms. 

\subsection{Clock propagation}
\label{clockpropagsection}

After a gate $U_k$ is applied, I need to ``propagate'' the pointer (the active state of the qubit pair
$r_k,s_k$) to the next pair of qubits $r_{k+1},s_{k+1}$.
This is done in two steps, as shown in Fig.\ref{figclockprop}.
\begin{figure}
	\begin{center}
	\includegraphics[width=3.6in]{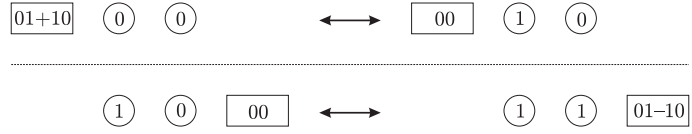}
	\end{center}
	\caption{Illustration of the two-step clock pointer propagation. \label{figclockprop}}
\end{figure}

For each step, I want to write a 3-local positive semidefinite Hamiltonian 
with terms acting on 4 consecutive qubits $r_k,s_k,c_{k+1},d_{k+1}$ 
(for the second step of the clock pointer propagation, the four qubits in play are $c_k,d_k,r_k,s_k$), 
with zero eigenvalue for the legal clock-propagation states, and perhaps
also some illegal clock states, which will be disallowed by other terms in the Hamiltonian ($H_{clock}$).
For the first step, these desired eigenvectors with zero eigenvalues are 
\begin{eqnarray}
	\ket{\alpha_{1}}_{r_k,s_k,c_{k+1},d_{k+1}} &=& \ket{00}_{r_k,s_k} \ket{00}_{c_{k+1},d_{k+1}}, \\
	\ket{\alpha_{2}}_{r_k,s_k,c_{k+1},d_{k+1}} &=& \frac{1}{\sqrt{2}}(\ket{01}-\ket{10})_{r_k,s_k} \ket{00}_{c_{k+1},d_{k+1}}, \nonumber \\
	\ket{\alpha_{3}}_{r_k,s_k,c_{k+1},d_{k+1}} &=& \frac{1}{2}(\ket{01}+\ket{10})_{r_k,s_k} \ket{00}_{c_{k+1},d_k+1}
				 																 + \frac{1}{\sqrt{2}}\ket{00}_{r_k,s_k} \ket{10}_{c_{k+1},d_{k+1}}, \nonumber\\
	\ket{\alpha_{4}}_{r_k,s_k,c_{k+1},d_{k+1}} &=& \ket{00}_{r_k,s_k} \ket{11}_{c_{k+1},d_{k+1}}. \nonumber
\end{eqnarray}
The state that I want to exclude (make it a nonzero eigenvector) is the legal clock state with incorrect pointer propagation:
\begin{eqnarray} \label{exclude1}
	\kets{\alpha^{\perp}}_{r_k,\dots,d_{k+1}} &=& 
			\frac{1}{2}(\ket{01}+\ket{10})_{r_k,s_k} \ket{00}_{c_{k+1},d_k+1}
		- \frac{1}{\sqrt{2}}\ket{00}_{r_k,s_k} \ket{10}_{c_{k+1},d_{k+1}}.
\end{eqnarray}
Let me present the Hamiltonian penalizing this state.
\begin{eqnarray}
	H^{(k)}_{clockprop1} &=& 
				\ket{10} \bra{10}_{c_{k+1},d_{k+1}}  \label{prop41}\\
				&+& \frac{1}{2}\left(\ket{01}+\ket{10}\right)\left(\bra{01}+\bra{10}\right)_{r_k,s_k} \nonumber \\
			  &-& \frac{1}{\sqrt{2}}\Big( \ket{0}\bra{1}_{r_k}+\ket{0}\bra{1}_{s_k} \Big) 
							\otimes \ket{10}\bra{00}_{c_{k+1},d_{k+1}} \nonumber \\
				&-& \frac{1}{\sqrt{2}}\Big( \ket{1}\bra{0}_{r_k}+\ket{1}\bra{0}_{s_k} \Big) 
							\otimes \ket{00}\bra{10}_{c_{k+1},d_{k+1}} \nonumber\\
	&+& 2\ket{11}\bra{11}_{r_k,s_k}. \nonumber
\end{eqnarray}
It is positive semidefinite, with eigenvalues 
0 ($\times 7$), 
1 ($\times 4$), 
2 ($\times 3$) and  
3 ($\times 2$).
Its zero energy eigenvectors are $\ket{\alpha_1}$, $\ket{\alpha_2}$, $\ket{\alpha_3}$, $\ket{\alpha_4}$ expressed above, 
and three illegal clock states, 
$\ket{00}_{r_k,s_k}\ket{01}_{c_{k+1},d_{k+1}}$, 
$(\ket{01}-\ket{10})\ket{01}$ and 
$(\ket{01}-\ket{10})\ket{11}$.
The state $\ket{\alpha^{\perp}}$ is an eigenvector of $H^{(k)}_{clockprop1}$  with eigenvalue 2.
This means that $H^{(k)}_{clockprop1}$ fulfills its job in punishing the legal states of the clock register 
\eqref{exclude1}, which do not correctly propagate the clock. 
This Hamiltonian term is a positive semidefinite operator, while Bravyi's original definition of quantum $k$-SAT
requires the terms in the Hamiltonian to be projectors. However, as I have shown at the end of Section \ref{ch4:qks},
quantum $k$-SAT with positive semidefinite operator terms is equivalent to quantum $k$-SAT with only projector terms.

For the second step, the desired zero energy eigenvectors are
\begin{eqnarray}
	\ket{\beta_{1}}_{c_k,d_k,r_k,s_k} &=& \ket{00}_{c_k,d_k} \ket{00}_{r_k,s_k}, \\
	\ket{\beta_{2}}_{c_k,d_k,r_k,s_k} &=& \frac{1}{\sqrt{2}}\ket{10}_{c_{k},d_{k}} \ket{00}_{r_k,s_k}
				+ \ket{11}_{c_{k},d_k} \frac{1}{2}(\ket{01}-\ket{10})_{r_k,s_k}, \nonumber\\
	\ket{\beta_{3}}_{c_k,d_k,r_k,s_k} &=& \ket{11}_{c_k,d_k} \frac{1}{\sqrt{2}}(\ket{01}+\ket{10})_{r_k,s_k}, \nonumber \\
	\ket{\beta_{4}}_{c_k,d_k,r_k,s_k} &=& \ket{11}_{c_{k},d_{k}} \ket{00}_{r_k,s_k}, \nonumber
\end{eqnarray}
and the state I want to exclude is
\begin{eqnarray}
	\kets{\beta^{\perp}}_{c_k,d_k,r_k,s_k} &=& 
			\frac{1}{\sqrt{2}}\ket{10}_{c_{k},d_{k}} \ket{00}_{r_k,s_k}
				- \ket{11}_{c_{k},d_k} \frac{1}{2}(\ket{01}-\ket{10})_{r_k,s_k}.
\end{eqnarray}
The Hamiltonian with these properties is a simple analogue of \eqref{prop41}:
\begin{eqnarray}
	H^{(k)}_{clockprop2} &=& 
					\ket{10} \bra{10}_{c_{k},d_{k}}   \label{prop42}\\
	&+&  \frac{1}{2}\left(\ket{01}-\ket{10}\right)\left(\bra{01}-\bra{10}\right)_{r_k,s_k} \nonumber\\
				 &-& \ket{11}\bra{10}_{c_{k},d_{k}} \otimes 
						\frac{1}{\sqrt{2}}\Big( -\ket{1}\bra{0}_{r_k}+\ket{1}\bra{0}_{s_k} \Big) \nonumber\\
					&-& \ket{10}\bra{11}_{c_{k},d_{k}} \otimes
						\frac{1}{\sqrt{2}}\Big( -\ket{0}\bra{1}_{r_k}+\ket{0}\bra{1}_{s_k} \Big) \nonumber\\
	&+& 2\ket{11}\bra{11}_{r_k,s_k}. \nonumber
\end{eqnarray}
This is again a positive semidefinite operator with eigenvalues 
0 ($\times 7$), 
1 ($\times 4$), 
2 ($\times 3$) and  
3 ($\times 2$).
Its zero energy eigenvectors are $\ket{\beta_1}$, $\ket{\beta_2}$, $\ket{\beta_3}$, $\ket{\beta_4}$ expressed above, 
and three illegal clock states, 
$\ket{01}_{c_{k},d_{k}}\ket{00}_{r_k,s_k}$, 
$\ket{01}(\ket{01}+\ket{10})$ and 
$\ket{00}(\ket{01}+\ket{10})$, which are penalized by $H_{clock}$.
The state $\ket{\beta^{\perp}}$ is an eigenvector of $H^{(k)}_{clockprop2}$ with eigenvalue 2,
which is what I intended $H^{(k)}_{clockprop2}$ to do.

Just as in the previous section, the total Hamiltonian leaves the decomposition into 
$\mathcal{H}_{legal} \oplus \mathcal{H}_{legal}^\perp$ invariant while all illegal clock states violate at least one term in $H_{clock}$. Again, up to a constant prefactor, the restriction of $H$ to the legal clock space is the same as that of the Hamiltonian in \eqref{ham4}. The proof of the necessary separation between positive and negative instances then follows the proof in the previous section. This concludes the proof that quantum 4-SAT with positive semidefinite operators made out of 3-local terms is QMA$_1$ complete.

\subsection{Discussion and further directions} \label{ch4:new3conclusions}

In the above, I proved that quantum 3-SAT for particles with dimensions $3\times 2 \times 2$ is QMA$_1$ complete. I have shown in Section \ref{ch4:qks} that quantum $k$-SAT with positive semidefinite operator terms 
(not just projectors) is equivalent to quantum $k$-SAT with projector terms. I presented a new 3-local construction of a quantum 4-SAT Hamiltonian with positive semidefinite operator terms, proving that quantum 4-SAT with 3-local interactions is QMA$_1$ complete.

The currently known complexities of classical and quantum satisfiability problems are shown in Table \ref{QSATtable}. 
Quantum 3-SAT contains classical 3-SAT and therefore is NP-hard. 
Unlike classical $k$-SAT, which is known to be NP-complete for $k\geq 3$, quantum $k$-SAT is only known to be QMA$_1$ complete for $k\geq 4$ \cite{lh:BravyiQ2SAT06}. In my opinion, it is unlikely that one can show that quantum 3-SAT ($k=3$) is also complete for QMA$_1$. One indication for this arises in my numerical explorations, where random instances of quantum 3-SAT for a reasonable number of clauses generally have no solutions, unless the clauses exclude non-entangled states. This may suggest that the hardness of quantum 3-SAT actually lies only in the classical instances (3-SAT) and a classical verifier circuit for quantum 3-SAT might exist.

Another reason comes from dimension counting. This argument, however, is only valid for the specific encoding of a circuit into the Hamiltonian I used. I worked with a tensor product space $\mathcal{H}_{work}\otimes \mathcal{H}_{clock}$, encoding the computation in the history state \eqref{ch4:psihistory}. Encoding an interaction of two qubits requires at least an $4+4=8$ dimensional space ($4$ for the two qubits before the interaction and $4$ for the qubits after the interaction).
On a first glance, a three-local projector on the space of two work qubits and one clock qubit ($2\times2\times2=8$) seems to suffice.
However, I must ensure that this interaction only occurs at a
specific clock time. When the two work qubits interact with just a single clock qubit 
(flipping it between states 
before/after interaction) ambiguities and legal-illegal clock state 
transitions are unavoidable. This transforms the problem 
from the SAT-type (determining whether a {\em simultaneous} ground state of {\em all} terms 
in the Hamiltonian exists) 
to the MAX-SAT type problem (determining the
properties of the ground state energy of the {\em sum of terms} in the Hamiltonian). 
A single clock qubit cannot both determine the exact time of an interaction and 
distinguish between the states before and after the interaction.
I managed to overcome this obstacle by using three-dimensional clock particles with 
states $d$, $a_1$ and $a_2$ and a $2\times2\times3=12$ dimensional space for encoding the interactions. 
I believe that this can not be further improved with more clever
clock-register realizations within the usual $\mathcal{H}=\mathcal{H}_{work}\otimes \mathcal{H}_{clock}$ framework. However, a recent novel idea by Eldar and Regev \cite{lh:RegevTriangle} presented in Section \ref{ch4:triangle} could be a step in this direction.
They introduce a novel `triangle' clock construction, and show that Quantum 2-SAT for particles with dimension $5\times 3$ is QMA$_1$-complete.

A different approach to encoding a quantum computation into the ground state of a Hamiltonian is the geometric clock (see Section \ref{ch4:clocks} found in the work of Aharonov et.al \cite{AQC:AvDKLLR05}. Their idea is to lay out the qubits in space in such a way, that the shape of the state (the locations of the work qubits moving around in the system) uniquely corresponds to a clock time. Their motivation was to show that adiabatic quantum computation \cite{AQC:AvDKLLR05} is polynomially equivalent to the circuit model. As a side result, they showed that the nearest-neighbor 2-local Hamiltonian problem with 6-dimensional particles is QMA complete. Actually, as Kempe et al. proved, even 2-local Hamiltonian with 2-dimensional particles (qubits) is QMA complete \cite{lh:KKR04}. I did not succeed to improve or reproduce the construction of quantum 3-SAT for particles with dimensions $3\times2\times2$ using the geometric clock framework. However, one can use the idea of a geometric clock to construct quantum 2-SAT for higher dimensional particles (qudits). I know that classical 2-SAT for particles with dimensions $3\times3$ contains graph coloring, and is thus NP-complete. 
Using a rather straightforward modification of the Aharonov et. al. construction, one can prove that quantum 2-SAT for
12-dimensional particles is QMA complete. Combining the work/clock and the geometric clock constructions, a much tighter result can be shown. Specifically, one can construct quantum 2-SAT for particles 
with dimensions $9\times4$ and prove that 
it is QMA$_1$ complete. However, this is now made obsolete by the result of Eldar and Regev presented in Section \ref{ch4:triangle} which does not use a geometric clock. It remains to be seen what are the minimal dimensions of particles for which quantum 2-SAT is QMA$_1$ complete.


\section{Quantum 2-SAT on a line}\label{ch4:q1}

In 2007, Aharonov, Gottesman, Irani and Kempe \cite{CA:GottesmanLine} proved a surprising result about the power of quantum systems on a line. One might have thought before that finding ground states or simulating time evolution of spin chains is not too hard a task with today's modern DMRG and MPS methods (see Chapter \ref{ch3mps}). However, \cite{CA:GottesmanLine} shows that finding ground states of 2-local Hamiltonians on a line of $d=12$ dimensional particles is QMA complete. Utilizing this construction, recently Schuch, Cirac and Verstraete \cite{CA:CiracLine} proved that even finding ground states of 2-local Hamiltonians on a line whose ground states are known to be Matrix Product States with constant dimension is NP-hard.

In this section I revisit \cite{CA:GottesmanLine} and show that the problem remains QMA complete even for particles with dimension $d=11$. Moreover, \cite{CA:GottesmanLine} shows that one can use a simplified nearest-neighbor Hamiltonian on a line of $d=9$-dimensional particles to perform universal Quantum Adiabatic Computation. This result was recently improved to a system of $8$-dimensional particles by Chase and Landahl \cite{CA:Landahl}. Concurrently with their work, myself with P. Wocjan in \cite{CA:d10} have proved that a translationally invariant Hamiltonian on a line of $10$-dimensional particles can be used for this purpose. When I release the translational invariance requirement, I get a Hamiltonian for a chain of $8$-dimensional particles as well. I present this result later in Section \ref{d10section}.

\subsection{A Line of $d=12$ Dimensional Particles (Quantum (12,12)-SAT)}

As discussed in Section \ref{ch4:clocks}, to show that a Local Hamiltonian (or Quantum-SAT) problem is QMA complete, one usually proceeds as follows. First, encode a quantum circuit into a progression of orthogonal states $\ket{\Psi_t}$ of a larger system, either by adding a clock register to the register holding the work qubits, or by using a geometric clock construction. The corresponding Hamiltonian is then constructed to check the proper transition rules for the progression of states $\ket{\Psi_t}$. What remains to be shown are the properties of the low-lying spectrum of $H$. 

\subsubsection{The $d=12$ Geometric clock on a line}

On a line, only nearest neighbor interactions are allowed. This rules out the use of a separate clock register, which means a geometric clock (see Section \ref{ch4:clocks}) is a necessity. This is also a reason for the dimensionality ($d=12$) of the particles. The state space of each particle consists of four two-dimensional subspaces 
\begin{eqnarray}
	\goA &:& \textrm{a qubit marked as `active'}, \\
	\goM &:& \textrm{a qubit labeled `messenger'}, \\
	\goD &:& \textrm{a qubit to the left of the active site}, \\
	\goR &:& \textrm{a qubit to the right of the active site},
\end{eqnarray}
and four more states, each of them 1-dimensional
\begin{eqnarray}
	\goX &:& \textrm{to the left of the qubits (`done/dead')}, \\ 
	\goO &:& \textrm{to the right of the qubits (`ready/unborn')}, \\
	\goT &:& \textrm{a `turn' state at the qubit sequence boundary}, \\
	\goB &:& \textrm{a `push' state used to move the qubits to the right}.
\end{eqnarray}
Altogether, the state space of each particle is 12-dimensional, with the structure
\begin{eqnarray}
	\cH_{12} = Q \oplus h = (q\otimes l) \oplus h = \big( \bigcirc \otimes \{\blacktriangleright, \vartriangleright, \times, \cdot\} \big) 
	\oplus \{ \goX,\goO,\goT,\goB \},
	\label{ch4:g12statespace}
\end{eqnarray}
where $Q$ is the subspace for the 4 types of qubits with $q$ the internal state space of a qubit and $l$ its four possible labels. The state space for the four other states is labeled $h$.

Without loss of generality, take a quantum circuit on $n$ qubits which consists of $K$ rounds of nearest neighbor gates as in Figure \ref{ch4:figurecircuit12}
\begin{figure}
	\begin{center}
	\includegraphics[width=3in]{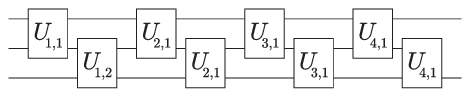} 
	\end{center}
	\caption{A quantum circuit for $n=3$ qubits with $K=4$ rounds of nearest neighbor gates $U_{k,s}$.
	 \label{ch4:figurecircuit12}}
\end{figure}
\begin{eqnarray}
	U = \left(U_{K,n-1}U_{K,n-2}\dots U_{K,1}\right) 
	\dots
		\left(U_{2,n-1}U_{2,n-2}\dots U_{2,1}\right)  
		\left(U_{1,n-1}U_{1,n-2}\dots U_{1,1}\right), 
\end{eqnarray}
where the gate $U_{k,s}$ acts on the pair of qubits $s,s+1$. I now encode the progression of $U$ into a set of states of a line of qudits
with length $nK$. The initial state $\ket{\psi_1}$ has a sequence of $n$ qubits on the left, and the rest of the chain is in the `ready' state $\goO$. The leftmost qubit is active $\goA$. Here I write out the initial state corresponding to a quantum circuit on $n=3$ qubits with $G=4\times 2$ nearest neighbor gates (here $K=4$, as in Figure \ref{ch4:figurecircuit12}).
\begin{eqnarray}
\ket{\psi_{1}} = \goA \goR \goR \goO \goO \goO \goO \goO \goO \goO \goO \goO
\end{eqnarray}
Using a few rules, a progression of states $\ket{\psi_{t}}$ is now constructed. The rules are prepared in such a way that there is always only one possible state $\ket{\psi_{t+1}}$ for a given state $\ket{\psi_t}$.

The first rule says that when the active spot $\goA$ is at position $kn+m$ where $1\leq k\leq K$ and $1 < m \leq n$, it can pass to the right. When this happens, the gate $U_{k,m}$ is first applied to the state of the corresponding qubits $(kn+m,kn+m+1)$, while their labels change from $\goA\goO$ to $\goD\goA$. 
\begin{eqnarray}
	1a\: &:&\: 
		\goA \goR
		\goes
		U_{k,m} \left(\goD \goA\right) \\
	&&\textrm{ on particles }(kn+m,kn+m+1). \nonumber
\end{eqnarray}
On the other hand, the active spot $\goM$ just moves to the right, changing the qubit labels without modifying their internal states.	
\begin{eqnarray}
	1b\: &:&\: 
		\goM \goR
		\goes
		\goR \goM
	\label{L12rule1}
\end{eqnarray}
Here is the corresponding progression of states (only 1$a$ applies here):
\begin{eqnarray}
 \goA \goR \goR \goO \goO \goO \goSO \label{ch4:g12start}\\
 \goD \goA \goR \goO \goO \goO \goSO \\
 \goD \goD \goA \goO \goO \goO \goSO
\end{eqnarray}
After these $n-1$ applications of rule 1$a$, the internal state of the $n$ qubits holds the state of the quantum circuit after the first round of gates.

The second set of rules involves the active spot reaching the front of the chain of qubits. Rule 2$a$ applies to particle pairs $(kn-1,kn)$ for integer $k$, while rule 2$b$ applies everywhere else. The reason for this is to ensure that when working out the progression of states backwards, there is always only one of the rules that applies in a given situation.
\begin{eqnarray}
	2a\: &:&\: 
		\goA \goO
		\goes
		\goD \goT
	\textrm{ on particles }(kn-1,kn), \\
	2b\: &:&\: 
		\goM \goO
		\goes
		\goD \goT
	\textrm{ everywhere else,}
	\label{L12rule2} 
\end{eqnarray}
Thus, after applying rule 2$a$ on particles $(1,2)$, I get
\begin{eqnarray}
\goD \goD \goA \goO \goO \goO \goSO \\
\goD \goD \goD \goT \goO \goO \goSO
\end{eqnarray}
The next three rules facilitate the sending of a message back to the left end of the chain.
\begin{eqnarray}
	3\: : \:
		\goT \goO 
		\goes
		\goB \goO 
\label{L12rule3} \\
	4\: : \:
		\goD \goB
		\goes 
		\goB \goR 
\label{L12rule4} \\
	5\: : \:
		\goX \goB
		\goes
		\goX \goT
\label{L12rule5} 
\end{eqnarray}
Also note that rule $3$ is simply $\goT \goes \goB$ for the rightmost particle of the line and rule $5$ is $\goB \goes \goT$ for the leftmost particle of the line. This produces the state progression
\begin{eqnarray}
\goD \goD \goD \goT \goO \goO \goSO \\
\goD \goD \goD \goB \goO \goO \goSO \\
\goD \goD \goB \goR \goO \goO \goSO \\
\goD \goB \goR \goR \goO \goO \goSO \\
\goB \goR \goR \goR \goO \goO \goSO \\
\goT \goR \goR \goR \goO \goO \goSO 
\end{eqnarray}
Observe now that the qubits have moved one step to the right from where they were in \eqref{ch4:g12start}. Finally, the last set of rules activates the leftmost of the qubits. Depending on the position this qubit is at, it will become either $\goA$ or $\goM$.
\begin{eqnarray}
	6a\: &:& \: 
		\goT \goR
		\goes
		\goX \goA
	\textrm{ for particles }(kn,kn+1), \\
	6b\: &:& \: 
		\goT \goR
		\goes
		\goX \goM
	\label{L12rule6}
\end{eqnarray}
The progression of states thus continues as
\begin{eqnarray}
\goT \goR \goR \goR \goO \goO \goSO \\
 \goX \goM \goR \goR \goO \goO \goSO  \label{mturn1}\\
 \goX \goD \goM \goR \goO \goO \goSO \\
 \goX \goD \goD \goM \goO \goO \goSO \\
 \goX \goD \goD \goD \goT \goO \goSO \\
 \goX \goD \goD \goD \goB \goO \goSO \\
 \goX \goD \goD \goB \goR \goO \goSO \\
 \goX \goD \goB \goR \goR \goO \goSO \\
 \goX \goB \goR \goR \goR \goO \goSO \\
 \goX \goT \goR \goR \goR \goO \goSO \\
 \goX \goX \goM \goR \goR \goO \goSO \label{mturn2}
\end{eqnarray}
and
\begin{eqnarray}
 \goX \goX \goM \goR \goR \goO \goSO \\
 \goX \goX \goD \goM \goR \goO \goSO \\
 \goX \goX \goD \goD \goM \goO \goSO \\
 \goX \goX \goD \goD \goD \goT \goSO \\
 \goX \goX \goD \goD \goD \goB \goSO \\
 \goX \goX \goD \goD \goB \goR \goSO \\
 \goX \goX \goD \goB \goR \goR \goSO \\
 \goX \goX \goB \goR \goR \goR \goSO \\
 \goX \goX \goT \goR \goR \goR \goSO \\
 \goX \goX \goX \goA \goR \goR \goSO \label{aturn}
\end{eqnarray}
Note that in steps \eqref{mturn1} and \eqref{mturn2} the last qubit is marked as a `messenger' $\goM$ by rule 6b. This rule is used at the left end of the qubit sequence until the qubits have moved $n$ steps to the right. Then, in \eqref{aturn}, the leftmost qubit is at position $kn+1$, which makes it become $\goA$ by rule 6a. The second round ($k=2$) of gate applications now begins.

\subsubsection{The Hamiltonian}
As in Section \ref{ch4:q3}, the Hamiltonian is an implementation of
\eqref{ch4:completeHam}.

The part checking the propagation of the computation is constructed from the above rules as
\begin{eqnarray}
	H = R^{(5)}_{(1)} + R^{(3)}_{(Kn)} 
		+ \sum_{s=1}^{nK-1} 
		  \sum_{r=1a}^{6b} 
			R^{(r)}_{(s,s+1)},
\end{eqnarray}
where each term $R^{(r)}$ is a projector acting on one or two neighboring qubits corresponding to the rule $r$. 

As an example, I write out the term $R^{(1a)}$ which checks the proper application of rule 1$a$.  This projector acts nontrivially on an 8-dimensional subspace of two neighboring particles spanned by $\goM \goR$ and $\goR \goM$.
Using the notation of \eqref{ch4:g12statespace}, this subspace can be thought of as
\begin{eqnarray}
	(q_1 \otimes l_1) \otimes (q_2 \otimes l_2),
\end{eqnarray}
where $q_1$ and $q_2$ are states of the two qubits, while their labels $(l_1, l_2)$ are in the state $(\hA,\hO)$ or 
$(\hX,\hA)$.
The corresponding projector (only for clock particles in position $kn+m,kn+m+1$) is 
\begin{eqnarray}
	R^{(1a)}_{(kn+m,kn+m+1)} &=& \half
			\ii_{q_1,q_2} \otimes
			\left(
			\ket{\hA\hO}\bra{\hA\hO} + 
			\ket{\hX\hA}\bra{\hX\hA} 
		\right)_{l_1 l_2}
		\\
		&-&	\half 
			\left(U_{k,m}\right)_{q_1,q_2} \otimes 
			\ket{\hX\hA}\bra{\hA\hO}_{l_1 l_2} \\
		&-& \half
			\left(U_{k,m}^\dagger\right)_{q_1,q_2} \otimes
			\ket{\hA\hO}\bra{\hX\hA}_{l_1 l_2}, 
\end{eqnarray}
where the gate $U_{k,m}$ is acting on the internal state of the two qubits.

The rest of the projectors are much simpler, therefore I write out only one of them. Rule 2$a$ involves the 4-dimensional subspace of two particles spanned by $\goA\goO$ and $\goX\goT$. 
\begin{eqnarray}
	P^{(2a)}_{(kn-1,kn)} &=& \half
			\left( \ii - \sigma_X \right)
\end{eqnarray}

\subsubsection{The low-lying spectrum of $H$}
What remains in the Hamiltonian construction are the legal clock state checking operators. First, analogously to the $\ket{01}\bra{01}$ check operators of the domain wall clock, I need
\begin{eqnarray}
	H_{check01} &=& 
		\sum_{A\in S_a}
		\sum_{B\in S_d}
		\ket{A}\bra{A} \otimes
		\ket{B}\bra{B}\\
		&+&
		\sum_{C\in S_r}
		\sum_{A\in S_a}
		\ket{C}\bra{C} \otimes
		\ket{A}\bra{A},
\end{eqnarray}
where
\begin{eqnarray}
	S_a &=& \{\goM,\goA,\goT,\goB\}, \\
	S_d &=& \{\goX,\goD\}, \\
	S_r &=& \{\goO,\goR\}.
\end{eqnarray}
Next, I need to check that there is an active site in the system at all by
\begin{eqnarray}
	H_{check1a0} &=& 
		\sum_{A\in S_d}
		\sum_{B\in S_r}
		\ket{A}\bra{A} \otimes
		\ket{B}\bra{B}.
\end{eqnarray}
Next, I need to check that there is not more than one active site in the system by adding a term
\begin{eqnarray}
	H_{check11} &=& 
		\sum_{A\in S_a}
		\sum_{B\in S_a}
		\ket{A}\bra{A} \otimes
		\ket{B}\bra{B}.
\end{eqnarray}
In fact, this is not sufficient to rule out that the two active sites wouldn't be at two places in the system that are spatially separated. However, Aharonov et al. \cite{CA:GottesmanLine} prove a {\em clairvoyance lemma}, which says that the expectation value of the complete Hamiltonian in a state with two active states is lower bounded by an inverse polynomial in $n$. This is because the two active sites have to move in the system, and thus will end up on neighboring sites.
Also, one needs to check whether the number of qubits in the system is right. Whether there are not too few of them can be checked by
\begin{eqnarray}
	H_{few} = \ket{\goO}\bra{\goO}_{c_n}
\end{eqnarray}
on the $n$-th particle from the left. On the other hand, whether there are not too many of them can be checked by
\begin{eqnarray}
	H_{many} = \ket{\goA \goR}\bra{\goA \goO}_{c_n,c_{n+1}}.
\end{eqnarray}
Finally, the initial state of the ancillae is checked by 
\begin{eqnarray}
	H_{init} &=& \sum_{k=m+1}^{n} \ket{1}\bra{1}_{\goO_{c_k}}
\end{eqnarray}
and the very first state in the sequence is initialized by  
\begin{eqnarray}
	H_{clockinit} &=& \ket{\goO}\bra{\goO}_{c_1} + 
	\ket{\goM}\bra{\goM}_{c_1}.
\end{eqnarray}
Altogether, for any quantum circuit $U$, Aharonov et. al. \cite{CA:GottesmanLine} constructed a Hamiltonian for $d=12$ dimensional particles on a line with nearest neighbor interactions, whose ground state energy encodes whether the quantum circuit $U$ can output $1$ with high probability on some state. Thus, Quantum (12,12)-SAT on a line is QMA$_1$ complete and Local Hamiltonian for $d=12$ particles on a line is QMA complete.


\subsection{Q-(11,11)-SAT on a line is QMA$_1$ complete.}\label{ch4:d11}
Here I present a modification of the above clock construction which requires only 11-dimensional qudits. The complexity of Q-(11,11)-SAT on the line is thus still QMA$_1$ complete (and Local Hamiltonian with $d=11$ qudits on a line is QMA complete). This result is my own previously unpublished work.

There are 8 dimensions required for the qubits in the previous clock construction. I needed two types of active qubits, one to send a message to the right $\goM$ and the other one $\goA$ to facilitate the application of gates. 
I also required two types of inactive qubits, one in the state `done' $\goD$ (to the left of the active spot), and the other one `ready' $\goR$ (to the right of the active spot). This can be reduced by having only one type of inactive qubit $\goR$, whose property (`ready' or `done') is determined by the parity of its position. To compensate for this, I introduce two more non-qubit states $\gox$ and $\goo$. The initial state is 
\begin{eqnarray}
\goSX \underbrace{\goA \goo \goR \goo \goR \goo}_{\textrm{3 qubits separated by }\goo} \goO \goO \goO \goO \goO \goO \goSO 
\label{d11start}
\end{eqnarray}
where I now follow each qubit by the non-qubit state $\goo$. The progression of the states $\ket{\Psi_t}$ becomes somewhat more involved. To move the qubits two spaces to the right, one goes through
\begin{eqnarray}
\goSX \goA \goo \goR \goo \goR \goo \goO \goO \goO \goO \goO \goO \goSO \label{d11start2}\\
\goSX \gox \goA \goR \goo \goR \goo \goO \goO \goO \goO \goO \goO \goSO \\
\goSX \gox \goR \goA \goo \goR \goo \goO \goO \goO \goO \goO \goO \goSO \\
\goSX \gox \goR \gox \goA \goR \goo \goO \goO \goO \goO \goO \goO \goSO \\
\goSX \gox \goR \gox \goR \goA \goo \goO \goO \goO \goO \goO \goO \goSO \\
\goSX \gox \goR \gox \goR \gox \goA \goO \goO \goO \goO \goO \goO \goSO \\
\goSX \gox \goR \gox \goR \gox \goR \goT \goO \goO \goO \goO \goO \goSO \\
\goSX \gox \goR \gox \goR \gox \goR \gox \goT \goO \goO \goO \goO \goSO \\
\goSX \gox \goR \gox \goR \gox \goR \gox \goB \goO \goO \goO \goO \goSO \\
\goSX \gox \goR \gox \goR \gox \goR \goB \goo \goO \goO \goO \goO \goSO \\
\goSX \gox \goR \gox \goR \gox \goB \goR \goo \goO \goO \goO \goO \goSO \\
\goSX \gox \goR \gox \goR \goB \goo \goR \goo \goO \goO \goO \goO \goSO \\
\goSX \gox \goR \gox \goB \goR \goo \goR \goo \goO \goO \goO \goO \goSO \\
\goSX \gox \goR \goB \goo \goR \goo \goR \goo \goO \goO \goO \goO \goSO \\
\goSX \gox \goB \goR \goo \goR \goo \goR \goo \goO \goO \goO \goO \goSO \\
\goSX \goB \goo \goR \goo \goR \goo \goR \goo \goO \goO \goO \goO \goSO \\
\goSX \goT \goo \goR \goo \goR \goo \goR \goo \goO \goO \goO \goO \goSO \\
\goSX \goX \goT \goR \goo \goR \goo \goR \goo \goO \goO \goO \goO \goSO \\
\goSX \goX \goX \goM \goo \goR \goo \goR \goo \goO \goO \goO \goO \goSO \label{d11mturn1}
\end{eqnarray}
Observe that in \eqref{d11mturn1}, the qubits have moved two spots to the right with respect to their starting position \eqref{d11start2}. This is an analogue of obtaining the state \eqref{mturn1} from the state \eqref{ch4:g12start} in the previous Section. After the qubits have moved $2n$ spots to the right, the leftmost of them becomes $\goA$ again, and the next round of gate applications commences.

The formal rules governing this progression of states are:
\begin{eqnarray}
	1a\: &:&\: 
		\goA \goo
		\goes
		\gox \goA
	\\
	1b\: &:&\: 
		\goM \goo
		\goes
		\gox \goM
	\label{L11rule1}
\end{eqnarray}
shifting the active qubit position. Next, when a qubit marked `apply gates' meets another qubit, depending on their position, a unitary gate is applied to the internal states of the two qubits as the second qubit gets the `apply gates' label. 
\begin{eqnarray}
	2a\: &:&\: 
		\goA \goR
		\goes
		U_{k,m} \left( \goD \goA \right) \\
		&&\textrm{ on particles }(2kn+2m,2kn+2m+1), \nonumber
\end{eqnarray}
where $k$ and $m$ are integers. Meanwhile, the transition rule for a `messenger' qubit meeting another qubit is just
\begin{eqnarray}
	2b\: &:&\: 
		\goM \goR
		\goes
		\goR \goM .
	\label{L11rule2}
\end{eqnarray}
When the qubit with the `apply gates' label is at the right end of the qubit sequence, it changes according to 
\begin{eqnarray}
	2a\: &:&\: 
		\goA \goO
		\goes
		\goD \goT \\
	&&\textrm{ on particles }(kn,kn+1),\nonumber
\end{eqnarray}
while everywhere else it changes according to
\begin{eqnarray}
	2b\: &:&\: 
		\goM \goO
		\goes
		\goD \goT.
	\label{L11rule2end} 
\end{eqnarray}
The reason for this distinction between rules 2$a$ and 2$b$ is that the rules applied backwards also have to be unique.
To send the active spot to the left now takes several steps:
\begin{eqnarray}
	3a\: &:&\: 
		\goT \goO
		\goes
		\gox \goT \\
	&&\textrm{ on particles }(2k+1,2k+2), \nonumber
\end{eqnarray}
on neighboring particles with (odd,even) locations, 
while on the other pairs of particles the transition rule is
\begin{eqnarray}
	3b\: &:&\: 
		\goT \goO
		\goes
		\goB \goO 
		\label{L11rule3}\\
	&&\textrm{ on particles }(2k+2,2k+3). 
	\nonumber
\end{eqnarray}
Next, the qubits are pushed to the right as the active symbol $\goB$ moves to the left. Then the active spot bounces back as it reaches the left end of the qubit sequence.
\begin{eqnarray}
	4a\: &:& \:
		\gox \goB
		\goes 
		\goB \goo 
\\
	4b\: &:& \:
		\goR \goB
		\goes 
		\goB \goR 
\label{L11rule4} \\
	5a\: &:& \:
		\goX \goB
		\goes
		\goX \goT
\\
	5b\: &:& \:
		\goT \goo
		\goes
		\goX \goT
\label{L11rule5} 
\end{eqnarray}
Finally, the leftmost qubit becomes active. When the particle pair involved is at position $(kn,kn+1)$, it means that the qubits have moved $2n$ spots to the right and another round of gate applications can start. This activation proceeds according to the rule
\begin{eqnarray}
	6a\: &:& \: 
		\goT \goR
		\goes
		\goX \goA\\
	&&\textrm{ on particles }(2kn-1,2kn), \nonumber
\end{eqnarray}
while everywhere else the last of the qubits is activated into the `messenger' state $\goM$ as
\begin{eqnarray}
	6b\: &:& \: 
		\goT \goR
		\goes
		\goX \goM
	\label{L11rule6}
\end{eqnarray}


\section{Quantum 2-SAT in general geometry (for qudits)}\label{ch4:q2general}
\subsection{The Triangle Clock and Q-(5,3)-SAT}
\label{ch4:triangle}
In a recent paper, Eldar and Regev \cite{lh:RegevTriangle}, present a novel idea. Their `triangle' clock construction abandons a simple linear progression of states $\ket{\Psi_t}$ encoding the progression of the quantum circuit $U$. This allows them to prove that certain variant of Quantum 2-SAT for higher spins is QMA$_1$ complete. 

The complexity of classical 2-Satisfiability with higher dimensional particles is known. Already the simplest variant, (3,2)-SAT, i.e. 2-SAT for a trit and a bit is NP-complete, as shown in Section \ref{ch1:sat23}. On the other hand, the complexity of Q-SAT for higher spins is an open question. In Section \ref{ch4:q3} I proved that Quantum (3,2,2)-SAT, i.e. Quantum 3-SAT for a qutrit and two qubits, is QMA$_1$ complete. Here I focus on the higher spin version of Quantum 2-SAT. Eldar and Regev's result which I now review is that Quantum (5,3)-SAT is QMA$_1$-complete. Here $(5,3)$ means that each projector in their Hamiltonian involves one cinquit (a particle with dimension $d=5$) and one qutrit ($d=3$).

\begin{figure}
	\begin{center}
	\includegraphics[width=4in]{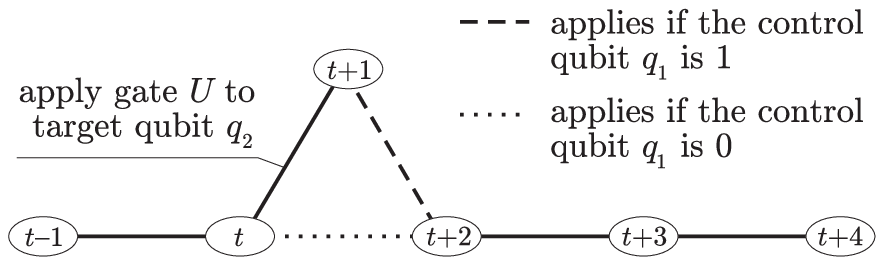} 
	\end{center}
	\caption{The `triangle' clock transition rules allowing to apply a controlled gate to two work qubits using interactions that involve only one work qubit at a time.
	\label{ch4:figuretriangle2}}
\end{figure}
The proof of QMA$_1$-hardness relies on the `triangle' clock construction depicted in Figure \ref{ch4:figuretriangle2}. 
The state 
\begin{eqnarray}
	\ket{\Psi_t} = \ket{\psi_t} \otimes \ket{t}_c
\end{eqnarray}
has an allowed transition to the state 
\begin{eqnarray}
	\ket{\Psi_{t+1}} = \left(\sigma_x^{(q_2)} \ket{\psi_t} \right) \otimes \ket{t+1}_c,
\end{eqnarray}
which has the operator $\sigma_x$ applied to a target work qubit $q_2$. Next, the transition from the state $\ket{\Psi_t}$ to the state
$\ket{\Psi_{t+2}}$ is conditioned on a control (work) qubit $q_1$ to be in the state $\ket{0}$. On the other hand, if the control qubit $q_1$ is in the state $\ket{1}$, the transition $\ket{\Psi_{t+1}}\goes\ket{\Psi_{t+2}}$ applies. In this fashion, a two-qubit controlled gate, in this case a CNOT, is applied to two work qubits, while interacting with only one of them at a time. The corresponding `triangle' Hamiltonian consists of three terms:
\begin{eqnarray}
	H_{\vartriangle} 
	&=& 
	\half \left[\ii\otimes(P_t + P_{t+1}) 
	- \sigma_x^{(q_2)} \otimes 
	\left(X_{t+1,t} + X_{t+1,t}^\dagger\right) \right]
	\label{ch4:Htriangle}
	\\
	&+& 
	\ket{0}\bra{0}_{q_1} \otimes 
		\half \left[P_t + P_{t+2}
	- 
	\left(X_{t+2,t} + X_{t+2,t}^\dagger\right) \right]
	\nonumber
	\\
	&+& 
	\ket{1}\bra{1}_{q_1} \otimes 
		\half \left[P_{t+1} + P_{t+2} 
		- 
	\left(X_{t+2,t+1} + X_{t+2,t+1}^\dagger\right) \right].
	\nonumber
\end{eqnarray}
The first term checks that for the transition between states $t$ and $t+1$, a $\sigma_x$ gate is applied to the target qubit $q_2$. The second and third terms check the transitions $\ket{\Psi_t} \rightarrow \ket{\Psi_{t+2}}$ and $\ket{\Psi_{t+1}} \rightarrow \ket{\Psi_{t+2}}$, conditioned on the state of the control work qubit $q_1$. The reader is welcome to check that the ground state energy of $H_{\vartriangle}$ is exactly zero, and the states in the ground state subspace of $H_{\vartriangle}$ are of the form
\begin{eqnarray}
	\ket{\Psi_{\vartriangle}} &=& c_{\vartriangle} \left(
	\dots + 
		\ket{\Psi_t} + \ket{\Psi_{t+1}} + \ket{\Psi_{t+2}}
		 +  \dots\right),
\end{eqnarray}
where
\begin{eqnarray}
	\ket{\Psi_t} &=& \ket{\psi}\otimes\ket{t}_c,\\
	\ket{\Psi_{t+1}} &=& \left(\sigma_{x}^{(q_2)} \ket{\psi} \right)\otimes\ket{t+1}_c,\\
	\ket{\Psi_{t+2}} &=& \left( \textrm{CNOT}_{q_1,q_2} \ket{\psi}\right)\otimes\ket{t+2}_c
\end{eqnarray}
and $\ket{\psi}$ is a state of the work register.

The goal of Eldar and Regev in \cite{lh:RegevTriangle} is to construct a Quantum 2-SAT Hamiltonian (for a qutrit and cinquit), and for that one needs to implement the clock in a 2-local fashion. The clock register is made of alternating particles with dimension $d=3$ and $d=5$. The 3-dimensional clock particle states are $\rO, \rA$ and $\rX$, while the 5-dimensional clock particle states are $\ro,\ra,\rb,\rc$ and $\rx$. The progression of clock states is similar to the domain wall combined with the pulse clock, as  discussed in Section \ref{ch4:clocks}:
\begin{eqnarray}
\ket{1}_c &=& \rA \ro \rO \ro \rO \ro \rO \ro \\
\ket{2}_c &=& \rX \ra \rO \ro \rO \ro \rO \ro \\
\ket{3}_c &=& \rX \rb \rO \ro \rO \ro \rO \ro \\
\ket{4}_c &=& \rX \rc \rO \ro \rO \ro \rO \ro \\
\ket{5}_c &=& \rX \rx \rA \ro \rO \ro \rO \ro \\
\ket{6}_c &=& \rX \rx \rX \ra \rO \ro \rO \ro \\
\ket{7}_c &=& \rX \rx \rX \rb \rO \ro \rO \ro \\
\ket{8}_c &=& \rX \rx \rX \rc \rO \ro \rO \ro \\
\ket{9}_c &=& \rX \rx \rX \rx \rA \ro \rO \ro \\
\ket{10}_c &=& \rX \rx \rX \rx \rX \ra \rO \ro \\
\ket{11}_c &=& \rX \rx \rX \rx \rX \rb \rO \ro \\
\ket{12}_c &=& \rX \rx \rX \rx \rX \rc \rO \ro \\
\ket{13}_c &=& \rX \rx \rX \rx \rX \rx \rA \ro \\
\ket{14}_c &=& \rX \rx \rX \rx \rX \rx \rX \ra
\end{eqnarray}
with a single active particle for a given state $\ket{t}_c$ (denoted by a black filled symbol). The three clock states $\ket{t}_c,\ket{t+1}_c$ and $\ket{t+2}_c$ involved in the triangle construction in Figure \ref{ch4:figuretriangle2} correspond to three consecutive states of the clock register in whose the active site is on one 5-dimensional particle, such as $\ket{2}_c,\ket{3}_c$ and $\ket{4}_c$. 

The propagation Hamiltonian $H_{prop}$ in \ref{ch4:completeHam} now consists of four types of terms. For a neighboring clock qutrit and cinquit I have
\begin{eqnarray}
	H_{prop(3,5)}^{t} &=& \half \left(
		\ket{\rA}\bra{\rA}_{(3)}\otimes \ii_{(5)} + 
		\ii_{(3)} \otimes \ket{\ra}\bra{\ra}_{(5)}\right) \\
		&-& \half \left( \ket{\rX \ra}\bra{\rA \ro}_{(3,5)} 
		+ \ket{\rA \ro}\bra{\rX \ra}_{(3,5)} 
		\right),
\end{eqnarray}
while for a neighboring clock cinquit and qutrit, the transition is checked by \begin{eqnarray}
	H_{prop(5,3)}^{t} &=& \half \left(
		\ket{\rc}\bra{\rc}_{(5)} \otimes \ii_{(3)} +
		\ii_{(5)} \otimes \ket{\rA}\bra{\rA}_{(3)}  \right) \\
		&-& \half \left(
		 \ket{\rx \rA}\bra{\rc \rO}_{(5,3)} 
		- \ket{\rc \rO}\bra{\rx \rA}_{(5,3)} 
		\right).
\end{eqnarray}

The third type of term in $H_{prop}$ is the above described `triangle' Hamiltonian \eqref{ch4:Htriangle}, implemented using the following projectors and transition operators:
\begin{eqnarray}
	P_{t} &=& \ket{\ra}\bra{\ra}_{(5)}, \label{ch4:clock35implement}\\
	P_{t+1} &=& \ket{\rb}\bra{\rb}_{(5)}, \nonumber\\
	P_{t+2} &=& \ket{\rc}\bra{\rc}_{(5)}, \nonumber\\
	X_{t+1,t} &=& \ket{\rb}\bra{\ra}_{(5)}, \nonumber\\
	X_{t+2,t} &=& \ket{\rc}\bra{\ra}_{(5)}, \nonumber\\
	X_{t+2,t+1} &=& \ket{\rc}\bra{\rb}_{(5)}. \nonumber
\end{eqnarray}
This term checks the application of a controlled-not (CNOT) gate on two work qubits in a 2-local fashion, as each of the three terms in $H_{\vartriangle}$ is a projector acting nontrivially on the Hilbert space of one cinquit $(5)$ and one qubit $(2)$.

So far, I have shown how Eldar and Regev check the application of a 2-local controlled gate on two work qubits by using projectors on one cinquit and one qutrit at a time. However, for universality one also needs to apply single qubit unitary gates. The final, fourth type of term in $H_{prop}$ is then
\begin{eqnarray}
	H_{prop(single)}^{t} &=& 
		\half \left( \ii \otimes (P_{t} + P_{t+1})
			- U \otimes X_{t+1,t} - \left(U \otimes X_{t+1,t}\right)^\dagger \right) \\
			&+&
			\ii \otimes \half \left( P_{t+1} + P_{t+2} -
			\left( X_{t+2,t+1} + X_{t+2,t+1}^\dagger\right) \right),
\end{eqnarray}
and it is implemented using \eqref{ch4:clock35implement}. The first line is a projector which checks the application of a unitary gate $U$ on a single work qubit between states $\ket{\Psi_{t}}$ and $\ket{\Psi_{t+1}}$. The second line is a projector checking a transition from $\ket{\Psi_{t+1}}$ to $\ket{\Psi_{t+2}}$ without changing the work qubits.

Finally, the legal clock subspace checking operators are the analogue of the operators $\ket{01}\bra{01}$ for the domain wall clock. They read
\begin{eqnarray}
	H_{clock(3,5)} &=& 
		\ket{\rX}\bra{\rX} \otimes \ket{\ro}\bra{\ro} \\
		&+& \ket{\rA}\bra{\rA} \otimes \left( \ii - \ket{\ro}\bra{\ro}\right) \\
		&+& \ket{\rO}\bra{\rO} \otimes \left( \ii - \ket{\ro}\bra{\ro}\right) 
\end{eqnarray}
for a neighboring qutrit and cinquit, while for a neighboring cinquit and qutrit they are
\begin{eqnarray}
	H_{clock(5,3)} &=& 
		\ket{\rx}\bra{\rx} \otimes \ket{\rO}\bra{\rO} \\
		&+& \left( \ii - \ket{\rx}\bra{\rx}\right) \otimes \ket{\rA}\bra{\rA} \\
		&+& \left( \ii - \ket{\rx}\bra{\rx}\right) \otimes \ket{\rX}\bra{\rX} 
\end{eqnarray}
Whether there is an active site in the clock register at all is taken care of by
\begin{eqnarray}
	H_{clockinit} = \ket{\rO}\bra{\rO}_{c_1} + \ket{\rX}\bra{\rX}_{c_L},
\end{eqnarray}
where $c_1$ and $c_L$ are the leftmost and the rightmost qutrits in the clock register.

Eldar and Regev then proceed to show that if the circuit $U$ accepts a state $\ket{\psi_{yes}}$, the history state corresponding to $\ket{\psi_{yes}}$ with extra terms for the states above the linear progression of states such as $\ket{\Psi_{t+1}}$ in Figure \ref{ch4:figuretriangle2} is the zero-energy ground state of $H$. Also, they show that in the `no' case, the ground state energy of $H$ is bounded from below by an inverse polynomial in $n$. Thus, Quantum (5,3)-SAT is QMA$_1$-complete. In the following sections, I build on this result.


\subsection{A QMA$_1$-complete 3-local Hamiltonian with restricted terms}
\label{ch4:triangleXZ}
Here I implement the 2-local {\em triangle clock} which originally uses cinquits and qutrits, using qubits. This way I obtain another QMA$_1$ complete 3-local Hamiltonian construction which does not require any large penalty terms. Moreover, all the terms involved are constructed from a restricted set of operators. This is an unpublished result obtained with Peter Love during the Computational Complexity of Quantum Hamiltonian Systems workshop in Leiden, Netherlands (July 2007).

As in Section \ref{ch4:q3}, I use a combined domain wall and pulse clock. First, replace each $d=3$ particle in the construction of the previous section with two qubits $w_1,w_2$ as
\begin{eqnarray}
	\rO &\goes& \mathtt{00}, \\
	\rA &\goes& \mathtt{10}, \\
	\rX &\goes& \mathtt{11},
\end{eqnarray}
as in the domain wall clock in Section \ref{ch4:clocks}. Next, replace each $d=5$ particle in the above construction with three qubits $p_1,p_2,p_3$ of a pulse clock as 
\begin{eqnarray}
	\ro &\goes& \mathtt{000}, \\
	\ra &\goes& \mathtt{100}, \\
	\rb &\goes& \mathtt{010}, \\
	\rc &\goes& \mathtt{001}, \\
	\rx &\goes& \mathtt{000}.
\end{eqnarray}
Note that the transitions between the active (pulse) states of the $d=5$ clock involve only two of the three qubits $p_i$. Therefore, the  clock register transitions between $\ket{t}_c$, $\ket{t+1}_c$ and $\ket{t+2}_c$ in the triangle construction of the previous section are now implemented 2-locally as 
\begin{eqnarray}
	\ra \goes \rb \qquad &:&
	\qquad
	\ket{\mathtt{10}}_{p_1,p_2}\goes 
	\ket{\mathtt{01}}_{p_1,w_2}, \\
	\ra \goes \rc \qquad &:&
	\qquad
	\ket{\mathtt{10}}_{p_1,p_3}\goes 
	\ket{\mathtt{01}}_{p_1,p_3}, \\
	\rb \goes \rc \qquad &:&
	\qquad
	\ket{\mathtt{10}}_{p_2,p_3}\goes 
	\ket{\mathtt{01}}_{p_2,p_3}.
\end{eqnarray}
Adding an interaction with a single work qubit at a time, this becomes 3-local. The propagation Hamiltonian is thus made from projectors on three qubits. This is true also for the transition from an active state of a qutrit clock particle to an active state of the neighboring cinquit clock particle. It is implemented 3-locally as
\begin{eqnarray}
	\rA \ro \goes \rX \ra \qquad &:& \qquad 
		\ket{\mathtt{10}}_{w_1,w_2} \ket{\mathtt{0}}_{p_1} \goes \ket{\mathtt{11}}_{w_1,w_2} \ket{\mathtt{1}}_{p_1},
\end{eqnarray}
while moving the active state from a cinquit clock particle to a qutrit is done as
\begin{eqnarray}
	\rc \rO \goes \rx \rA \qquad &:&
	\qquad
	\ket{\mathtt{1}}_{p_3} \ket{\mathtt{00}}_{w_1,w_2} \goes \ket{\mathtt{0}}_{p_3} \ket{\mathtt{10}}_{w_1,w_2}.
\end{eqnarray}

The legal clock state progression is now
\begin{eqnarray}
	\ket{0}_c &=& \ket{\mathtt{10\:\:000\:\:00\:\:000\:\:00\:\:000\:\:00}} \\
	\ket{1}_c &=& \ket{\mathtt{11\:\:100\:\:00\:\:000\:\:00\:\:000\:\:00}} \\
	\ket{2}_c &=& \ket{\mathtt{11\:\:010\:\:00\:\:000\:\:00\:\:000\:\:00}} \\
	\ket{3}_c &=& \ket{\mathtt{11\:\:001\:\:00\:\:000\:\:00\:\:000\:\:00}} \\
	\ket{4}_c &=& \ket{\mathtt{11\:\:000\:\:10\:\:000\:\:00\:\:000\:\:00}} \\
	\ket{5}_c &=& \ket{\mathtt{11\:\:000\:\:11\:\:100\:\:00\:\:000\:\:00}} \\
	\ket{6}_c &=& \ket{\mathtt{11\:\:000\:\:11\:\:010\:\:00\:\:000\:\:00}} \\
	\ket{7}_c &=& \ket{\mathtt{11\:\:000\:\:11\:\:001\:\:00\:\:000\:\:00}} \\
	\ket{8}_c &=& \ket{\mathtt{11\:\:000\:\:11\:\:000\:\:10\:\:000\:\:00}} \\
	\ket{9}_c &=& \ket{\mathtt{11\:\:000\:\:11\:\:000\:\:11\:\:100\:\:00}} \\
	\ket{10}_c &=& \ket{\mathtt{11\:\:000\:\:11\:\:000\:\:11\:\:010\:\:00}} \\
	\ket{11}_c &=& \ket{\mathtt{11\:\:000\:\:11\:\:000\:\:11\:\:001\:\:00}} \\
	\ket{12}_c &=& \ket{\mathtt{11\:\:000\:\:11\:\:000\:\:11\:\:000\:\:10}} 
\end{eqnarray}

However, this is not a proof that Quantum 3-SAT is QMA$_1$ complete. I have not found a way to make the terms checking the legal clock states 3-local projectors. If I want the term which ensures proper coupling of the domain wall and the pulse clock to be a non-negative operator, it necessarily becomes 5-local. On the other hand, I have found a 5-local positive operator made out of 2-local operators, which does the job. This is reminiscent of the idea presented in Section \ref{ch4:q3}. 

The operators checking whether a state is in a legal clock subspace are the following. First, for the domain-wall clock particles, the operator
\begin{eqnarray}
	H_{clock(wall)} &=& \sum_{\langle k,l \rangle} \ket{01}\bra{01}_{w_k,w_l}
\end{eqnarray}
where $w_k, w_l$ are consecutive domain wall qubits, checks the domain wall clock is correctly encoded. Next, I check whether the pulse clocks have at most one pulse in each of them by 
\begin{eqnarray}
	H_{clock(pulse)} &=& \sum_{k} 
	\left(
	\ket{11}\bra{11}_{p^{(k)}_1,p^{(k)}_2}
	+\ket{11}\bra{11}_{p^{(k)}_1,p^{(k)}_3}
	+ \ket{11}\bra{11}_{p^{(k)}_2,p^{(k)}_3}	
	\right),
\end{eqnarray}
where $p^{(k)}_1,p^{(k)}_2,p^{(k)}_3$ is the $k$-th pulse clock triplet. Finally, I need to check that the domain wall and the pulse clocks are correctly coupled. In other words, I need to check that the state in which the pulse clock was {\em not} initialized at the domain wall
\begin{eqnarray}
&&\ket{\dots \mathtt{11\:\:000\:\:00} \dots}
\end{eqnarray}
is illegal, while the states 
\begin{eqnarray}
&&\ket{\dots \mathtt{00\:\:000\:\:00} \dots} \\
&&\ket{\dots \mathtt{10\:\:000\:\:00} \dots} \\
&&\ket{\dots \mathtt{11\:\:100\:\:00} \dots} \label{pulseleft1}\\
&&\ket{\dots \mathtt{11\:\:010\:\:00} \dots} \\
&&\ket{\dots \mathtt{11\:\:001\:\:00} \dots} \\
&&\ket{\dots \mathtt{11\:\:000\:\:10} \dots} \label{pulseleft4}\\
&&\ket{\dots \mathtt{11\:\:000\:\:11} \dots} \label{pulseleft5}
\end{eqnarray}
are perfectly fine. For this task, I now construct two versions of a  5-local positive semidefinite operator $H_{clockpulse}$, the first made from 3-local and the second made from 2-local terms. In both cases, $H_{clockpulse}$ acts nontrivially on the space of three pulse qubits $p_i$ and two neighboring domain wall qubits $w_i$. 

In the first case, $H_{clockpulse}$ uses at most 3-local interactions:
\begin{eqnarray}
	H_{clockpulse}^{(3-loc)} &=& \ket{1}\bra{1}_{w_1} \otimes \left( 
		\sigma_{z}^{(p_1)} +
		\sigma_{z}^{(p_2)} +
		\sigma_{z}^{(p_3)} - 1
	\right) \otimes \ket{0}\bra{0}_{w_2} \nonumber\\
		&+& \sum_{i\neq j} \ket{11}\bra{11}_{p_i,p_j}.
\end{eqnarray}
It checks whether there is a single $1$ in the pulse clock, when the neighboring domain wall qubits are $1$ and $0$. 

The second variant uses only 2-local interactions:
\begin{eqnarray}
	H_{clockpulse}^{(2-loc)} &=& \ket{1}\bra{1}_{w_1} \otimes \left( 
		\sigma_{z}^{(p_1)} +
		\sigma_{z}^{(p_2)} +
		\sigma_{z}^{(p_3)} +
		\sigma_{z}^{(w_2)} - 2
	\right)  \nonumber \\
			&+& \sum_{\langle i,j \rangle} \ket{11}\bra{11}_{p_i,p_j} 
			+ \sum_{i} \ket{1}\bra{1}_{p_i} \otimes \ket{1}\bra{1}_{w_2},
\end{eqnarray}
checking that when a domain wall qubit to the left of a pulse clock is $1$, exactly one of the following 4 qubits (3 pulse and 1 domain wall) must be $1$, as seen in \eqref{pulseleft1}-\eqref{pulseleft5}.

Altogether, I constructed a Quantum 5-SAT Hamiltonian from at most 3-local interaction terms. Each of the $O(n)$ terms in this Hamiltonian has constant norm. Moreover, I retained the properties of the previous construction, as I only implemented the Hamiltonian \eqref{ch4:completeHam} in a different system. Therefore, Quantum 5-SAT made from 3-local interactions of a restricted type is QMA$_1$ complete. Moreover, 3-local Hamiltonian with constant norm terms of a restricted type is QMA-complete.


\section{Train Switch Construction: Universal Quantum Computation using a Quantum 3-SAT Hamiltonian}
\label{ch4:train}
In this Section, building on Regev's triangle construction idea, I prove that one can perform universal quantum computation using a Quantum 3-SAT Hamiltonian $H_{3S}$ (for qubits). It is suitable for an Adiabatic Quantum Algorithm where I slowly change the Hamiltonian from a simple starting Hamiltonian $H_B$ to $H_{3S}$, or for a Hamiltonian Computer model where I let an easily prepared starting state evolve with $H_{3S}$ for not too long a time.

\begin{figure}
	\begin{center}
	\includegraphics[width=4.5in]{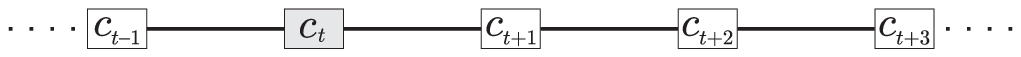} 
	\end{center}
	\caption{The `pulse' clock construction consists of $T+1$ clock qubits on a line, where only one of them is in the state $\ket{1}$ at a time.
	\label{ch4:figurepulse}}
\end{figure}
As I want to perform a quantum computation, as opposed to constructing a Hamiltonian with a unique ground state, I have the advantage that I can choose the initial state of the system. This allows me to use the pulse clock (see Section \ref{ch4:clocks} and Figure \ref{ch4:figurepulse}) 
\begin{eqnarray}
	\ket{t}_c = \ket{0_{c_0} \dots 0_{c_{t-1}} 1_{c_{t}} 0_{c_{t+1}} \dots 0_{c_L}}
\end{eqnarray}
as my starting point. The main innovation I add to it is the following {\em train switch} clock construction.
\begin{figure}
	\begin{center}
	\includegraphics[width=4.5in]{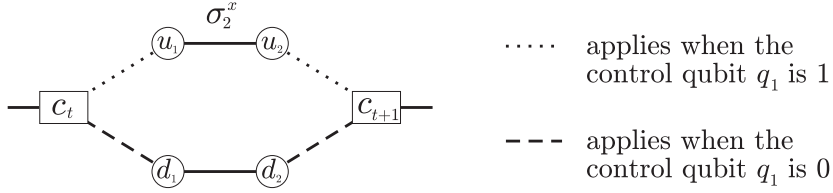} 
	\end{center}
	\caption{The `train switch' clock construction for the application of a CNOT gate using only 3-local terms. There are four extra qubits inserted between $c_t$ and $c_{t+1}$. The transitions from the original clock qubits to the gadget qubits are controlled by one of the work qubits. During the transition from $u_1$ to $u_2$, the target  work qubit is flipped. This gadget checks the application of a CNOT gate using 3-local projectors on qubits.
	\label{ch4:figuretrain}}
\end{figure}
My goal is to use a 3-local projector to check whether a 2-qubit controlled gate CNOT$_{q_1,q_2}$ was properly applied to the work register when transitioning from the state $\ket{\Psi_t}$ to $\ket{\Psi_{t+1}}$. For this, I replace two clock qubits $c_t$ and $c_{t+1}$ with a six-qubit gadget as in Figure \ref{ch4:figuretrain}. If the control qubit $q_1$ is $\ket{1}$, the active site in the clock register follows the upper rail, and takes the lower rail otherwise. Moreover, when the active site in the clock register moves from $u_1$ and $u_2$, the target qubit $q_2$ is flipped. This effectively applies a CNOT gate between $q_1$ and $q_2$, between the states $\ket{\Psi_t}$ and $\ket{\Psi_{t+1}}$. In detail, the transition checking operators are
\begin{eqnarray}
	H_{t,u_1} &=& \half \left( \ii \otimes 
		\left( P_{t} + P_{u_1} \right)
 			- \ket{1}\bra{1}_{q_1} \otimes 
 			\left( X_{u_1,t} + X_{u_1,t}^\dagger \right)
		\right), \\ 
	H_{u_1,u_2} &=& \half \left(
	\ii \otimes  
		\left( P_{u_1} + P_{u_2} \right)
			- \sigma_{x}^{(q_2)} \otimes 
			\left( X_{u_2,u_1} 
				+ X_{u_2,u_1}^\dagger \right)
		\right), \\
	H_{u_2,t+1} &=& \half \left(
	\ii \otimes 
		\left( P_{u_2} + P_{t+1} \right)
			- \ket{1}\bra{1}_{q_1} \otimes 
			\left( X_{t+1,u_2} + X_{t+1,u_2}^\dagger \right)
		\right)
\end{eqnarray}
for the upper rail, with the projectors $P$ and time-increase operators $X$ implemented as before for the pulse clock (see Section \ref{ch4:clocks}):
\begin{eqnarray}
	P_t &=& \ket{1}\bra{1}_{c_t}, \\
	X_{b,a} &=& \ket{10}\bra{01}_{b,a}.
\end{eqnarray}
Analogously, for the lower rail I write
\begin{eqnarray}
	H_{t,d_1} &=& \half \left(
		\ii \otimes 
		\left( P_{t} + P_{u_1} \right)
			-  \ket{0}\bra{0}_{q_1} \otimes 
			\left( X_{d_1,t} + X_{d_1,t}^\dagger \right)
		\right), \\
	H_{d_1,d_2} &=& \half \left(
		\ii \otimes 
		\left( P_{d_1} + P_{d_2} \right)
			-  \ii \otimes 
			\left( X_{d_2,d_1} 
				+ X_{d_2,d_1}^\dagger \right)
			\right), \\
	H_{d_2,t+1} &=& \half \left(
	\ii \otimes 
		\left( P_{d_2} + P_{t+1} \right)
			-  \ket{0}\bra{0}_{q_1} \otimes 
			\left( X_{t+1,d_2} + X_{t+1,d_2}^\dagger \right)
		\right).
\end{eqnarray}
Let me examine the action of this Hamiltonian on the state 
\begin{eqnarray}
	\ket{\Psi_t} = \ket{\psi_t}\otimes \ket{t}_c
		= \left(
			 a 
			\ket{0}_{q_1} \ket{\alpha}_{q_2,\dots}
		+ 	b 
			\ket{1}_{q_1} \kets{\beta}_{q_2,\dots}
		\right) 
		\otimes \ket{t}_c.
\end{eqnarray}
When I consider only the forward moving terms in 
\begin{eqnarray}
	H_{train(t,t+1)} = 
	  H_{t,u_1} 
	+ H_{u_1,u_2}
	+ H_{u_2,t+1}
	+ H_{t,d_1}
	+ H_{d_1,d_2}
	+ H_{d_2,t+1}
\end{eqnarray}
(forgetting the identity and time-decreasing terms), I obtain the following three states:
\begin{eqnarray}
	\kets{\Psi^{(u_1,d_1)}_t} 
	&=& 
		 a 
		\ket{0}_{q_1} \ket{\alpha}_{q_2,\dots}
		\otimes \ket{d_1}_c
		+
		b 
		 \ket{1}_{q_1} \kets{\beta}_{q_2,\dots}
		 \otimes \ket{u_1}_c, 
		\\
	\kets{\Psi^{(u_2,d_2)}_t}  
	&=& 
		 a 
		\ket{0}_{q_1} \ket{\alpha}_{q_2,\dots}
		\otimes \ket{d_2}_c
		+
				 b 
		 \ket{1}_{q_1} 
		 \kets{\beta'}_{q_2,\dots}
		 \otimes \ket{u_2}_c 
, 
		\\
	\kets{\Psi_{t+1}} 
	&=& \left(
			 a 
			\ket{0}_{q_1} \ket{\alpha}_{q_2,\dots}
		+ 	b 
			\ket{1}_{q_1} 
			\kets{\beta'}_{q_2,\dots}
		\right) 
		\otimes \ket{t+1}_c,
\end{eqnarray}
where 
\begin{eqnarray}
	\ket{\beta'}_{q_2,\dots} = 
	\sigma_{x}^{(q_2)} \kets{\beta}_{q_2,\dots}
\end{eqnarray} 
is the state obtained of the work qubits (besides $q_1$) after the work qubit $q_2$ was flipped. Observe that the work qubits of the state $\kets{\Psi_{t+1}}$ have the CNOT gate applied to $q_1,q_2$, as I wanted, therefore
\begin{eqnarray}
	\kets{\Psi_{t+1}} 
	&=& \left(
		\textrm{CNOT}_{q_1,q_2}\ket{\psi_t}
	\right)
	\otimes\ket{t+1}_c.
\end{eqnarray}

For a single qubit gate application, the transition checking is much simpler. At places where I do not insert the train switch gadget, I use the usual single qubit gate checking operator
\begin{eqnarray}
	H_{single(t)} = \half \left( 
		\ii \otimes (P_{t} + P_{t+1})
		- U \otimes X_{t+1,t}
		- \left(U \otimes X_{t+1,t}\right)^\dagger
	\right).
\end{eqnarray}

The Hamiltonian 
\begin{eqnarray}
	H_{3S} &=& 
		\sum_{t: U_t^{(q_1)}} H_{single(t)}
		+ \sum_{t: U_t^{(q_1,q_2)}} H_{train(t,t+1)},
		\label{ch4:h3s}
\end{eqnarray}
corresponding to a quantum circuit $U$ is a Quantum 3-SAT Hamiltonian, because it is composed of 3-local projectors on qubits. 

The progression of states $\ket{\Psi_t}$ with the occasional intermediate states $\kets{\Psi_{t}^{(u_1,d_1)}}$ and $\kets{\Psi_{t}^{(u_2,d_2)}}$, when a CNOT gate is on order, has the geometry of a line. Each of the states $\ket{\Psi_t}$ is connected by a transition only to two neighboring states. 
Let me call $\cH_{legal}$ the subspace of $\cH$ spanned by the states $\ket{\Psi_t}$. Because I constructed it so, the Hamiltonian $H_{3S}$ \eqref{ch4:h3s} does not induce transitions between the subspace $\cH_{legal}$ and $\cH_{legal}^\perp$. The time evolution of an initial state in $\cH_{legal}$ is thus governed solely by the restriction of $H_{3S}$ to $\cH_{legal}$.
Moreover, in the basis $\ket{\Psi_t}$ this restriction has the form
\begin{eqnarray}
	H_{3S}\Big|_{\cH_{legal}} &=& 
	\left[\begin{array}{rrrrrr}
		\half 	& -\half	& 				&  				&					&					\\
		-\half 	& 1				&	-\half 	&  				&					&					\\
						& -\half	& 1      	& -\half  &					&					\\
		 				& 			  & -\half 	& \ddots	& \ddots  &					\\
		 				& 			  &  				&	\ddots	& 1			 	&	-\half	\\
		 				& 			  &  				&	      	& -\half 	&	\half		
	\end{array}\right],
\end{eqnarray}
same as Kitaev's propagation Hamiltonian (see Section \ref{ch1:Kitaev5} and Appendix \ref{appProp}). It is a Hamiltonian for a quantum walk on a line of length $L' = poly(n)$, where $n$ is the number of qubits the quantum circuit $U$ acts on. Here $L' = 1 + L_1 + 3 L_2$ where $L_1$ is the number of single qubit gates in $U$ and $L_2$ is the number of CNOT's in the circuit. 
\begin{figure}
	\begin{center}
	\includegraphics[width=5.5in]{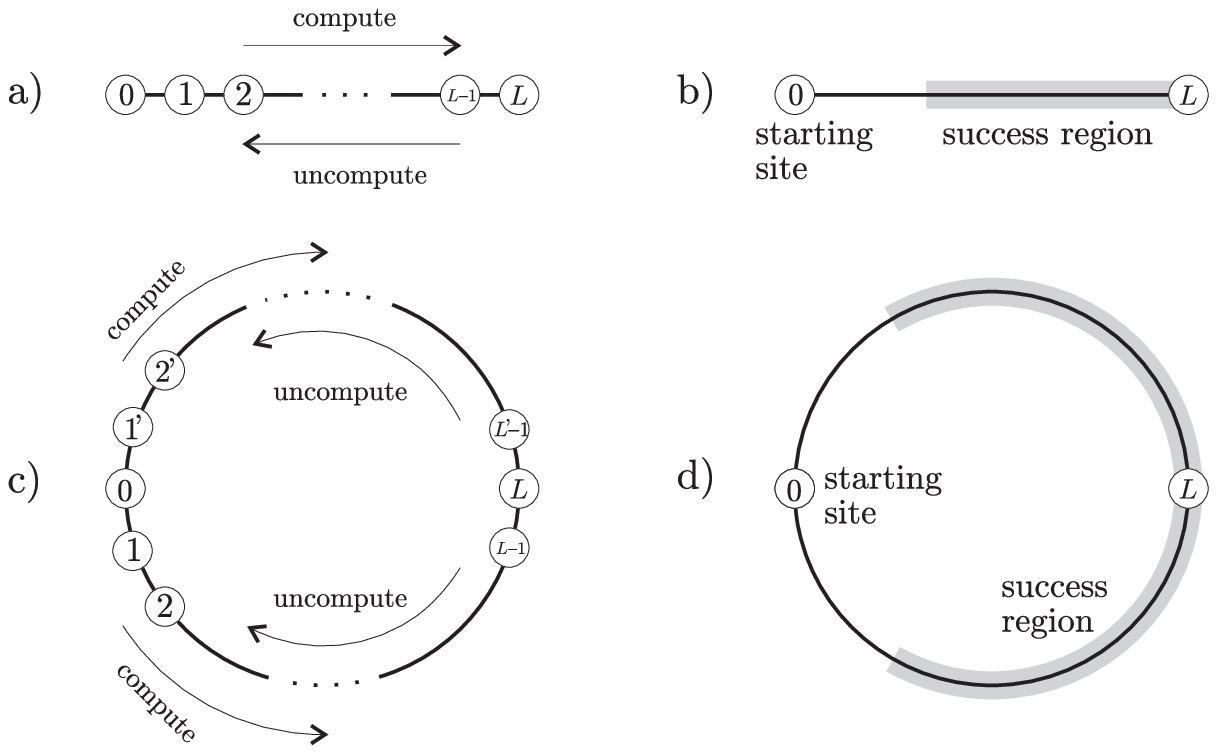} 
	\end{center}
	\caption{a) Computation on a line of states $\ket{\Psi_t}$. b)
	The line corresponding to a circuit $U$ padded with extra identity gates. The region where the computation is done is marked. c) Computation on a cycle of states $\ket{\Psi_t}$ coming from two copies of the computation on a line. d) The cycle corresponding to a padded circuit $U$.
	\label{ch4:figurecircle}}
\end{figure}
To avoid complications in the analysis of the required running time coming from the endpoints, I can change the `line' of states $\ket{\Psi_t}$ into a circle of length $2L'$ as in Figure \ref{ch4:figurecircle}. First, double the number of qubits in the clock register as
\begin{eqnarray}
	\ket{c_0 c_1 \dots c_L} \goes \ket{c_0 c_1 \dots c_L} \otimes \ket{c'_0 c'_1 \dots c'_L},
\end{eqnarray}
with the second `line' $\ket{c'_0 c'_1 \dots c'_L}$ having analogous transition rules. Second, identify the endpoint qubits, i.e. $c_0 \equiv c'_0$ and $c_L \equiv c'_L$. This gives the clock register the geometry of a cycle, with a unique state $\ket{0}_c = \ket{1_{c_{0}} 0 \dots 0}$ denoting time $t=0$. The active site (spin up) in the clock register can proceed towards $c_L$ both ways, as in Figure \ref{ch4:figurecircle}.
This new set of states $\ket{\Psi_t}$ with the geometry of a circle then defines the subspace $\cH_{legal}^\circ$. The Hamiltonian $H_{3S}$ restricted to this subspace is 
\begin{eqnarray}
	H_{3S}\Big|_{\cH_{legal}^\circ} = \half \left(\ii - B_\circ\right),
\end{eqnarray} 
where $B_\circ$ is the adjacency matrix for a cycle of length $2L$. The dynamics of this system is the quantum walk on a cycle, and I analyze it in detail in Appendix \ref{d20proof}. 
There I show that when starting from the state on site A of the line and letting the system evolve for a time chosen uniformly at random between zero and a number not larger than $O(L \log^2 L)$, the probability to find the state farther than $L/3$ from the starting point is close to $\frac{2}{3}$.
This corresponds to finding a state with a spin up on clock qubit $c_{t\geq L/3}$ (lower part of the circle) or $c_{t'\geq L/3}$ (upper part of the circle). In both cases, because the circuit $U$ was padded with identity gates, the state of the work qubits is then $\ket{\psi_L}$ and contains the output of the quantum circuit $U$. Therefore, one can simulate a quantum computation $U$ with $L_1$ single qubit gates and $L_2$ CNOT's with this system as follows:
\begin{enumerate}
\item Construct a clock register with $C=1+L_1 + 4L_2$ qubits and transition rules as described by the {\em train switch} construction.
\item Pad the clock register with extra $2C$ qubits (equivalent to padding the circuit $U$ with identity gates).
\item Double the clock register and make it into a cycle with length $L = 6C$.
\item Initialize the system in the state
\begin{eqnarray}
	\ket{\Psi_0} = \ket{00\dots 0}\otimes \ket{1_{c_0}0\dots 0}.
\end{eqnarray}
\item Let it evolve with $H_{3S}$ \eqref{ch4:h3s} for a time chosen uniformly at random between zero and
$\tau\leq O(L\log^2 L)$. 
\item Measure the clock register qubits in the success area (farther than $L/3$ from the initial site, see also Figure \ref{ch4:figurecircle}). With probability close to $\frac{2}{3}$ you will find the active site (spin up) there. The work register now contains the output of the quantum circuit $U$. Restart otherwise.
\end{enumerate}

With this I have shown that Quantum 3-SAT Hamiltonians are powerful enough to perform universal quantum computation in the Hamiltonian Computer model based on a quantum walk. 
When the number of qubits and gates involved in the circuit is $L$, what I need for my Quantum 3-SAT Hamiltonian computer are thus $O(L)$ qubits, a Hamiltonian with $O(L)$ projector terms with norm $O(1)$, resulting in $\norm{H}=O(L)$, and a running time of order $O(L\log^2 L)$. The rescaled required resources (time $\times$ energy) scale as $\tau \cdot \norm{H} = O(L^2 \log^2 L)$.
The initial part of this Hamiltonian, turning off the initial Hamiltonian pinning the initial state in $\ket{\Psi_0}$ can be turned off and $H_{3S}^{\circ}$ turned on adiabatically or quickly, as a quick change will not destroy the computation as discussed in Section \ref{ch2:hc}. Moreover, Seth Lloyd \cite{AQC:Lloydunpublished} showed that this Hamiltonian Computer model is protected by an energy gap between the legal subspace $\cH_{legal}^{\circ}$ and the subspace orthogonal to it. Noise in the system will induce transitions away from $\cH_{legal}^{\circ}$, but their energy cost is going to be scaling like $O(L^{-1})$. The computation is thus protected by an energy barrier of the order $O(L^{-1})$. 

\chapter{Hamiltonian Quantum Cellular Automata}
\label{ch5hqca}

Can universal quantum computation be performed by time evolving a simple system by a translationally invariant, time-independent Hamiltonian? I present ways how to do it in this chapter, largely based on the paper \cite{CA:d10}, so far posted on the arXiv.

\mypaper{Hamiltonian Quantum Cellular Automata in 1D}{Daniel Nagaj, Pawel Wocjan}{We construct a simple translationally invariant, nearest-neighbor Hamiltonian on a chain of $10$-dimensional qudits that makes it possible to realize universal quantum computing without any external control during the computational process. We only require the ability to prepare an initial computational basis state which encodes both the quantum circuit and its input. The computational process is then carried out by the autonomous Hamiltonian time evolution. 
After a time polynomially long in the size of the quantum circuit has passed, the result of the computation is obtained with high probability by measuring a few qudits in the computational basis.


This result also implies that there cannot exist efficient classical simulation methods for generic translationally invariant nearest-neighbor Hamiltonians on qudit chains, unless quantum computers can be efficiently simulated by classical computers (or, put in complexity theoretic terms, unless BPP=BQP).
}


This chapter is organized as follows. First, in Section \ref{ch5:hqcaintro} I introduce classical cellular automata, quantum cellular automata and the Hamiltonian Quantum Cellular Automaton model. Then, in Section \ref{d10construct} I present a HQCA in 1D with cell size $d=10$. I give another HQCA construction in 1D with cell size $d=20$ in Section \ref{d20construct}. Throughout the required runtime analysis, Appendix \ref{d20proof} on the quantum walk on a line and Appendix
\ref{d10proof} about diffusion of free fermions on a line is referenced.


\section{Introduction}
\label{ch5:hqcaintro}

One of the most important challenges in quantum information science is to identify quantum systems that can be controlled in such a way that they can be used to realize universal quantum computing. The quantum circuit model abstracts from the details of concrete physical systems and states that the required elementary control operations are: (i) initialization in basis states, (ii) implementation of one and two-qubit gates, and (iii) measurement of single qubits in basis states.  Meanwhile, many other models have been proposed such as measurement-based quantum computing \cite{CA:RB-oneway:00,CA:Nielsen-Measurement:03,CA:Leung-Measurement:04,CA:CLN-Measurement:05}, adiabatic quantum computing \cite{AQC:FarhiScience,AQC:AvDKLLR05}, or topological quantum computing \cite{CA:anyons1,CA:anyons2,CA:anyons3,CA:anyons4,CA:anyons5,CA:anyons6,CA:anyons7} that reduce or modify the set of elementary control operations.  However, the common principle underlying all these models is that the computation process is always driven by applying a sequence of control operations.

Instead, I consider a model that does not require any control during the computational process. This model consists of a quantum system with a Hamiltonian that makes it possible to realize universal quantum computing by the following protocol: (1) prepare an initial state in the computational basis that encodes both the program and input, (2) let the Hamiltonian time evolution act undisturbed for a sufficiently long time, and (3) measure a small subsystem in the computational basis to obtain the result of the computation with high probability.  I refer to this model as a {\em Hamiltonian quantum computer} and more specifically as a {\em Hamiltonian quantum cellular automaton} (HQCA) provided that the Hamiltonian acts on qudits that are arranged on some lattice, is invariant with respect to translations along the symmetry axis of the lattice, and contains only finite range interactions. Most natural Hamiltonians have these properties, so it is important to construct HQCA that are as close as possible to natural interactions.

Hamiltonian QCA are related to the more usual discrete-time QCA (for further review of the different types of quantum cellular automata I refer the reader to \cite{CA:PDC-ModelsQCA:05}). However, while the evolution of discrete-time QCA proceeds in discrete update steps (corresponding to tensor products of local unitary operations, see e.g. \cite{CA:Raussendorf-QCA:05,CA:ShepherdPRL:06}), the states of Hamiltonian QCA change in a continuous way according to the Schr\"odinger equation (with a time-independent Hamiltonian).
For this reason, Hamiltonian QCA are also called {\em continuous-time} QCA \cite{CA:PDC-ModelsQCA:05}. 
Also, in the HQCA model, all the couplings (interactions)
are present all the time, while for the the discrete-time QCA,
the execution of updates on overlapping cells is synchronized by external control.
Therefore, the nearest-neighbor interactions of a HQCA have to include a mechanism that ensures that the logical transformations are carried out in the correct order. 

The motivation to consider Hamiltonian computers is threefold.  First, it is a fundamental question in the thermodynamics of computation how to realize computational processes within a closed physical system.  Such Hamiltonian computers were presented and discussed by Benioff \cite{CA:Benioff:80}, Feynman \cite{lh:Feynman2}, and Margolus \cite{CA:Margolus:90}. Second, Hamiltonian quantum cellular automata could lead to new ideas for reducing the set of necessary control operations in current proposals for quantum computing by using the inherent computational power of the interactions.  HQCA are at one end of the spectrum of possible implementations; more realistic perspectives for quantum computing could arise by combining this model with more conventional models involving external control operations throughout the computation. Third, this model can show the limitations of current and future methods in condensed matter physics for simulating the time evolution of translationally invariant systems. If evolving with a certain Hamiltonian can realize universal quantum computing, then there cannot exist any classical method for efficiently simulating the corresponding time evolution unless classical computers are as powerful as quantum computers (BPP=BQP).

The first theoretical computational models based on a single time-independent Hamiltonian go back to \cite{CA:Benioff:80, lh:Feynman2, CA:Margolus:90}. However, these Hamiltonian computers were not explicitly designed for realizing universal {\em quantum} computing. Margolus' model \cite{CA:Margolus:90} has the attractive feature that it is laid out on a
$2$-dimensional lattice with translationally invariant, finite-range interactions. (In \cite{CA:Biafore:94} it was argued that the part of the Hamiltonian responsible for the synchronization in a $1$-dimensional variant is close to real interaction in solid states.) However,  this scheme does not satisfy the requirement (1) since its initial state has to be prepared in a superposition. Building upon Margolus' idea, a translationally invariant Hamiltonian 
universal for quantum computing even if the initial state is restricted to be a canonical basis state was given in \cite{CA:JW:05}. This model requires $10$-local, finite-range interactions among qubits on a $2$-dimensional rectangular lattice wrapped around a cylinder.
Subsequently, it was established in \cite{CA:Janzing:07} that nearest-neighbor interactions among qutrits on a $2$-dimensional lattice suffice. However, the Hamiltonian of \cite{CA:Janzing:07} is translationally invariant only when translated over several lattice sites.
A different approach was taken by Vollbrecht and Cirac in \cite{CA:VC-QCA:07}, showing that one can implement universal quantum computation with a translationally invariant, nearest-neighbor Hamiltonian on a chain of $30$-dimensional qudits. 

I present two different simplified HQCA constructions on one-dimensional qudit chains. In both models, I think of the qudit chain as composed of two registers, data and program. The work qubits I compute on are located at a static location in the data register. Driven by the autonomous Hamiltonian time evolution, the program sequence contained in the program register moves past the work qubits and the gates are applied to them. After I let the system evolve for a time not larger than a polynomial 
in the length of the program,
 I measure one or two qudits in the computational basis to read out the output of the computation with high probability.

My first construction is for a chain of $10$-dimensional qudits and is related to the ideas of \cite{CA:VC-QCA:07}. The mechanism behind the progress of the program sequence in this particular model can be thought of as the diffusion of a system of free fermions on a line.
My second construction uses qudits with dimension $d=20$ and is inspired by \cite{CA:JWZ-measuring:07}, utilizing a technique of \cite{CA:GottesmanLine} to transport the program. Here, the mechanism for the progress of the computation can be thought of as a quantum walk on a line.


\section{The HQCA in 1D with cell size $d=10$}
\label{d10section}

I present a simple universal HQCA on a chain of qudits with dimension $d=10$. First, I encode the progression of a quantum circuit $U$ on $N$ qubits into a set of states $\kets{\varphi_{\sigma}}$ of a chain of qudits with length $L=poly(N)$. Second, I give a translationally invariant nearest-neighbor Hamiltonian on this chain of qudits, which induces a quantum walk on the set of states $\kets{\varphi_{\sigma}}$. Finally, using a mapping to a system of free fermions in 1D, I prove that when I initialize the qudit chain in an easily determined computational basis state and let the system evolve for a time $\tau\leq \tau_{10} = O(L\log L)$ chosen uniformly at random, I can read out the result of the quantum circuit $U$ with probability
$p_{10}\geq \frac{5}{6} - O\left(\frac{1}{\log L}\right)$ by measuring one of the qudits in the computational basis. I then show that this is enough to ensure universality of my HQCA for the class BQP.


\subsection{The Construction}
\label{d10construct}

The gate set \{Toffoli, Hadamard\} is universal for quantum computation \cite{CA:Shi:02}. With only polynomial overhead, one can simulate a circuit consisting of these gates using only the gate $W$ (controlled $\frac{\pi}{2}$ rotation about the $y$-axis) 
\begin{eqnarray}
	W = \left[
		\begin{array}{cccc}
			1 & 0 & 0 & 0\\
			0 & 1 & 0 & 0 \\
			0 & 0 & \frac{1}{\sqrt{2}} & -\frac{1}{\sqrt{2}} \\
			0 & 0 & \frac{1}{\sqrt{2}} & \frac{1}{\sqrt{2}}
		\end{array}
	\right]
\end{eqnarray}
if it can be applied to any pair of qubits.
Let me consider implementing universal quantum computation on a qubit chain using only nearest neighbor gates. Let me also restrict the use of the $W$ gate so that the control qubit has to be to the left of the target qubit. Using only polynomially many additional swap gates $S$, one can still do universal quantum computation on a qubit chain. 
Thus given a quantum circuit $U'$ on $N'$ qubits with $poly(N')$ generic two-qubit gates, I can transform it into a circuit $U$ on a chain of $N=poly(N')$ qubits with nearest neighbor gates $W$ (with control on the left) and $S$ without loss of universality. I then add identity to my gate set and further transform the circuit $U$ to have the following form (see Figure \ref{CircuitFigure}). Rewrite the circuit as $K$ sequences of nearest neighbor gates $U_{k,g} \in \{W,S,I\}$,
where gate $U_{k,g}$ belongs to the $k$-th sequence and acts on the pair of qubits $(g,g+1)$:
\begin{eqnarray}
	U = (U_{K,N-1}\dots U_{K,1}) \cdots (U_{1,N-1}\dots U_{1,1}).
	\label{circuitsequence}
\end{eqnarray}
\begin{figure}
	\begin{center}
	\includegraphics[width=4in]{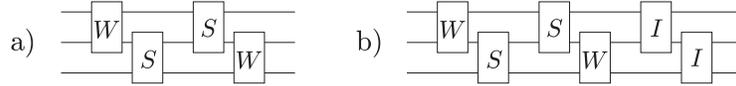}
	\caption{a) A quantum circuit consisting of two sequences of gates acting on nearest neighbors. b) The previous circuit with a third sequence of identity gates added.}
	\label{CircuitFigure}
	\end{center}
\end{figure}

%
%
%
I wish to encode the progression of the circuit $U$ into the states of a chain of qudits with dimension $10$, with length $L=poly(N)$.
The basis states of each qudit $\ket{q} = \ket{p}\otimes \ket{d}$ are 
constructed as a tensor product of a $5$-dimensional program register $p \in\{ \bul,\gat,W,S,I\}$ and a $2$-dimensional data register $d \in \{0,1\}$.
I start by writing the initial product state 
\begin{eqnarray}
	\kets{\varphi} = \bigotimes_{j=1}^{L} \left(\ket{p_j} \otimes \ket{d_j}\right)_j,
	\label{start10product}
\end{eqnarray}
with $p_j$ and $d_j$ as follows
(here I give an example for the circuit in Figure \ref{CircuitFigure}a):
\begin{eqnarray}
	\begin{array}{c|cccccccccccc}
		j  &   1      & & & & \cdots
		               &  M    & \cdots	
		              & & & &  & 2M \\
		\hline
		p_{j}  & \bul &  \bul & \gat & \bul & \bul & \gat &
		\iga & \wga & \sga & 
		\iga & \sga & \wga 
		\\		
		d_{j}  & 0      & 0    & 0    & 0    & 0    & 0    & 
		w_1  & w_2  & w_3  &
		0    & 0    & 0        		
 	\end{array} \label{startstate10}
\end{eqnarray}
The qudit chain has length 
\begin{eqnarray}
	L = 2M = 2KN,
\end{eqnarray}
where $K$ is the number of gate sequences in \eqref{circuitsequence}.
The left half of the top (program) register contains $K$ pointer symbols $\gat$ at positions $kN$ for $k=1\dots K$ and empty symbols $\bul$ everywhere else. The right half holds the program in the form
\begin{eqnarray}
	I \underbrace{U_{1,1} \dots U_{1,N-1}}_{\textrm{gate sequence 1}}
	I \underbrace{U_{2,1} \dots U_{2,N-1}}_{\textrm{gate sequence 2}}
	I \dots
	I \underbrace{U_{K,1} \dots U_{K,N-1}}_{\textrm{last sequence}}\, ,
	\label{program10}
\end{eqnarray} 
with $U_{k,g} \in \{W,S,I\}$ and each sequence preceded by an identity gate. The bottom (data) register contains $N$ work qubits (labeled $w_n$ in the table)
at positions $M+n$ for $n=1\dots N$ and qubits in the state $\ket{0}$ everywhere else. I designate $w_N$ as the readout qubit.

I now give two simple rules, many applications of which generate the set of states $\{\ket{\varphi_{\sigma}}\}$ from the initial state $\kets{\varphi}$. I label the states by $\sigma$, the description of the sequence of rules I choose to apply to the initial state. I invite the reader to work out what happens to $\ket{\varphi}$ \eqref{startstate10} as the rules are applied, noting that there are usually several possible rules that can be applied to a given state.
The first rule says that the symbols $A\in\{W,S,I\}$ (from now on I call them gates and think of them as particles) in the program register can move one step to the left, if there is an empty spot there:
\begin{eqnarray}
	\begin{array}{rccc}
	1\,:& 
		\band{\bul}{\aga}
		&\goes&
		\band{\aga}{\bul}
	\end{array}	\label{rule10x1}
\end{eqnarray}
The second rule concerns what happens when a gate meets a pointer symbol: 
\begin{eqnarray}
	\begin{array}{rccc}
	2\,:& 
		\four{\gat}{\aga}{x}{y}
		&\goes&
		\flour{\aga}{\gat}{A(x,y)}
	\end{array}	\label{rule10x2}
\end{eqnarray}
In this case, the gate moves to the left, the pointer $\gat$ gets pushed to the right, and the gate $A\in\{W,S,I\}$ is applied to the qubits in the data register below.

The initial state has $K$ pointers, one for each sequence of gates in the circuit. I constructed the initial state \eqref{startstate10} in such a way that as a gate $U_{k,g}$ from the $k$-th gate sequence moves to the left, it meets the $k$-th (counting from the right) pointer $\gat$ exactly above the work qubits to which the gate was intended to be applied ($g,g+1$). However, one also needs to consider what happens when a gate meets a pointer for a different sequence. If this happens over a pair of extra qubits in the state $\ket{0}\ket{0}$, the qubits stay unchanged because the gate is either the controlled gate $W$, a swap gate or the identity. The second possibility is that a gate meets a pointer above the boundary of the work qubits (i.e. $\ket{0}_{M}\ket{w_1}_{M+1}$ or $\ket{w_N}_{M+N}\ket{0}_{M+N+1}$). 
I ensured that it is going to be the identity gate by inserting one in front of each gate sequence in \eqref{program10}. This implies that the leftmost (or the rightmost) work qubit again stays unchanged. Therefore, I am sure that the state $\ket{\varphi_{\sigma}}$ in which the gate particles have all moved to the left half of the chain contains the result of the quantum computation $U$ in the state of the work qubits, while the additional qubits in the data register remain in the state $\ket{0}$.

Let me now allow all the rules to be applied backwards as well, opening the possibility of returning back to the initial state $\ket{\varphi}$ (undoing the computation). 
Although the number of possible sequences of rule (and backward rule) applications then becomes infinite, the space of states $\{\ket{\varphi_{\sigma}}\}$ is nevertheless finite-dimensional.
In every state $\ket{\varphi_{\sigma}}$, the $M$ gates $\{W,S,I\}$ in the program register that started at positions $\{M+1, \dots, 2M\}$ occupy some combination $C=\{a_{1}^{(C)},\dots,a_{M}^{(C)}\}$ of the $L$ sites of the chain. Given the rules \eqref{rule10x1}-\eqref{rule10x2}, the order of the gates in the program register cannot change, i.e. $a_{k}^{(C)}<a_{m}^{(C)}$ for $k<m$. 
Because the positions of the pointers 
$\gat$ and the state of the data register are also uniquely determined by $\{a_{1}^{(C)},\dots,a_{M}^{(C)}\}$, I can label the states I constructed by $\ket{\varphi_{C}}$ with 
$C=1\dots \combine{L}{M}$.

I now give the universal translationally invariant Hamiltonian 
as a sum of translationally invariant terms
\begin{eqnarray}
	H_{10} = - \sum_{j=1}^{L-1} \left( R+R^\dagger
		\right)_{(j,j+1)},
		\label{ourH10}
\end{eqnarray}
where $R$ corresponds to the rules \eqref{rule10x1}-\eqref{rule10x2}
and acts on two neighboring qudits as
\begin{eqnarray}
	R = 
		\sum_{A\in\{W,S,I\}}
		\left[
			 \ket{\aga \bul}\bra{\bul \aga}_{p_1, p_2}  
			 \otimes 
			 \ii_{d_1,d_2}
		+
			 \ket{\aga\gat}\bra{\gat\aga}_{p_1,p_2}		
			 \otimes
			 A_{d_1,d_2}
		 \right],
\end{eqnarray}
where $p$ stands for the program register and $d$ for the data register of the respective qudit. 

Recall that my model of computation consists of initializing the qudit chain in the state $\kets{\varphi}$ and 
evolving the system for a time $\tau \leq \tau_{10}$ which I will determine in Section \ref{d10time}. Finally, I want to show that when I measure the output qubit $w_N$ in the data register, I will read out the result of the quantum computation $U$ with high probability.  


\subsection{Required Time Analysis}
\label{d10time}

The time evolved state $\kets{\varphi(\tau)}$ (obtained from
$\kets{\varphi}$ by evolving with $H_{10}$ for time $\tau$)
is a superposition of the states $\ket{\varphi_{C}}$. 
I can write it as 
\begin{eqnarray}
	\kets{\varphi(\tau)} = 
		\sum_{C=1}^{\combine{L}{M}} \varphi_{C}(\tau) 
		\,\kets{\pi_C}_{program} \otimes
		 \kets{\theta_C}_{work} \otimes \kets{\alpha}_{extra},
		 \label{evolved10a}
\end{eqnarray}
where $\kets{\pi_C}$ is the state of the program register of the chain in the state $\kets{\varphi_C}$, the corresponding state of the work qubits is $\kets{\theta_C}$ and $\ket{\alpha}$ is the state (all zero) of the extra data qubits.
The state $\kets{\theta_C}$ of the work qubits holds the output state of the computation $U$, if all of the gate ``particles'' have moved to the left of the work qubits.
Let me now choose some number $f$ 
and pad the qubit chain with $(f-1)M$ empty sites on the left and $M$ sites containing $I$ in the program register on the right:
\begin{eqnarray}
	\kets{\varphi^{pad}} &=& 
		\left[\begin{array}{c}
		\bul\\0
		\end{array}\right]^{\otimes (f-1)M}
		\otimes \kets{\varphi} \otimes 
		\left[\begin{array}{c}
		\iga\\0
		\end{array}\right]^{\otimes M}.
	\label{startstate10padded}
\end{eqnarray}
The original chain had length $2M$, so the length of this padded chain is $L=(f+2)M$. The program register
of the initial state $\kets{\varphi^{pad}} = \kets{\varphi_{C=1}}$ now has $2M$ gate ``particles'' $\{W,S,I\}$ at positions $a_m^{(1)} = fM+m$ with $m=1\dots2M$. 
On this modified chain, every state $\ket{\varphi_{C}}$ 
in which the first $M$ gate ``particles'' are located in the first $fM$ sites of the chain ($a_m^{(C)} \leq fM$ for $m\leq M$) contains the finished computation in the state of its work qubits $\ket{\theta_C}$.
Note that for all these states, the state of the work qubits is the same and equal to $\ket{\theta_U} = U \ket{w_1 \dots w_N}$, as it does not change under the extra identity gates I added.
I can now rewrite the time evolved state \eqref{evolved10a} as
\begin{eqnarray}
	\kets{\varphi(\tau)} 
		&=&	\left( 
			\sum_{a_M^{(C)} > fM} \varphi_{C}(\tau) 
			\,\kets{\varphi_C} 
			\right)
		+ \kets{\pi}_{prog} \otimes
			 \kets{\theta_U}_{work} \otimes 
			 \kets{\alpha}_{extra},
		 \label{evolved10b} 
\end{eqnarray}
where the sum in the first term is over the set of positions of the $2M$ gate particles in which the $M$-th particle is still near the right end of the chain and the computation is thus not finished yet. Meanwhile,
\begin{eqnarray}		
	\kets{\pi}_{prog} = \sum_{a_M^{(C)} \leq fM} \varphi_{C}(\tau)\,
		\kets{\pi_C}_{prog}.
\end{eqnarray}
is a superposition of the program register states which correspond to an executed computation. 

Recall the definition of the quantum complexity class BQP (Definition \ref{defbqp} in Section \ref{ch1:qcomplex}). There is a verifier circuit $U$ associated with each problem. Denote $x$ a problem instance. As the problems in BQP are decision problems, $x\in (L_{yes} \cup L_{no})$, where $L_{yes}$ is the set of instances with the answer `yes'. Now, when $x\in L_{yes}$, the probability of the circuit $U$ outputting `yes' is not smaller than $\frac{2}{3}$. On the other hand, when $x\in L_{no}$, the probability of the circuit outputting `yes' is not greater than $\frac{1}{3}$.
Let me assume the worst case for the circuit $U$, i.e. that the circuit outputs `yes' on a good proof of a `yes' instance with probability $p_U = \frac{2}{3}$.
In the language of spins, the circuit $U$ outputs `yes' when I measure spin up on the output qubit. Therefore, the expected value of measuring $\sigma_{z}$ on output qubit of the circuit $U$ is bounded from below by
\begin{eqnarray}
	\langle \sigma^{(z)}_{w_N} \rangle^{circuit}_{yes} \geq 
	1\times p_U + (-1)\times(1-p_U) = 2p_U-1 = \frac{1}{3}
	\label{yesoutput}
\end{eqnarray}
when $x\in L_{yes}$. Analogously, when $x\in L_{no}$, it is bounded from above by
\begin{eqnarray}
	\langle \sigma^{(z)}_{w_N} \rangle^{circuit}_{no} \leq 
	-2p_U + 1
	= - \frac{1}{3}.
	\label{nooutput}
\end{eqnarray}

To solve BQP problems with my automaton, I need to distinguish the `yes' from the `no' cases, i.e. I need to show that the expectation value of measuring $\sigma^{(z)}$ on the output qubit of my automaton at a random time $\tau\leq \tau_{10}$ is greater than zero in the `yes' case, and smaller than zero in the `no' case. 
First, call $p_{10}$ the probability to find a state where the computation is finished.
Let me consider a $yes$ instance ($x\in L_{yes}$).
Recalling \eqref{evolved10b} and observing that $\ket{\pi}_{prog}$ is orthogonal to the states of the program register in which the computation is not finished, I have
\begin{eqnarray}
	\langle \sigma^{(z)}_{w_N} \rangle_{yes} =
	\bra{\varphi(\tau)} \sigma^{(z)}_{w_N} 	\ket{\varphi(\tau)}
		&=& p_{10} 
			\underbrace{
				\bra{\theta_U} 
				\sigma^{(z)}_{w_N} 
				\ket{\theta_U}_{work} 
			}_{\textrm{output of }U}
			+ (1-p_{10}) 
				\bra{\varphi'}\sigma^{(z)}_{w_N} 
				\ket{\varphi'},
		\label{expect10}
\end{eqnarray}
where $\ket{\varphi'}$ is the normalized first term in \eqref{evolved10b}.
The first term is the circuit output \eqref{yesoutput}, therefore
\begin{eqnarray}
	\bra{\theta_U} 
			\sigma^{(z)}_{w_N} 
			\ket{\theta_U}_{work}
	=
	\langle \sigma^{(z)}_{w_N} \rangle^{circuit}_{yes}
	\geq 
		2p_U - 1.
\end{eqnarray}
The second term can be bounded from below (adversarially, i.e. for every time the computation is not finished, the output qubit gives the opposite of the correct answer) by
\begin{eqnarray}
	\bra{\varphi'}\sigma^{(z)}_{w_N}	\ket{\varphi'} \geq -1. 
\end{eqnarray}
Putting it together, the expectation value is bounded by
\begin{eqnarray}
	\langle \sigma^{(z)}_{w_N} \rangle_{yes} 
		&\geq& p_{10} (2p_U - 1) - (1-p_{10})
		= 2p_{10}p_{U} - 1.
	\label{expect10done}
\end{eqnarray}
Analogously, for the $x\notin L$ case, I obtain
\begin{eqnarray}
	\langle \sigma^{(z)}_{w_N} \rangle_{no} 
		&\geq& - 2p_{10}p_{U} + 1.
	\label{expect10doneNO}
\end{eqnarray}


I will now prove that when I choose the time $\tau$ uniformly at random in $(0,\tau_{10})$, with $\tau_{10} = poly(M)$, the probability of finding a state with the computation executed (with $a_M\leq fM$) is $p_{10} \geq \frac{5}{6} -  O\left(\frac{L}{\tau_{10}}\right)$
with $L=(f+2)M$.

Let me analyze the time evolution of $\kets{\varphi^{pad}}$ \eqref{startstate10padded} under $H_{10}$ \eqref{ourH10}. I can restrict the analysis to the program register of the chain, as the content of the data register in the time-evolved state $\ket{\varphi(\tau)}$ is completely determined by the content of the program register.
The data register does not hinder the time evolution of the program register in any way. In fact, there exist bases, in which $H_{10}$ is identical (as a matrix) to $H_{10}$ restricted to the program register. Moreover, let me consider a further mapping of the system restricted to the program register to a line of qubits with length $L=(f+2)M$ as follows. Map the states $\{\gat,\bul\}$ to the state $\ket{0}$, and the states $\{W,S,I\}$ to the state $\ket{1}$.
The mapping of $H_{10}$ to this system is a sum of hopping terms 
\begin{eqnarray}
	H_q = - \sum_{j=1}^{L-1} (\ket{10}\bra{01}+\ket{01}\bra{10})_{j,j+1}.
	\label{spinham}
\end{eqnarray}
Using the Wigner-Jordan transformation I can define the
operators
\begin{eqnarray}
	b^{\dagger}_j &=& \sigma^{z}_1 \dots \sigma^{z}_{j-1} \otimes 
		\ket{1}\bra{0}_j \otimes \ii_{j+1, \dots, L}, 
			\label{wigner1} \\
	b_j &=& \sigma^{z}_1 \dots \sigma^{z}_{j-1} \otimes 
		\ket{0}\bra{1}_j \otimes \ii_{j+1, \dots, L}. 
			\label{wigner2}
\end{eqnarray}
As $b^{\dagger}_j$ and $b_j$ have the required properties $\{b_i,b^{\dagger}_j\} = \delta_{ij} \ii$ and $b_j^2 = b_j^{\dagger 2} = 0$,
they can be viewed as the creation and annihilation operators for a fermion at site $j$. Rewriting \eqref{spinham} in terms of \eqref{wigner1}-\eqref{wigner2}, I obtain 
\begin{eqnarray}
	H_f = - \sum_{j=1}^{L-1} b^{\dagger}_j b_{j+1} + h.c.,
	\label{Hf}
\end{eqnarray}
a Hamiltonian for a system of free fermions in second quantization.
Following the mapping, the initial state $\kets{\varphi^{pad}}$ of the qudit chain thus corresponds to the state of the fermionic system $\ket{\Psi} = b^{\dagger}_{fM+1} \dots b^{\dagger}_{(f+2)M} \ket{0}$ with $2M$ fermions on the right end of the line (here $\ket{0}$ is the state with no fermions).
I now use the following Lemma proved in Appendix \ref{d10proof}:

\begin{lemma}
\label{difflemma}
	Consider the state $\ket{\Psi}$ of $2M$ fermions on the right end of a line with $L=(f+2)M$ sites. Let the system evolve for a time chosen uniformly at random between $0$ and $\tau_{10}$ with the Hamiltonian given by \eqref{Hf} and measure the number of fermions in the region $1\leq x \leq fM$. The probability to measure a number greater than $M$
is $p_{10} \geq \frac{f-2}{f+2} - O\left(\frac{L}{\tau_{10}}\right)$.
\end{lemma}

Let me choose $f=22$ and $\tau_{10} = O(L\log L) = O(M\log M)$. Following the mapping I did from my qudit chain backwards, this implies that when I initialize the qudit chain of length $L=24M$ in $\kets{\varphi^{pad}}$ as in \eqref{startstate10padded} and let it evolve with $H_{10}$ \eqref{ourH10} for a random time 
$\tau\leq \tau_{10}$,
the probability for the chain to be in a state where the gate particles have moved sufficiently to the left for the computation to be done ($a_M \leq fM$) is 
\begin{eqnarray}
	p_{10} > \frac{5}{6} - O\left(\frac{1}{\log M}\right).
\end{eqnarray}  
Therefore, equations \eqref{expect10done} and \eqref{expect10doneNO} now read
\begin{eqnarray}
	\langle \sigma^{(z)}_{w_N} \rangle_{yes} 
		&\geq& \phantom{-}\frac{1}{9} - O\left(\frac{1}{\log M}\right), \nonumber
		\\
	\langle \sigma^{(z)}_{w_N} \rangle_{no} 
		&\leq& -\frac{1}{9} + O\left(\frac{1}{\log M}\right).
	\label{sigmayesno}
\end{eqnarray}
Therefore, one can recognize any language in BQP using the HQCA I described above.

As an aside, note that there is a way to determine that I obtained a state in which the computation has been done with certainty (and thus getting rid of the second term in \eqref{expect10}). I could have chosen to measure all the program qudits to the right of the first work qubit and check whether all the $S$ and $W$ are gone. This happens with the above probability $p_{10}$, and the postselected state of the work qubits now surely contains the output of the circuit $U$. Note also that I can think of the state of all the work qubits as the circuit output, as compared to only the last work qubit. Nevertheless, thinking only about the last work qubit is enough to ensure universality of my HQCA for the class BQP.


\section{The HQCA in 1D with cell size $d=20$}
\label{d20section}

I now present the second construction, a HQCA for a chain of $20$-dimensional qudits.
As in Section \ref{d10construct}, I describe an encoding of the progression of a quantum circuit $U$ into a set of states of a qudit chain. However, the geometry of this set of states $\ket{\psi_t}$ will be now much simpler, as I can label them by a ``time'' label $t=1\dots L=poly(N)$, thinking of the set of states as a
``line''.
The Hamiltonian $H_{20}$ I construct induces a quantum walk on this ``line'' of states. I conclude by proving that when I let the initial state $\ket{\psi_0}$ evolve with $H_{20}$ for a time $\tau$
chosen uniformly at random between $0$ and $\tau_{20} = O(L\log L)$, I can read out the result of the quantum computation $U$ with probability $p_{20}\geq \frac{5}{6}-O\left(\frac{L}{\tau_{20}}\right)$ by measuring two of the qudits in the computational basis.

\subsection{The Construction}
\label{d20construct}

I encode the progression of a quantum circuit $U$ in the form \eqref{circuitsequence} (see also Figure \ref{CircuitFigure}) into a set of states $\ket{\psi_t}$ of a qudit chain with length $L = (2K-1)(N+1) + 2$. As in Section \ref{d10construct}, each qudit consists of a program register and a data register. 
The data register is again two-dimensional, but the program register can now be in the following $10$ states:
\begin{eqnarray}
	\wga,\sga,\iga &:& \textrm{the program sequence,} 
		\nonumber\\
	\wci,\sci,\ici &:& \textrm{marked characters in the program sequence, used to propagate} \nonumber\\
			&& \textrm{the active spot to the front (left) of the program sequence,} \nonumber\\
	\gat &:& \textrm{apply gate symbol,} \nonumber\\
	\mov &:& \textrm{shift program forward,} \nonumber \\
	\tur &:& \textrm{a turn-around symbol,} \nonumber \\
	\bul &:& \textrm{empty spot (before/after the program)}. \nonumber
\end{eqnarray}

Similarly to \eqref{start10product} and \eqref{startstate10},
the initial product state $\ket{\psi_0} = \bigotimes_{j=1}^{L} \left(\ket{p_j} \otimes \ket{d_j}\right)_j$
is given by 
(the following is an example for the circuit in Figure \ref{CircuitFigure}a)
\begin{eqnarray}
	\begin{array}{c|cccccccccccccc}
		     j & 1 & \cdots     &      &      &      &      & 
		&&&&&& \cdots & L\\
		\hline
		p_j & \bul & \bul & \bul & \bul & \bul & \bul &
		\iga & \wga & \sga & 
		\iga & \iga & \sga & \wga & \tur 
		\\		
		d_j & 0    & 1    & 0    & 0    & 0    & 1    & 
		w_1  & w_2  & w_3  &
		1    & 0    & 0    & 0    & 1		
 	\end{array} \label{startstate}
\end{eqnarray}
In general, the data register contains $N$ work qubits (labeled $w_n$ in my example) at positions $(K-1)(N+1)+2+n$ for $n=1:N$ (counting from the left). Qubit $w_N$ is the designed output qubit for the computation, i.e. once the computation is done, $w_N$ contains the output of $U$. Next, the data register contains qubits in the state $\ket{1}$ at positions $(k-1)(N+1) + 2$ for $k=1\dots 2K$ and qubits in the state $\ket{0}$ everywhere else. The 1's serve as sequence boundary markers.
The program register has 
empty symbols $\bul$ on the left, and then it contains the program in the form 
\begin{eqnarray}
	I\, 
	\underbrace{U_{1,1} \dots U_{1,N-1}}_{1^{\textrm{st}}\textrm{ gate sequence}} 
	I\, I\, 
	\underbrace{U_{2,1} \dots U_{2,N-1}}_{2^{\textrm{nd}}\textrm{ gate sequence}} 
	I\, I\, 
	\:\cdots\:
	I\, I\, 
	\underbrace{U_{K,1} \dots U_{K,N-1}}_{\textrm{last gate sequence}},
	\label{program}
\end{eqnarray}
with the program written from left to right.
In the example given in \eqref{startstate}, the first gate sequence (see Figure \ref{CircuitFigure}a) is $WS$ and the second gate sequence is $SW$. Finally, the last qudit in the program register is in the state $\tur$, marking an active spot in the computation. 

I now give the rules to obtain the sequence of states $\ket{\psi_t}$ from $\ket{\psi_0}$. These rules are constructed so that there is always only one of them that can be applied to a given state $\ket{\psi_t}$, thus giving me a unique state $\ket{\psi_{t+1}}$.  
(Also, using the rules backwards, one obtains a unique $\ket{\psi_{t-1}}$ from $\ket{\psi_t}$). The first three are
\begin{eqnarray}
	\begin{array}{rccc}
	\raiseonebox 
	1\,:& 
		\band{\aga}{\tur} 
		&\goes&
		\band{\aci}{\bul} 
		\\
	\raiseonebox 
	2\,:& 
		\band{\aga}{\bci}
		&\goes&
		\band{\aci}{\bga}
		\\
	3\,:& 
		\band{\bul}{\aci}
		&\goes&
		\band{\tur}{\aga} 
	\end{array}
	\label{rule123}
\end{eqnarray}
where $A,B$ stands for either $W,S$ or $I$. 
These rules ensure the passing of the active spot from the back end (right side) of the program to the front (left side), without modifying the data register or the order of the gates in the program sequence. Next, 
\begin{eqnarray}
	\begin{array}{rccccrccc}
	4a\,:&
		\triUR{\bul}{\tur}{1} 
		&\goes&
		\triUR{\bul}{\gat}{1}
	&\quad&
	4b\,:& 
		\triUR{\bul}{\tur}{0}
		&\goes&
		\triUR{\bul}{\mov}{0}
	\end{array} \label{rule4}
\end{eqnarray}
After the active spot has moved to the front of the program, there are two possibilities. The turn symbol $\tur$ can change to the apply gate symbol $\gat$ (rule 4a), or to the shift program symbol $\mov$ (rule 4b), depending on whether the data qubit below contains the sequence boundary marker state $1$. Afterwards, for the states containing the apply gate symbol $\gat$, I have:
\begin{eqnarray}
	\begin{array}{rccccrccc}
	5a\,:& 
		\four{\gat}{\aga}{x}{y}
		&\goes&
		\flour{\aga}{\gat}{A(x,y)}
	&\quad&
	6a\,:& 
		\triUL{\gat}{\bul}{1}
		&\goes&
		\triUL{\tur}{\bul}{1}
	\end{array}	\label{rule5}
\end{eqnarray}
When applying rule 5a, the apply gate symbol $\gat$ moves 
to the right, while a gate from the program sequence is applied to the qubits in the data register below. 
Applying the rule repeatedly, the $\gat$ symbol moves to the right end of the program sequence. 
As an example, I now write out the state $\ket{\psi_{12}}$ that I obtained from the state $\ket{\psi_0}$ applying rules 1, 2 (6 times), 3, 4a and 5a (3 times) from the state $\ket{\psi_0}$. 
\begin{eqnarray}
	\ket{\psi_{12}} &=& \quad \Big[ \begin{array}{cccccccccccccc}
		\bul & \bul & \bul & \bul & \bul & \iga & 
		\wga & \sga & \gat & 
		\iga & \iga & \sga & \wga & \bul 
		\\		
		0    & 1    & 0    & 0    & 0    & 1    & 
		\multicolumn{3}{c}{\ket{\,\,\dots\theta\dots\,\,}}
		 &
		1    & 0    & 0    & 0    & 1		
 	\end{array} \Big], \label{psi12}
\end{eqnarray}
where $\ket{\dots\theta\dots}$ stands for the state of the three work qubits after the gates $W_{12}$ and then $S_{23}$ were applied to them. Let me now take a closer look at the marker qubits (all qubits in the data register except for the work qubits $w_n$) and the application of rule 5a. The marker qubits stay unchanged for all $\ket{\psi_t}$. The gate applied to pairs $\ket{0}\ket{1}$ and $\ket{1}\ket{0}$ of marker qubits or the pairs of qubits $\ket{1}\ket{q_1}$ and $\ket{q_N}\ket{1}$ (the left and right ends of the work qubit sequence) is always $I$, because of the identity gates I inserted between sequences of gates in the program \eqref{program}.
Finally, the qubit pairs $\ket{0}\ket{0}$ between the 1 markers do not change under the swap operation or the $W$ gate (a controlled gate). 

After the apply gate $\gat$ symbol gets to the end of the sequence, it changes into the turn symbol $\tur$ via rule 6a. Note that the boundary markers in the data register are spaced in such a way, that the $\gat$ symbol will arrive at the right end of the sequence when the qubit below is in the state 1. 
Note that at the very end of the line, rule 6a needs to be modified to involve only the two particles directly above each other.
Using rule 6a, $\gat$ will then change into the turn symbol $\tur$. 
 After applying rules 1, 2 (6 times) and 3, the active spot again moves to the left of the program. Because the $\tur$ symbol is now above a 0 marker qubit, rule 4b can be used, and I get a state with the shift program symbol $\mov$. Finally, here are the last two rules:
\begin{eqnarray}
	\begin{array}{rccccrccc}
	5b\,:& 
		\band{\mov}{\aga}
		&\goes&
		\band{\aga}{\mov}
	&\quad&
	6b\,:& 
		\triUL{\mov}{\bul}{0}
		&\goes&
		\triUL{\tur}{\bul}{0}
	\end{array}	\label{rule6}
\end{eqnarray}
where again $A$ stands for either $\wga,\sga$ or $\iga$. 
Rule 5b makes the program shift to the left while the $\mov$ symbol moves to the right. Finally, rule 6b deals with what happens when the $\mov$ symbol arrives at the end of the program sequence. Because of the way I constructed the data register in $\ket{\psi_0}$, the data
qubit below the $\mov$ symbol will then be in the state $0$, so that the $\mov$ symbol changes to the turn symbol $\tur$. The reason why I need to look at the qubit in the data register below the $\tur$ symbol in rules 6a and 6b is that when I apply the rules backwards (making $\ket{\psi_{t-1}}$ from $\ket{\psi_t}$), again only one of them applies for each $\ket{\psi_t}$.

After applying rule 1, 2a (6 times), 3 and 4b, the $\mov$ symbol appears again and starts shifting the program further to the left. After several rounds of this, when the program shifts to the left by $N+1$, rule 4a can be used again (as the $\tur$ symbol will be above a 1 marker qubit), and subsequently, the $\gat$ symbol facilitates the application of the second sequence of gates to the work qubits.

After many applications of the above rules, I arrive at the state $\ket{\psi_L}$, for which none of the (forward) rules apply. 
\begin{eqnarray}
	\ket{\psi_L} &=& \quad \Big[ \begin{array}{cccccccccccccc}
		\tur & \iga & \wga & \sga & \iga & \iga & 
		\sga & \wga & \bul & 
		\bul & \bul & \bul & \bul & \bul
		\\		
		0    & 1    & 0    & 0    & 0    & 1    & 
		\multicolumn{3}{c}{\ket{\,\,\dots\theta'\dots\,\,}}  &
		1    & 0    & 0    & 0    & 1		
 	\end{array} \Big]. \label{endstate}
\end{eqnarray}
This is the state in which the program has moved to the left of the qudit chain, and all sequences of gates have been applied to the qubits in the data register.
The state $\ket{\dots\theta'\dots}$ is thus the output state of the circuit $U$ and the last of the work qubits ($w_N$) holds the output of the quantum computation.

Starting from \eqref{startstate}, I have constructed the set of states $\ket{\psi_t}$ for $t=0\dots L$ with $L 
= O(K^2 N^2) = poly(N)$. As $t$ grows, these states encode the progress of a quantum circuit $U$. 
What is the geometry of this set of states? They are labeled by a discrete label $t$, with the state $\ket{\psi_{t}}$ 
obtainable only from the states $\ket{\psi_{t-1}}$ and $\ket{\psi_{t+1}}$ using the above rules and their backward applications.
Therefore, the states $\ket{\psi_t}$ can be thought of as position basis states on a line of length $L+1$
\begin{eqnarray}
	\ket{\psi_{t}} \quad \leftrightarrow \quad \ket{t}_{line},
	\label{d20map}
\end{eqnarray}
where $t=0\dots L$.

Let me choose a Hamiltonian $H_{20}$ for this system 
as a sum of translationally invariant terms:
\begin{eqnarray}
	H_{20} = - \sum_{i=1}^{L-1} 
		\sum_{k=1}^{6b} \left( P_{k} + P_{k}^{\dagger} \right)_{(i,i+1)}
		\label{ourH}
\end{eqnarray}
where the terms $P_k$ correspond to the rules 1-6b \eqref{rule123},\eqref{rule4},\eqref{rule5} and \eqref{rule6} and act on two neighboring qudits as 
\begin{eqnarray}
	P_1 &=& 
		\sum_{A\in\{W,S,I\}}			
			\ket{\aci\bul}\bra{\aga\tur}_{p_1,p_2}
		 \otimes 
		\ii_{d_1,d_2}	
			, \\
	P_2 &=& 
		 \sum_{A,B\in\{W,S,I\}} 
		\ket{\aci\bga}\bra{\aga\bci}_{p_1,p_2}
		\otimes
		\ii_{d_1,d_2} 
		,	 \\
	P_3 &=& 
		 \sum_{A\in\{W,S,I\}}
			\ket{\tur\aga}\bra{\bul\aci}_{p_1,p_2}
		\otimes	\ii_{d_1,d_2}, 
\end{eqnarray}
and
\begin{eqnarray}
	P_{4a} &=& 
			\ket{\bul\gat}\bra{\bul\tur}_{p_1,p_2}
		  \otimes 
			\ii_{d_1}
		  \otimes 
			\ket{1}\bra{1}_{d_2}, \\
	P_{4b} &=& 
		\ket{\bul\mov}\bra{\bul\tur}_{p_1,p_2}
			\otimes
			\ii_{d_1}\otimes \ket{0}\bra{0}_{d_2}, \\
	P_{5a} &=& 
		\sum_{A\in\{W,S,I\}}  
			\ket{\aga\gat}\bra{\gat\aga}_{p_1,p_2}
			\otimes
			A_{d_1,d_2}, \\
	P_{5b} &=& 
		\sum_{A\in\{W,S,I\}} \ket{\aga\mov}\bra{\mov\aga}_{p_1,p_2}
			\otimes
			\ii_{d_1,d_2}, \\
	P_{6a} &=& 
		\ket{\tur\bul}\bra{\gat\bul}_{p_1,p_2}
			\otimes
			\ket{1}\bra{1}_{d_1} \otimes \ii_{d_2}, \\
	P_{6b} &=& 
			\ket{\tur\bul}\bra{\mov\bul}_{p_1,p_2}
			\otimes
			\ket{0}\bra{0}_{d_1} \otimes \ii_{d_2}.
\end{eqnarray}
When thinking of the set of states $\ket{\psi_t}$ as the set of positions of a particle on a line \eqref{d20map}, $H_{20}$  becomes 
\begin{eqnarray}
	H_{line} = - \sum_{t=0}^{L-1} \big(
		\ket{t}\bra{t+1} + \ket{t+1}\bra{t} 
			\big).
	\label{Hwalk}
\end{eqnarray}
This is the Hamiltonian of a (continuous-time) quantum walk on a line of length $L+1$. Therefore, $H_{20}$ induces a quantum walk on the ``line'' of states $\ket{\psi_t}$.


\subsection{Required Evolution Time Analysis}
\label{d20time}

After initializing the qudit chain in the state $\ket{\psi_0}$ and evolving with $H_{20}$ for time $\tau$, the final step 
is to read out the output of the computation. 
As in Section \ref{d10time}, I need to ensure that the probability of finding the chain of qudits in a state where the computation was performed completely is high. To raise this probability, I choose to pad the program ($K$ sequences of gates) with another $5K$ sequences of identity gates and redo the construction in the previous section. The length of the qudit chain thus becomes $L = (2(6K)-1)(N+1)+2$. The states $\kets{\psi_{t>L/6}}$ (with $L$ modified) now all contain the result of the quantum circuit $U$ in the readout qubit $w_N$, as the relevant gates have been applied to the work qubits in those states. Note that as the extra identity gates pass by, the state of the work qubits does not change. 

The readout procedure consists of two steps. First, measure the qudit $p_{L-K(N+1)}$ in the program register (the qudit with distance from the right end of the chain equal to the length of the original program).
Let me call $p_{20}$ the probability to measure $\bul$ (which would mean the program has moved to the left of the qudit I just measured). When this happens, I am assured that the work register is in a state in which the computation is done. Second, I measure $w_N$, the last of the work qubits, and read out the result of the computation $U$.
I now prove that when I choose to measure $p_{L-K(N+1)}$
at a random time $0\leq \tau\leq \tau_{20}$ with $\tau_{20} = poly(N)$, the probability $p_{20}$ of obtaining the state $\bul$ is close to $\frac{5}{6}$.

To simplify the notation, let me label the states $\ket{\psi_t}$ as $\ket{t}$. In this basis, the Hamiltonian \eqref{ourH} is the negative of the adjacency matrix of a line graph with $L+1$ nodes. 
I now use the following lemma about the quantum walk on a line proved in Appendix \ref{d20proof}:
\begin{lemma}
\label{linelemma}
Consider a continuous time quantum walk on a line of length $L+1$, where the Hamiltonian is the negative of the adjacency matrix for the line. Let the system evolve for a time $\tau$ chosen uniformly at random between $0$ and $\tau_{20}$, starting in a position basis state $\ket{c}$.
The probability to measure a state $\ket{t}$ with $t>L/6$ is then $p_{20}\geq \frac{5}{6}-O\left(\frac{L+1}{\tau_{20}}\right)$. 
\end{lemma}

This implies that when I initialize the qudit chain in the state $\ket{\psi_0}$ (corresponding to the leftmost state on the line $\ket{c}=\ket{1}$) and let it evolve with $H$ for a random time $\tau\leq \tau_{20}$ with $\tau_{20} = O(L\log L)$, the probability to find a state with $t>L/6$ is close to $\frac{5}{6}$. Therefore, when I measure the program qudit $p_{L-K(N+1)}$, I will obtain $\bul$ with probability close to $\frac{5}{6}$. 
Finally, when I subsequently measure the work qubit $w_N$, I will obtain the result of the quantum circuit $U$. 

Note that I can also avoid this postselection procedure and simply measure the output qubit. The analysis of the outcome would then follow what I did above in Section \ref{d10time}, resulting in \eqref{sigmayesno} again, with $M$ replaced by $L$.

\chapter{Conclusions}
\label{conclusions}

In this thesis I presented three ways to look at local Hamiltonians in quantum computation. Some of these Hamiltonians are useful for building a quantum computer, some of them can be simulated classically, while the properties of others are hard to determine even on a quantum computer.

First, in Chapter \ref{ch2adiabatic} I showed that local structure is a necessity for Adiabatic Quantum Computing to work, and also that a simple 3-local Hamiltonian can be used in the AQC model to simulate any quantum circuit much more effectively than had been known. Second, in Chapter \ref{ch3mps} I presented a numerical method to find the ground state of a translationally invariant Hamiltonian with an infinite tree geometry of interactions, investigated two particular AQC Hamiltonian models, and found encouraging results about the phase transitions in these systems. Third, in Chapter \ref{chQsat} I obtained several new results about the complexity of Quantum $k$-SAT and Local Hamiltonian, implying that finding the ground state properties of rather simple local Hamiltonians is unlikely to be possible even on a quantum computer. Finally, returning to the idea of using local Hamiltonians for building a quantum computer, in Chapter \ref{ch5hqca} I presented two Hamiltonian Quantum Cellular Automata, showing a way to perform quantum computation in continuous time with a translationally invariant, time-independent Hamiltonian on a line, reminiscent of universal Turing Machines.

\subsubsection{Adiabatic Quantum Computing 
\\-- how to use it and how to make it fail}
My results about how not to design quantum adiabatic algorithms are a response  to researchers who made a bad choice of the initial Hamiltonian and then showed that the algorithm doesn't work. I showed that throwing away local structure in the Hamiltonian maps the AQC algorithm into running the unstructured search problem backwards. The running time of any such attempt must then take time of order at least $\sqrt{2^n}$, dashing the hopes of obtaining an exponential speedup. On the other hand, it has not been ruled out that a properly chosen initial Hamiltonian could lead to algorithmic success. A recent numerical study of the NP-complete {\em Exact Cover} problem up to 128 bits shows promising, inverse-polynomial scaling of the eigenvalue gap governing the required runtime of the AQC algorithm \cite{AQC:SmelyanskiyEC}.
Besides using AQC for solving locally constrained optimization problems, it remains an interesting alternative to the conventional quantum circuit model, because it is possible to simulate quantum circuits using AQC Hamiltonians. Moreover, I showed that the previously estimated required running time $\tau \propto L^{14}$ (where $L$ is the number of gates in the circuit and the norm of the Hamiltonian is also proportional to $L$) for such simulations 
was grossly overestimated. When using an AQC Hamiltonian made from 3-local projectors I constructed in Section \ref{ch4:train}, the running time needs to scale only as $\tau \cdot \norm{H} \propto L^2 \log^2 L$. In fact, I believe the $\log$ factors in the required running time can be omitted, as I required them to prove strong convergence of the underlying quantum walk (see Appendix \ref{d20proof}), while showing weak convergence would be enough\footnote{Here weak convergence means that the time averaged probability distribution $p_{\tau}(x|0)$, averaged between $\tau=0$ and $\tau=\tau_0$, converges to the limiting distribution $\pi(x|0)$ as $\left|\sum_{x\in SR} \left(\bar{p}_{\tau_0}(x|0) - \pi(x|0)\right)\right|\leq \ep$, where $SR$ is a region covering two thirds of the line. On the other hand, strong convergence means that $\bar{p}_\tau(x|0)$ converges to $\pi(x|0)$ at every point $x$, i.e. 
$\sum_{x\in A} \left|\bar{p}_{\tau_0}(x|0) - \pi(x|0)\right|\leq \ep$.}.
Moreover, such computation is inherently protected against decoherence by an energy gap of the order $L^{-1}$. On the other hand, we do not yet know how to make the AQC model fault tolerant \cite{AQC:JorThesis}.
It will be interesting to see whether the first useful quantum computer will be based on the quantum circuit model, use measurement-based computation \cite{CA:CLN-Measurement:05}, utilize topological quantum computing \cite{CA:anyons1}, or whether it will be based on  ``analog'' adiabatic evolution \cite{AQC:FGGS00}. 

\subsubsection{Matrix Product State ansatz based numerics \\
-- investigating spins on an infinite tree}
It is hard classically to find the spectral properties of
(and to simulate the time evolution with) AQC Hamiltonians for many qubits. This is why I chose to investigate two simple translationally invariant systems. The method I used is based on a local approximate description of a quantum state, the Matrix Product State ansatz. I improved the imaginary time evolution method for finding the ground state energy of a Hamiltonian within the MPS ansatz, and showed how to apply it for a system with an infinite tree geometry of interactions. I then investigated the phase transitions in the transverse field Ising model and the Not 00 model on the Bethe lattice. The phase transitions I found are not first-order. This may have positive implications for the minimum gap scaling of the corresponding AQC Hamiltonians on finite trees. While the two models I considered do not encode computationally interesting optimization problems, they are an important stepping stone for further research. Showing that the systems with infinite tree geometry have similar properties as large finite trees, adding randomness and loops to the underlying graph structure, and finding the properties of more interesting Hamiltonians are now open for research.

Matrix Product State methods were thought to work especially well in 1D. However, after Aharonov et. al. \cite{CA:GottesmanLine} showed that Quantum 2-SAT on a line for particles with dimension $d=12$ is QMA$_1$ complete, it became clear that MPS methods are not a panacea for 1D systems, especially for spins with higher dimensions. Moreover, recently Schuch et. al. showed that the ground state of a particular frustration free Hamiltonian in 1D can encode solutions to NP-complete problems, even though the Hamiltonian has a polynomial eigenvalue gap and the ground state is exactly described by a low-dimension MPS \cite{CA:CiracLine}. Classical MPS-based methods for finding this ground state are thus bound to fail. How is this possible? The imaginary time evolution method I proposed is not complete. Its weakness must lie in the method for bringing the state back to the MPS ansatz, implying that the procedure can get stuck in local minima, although the required dimensions of the MPS description of the state are small. However, this didn't happen for the two Hamiltonians I investigated. DMRG (and MPS) based methods remain a very useful tool for spin-$\half$ systems in 1D and on trees. Nevertheless, new methods such as PEPS are necessary in higher dimensional geometry. At this front, exciting new results have been recently found by Verstraete at al. \cite{MPS:PEPS1,MPS:PEPS2}.

\subsubsection{Quantum Satisfiability \\
-- an interesting analogue of a classical problem}
The complexity of the Local Hamiltonian problem was sorted out by the paper by Oliveira and Terhal \cite{lh:Terhal2D}, who showed that 2-local Hamiltonian on a 2D grid is QMA complete. As $k$-local Hamiltonian is the quantum analogue of MAX-$k$-SAT, it is interesting to look at Quantum $k$-SAT problem, constructed from a sum of $k$-local projectors, which is the analogue of classical Satisfiability. The solutions to Q-$k$-SAT are states which are annihilated by all of the terms in the Hamiltonian. This problem seems easier than $k$-local Hamiltonian, but already for $k=3$ it contains the NP-complete problem 3-SAT. Is Quantum 3-SAT then QMA$_1$ complete? This big question remains unanswered here. As showed by Cook and Levin, the answer to a classical 3-SAT instance can encode the evaluation of a classical circuit, because classical information can be freely copied. However, quantum information cannot be cloned, forcing us to store the information about a computation in orthogonal states such as $\ket{\psi_t}\otimes \ket{t}_c$ by adding a clock register to the workspace. So far, no one has been able to show how to encode the clock register in such a way that its transitions would be checkable by 3-local projectors while being coupled to the application of a quantum circuit. So far, my attempt using projectors on one qutrit and two qubits, the QMA$_1$-complete problem (3,2,2)-SAT comes closest to this goal. Wholly new ideas, such as Eldar and Regev's triangle construction, which I successfully built on in Section \ref{ch4:train}, are required to answer the question of complexity of Q-3-SAT. There I have shown that a Q-3-SAT Hamiltonian can be used as a Hamiltonian computer to perform any quantum computation. However, the question of the complexity of the original Q-3-SAT problem (as a question about a Hamiltonian: is there a zero energy ground state of $H$) remains open. It might even be possible that a classical verifier circuit exists for it.

In the course of investigating Q-$3$-SAT, I obtained another result greatly strengthening previous results about the Local Hamiltonian problem. The 3-local Hamiltonian was proved to be QMA complete in \cite{lh:KR03}. When the Hamiltonian in \cite{lh:KR03} is rescaled to have norm $O(1)$, the separation in the question about its ground state energy (whether $E_0$ is lower than $a$ or greater than $b$) scales like $b-a \propto L^{-10}$. On the other hand, when using my new 3-local Hamiltonian construction of Section \ref{ch4:q3}, I showed that even for $b-a \propto L^{-4}$ (and $\norm{H}=O(1)$), the 3-local Hamiltonian problem remains QMA-complete. If one could show this for $b-a \propto 1$, it would result in the quantum analogue of the PCP theorem\footnote{
PCP (Probabilistically Checkable Proofs): Any 3-SAT instance can be reformulated so that any purported proof of the new instance is probabilistically checkable by looking at only a constant number of randomly chosen clauses.}, the basis of many results about hardness of approximation and a great achievement of theoretical computer science \cite{PCP1,PCP2}.
It is unlikely that a method based on perturbation theory gadgets, such as one employed in when transforming 3-local Hamiltonian to 2-local Hamiltonian \cite{lh:KKR04}, can lead to a quantum PCP theorem. After rescaling the Hamiltonian back to norm $O(1)$, each round of perturbation theory necessarily shrinks the gap in the spectrum. Although the separation $b-a$ is not the gap $\Delta = E_2-E_1$ of the Hamiltonian, it is related to it as follows. The low energy states in the `yes' instance of a Local Hamiltonian problem encode a correct computation with a positive result. The excited states encode either a wrong computation or a negative final result. The eigenvalue separation between these is $\Delta$. On the other hand, in the `no' instance of Local Hamiltonian, the ground state must again have the computation history wrong or the final result negative. It is very likely that the ground state energy is going to be smaller than the gap $\Delta$ in the `yes' instance. It surely was the case for the Hamiltonians we have encountered so far, where $\Delta \propto L^{-3}$, while the largest separation was $b-a \propto L^{-4}$ (after rescaling the Hamiltonian to $\norm{H}=O(1)$). Is it then possible to ever obtain a better scaling of the necessary separation $b-a$? If it is tied to the $\Delta$ as I described, the outlook is bleak. If the ground state of a Hamiltonian with gap scaling even like $\Delta \propto L^{-1}$ could encode any quantum computation, it would mean we could use the Hamiltonian as a computer and obtain the result in time shorter than what is physically possible. Thus, when searching for a PCP-like quantum theorem, it seems necessary to find new Hamiltonian encodings which do not tie $\Delta $ and $b-a$, or whose underlying idea is not that of Feynman's Hamiltonian computer.

I also showed that Quantum 2-SAT in 1D for particles with dimension $d=11$ is QMA$_1$ complete. It would be interesting to look at the problem from the other end and classify the hardness of Quantum $2$-SAT on a line for lower-dimensional particles, perhaps relating it to the {\em matrix consistency} problem in 1D \cite{lh:YiKai}. My preliminary results show that the ground states of random instances of Quantum 2-SAT in 1D for particles with dimension $d=4$, with 4 clauses per spin could be hard to approximate classically.

\subsubsection {Hamiltonian Quantum Cellular Automata \\
-- universal computing with a simple Hamiltonian}
I proposed a model of quantum computation based on a time-independent, translationally invariant Hamiltonian in Chapter \ref{ch5hqca}. 
The best realization of a HQCA that I found requires $10$ dimensional particles placed on a line with nearest neighbor interactions. I initialize the program and data by preparing a computational basis state of the chain, and use a diffusion-like mechanism to run the program through the data. What errors could plague this model? First, it is prone to an Anderson localization problem in the motion of the program, requiring sensitive tuning of the transitions. Other error preventing measures (such as encoding an already fault tolerant circuit into the HQCA) can be taken, but the model cannot be made fault tolerant, unless the number of `gate particles' in the system is kept fixed. This could be done by using actual particles whose number is conserved. However, three kinds of `gate particles' are needed. Unless this is remedied, this model will remain purely theoretical. On the other hand, it shows interesting results about the hardness of simulation of simple translationally invariant Hamiltonians. It would be interesting to find out how much can the dimension of the particles be lowered, keeping the quantum computing universality of the model. One possible new research direction I have embarked upon so far is a programmable scheme in 2D, based on 2-local nearest neighbor interactions of particles with dimensions $3$ and $5$. For this, I must increase the size of the unit cell (for translational invariance) to about 20 sites.

\vskip12pt

There is a lot to be learned when looking at the underlying machinery of quantum mechanical systems -- local Hamiltonians. They can lead to discoveries of new quantum algorithms such as the NAND tree algorithm \cite{FarhiNAND}, classical simulation methods such as iTEBD \cite{MPS:Vidal1Dinfinite}, proofs of complexity results such as the QMA completeness of Quantum (12,12)-SAT on a line \cite{CA:GottesmanLine}, or a scalable Hamiltonian computer design based on a feasible physical Hamiltonian such as the superconducting AQC of \cite{KaminskyScalable}. I hope this thesis will serve as a pointer to these directions not just for myself, but for the reader as well.

\appendix
\chapter{Kitaev's propagation Hamiltonian}
\label{appProp}
In this Appendix, I look at Kitaev's propagation Hamiltonian \eqref{ch1:Hprop} in more detail and prove a lower bound on its second lowest eigenvalue. 

Consider the system consisting of two registers, the first of which holds $n$ work qubits and the second of which holds $L+1$ clock qubits:
\begin{eqnarray}
	\cH = \cH_{work} \otimes \cH_{clock}.
\end{eqnarray}
Kitaev's propagation Hamiltonian \eqref{ch1:Hprop} is 
\begin{eqnarray}
	H_{prop} &=& \frac{1}{2}\sum_{t=0}^{L} H_{prop}^t, 
	\label{appProp:Hprop}
	\\
	H_{prop}^t &=& 
				\ii \otimes 
				\left(P_{t} + P_{t+1} \right) 
				- \left( X_{t+1,t} \otimes U_t 
				+ X^{\dagger}_{t+1,t} \otimes U^\dagger_t \right)
	,
\end{eqnarray}
where 
\begin{eqnarray}
	P_{t} &=& \ket{t}\bra{t}
\end{eqnarray}
is the projector onto the state $\ket{t}$ of the clock register, 
\begin{eqnarray}
	X_{t+1,t} &=& \ket{t+1}\bra{t}, 
\\
	X^\dagger_{t+1,t} &=& \ket{t}\bra{t+1}, 
\end{eqnarray}
increase and decrease the state of the clock register by one, and $U_t$ acts on two of the work qubits.
To diagonalize this Hamiltonian, it is convenient to make a basis transform
\begin{eqnarray}
	\ket{s} \otimes \ket{t} 
	\qquad \rightarrow \qquad 
	\ket{\Psi_{s,t}} =(U_t U_{t-1} \dots U_2 U_1) \ket{s} \otimes \ket{t},
	\label{appProp:basischange}
\end{eqnarray}
where $\ket{s}$ are computational basis states of the work qubit register. In this new basis, $H_{prop}$ \eqref{appProp:Hprop} has a block diagonal form, 
\begin{eqnarray}
	H_{prop} &=& 
	\left[\begin{array}{cccc}
		A 			&  			& 	&  \\
		 			& A 			& 			 	&  \\
						&  				& \ddots	&   \\
		 			& 			  &  	& A
	\end{array}\right],
\end{eqnarray}
with each $(L+1) \times (L+1)$ block 
\begin{eqnarray}
	A &=& 
	\left[\begin{array}{cccccc}
		\half 	& -\half	& 				&  				&					&					\\
		-\half 	& 1				&	-\half 	&  				&					&					\\
						& -\half	& 1      	& -\half  &					&					\\
		 				& 			  & -\half 	& \ddots	& \ddots  &					\\
		 				& 			  &  				&	\ddots	& 1			 	&	-\half	\\
		 				& 			  &  				&	      	& -\half 	&	\half		
	\end{array}\right]
\end{eqnarray}
corresponding to a basis state $\ket{s}$. Note that the Hamiltonian does not mix sectors with different $\ket{s}$. This matrix can be decomposed as
\begin{eqnarray}
	A = \ii - \half \left( B + P_0 + P_{L} \right),
\end{eqnarray}
where $B$ is the adjacency matrix for a line of length $L+1$ and $P_0$ and $P_{L}$ are projectors on its ends. 
The eigenvectors of $A$ are plain waves \cite{KitaevBook},
\begin{eqnarray}
	\ket{\phi_0} &=& \frac{1}{\sqrt{L+1}} \sum_{t=0}^{L} \ket{t}, \\
	\ket{\phi_j} &=& 
			\sqrt{\frac{2}{L+1}} \sum_{t=0}^{L} 
			\cos\left( p_j \left( t + \frac{1}{2}\right)\right) \ket{t}, 
		\label{appProp:eigs}
\end{eqnarray}
with momenta
\begin{eqnarray}
	p_j &=& \frac{\pi j}{L+1}, \qquad j=0,\dots,L,
\end{eqnarray}
corresponding to the eigenvalues 		
\begin{eqnarray}
	\lambda_j &=& 1- \cos p_j.
\end{eqnarray}
The smallest eigenvalue is $\lambda_0 = 0$,
and the first nonzero eigenvalue is 
\begin{eqnarray}
	\lambda_1 = 1- \cos \left(\frac{\pi}{L+1}\right),
\end{eqnarray}
For large $L$, this eigenvalue is bounded from below by 
\begin{eqnarray}
	\lambda_1 \geq \frac{c_1}{L^2},
\end{eqnarray}
for some constant $c_1$. The eigenvalue gap between the ground state and the first excited state of $H_{prop}$ is thus bounded from below by
\begin{eqnarray}
	\Delta_{H_{prop}} \geq \frac{c_1}{L^2}.
\end{eqnarray}
The dynamics of this system is a quantum walk on a line with the two additional boundary terms. 

\chapter{Dyson Series for an AQC Hamiltonian}
\label{appDyson}

Consider a time-dependent AQC Hamiltonian (see Chapter \ref{ch2adiabatic})
\begin{eqnarray}
	H(s) = (1-s) H_B + s H_P = 
		\underbrace{H_B}_{A} 
		+ \frac{t}{T} \underbrace{(H_P-H_B)}_{B},
		\label{dysonH}
\end{eqnarray}
where $s=t/T$ and $H_B$ does not commute with $H_P$. Assume also that all the terms inside $H_B$ commute with each other, and so do the terms within $H_P$.
My goal in this section is to find a formula of the form
\begin{eqnarray}
	W_{t_0+\Delta t,t_0} = e^{-i a_2 H_B} e^{-i b_2 H_P} e^{-i c_2 H_B},
	\label{expformula}
\end{eqnarray}
approximating the time evolution operator $U_{t_0 + \Delta t,t_0}$
corresponding to \eqref{dysonH} up to order $(\Delta t)^2$ and possibly higher orders as well.

The time evolution operator for a time-dependent Hamiltonian $H(t)$ that doesn't commute at different times can be expressed using the Dyson series
\begin{eqnarray}
	U_{t_0 + \Delta t,t_0} = \ii + \sum_{n=1}^{\infty} 
		\left(\frac{-i}{\hbar}\right)^n
		\int_{t_0}^{t_0 + \Delta t} dt_1 
		\int_{t_0}^{t_1} dt_2 
		\cdots
		\int_{t_0}^{t_{n-1}} dt_n \:
		H(t_1) 
		H(t_2)
		\dots 
		H(t_n).
		\label{dyson}
\end{eqnarray}
In the following, I take $\hbar = 1$. 

The zeroth order term in \eqref{dyson} for $H$ \eqref{dysonH} is just $U_0 = \ii$. The first order term is
\begin{eqnarray}
	U_1 &=& (-i)
		\int_{t_0}^{t_0+\Delta t}
		H(t_1) 
		dt_1
		 \\
		&=& \Delta t (-i)
		\left(
			A + B
				\left(
				\frac{2t_0+ \Delta t}{2T}
				\right)
		\right), \nonumber
\end{eqnarray}
and the second order term is
\begin{eqnarray}
	U_2 &=& -
		\int_{t_0}^{t_0+\Delta t}
		H(t_1)
		\left(
		\int_{t_0}^{t_1}
		H(t_2)\,dt_2
		\right) 
		dt_1
		 \\
		&=& 
			- \frac{(\Delta t)^2}{2} A^2
			- \frac{(\Delta t)^2}{2T} 
					\left(t_0 + \frac{\Delta t}{3}\right)
			   AB \\
			&-& \frac{(\Delta t)^2}{T}
				\left( \frac{t_0}{2} 
					+ \frac{(\Delta t)}{3} \right) 
			   BA \nonumber\\
			&-& \frac{(\Delta t)^2}{2T^2}
			\left(t_0^2 +
				\frac{(\Delta t) t_0}{2} +
				\frac{(\Delta t)^2}{4}
			\right)
			   BB. \nonumber
\end{eqnarray}
Neglecting all terms of order higher than $(\Delta t)^2 \frac{t_0^2}{T^2}$ in $U_0 + U_1 + U_2$, I obtain
\begin{eqnarray}
	V_2 &=& 
	1 -i (\Delta t) A - i  (\Delta t) \left( \frac{t_0}{T} +\frac{\Delta t}{2T}\right) B \\
	&-& \frac{(\Delta t)^2}{2} AA
	- \frac{(\Delta t)^2 t_0}{2T} (AB+BA)
	- \frac{(\Delta t)^2 t_0^2}{2T^2} BB.
	\nonumber
\end{eqnarray}
Rewriting $A$ and $B$ back in terms of $H_B$ and $H_P$, I obtain
\begin{eqnarray}
	V_2 
		&=& 
	\ii  \nonumber\\
	&-& i 
		\left( 1-\frac{t_0}{T}-\frac{\Delta t}{2T}\right)
		(\Delta t) H_B \nonumber\\
	&-& i 
		\left( \frac{t_0}{T}+\frac{\Delta t}{2T}\right)
		(\Delta t) H_P \nonumber\\
	&-& \half \left( 1-\frac{t_0}{T} \right)^2 
		(\Delta t)^2 H_B H_B \nonumber\\
	&-& \frac{t_0}{2T} \left( 1-\frac{t_0}{T} \right) 
		(\Delta t)^2 (H_B H_P + H_P H_B) \nonumber\\
	&-& \half \left(\frac{t_0}{T} \right)^2 
		(\Delta t)^2 H_P H_P.
		\label{wantapprox}
\end{eqnarray}

I now want to find a formula of the form \eqref{expformula}, approximating \eqref{wantapprox}.
The simplest version is
\begin{eqnarray}
	W_1 &=& e^{-ia_1 H_B} e^{-ib_1 H_P} \\
		&\approx &
			(1 - i a_1 H_B)(1 - i b_1 H_P),  \nonumber
\end{eqnarray}
which approximates \eqref{wantapprox} to order $(\Delta t)$ when
\begin{eqnarray}
	a_1 &=& \left(\frac{t_0}{T}\right) \Delta t, \\
	b_1 &=& \left(1-\frac{t_0}{T}\right) \Delta t.
\end{eqnarray}
This agrees with the usual first-order Trotter formula. 

Next, let me look at
\begin{eqnarray}
	W_2 &=& e^{-ia_2 H_B} e^{-ib_2 H_P} e^{-ic_2 H_B}
	\label{appb:second}
	\\
		&\approx &
			\left(1 - i a_2 H_B - \frac{a_2^2}{2} H_B^2\right)
			\left(1 - i b_2 H_P - \frac{b_2^2}{2} H_P^2\right)
			\left(1 - i c_2 H_B - \frac{c_2^2}{2} H_B^2\right).
		\nonumber
\end{eqnarray}
When choosing
\begin{eqnarray}
	a_2 = c_2 &=& \left(1-\frac{t_0}{T} - \frac{\Delta t}{2T}\right) \frac{\Delta t}{2}, \\
	b_2 &=& \left(\frac{t_0}{T} + \frac{\Delta t}{2T}\right) \Delta t,
\end{eqnarray}
this formula agrees with \eqref{wantapprox} to order $(\Delta t)^2$.
It has extra corrections terms, not found in the usual second order Trotter formula derived for a time-independent Hamiltonian $H$
\begin{eqnarray}
	W^T_2 &=& e^{-i a_2^T H_B} e^{-i b_2^T H_P} e^{-i c_2^T H_B}, \\
	a_2^T = c_2^T &=& \left(1-\frac{t_0}{T}\right) \frac{\Delta t}{2}
		, \\
	b_2^T &=& \left(\frac{t_0}{T}\right) \Delta t.
\end{eqnarray}

\chapter{Continuous-time Quantum Walks in 1D}
\label{d20proof}

First, in Section \ref{plainline} I analyze the continuous-time quantum walk on a line and prove two useful lemmas about the mixing of this walk used in Section \ref{d20time} and Appendix \ref{d10proof}. Second, 
in Section \ref{circleline} I analyze the quantum walk on a cycle and prove another mixing lemma used in the proof of universality of the {\em train switch} Hamiltonian Computer in Section \ref{ch4:train}.

\section{Quantum Walk on a Line}
\label{plainline}
Consider a continuous time quantum walk on a line of length $L$, where the Hamiltonian is the negative of the adjacency matrix for the line 
\begin{eqnarray}
	H_1 = - \sum_{j=1}^{L-1} 
	\left(\ket{j}\bra{j+1} + \ket{j+1}\bra{j}\right).
	\label{H1hamiltonian}
\end{eqnarray}
The eigenvalues of this Hamiltonian are
\begin{eqnarray}
	\lambda_j = - 2 \cos \left(\frac{j\pi}{L+1}\right),
	\label{eigenval}
\end{eqnarray}
for $j=1\dots L$, while the corresponding eigenvectors $\kets{\phi^{(j)}} = \sum_{k=1}^{L} \phi^{(j)}_k \ket{k}$ have components
\begin{eqnarray}
	\phi^{(j)}_k = \sqrt{\frac{2}{L+1}} \, 
		\sin \left( \frac{j k \pi}{L+1} \right).
	\label{eigenvec}
\end{eqnarray}
Consider the time evolution of a particular basis state $\ket{c}$. The probability of finding the system in the 
basis 
state $\ket{m}$ at some time $\tau$ can be found by expanding $\ket{c}$ and $\ket{m}$ in the basis of the eigenvectors \eqref{eigenvec}:
\begin{eqnarray}
	p_{\tau}(m|c) = \left| 
				\bra{m} e^{-iH\tau} \ket{c}
			\right|^2
		= \sum_{j,k=1}^{L} e^{-i (\lambda_j -\lambda_k) \tau} 
			\phi^{(j)}_m
			\phi^{(j)*}_c
			\phi^{(k)*}_m
			\phi^{(k)}_c.
\end{eqnarray}
Because the time evolution (according to the Schr\"odinger equation) is unitary, this probability $p_{\tau}(m|c)$ does not converge. On the other hand, let me define the time average of $p_{\tau}(m|c)$
for time $0\leq\tau\leq\tau_{20}$ as
\begin{eqnarray}
	\bar{p}_{\tau_{20}} (m|c) 
		= \frac{1}{\tau_{20}} \int_0^{\tau_{20}} p_{\tau} (m|c) d\tau.
	\label{Paverage}
\end{eqnarray}
As I will show below in Lemma \ref{linelemma}, this average probability distribution does converge to a limiting distribution $\pi(m|c)$, defined
as the $\tau_{20}\rightarrow\infty$ limit of the average probability distribution \eqref{Paverage}. All the eigenvalues \eqref{eigenval} are different, so I can express the limiting distribution as
\begin{eqnarray}
	\pi(m|c) 
	= \lim_{\tau_{20}\rightarrow\infty} \bar{p}_{\tau_{20}}(m|c) 
	= \sum_{j=1}^{L} \big|\phi^{(j)}_m\big|^2 
		\big|\phi^{(j)}_c\big|^2,
	\label{deflimdistrib}
\end{eqnarray}
which in this case is
\begin{eqnarray}
	\pi(m|c) 
	 = \frac{2+\delta_{m,c} + \delta_{m,L+1-c}}{2(L+1)}.
	 \label{Plimiting}
\end{eqnarray}

According to the following lemma, the average probability distribution \eqref{Paverage} converges to the limiting distribution $\pi(m|c)$.
\begin{lemma}
\label{convergelemma}
Consider a continuous time quantum walk on a line of length $L$, where the Hamiltonian is the negative of the adjacency matrix for the line. Let the system evolve for time $\tau \leq \tau_{20}$ chosen uniformly at random, starting in a position basis state $\ket{c}$.
The average probability distribution $\bar{p}_{\tau_{20}}(\cdot|c)$ converges to the limiting probability distribution $\pi(\cdot|c)$ as
\begin{eqnarray}
	\sum_{m=1}^{L} \left| \bar{p}_{\tau_{20}}(m|c) - \pi(m|c) \right|
	\leq 
		O\left(\frac{L}{\tau_{20}}\right).
	\label{limitinglemma}
\end{eqnarray}
\end{lemma}

\begin{proof}
First, recall Lemma 4.3 of \cite{CA:AAKVwalk:01} for the total variation distance of the probability distribution $\bar{p}_{\tau_{20}}$ from the limiting distribution, saying
\begin{eqnarray}
	\sum_m \left| \bar{p}_{\tau_{20}}(m|c) - \pi(m|c) \right| \leq 
		\frac{2}{\tau_{20}} \sum_{\lambda_j\neq \lambda_k} 
		\frac{\big|\phi^{(j)}_c\big|^2}{ |\lambda_j-\lambda_k|}.
	\label{convergebound}
\end{eqnarray}
Using \eqref{eigenval} and \eqref{eigenvec}, I can bound the expression on the right of \eqref{convergebound}. When $j$ is close to $k$, i.e. $|j-k|\leq C_1$, I can write
\begin{eqnarray}
	\frac{\big|\phi^{(j)}_c\big|^2}{|\lambda_j-\lambda_k|} < 2.
\end{eqnarray}
On the other hand, for $|j-k|>C_1$ I can bound
\begin{eqnarray}
	\frac{\big|\phi^{(j)}_c\big|^2}{|\lambda_j-\lambda_k|} < \frac{C_2}{L+1},
\end{eqnarray}
with $C_1$ and $C_2$ constants independent of $L$. 
Inserting into \eqref{convergebound}, I have
\begin{eqnarray}
	\sum_{m=1}^{L} \left| \bar{p}_{\tau_{20}}(m|c) - \pi(m|c) \right|
	\leq 
		\frac{8 C_1 L}{\tau_{20}} + \frac{C_2 L}{\tau_{20}}
		= O\left(\frac{L}{\tau_{20}}\right).
\end{eqnarray}
\end{proof}

Using Lemma \ref{convergelemma}, I now prove a useful result utilized in the time analysis of the $d=20$ HQCA in Section \eqref{d20section}. 

\begin{lem4}
Consider a continuous time quantum walk on a line of length $L$, where the Hamiltonian is the negative of the adjacency matrix for the line. Let the system evolve for a time $\tau \leq \tau_{20}$ chosen uniformly at random, starting in a position basis state $\ket{c}$.
The probability to measure a state $\ket{t}$ with $t>L/6$ is then bounded from below as $p_{20}\geq\frac{5}{6}-O\left(\frac{L}{\tau_{20}}\right)$. 
\end{lem4}

\begin{proof}
The probability to measure a state $\ket{t}$ with $t>L/6$ at time $\tau\leq \tau_{20}$ chosen uniformly at random is
\begin{eqnarray}
	p_{20} = \sum_{m>\frac{L}{6}} 
			 \bar{p}_{\tau_{20}}(m|c).
\end{eqnarray}
Starting with \eqref{limitinglemma}, I have
\begin{eqnarray}
	O\left(\frac{L}{\tau_{20}}\right) 
	&\geq& 
		\sum_{m=1}^{L} 
			\left| \bar{p}_{\tau_{20}}(m|c) - \pi(m|c) \right| \\
	&\geq& 
		\sum_{m>\frac{L}{6}} 
			\left| \bar{p}_{\tau_{20}}(m|c) - \pi(m|c) \right| \\
	&\geq& 
		\left| \sum_{m>\frac{L}{6}} 
			 \bar{p}_{\tau_{20}}(m|c) 
		-
		\sum_{m>\frac{L}{6}} 
				\pi(m|c) \right| \\
	&=& 
		\left| p_{20} - \frac{5}{6} + O\left(\frac{1}{L}\right) 
			\right|.
\end{eqnarray} 
Therefore, the probability of finding the chain in state $\ket{\psi_{t>L/6}}$ at a random time $\tau \leq \tau_{20}$ is thus bounded from below by
\begin{eqnarray}
p_{20} \geq \frac{5}{6} - O\left(\frac{L}{\tau_{20}}\right).
\end{eqnarray}
\end{proof}

Also, Lemma \ref{linelemma} can be easily generalized for any desired probability $q$. Measuring a state which started in a position basis state $\ket{c}$ at a random time $\tau \leq \tau_{q}$ chosen uniformly at random, the probability to measure a state $\ket{t}$ with $t>(1-q)L$ is then bounded from below by $p \geq q - O\left(\frac{L}{\tau_{q}}\right)$. 



\section{Quantum Walk on a Circle}
\label{circleline}

If the geometry of the system is a closed loop of length $L$ instead of a line, the Hamiltonian \eqref{H1hamiltonian} gets an additional wrap-around term.
\begin{eqnarray}
	H_{loop} = 
	-\left(\ket{L}\bra{1} + \ket{1}\bra{L}\right)
	- \sum_{j=1}^{L-1} 
	\left(\ket{j}\bra{j+1} + \ket{j+1}\bra{j}\right).
	\label{HLoophamiltonian}
\end{eqnarray}
The eigenvalues of this Hamiltonian are
\begin{eqnarray}
	\lambda_j = - 2 \cos \left(p_j\right),
	\label{cycleeigenval}
\end{eqnarray}
corresponding to plain waves with momenta
\begin{eqnarray}
	p_j = \frac{2\pi j}{L},
	\label{cyclemomenta}
\end{eqnarray}
for $j=0\dots L-1$. The corresponding eigenvectors $\kets{\phi^{(j)}} = \sum_{k=1}^{L} \phi^{(j)}_k \ket{k}$ have components
\begin{eqnarray}
	\phi^{(0)}_k &=& \frac{1}{\sqrt{L}}, \\
	\phi^{(j)}_k &=& \sqrt{\frac{2}{L}} e^{ip_j k}, \qquad j=1,\dots,L-1,
	\label{cycleeigenvec}
\end{eqnarray}
These can be combined to make real eigenvectors. For my analysis, it will be enough to consider a line with even length, and only the cosine plain waves:
\begin{eqnarray}
	\phi^{(0)}_k &=& \frac{1}{\sqrt{L}}, \\
	\phi^{(j)}_k &=& \sqrt{\frac{2}{L}} \cos(i p_j k), \qquad j=1,\dots,\frac{L}{2}.
	\label{cycleeveneigenvec}
\end{eqnarray}

The limiting distribution on the cycle when starting from site $c$ is 
\begin{eqnarray}
	\pi(m|c) 
	 = \frac{1}{L}- \frac{2}{L^2}.
	 \label{Pcyclelimiting2}
\end{eqnarray}
for all points $m$ except for $m=c$ (return back) and $m=L+1-c$
(the point across the cycle), where I have
\begin{eqnarray}
	\pi(m|c) 
	 = \frac{2}{L} - \frac{2}{L^2}.
	 \label{Pcyclelimiting1}
\end{eqnarray}
As in Section \ref{plainline}, I will again utilize Lemma 4.3 of \cite{CA:AAKVwalk:01} to prove the convergence of the time-averaged probability distribution to this limiting distribution. On the right side of \eqref{convergebound}, I now have
\begin{eqnarray}
	\big|\phi_0^{(j)}\big|^2 \leq \frac{2}{L}.
\end{eqnarray}
The sum over the non-equal eigenvalues
\begin{eqnarray}
	S_{\circ} = \sum_{\lambda_j\neq \lambda_k} 
		\frac{1}{ |\lambda_j-\lambda_k|}
		\label{cycleeigsum}
\end{eqnarray}
is now more complicated, because of the degeneracy of the spectrum.
\begin{figure}
	\begin{center}
	\includegraphics[width=3.5in]{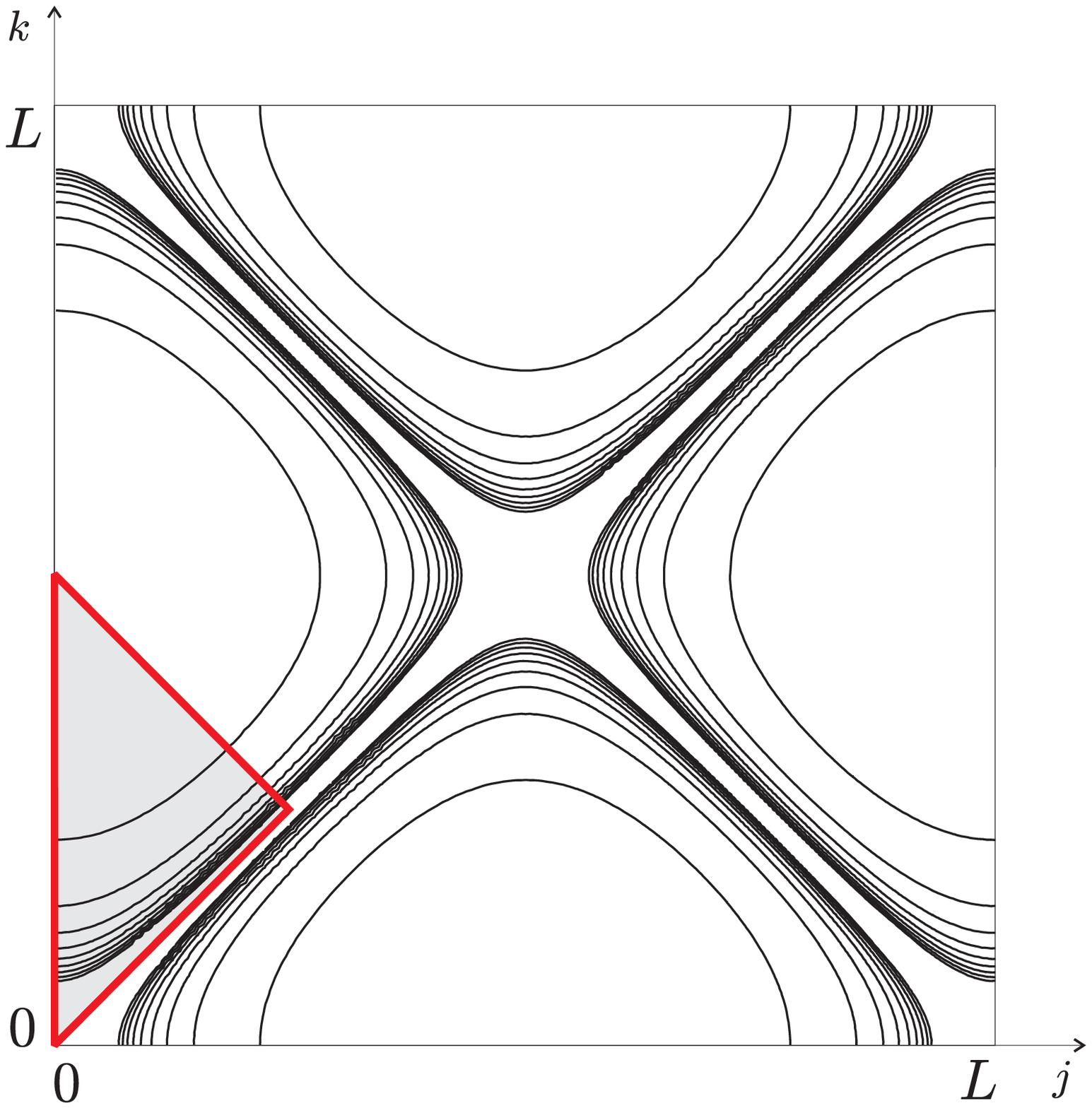}
	\end{center}
	\caption{A contour plot of $\left|\lambda_j - \lambda_k\right|^{-1}$ with $\lambda_j$ from \eqref{cycleeigenval}. The sum in \eqref{cycle16} is over the marked region.}
	\label{appc:figcircleeig}
\end{figure}
I plot $|\lambda_j-\lambda_k|^{-1}$ in Figure \ref{appc:figcircleeig}.
Because of the symmetrical way $\lambda_j$ \eqref{cycleeigenval} arise, assuming $L$ is divisible by 4, I can express \eqref{cycleeigsum} as
\begin{eqnarray}
	S_{\circ} 
	&=& 
	\sum_{\lambda_j\neq\lambda_k} 
		\frac{1}{ |\lambda_j-\lambda_k|} \\
	&\leq& 
	16 \sum_{j=0}^{L/4} \sum_{k=j+1}^{L/2-j}
		\frac{1}{ |\lambda_j-\lambda_k|}
		\label{cycle16}
	\\
	&=& 
	16 \underbrace{\sum_{k=1}^{L/2}
		\frac{1}{ |\lambda_0-\lambda_k|}}_{A_0}
		+
	16 \underbrace{\sum_{j=1}^{L/4} \sum_{k=j+1}^{L/2-j}
		\frac{1}{ |\lambda_j-\lambda_k|}}_{A_1}.
		\label{cyclesumbounds}
\end{eqnarray}
The term $A_0$ comes from $j=0$. I can bound the sum by an integral, taking $x=\frac{\delta k}{L}$, obtaining
\begin{eqnarray}
	A_0 = \frac{1}{2}\sum_{k=1}^{L/2}
		\frac{1}{ 1-\cos\left(2\pi \frac{k}{L}\right)} 
	\leq \frac{L}{2} \int_{\frac{1}{L}}^{\frac{1}{2}}
		\frac{dx}{ 1-\cos(2\pi x)}.
\end{eqnarray}
For $x \in \left[0,\half\right]$, I can bound
\begin{eqnarray}
	1-\cos{2 \pi x} \geq 2 \pi x^2,
\end{eqnarray}
resulting in
\begin{eqnarray}
	A_0 \leq \frac{L}{4\pi} \int_{\frac{1}{L}}^{\frac{1}{2}}
		\frac{dx}{x^2}
		= \frac{1}{4\pi} \left(L^2 - 2L\right) = O(L^2).
		\label{a0bound}
\end{eqnarray}
I bound the other term, $A_1$, in \eqref{cyclesumbounds} by an integral as well:
\begin{eqnarray} 
	A_1 &=& \frac{1}{2}\sum_{j=1}^{L/4} \sum_{k=j+1}^{L/2-j}
		\frac{1}{ \cos\left(2\pi\frac{k}{L}\right)
			- \cos\left(2\pi\frac{j}{L}\right)} \\
		&\leq&
		\frac{L^2}{2} \int_{\frac{1}{L}}^{\frac{1}{4}-\frac{1}{L}}
		dy 
		\int_{\frac{1}{L}+y}^{\frac{1}{2}-y}
		\frac{dx}{\cos\left(2\pi x\right)
			- \cos\left(2\pi y\right)}.		
\end{eqnarray} 
Again, I can lower bound the difference in eigenvalues for $y\in \left[0,\frac{1}{4}\right]$ and $x\in \left[y,\frac{1}{2}-y\right]$ by
\begin{eqnarray}
	\cos\left(2\pi x\right)
			- \cos\left(2\pi y\right) \geq 2\pi (x^2-y^2),
\end{eqnarray}
allowing me to write
\begin{eqnarray} 
	A_1 &\leq& 
		\frac{L^2}{4\pi} \int_{\frac{1}{L}}^{\frac{1}{4}-\frac{1}{L}}
		dy 
		\int_{\frac{1}{L}+y}^{\frac{1}{2}-y}
		\frac{dx}{x^2-y^2} \\
	&=&
		\frac{L^2}{4\pi} \int_{\frac{1}{L}}^{\frac{1}{4}-\frac{1}{L}}
		dy 
		\left[
			\frac{\log\left(\frac{x-y}{x+y}\right)}{2y}
		\right]_{\frac{1}{L}+y}^{\frac{1}{2}-y} \\
	&=&
		\frac{L^2}{4\pi} \int_{\frac{1}{L}}^{\frac{1}{4}-\frac{1}{L}}
		\frac{dy}{2y} 
		\underbrace{\left[
			\log\left(1-4y\right)
			+ 
			\log\left(1+2yL\right)
		\right]}_{R}.
		\label{a1final}
\end{eqnarray} 
As $y$ in \eqref{a1final} is at most $\frac{1}{4}-\frac{1}{L}$,
I can bound $R$ by
\begin{eqnarray}
	|R| \leq \log L.
\end{eqnarray}
Finally, this results in
\begin{eqnarray}
	A_1 \leq \frac{L^2 \log L}{2\pi} \int_{\frac{1}{L}}^{\frac{1}{4}-\frac{1}{L}}
		\frac{dy}{2y} 
		= \frac{L^2 \log L}{2\pi} 
		\underbrace{\log \left(\frac{\frac{1}{4}-\frac{1}{L}}{\frac{1}{L}}\right)}_{\leq \log L}
		\leq O(L^2 \log^2 L).
		\label{a1bound}
\end{eqnarray}
Putting \eqref{a0bound} and \eqref{a1bound} into \eqref{cyclesumbounds}, I obtain
\begin{eqnarray}
	S_{\circ} \leq O(L^2 \log^2 L).
\end{eqnarray}
Lemma 4.3 of \cite{CA:AAKVwalk:01} (see \eqref{convergebound}) then reads
\begin{eqnarray}
	\sum_m \left| \bar{p}_{\tau_{20}}(m|c) - \pi(m|c) \right| \leq 
		\frac{2}{\tau_{\circ}} \frac{1}{L} O(L^2 \log^2 L)
		= O\left(\frac{L \log^2 L}{\tau_\circ}\right).
	\label{convergeboundcycle}
\end{eqnarray}
Thus, for $\tau_{circ} = \ep O(L \log^2 L)$, the time-averaged distribution becomes $\ep$-close to the limiting distribution. Using the bound on the total variation distance I just proved, it is straightforward to obtain the following lemma which I use in Section \ref{ch4:train}:
\begin{lemma}
	\label{cyclelemma}
Consider a continuous time quantum walk on a cycle of length $L$ (divisible by 4), where the Hamiltonian is the negative of the adjacency matrix for the cycle. Let the system evolve for a time $\tau \leq \tau_{\circ}$ chosen uniformly at random, starting in a position basis state $\ket{0}$.
The probability to measure a position state $\ket{t}$ farther than $L/6$ from the starting point (the farther two thirds of the cycle) is then bounded from below as $p_{\circ}\geq\frac{2}{3}-\frac{1}{3L}-O\left(\frac{L \log^2 L}{\tau_{\circ}}\right)$. 
\end{lemma}
\begin{proof}
The proof is analogous to the proof of Lemma \ref{convergelemma} in Section \ref{plainline}. 
Let me call the farther two thirds of the cycle (see Figure \ref{ch4:figurecircle} in Section \ref{ch4:train}) the success region (SR). When I choose the time $\tau\leq \tau_{\circ}$ uniformly at random, the probability to measure a state $\ket{t}$ with $t\in SR$ is\begin{eqnarray}
	p_{\circ} = \sum_{m\in SR} 
			 \bar{p}_{\tau_{\circ}}(m|c),
\end{eqnarray}
Using the bound on the total variation distance \eqref{convergeboundcycle} I just proved and the formulae for the limiting distribution \eqref{Pcyclelimiting2},\eqref{Pcyclelimiting1}, I have
\begin{eqnarray}
	O\left(\frac{L \log^2 L}{\tau_{\circ}}\right) 
	&\geq& 
		\sum_{m=1}^{L} 
			\left| \bar{p}_{\tau_{\circ}}(m|c) - \pi(m|c) \right| \\
	&\geq& 
		\sum_{m\in SR} 
			\left| \bar{p}_{\tau_{\circ}}(m|c) - \pi(m|c) \right| \\
	&\geq& 
		\left| \sum_{m\in SR} 
			 \bar{p}_{\tau_{\circ}}(m|c) 
		-
		\sum_{m\in SR} 
				\pi(m|c) \right| \\
	&=& 
		\left| p_{\circ} - \frac{2}{3} + \frac{1}{3L} + O\left(\frac{1}{L}\right) 
			\right|.
\end{eqnarray} 
Therefore, the probability of finding the chain in state $\ket{\psi_{t\in SR}}$ at a random time $\tau \leq \tau_{\circ}$ is thus bounded from below by
\begin{eqnarray}
p_{\circ} \geq \frac{2}{3} - \frac{1}{3L} - O\left(\frac{L \log^2 L}{\tau_{\circ}}\right).
\end{eqnarray}
\end{proof}
It is thus enough to wait a random time not larger than $O(L\log^2 L)$ to find the state of the system in the success region with probability close to two thirds. 
\chapter{Diffusion of Fermions on a Line}
\label{d10proof}
Here I prove Lemma \ref{difflemma}, a result about the mixing of a discrete free fermion gas, used in Section \ref{d10section}.

\begin{lem3}
	Consider the state 
\begin{eqnarray}
	\ket{\Psi_0} = b^{\dagger}_{fM+1} b^{\dagger}_{fM+2} \dots b^{\dagger}_{fM+2M} 
	\ket{0}.
\end{eqnarray}
of $2M$ fermions on the right end of a line with $L=(f+2)M$ sites. Let the system evolve for a time chosen uniformly at random between $0$ and $\tau_{10}$ with the Hamiltonian 
\begin{eqnarray}
	H_f = - \sum_{j=1}^{L-1} b_j^\dagger b_{j+1} + h.c.
\end{eqnarray}
and measure the number of fermions in the region $1\leq x \leq fM$. The probability to measure a number greater than $M$
is $p_{10} \geq \frac{f-2}{f+2} - O\left(\frac{L}{\tau_{10}}\right)$.
\end{lem3}

\begin{proof}
Let us start with the outline of the proof. We look at the fermionic system in both first and second quantization to obtain an expression for the time evolution of the creation and annihilation operators in the Heisenberg picture, mapping it to a quantum walk on a line. We then consider the observable $X$, the number of particles sufficiently far from the right end of the line. We will show that when we choose the time to measure $X$ uniformly at random between $0$ and $\tau_{10}$, the expected value we will obtain is approaching a number close to $2M$. To show this, we will express the expected value of $X$ in the time-averaged state of the system using the results from a quantum walk on a line. Finally, because the number of particles in the system is $2M$, we will deduce that the probability to measure a number less than $M$ is then small.

Observe that $H_f$ is the Hamiltonian of a free fermion gas on a line in second quantization (a special case of the XY model). 
The time evolution of the state $\ket{\Psi_0}$ can be obtained by looking at the problem back in the first quantization, where we write $\ket{\Psi_0}$ as
\begin{eqnarray}
	\ket{\Psi_0} = \Big[ \ket{\phi_{fM+1}}\otimes \ket{\phi_{fM+2}} 
		\otimes \cdots \otimes \ket{\phi_{fM+2M}} \Big]^{-},
\end{eqnarray}
with $\ket{\phi_j}=\ket{j}$ in the position basis and $[\,\cdot\,]^{-}$ the standard antisymmetrization operator. We first solve for the time evolution of the corresponding one-particle wavefunction $\ket{\phi_j(\tau)}$ with the Hamiltonian
\begin{eqnarray}
	H_1 = - \sum_{j=1}^{L-1} 
	\left(\ket{j}\bra{j+1} + \ket{j+1}\bra{j}\right), 
	\label{H1hamiltonianAppendix}
\end{eqnarray}
and then obtain the solutions for the many-particle problem by antisymmetrization as
\begin{eqnarray}
	\ket{\Psi(\tau)} = \Big[ \ket{\phi_{fM+1}(\tau)}\otimes 
		\ket{\phi_{fM+2}(\tau)}
		\otimes \cdots \otimes \ket{\phi_{fM+2M}(\tau)} \Big]^{-}.
\end{eqnarray} 
The eigenfunctions of $H_1$ (quantum walk on a line) are plain waves (as in \eqref{eigenval} and \eqref{eigenvec}), and the time evolved states $\ket{\phi_j(\tau)}$ thus readily available. Let us define the unitary matrix $u(\tau)$ by 
\begin{eqnarray}
	\ket{j(\tau)} = \sum_{k=1}^{L} u_{jk}(\tau) \ket{k}.
\end{eqnarray}
Returning to the second quantized system, the time evolution of the creation and annihilation operators in the Heisenberg picture is then
\begin{eqnarray}
	b^{\dagger}_j(\tau) = \sum_{k=1}^{L} 
			u_{jk}(\tau) b^{\dagger}_k, 
	\qquad \qquad
	b_j(\tau) = \sum_{k=1}^{L} 
		u^*_{jk}(\tau) b_k. 
		\label{uheisenberg}
\end{eqnarray}

Consider now the observable $X$, the number of particles in the first $fM$ sites of the line with length $L=(f+2)M$
\begin{eqnarray}
	X &=& \sum_{m=1}^{fM} \hat{n}_{m} \label{Xdef}.
\end{eqnarray}
Its expectation value at time $\tau$ is
\begin{eqnarray}
	E_{\tau}(X) &=& 
	\sum_{m=1}^{fM} \bra{\Psi(\tau)} \hat{n}_{m} \ket{\Psi(\tau)}. 
	\label{Xexpect}
\end{eqnarray}
The number operator for site $m$ is $\hat{n}_m = b_m^{\dagger} b_m$. We can go to the Heisenberg picture and use \eqref{uheisenberg} to write
\begin{eqnarray}
	\bra{\Psi(\tau)} \hat{n}_{m} \ket{\Psi(\tau)}
	&=& 
	\bra{\Psi_0} 
		 b_m^{\dagger}(\tau) b_m(\tau)
	\ket{\Psi_{0}} \\
	&=& 
	\sum_{c=1}^{L} \sum_{d=1}^{L} u_{mc}(\tau)u^*_{md}(\tau) \bra{\Psi_0} 
		  b_c^{\dagger} b_d
	\ket{\Psi_{0}} \\
&=& 
	\sum_{c=1}^{L} 
	 \left| u_{mc}(\tau) \right|^2 \bra{\Psi_0} 
		  b_c^{\dagger} b_c
	\ket{\Psi_{0}} \\
&=& 
	\sum_{c=fM+1}^{L} \underbrace{\left| u_{mc}(\tau) \right|^2
	}_{p_{\tau}(m|c)},
	\label{nexpect10A}
\end{eqnarray}
where each term $\left| u_{mc}(\tau) \right|^2 = p_{\tau}(m|c)$ can be thought of as the probability of finding a particle at site $m$ at time $\tau$ when it started from the site $c$ and performed a quantum walk on a line, according to \eqref{H1hamiltonianAppendix}. Inserting this into \eqref{Xexpect}, the expected number of particles not in the rightmost part of the chain at time $\tau$ is 
\begin{eqnarray}
	E_{\tau}(X) 
	=
	\sum_{c=fM+1}^{L} 
	\left( \sum_{m=1}^{fM} p_{\tau}(m|c) \right)
	.
	\label{nexpect10}
\end{eqnarray}

Let us now choose the time $\tau$ uniformly at random between $0$ and $\tau_{10}$. The average value of $X$ (the expectation value in the time-average state) is
\begin{eqnarray}
	\bar{E}_{\tau_{10}}(X) = \tav{E_{\tau}(X)}.
\end{eqnarray}
For a quantum walk on a line, the time-averaged probability \eqref{Paverage} of finding a particle that started at position $c$ at final position $m$ converges to the limiting distribution \eqref{Plimiting}
according to Lemma 1 \eqref{limitinglemma} proven in Appendix \ref{d20proof}.
Using this fact, we can show that the expectation value $\bar{E}_{\tau_{10}}(X)$ in the time-averaged state converges to the limiting expectation value 
\begin{eqnarray}
 	\bar{E}(X) = \sum_{m\leq fM} \sum_{c>fM} \pi(m|c)
\end{eqnarray}
as
\begin{eqnarray}
	\left|
		\bar{E}_{\tau_{10}}(X) 
		- 
	    \bar{E}(X) \right|
	\leq O\left(\frac{LM}{\tau_{10}}\right).
	\label{expectconverge}
\end{eqnarray}
Recalling the limiting probability distribution for a quantum walk on a line of length $L$ \eqref{Plimiting}, we have
\begin{eqnarray}
 	\bar{E}(X) 
	&=& \sum_{m\leq fM} \sum_{c>fM} \pi(m|c) \\
	&=& fM\times 2M \times 	\frac{2}{2(L+1)}
	+ 2M \times 	\frac{1}{2(L+1)} \\
	&=& 
	2M\left(\frac{f}{f+2}\right) + O\left(1\right)
	.
\end{eqnarray}
Putting this into \eqref{expectconverge}, the average value of $X$ when the time $\tau\leq \tau_{10}$ is chosen uniformly at random is bounded from below as
\begin{eqnarray}
 	\bar{E}_{\tau_{10}}(X) 
	\geq
	 2M\left(\frac{f}{f+2}\right) - O\left(\frac{LM}{\tau_{10}}\right)
	.
\end{eqnarray}
We want to find the probability of measuring $X>M$. 
First, the maximum possible value we could measure at any time is $2M$, the number of particles in the system. Second, the average value $\bar{E}_{\tau_{10}}(X)$ at time $\tau$ chosen randomly is close to $2M$. Therefore, the fraction $\Delta$ of times at which we measure a number significantly lower than $2M$ must be small. Let us bound $\Delta$ in the worst case scenario. This is when each unsuccessful measurement yields $X=M$, and each successful measurement gives us $2M$. 
We then have
\begin{eqnarray}
	\Delta M + (1-\Delta)  2M &\geq& \bar{E}_{\tau_{10}}(X), \\
	\Delta  &\leq& \frac{ 
		 2\bar{E}_{\tau_{10}}(X) - M}{M}.
\end{eqnarray}
Hence we arrive at the desired bound on the probability to measure $X>M$:
\begin{eqnarray}
	p_{10} &=& 1-\Delta 
	\geq \frac{ 
		 2M\left(\frac{f}{f+2}\right) - O\left(\frac{LM}{\tau_{10}}\right)
		  - M
		 }{
		 M
		 }
	= \frac{f-2}{f+2} - O\left(\frac{L}{\tau_{10}}\right).
\end{eqnarray}
\end{proof}

\chapter{Transitions in a two level system}
\label{appTwolevel}

In this Appendix I give the analysis of transitions in a simple two-level system referenced in Section \ref{ch2:projector}. Consider a two level system with Hamiltonian
\begin{eqnarray}
  H(s) = E_0(s) \ket{\phi_0(s)}\bra{\phi_0(s)} + E_1(s) \ket{\phi_1(s)}\bra{\phi_1(s)},
\end{eqnarray}
which varies smoothly with $s=t/T$. Here $\ket{\phi_0(s)}$ and $\ket{\phi_1(s)}$ are orthonormal for all $s$.
The Schr\"{o}dinger equation reads
\begin{eqnarray}
  i \deriv{}{s} \ket{\psi} = T H(s) \ket{\psi}.
\end{eqnarray}
The two energy levels in the system are separated by a gap 
\begin{eqnarray}
	g(s)=E_1(s)-E_0(s),
\end{eqnarray}
which I assume is always larger than $0$. Let me introduce $\theta$ (with the dimension of energy) as 
\begin{eqnarray}
	\theta(s) = \int_0^s g(s') \,\textrm{d}s',
\end{eqnarray}
and let
\begin{eqnarray}
	\ket{\psi(s)} = c_0(s) e^{-iT \int_0^s E_0(s')\,\textrm{d}s'}\ket{\phi_0(s)} + 
		c_1(s) e^{-iT \int_0^s E_1(s')\,\textrm{d}s'}\ket{\phi_1(s)}. 		\label{statesolution}
\end{eqnarray}
I pick the phases of $\ket{\phi_1(s)}$ and $\ket{\phi_0(s)}$ such that 
$\bra{\phi_1(s)}\deriv{}{s}\ket{\phi_1(s)}=\bra{\phi_0(s)}\deriv{}{s}\ket{\phi_0(s)}=0$.
Plugging \ref{statesolution} into the Schr\"odinger equation gives
\begin{eqnarray}
	\deriv{c_0}{s} &=& c_1 e^{-iT\theta} \Big\langle \phi_1 \Big| \deriv{}{s} \Big| \phi_0\Big\rangle^*, \\
	\deriv{c_1}{s} &=& -c_0 e^{iT\theta} \Big\langle \phi_1 \Big| \deriv{}{s} \Big| \phi_0\Big\rangle,
\end{eqnarray}
or equivalently,
\begin{eqnarray}
	\deriv{c_0}{\theta} &=& c_1 e^{-iT\theta} f^*, \\
	\deriv{c_1}{\theta} &=& -c_0 e^{iT\theta} f,
\end{eqnarray}
where 
\begin{eqnarray}
	f(\theta) \equiv \Big\langle \phi_1 \Big| \deriv{}{\theta} \Big| \phi_0\Big\rangle
		 = -\frac{1}{g} \Big\langle \phi_1 \Big| \deriv{H}{\theta} \Big| \phi_0\Big\rangle
		 = -\frac{1}{g^2} \Big\langle \phi_1 \Big| \deriv{H}{s} \Big| \phi_0\Big\rangle.
\end{eqnarray}
Now let $\theta(1)=\bar{\theta}$. 
I started with $c_1(0)=0$ and I want to find the transition amplitude at $s=1$ which is 
\begin{eqnarray}
   c_1(\bar{\theta}) &=& -\int_0^{\bar{\theta}} c_0 e^{iT\theta}f\,\textrm{d}\theta \\
 	&=& \left[-c_0 f \frac{e^{iT\theta}}{iT}\right]^{\bar{\theta}}_0 + 
   \frac{1}{iT} \int_0^{\bar{\theta}} e^{iT\theta} \left( \deriv{c_0}{\theta}f+ c_0 \deriv{f}{\theta}\right) \textrm{d}\theta \\
 	&=& \frac{1}{T}\left( \left[i c_0 f e^{iT\theta} \right]^{\bar{\theta}}_0 
    -i \int_0^{\bar{\theta}} \left( c_1 f f^*+ e^{iT\theta} c_0 \deriv{f}{\theta} \right) \textrm{d}\theta \right). 
\end{eqnarray}
Now $|c_0|\leq 1$ and $|c_1|\leq 1$. 
As long as the gap does not vanish $|f(\theta)|$ and $\left|\deriv{f}{\theta}\right|$ are bounded, therefore
$\left|c_1(\bar{\theta})\right| = O\left(\frac{1}{T}\right)$.
The probability of transition to the excited state for a two-level system with a nonzero gap is thus 
\begin{eqnarray}
     \left|c_1(\bar{\theta})\right|^2 = O(T^{-2}).
\end{eqnarray}


\begin{singlespace}
\bibliography{main}
\bibliographystyle{plain}
\end{singlespace}

\end{document}